\documentclass[10pt]{amsart}
\usepackage{enumerate,amsmath,amssymb,latexsym,
amsfonts, amsthm, amscd}

\usepackage{enumerate,amsmath,amsthm,amssymb,latexsym,
amsfonts,amscd}
\usepackage[all]{xy}

\title[The Jacobi curvature operator and the gravitational field]{The Jacobi curvature operator in semi-Riemannian geometry and \\ the energy momentum stress
tensor\\
 of the gravitational field}

\author[Maurice J. Dupr\'e]{M\lowercase{aurice} J. D\lowercase{upr\'e}\\D\lowercase{epartment of} M\lowercase{athematics}
\\T\lowercase{ulane} U\lowercase{niversity}
\\N\lowercase{ew} O\lowercase{rleans}, LA 70118\\\lowercase{email:  mdupre@tulane.edu}\\30  D\lowercase{ecember} 2020}
\address{DEPARTMENT OF MATHEMATICS\\TULANE UNIVERSTIY\\NEW ORLEANS, LA 70118}
\email{mdupre@tulane.edu}

\theoremstyle{plain}

\newtheorem{proposition}{Proposition}[section]
\newtheorem{theorem}{Theorem}[section]
\newtheorem{corollary}{Corollary}[section]
\newtheorem{postulate}{Postulate}[section]

\theoremstyle{definition}
\newtheorem{definition}{Definition}[section]
\newtheorem{example}{Example}[section]

\numberwithin{equation}{section}

\renewcommand{\a}{\alpha}

\newcommand{\gt}{\tau}
\newcommand{\x}{\chi}
\newcommand{\A}{\boldsymbol{\mathcal{A}}}
\newcommand{\B}{\mathcal B}

\newcommand{\D}{\mathcal D}

\newcommand{\J}{\mathcal J}
\newcommand{\K}{\mathcal K}

\newcommand{\M}{\mathcal M}

\newcommand{\R}{\mathcal R}
\newcommand{\T}{\mathcal T}
\newcommand{\U}{\mathcal U}

\newcommand{\W}{\mathcal W}

\newcommand{\bC}{\mathbb{C}}

\newcommand{\bN}{\mathbb{N}}

\newcommand{\bR}{\mathbb{R}}

\newcommand{\ra}{\rightarrow}

\newcommand{\lra}{\longrightarrow}

\newcommand{\del}{\partial}

\newcommand{\med}{\medbreak}

\begin{document}

\maketitle

\begin{abstract}
We give a justification for generalizing the notion of energy momentum stress tensor in general relativity which results in the Jacobi curvature tensor being the embodiment of total energy momentum stress including gravity in general relativity.

\end{abstract}

%%%FIND THE PROPER SUBJECT CLASSIFICATION AND KEY WORDS%%%%%%%%%%%

\med \textbf{Mathematics Subject Classification (2000)} : 83C05,
83C40, 83C99.

\med \textbf{Keywords} : Gravity, general relativity, Einstein
equation, energy density.

\section{INTRODUCTION}

\med

The problem of characterizing the total energy momentum stress including that of gravity in general relativity goes back to the original papers of Albert Einstein, where a pseudo-tensor is used to characterize the energy momentum stress of the gravitational field.  Since that time, numerous attempts have been made to deal with the total energy problem, in fact way too many to possibly list them all.  Suffice it to say, that all have failed to adequately solve the problem.  The basic reason as noted by P. A. M. Dirac, is that one must either give up localizability or conservation \cite{D}.  More recently, E. Curiel \cite{CURIEL1} has made this assertion very precise.   In preliminary versions of this paper \cite{DUPRE}, \cite{DUPRE2}, and the improvements in \cite{DUPRE3} and \cite{DUPRE4}, this author advocated giving up energy conservation on the grounds that the gravitational energy momentum stress tensor of the gravitational field could be motivated in such a way as to easily give the Einstein equation for gravity in general relativity.  The results showed that the Ricci tensor, $Ric,$ should be proportional to the total energy momentum stress including all matter and fields as well as that of gravity, and of course $Ric$ is not conserved.  However, the arguments in \cite{DUPRE} , \cite{DUPRE2}, \cite{DUPRE3}, and \cite{DUPRE4} can be  interpreted as giving the part of the gravitational energy momentum stress of the gravitational field itself which serves as source for gravity, so should properly be relabeled as {\it gravitational SOURCE energy momentum stress}.  This means that $(1/4 \pi G)Ric$ should be regarded as the {\it total gravitational SOURCE} energy momentum stress tensor, including all matter and fields as well as that due to gravity itself.  Thus we may need to admit other forms of energy momentum stress which do not serve as gravity sources in order to maintain conservation of energy momentum stress in some sense.  In other words, we definitely need to go "outside of the box" to deal with the energy of gravity.

Let us ask ourselves what we really wish to have in our energy momentum stress tensor for everything including gravity.  First, it should be conserved in some way.  Second, it should be a tensor field in some sense-that is, it must be independent of the choice of coordinates.  Third, it should allow the gravitational source energy momentum stress to be easily extracted mathematically, in an invariant manner independent of choices of frames or coordinates.  Fourth, it should be of the second rank in some sense as is the ordinary energy momentum stress tensor.  Fifth, as a second rank object, it should be symmetric in some sense as is the ordinary energy momentum stress tensor. Sixth, it should include the entire Riemann curvature tensor in the sense that knowledge of one gives the other, since it is generally felt that if there is curvature, then there must be some sort of energy momentum stress.

We will see that we should look to the Jacobi curvature operator, $\J,$ in order to satisfy all these wishes.  It is a tensor-valued tensor field.  That is to say, at each point of spacetime its value is a second rank tensor whose values are themselves second rank tensors instead of being merely numerical valued. Here, we can look to quantum mechanics as a precedent for the idea of promoting numbers to being operators.  But in fact, we will see that the Jacobi curvature operator is a symmetric second rank tensor field whose value at each spacetime event is in fact itself a symmetric tensor in the sense that it is a self-adjoint operator valued second rank tensor on the tangent space at the event, relative to the inner product on the tangent space given by the metric tensor at the event.  Again this is similar to the idea that observable numerical quantities become promoted to be self-adjoint operators in quantum mechanics.  So, the Jacobi curvature operator is thus mathematically equivalent to a fourth rank tensor in the ordinary sense, just like the Riemann curvature operator, which in these terms is a skew adjoint operator valued second rank antisymmetric tensor.  That is to say it is a tensor valued tensor field whose value at each point is a skew adjoint operator valued second rank tensor on the tangent space relative to the inner product on the tangent space given by the metric tensor.  As numerous authors have pointed out that the energy of gravity is somehow in the curvature, it is natural to conclude that the anti-symmetry of the Riemann curvature operator has prevented its consideration for being the energy momentum stress tensor of gravity and all matter and fields.  In fact, the Jacobi curvature operator can be easily shown to have exactly the same information as the Riemann curvature operator.  However, $\J$ obeys a type of symmetry law which is a natural extension of the notion of symmetry for numerical valued tensors, that is $\J$ is a self-adjoint operator valued second rank symmetric tensor, making it equivalent to a fourth rank tensor with several symmetries which we will make clear.  In particular, for an observer with velocity $u$ in spacetime at an event $q,$ it turns out that $\J(u,u)u=0,$ so $\J(u,u)$ is in fact a self-adjoint operator on $u^{\perp}$ which is the Euclidean subspace of the tangent space at $q$ orthogonal to $u,$ so is the tangent Euclidean space for the observer at $q.$ Moreover, the trace of the Jacobi curvature operator gives the Ricci tensor and the symmetry of the Ricci tensor is an obvious consequence of the type of symmetry obeyed by $\J.$  Thus the gravitational source energy momentum stress, $(1/4 \pi G)Ric$ is simply extracted by forming the trace of $(1/4 \pi G)\J,$ in the sense of composing with the trace.  Since the Jacobi operator, $\J,$ and the Riemann curvature operator, $\R,$ have exactly the same information, the second Bianchi identity which is really the vanishing of an exterior derivative in a generalized sense (in the sense that $\R$ is a tensor valued 2-form) can be regarded as the true conservation law for total energy momentum stress.  Precisely then, $(1/4 \pi G)\J$ will be seen to satisfy all our desires, and should be regarded as being the total energy momentum stress of everything including gravity and all matter and fields serving as sources of gravity.  Since we are extracting the source gravitational energy momentum stress by taking the trace of $(1/4 \pi G)\J,$ it follows that the gravitational source energy momentum stress at any event is the sum of the diagonal components of the values of $(1/4 \pi G)\J,$ using any frame for the tangent space to spacetime at that event.  In particular, an observer sees the Jacobi Curvature operator as being a spatial self adjoint operator which he can diagonalize.  The resulting eigenvalues are the principal sectional curvatures along his timeline, and their sum is thus the source gravitational energy for that observer.  Consequently, it makes sense that each principal curvature should be considered as an energy.  That is to say, if $A+B+C$ is an energy in complete generality, then it must be that each of the summands must be energies. This also means the other off diagonal components of the values of $(1/4 \pi G)\J$ should be considered as parts of the total energy momentum stress which do not contribute to the source of gravity.  This is because a general coordinate change mixes the diagonal and off diagonal components, so it follows that all components should be considered as energies.  Thus, even though the off diagonal parts do or might not vanish in the vacuum, they do not in and of themselves produce the appearance of gravitating matter or fields in the vacuum, which of course would be paradoxical-that is, if something is gravitating in a region, then that region is not a vacuum, by definition.  However, because of the Bianchi identities, in the presence of ordinary matter and fields, some of the energy in the sum of the diagonal components as well as the  off diagonal components of $(1/4 \pi G)\J$ can be converted to (mix in or feed into) the tracial components causing gravitational source energy momentum stress to appear in the ordinary matter and fields, and thus accounting for the non-conservation of the total source energy momentum stress of everything including gravity, $(1/4 \pi G)Ric.$  Thus, the fact that $\J$ has trace zero in the vacuum is sort of a delicate energy balance which is destroyed in the presence of matter and fields other than gravity.  It is our view here that all the components of $(1/4 \pi G)\J$ as viewed by a given observer should count as energy momentum stress of everything including the gravitational field, but it is only the tracial part of this curvature operator which can serve as a source of gravity.

We will look at these tensor-valued tensors in more detail and when compared with Newtonian gravity, we will see that the Einstein equation for gravity is a natural extension of Newton's Law for gravity.  It is merely the correction to Newton's Law required by putting in the gravitational source energy momentum stress due to gravity itself and requiring that spacetime is a Lorentz manifold.  Since it is of little trouble to generalize our considerations to spacetime dimensions other than four, will will work in general spacetime dimension, and see that the mathematics dictates that spacetime dimension four is the only dimension in which the Einstein equation gives automatic conservation of the ordinary matter and fields.  That is, we will not need to assume that the energy momentum stress tensor of all matter and fields other than gravity has zero divergence for our proof of the Einstein equation.  This means that conservation of ordinary matter and fields in spacetime dimension 4 is a consequence or corollary of our theorem proving the Einstein equation from precise mathematical postulates.  In fact, much of our mathematical development will require only a semi-Riemannian manifold, so we will begin with semi-Riemannain geometry.

In the semi classical treatment of forces and force fields, the force fields also arise as curvature operators on principal Lie group bundles (gauge theory) and their associated vector bundles through representations of the Lie group.  Much work has been done trying to treat gravity as a force by treating the group of all diffeormorphisms of spacetime as the gauge group in an approach to quantum gravity.   For a detailed and precise mathematical treatment of the ordinary forces, a good reference is the little book by David Bleecker, \cite{BLEECKER}.  For a slightly more casual approach, one can consult the book by John Baez and Javier P. Muniain, \cite{BAEZ&MUNIAIN}.  As the case of gravity entails the theory of infinite dimensional Lie groups and principal Lie group bundles with manifolds and Lie groups modeled on infinite dimensional Frech\'et vector spaces, the mathematical technicalities are problematic.  An interested reader should consult the extensive article by R. S. Hamilton in the Bulletin of the American Mathematical Society, \cite{HAMILTON}.  Of course, for manifolds modeled on Banach spaces, the technicalities are manageable (see for instance \cite{DUPREGLAZE1} and \cite{DUPREGLAZE2}), but that is not the case for the diffeomorphism group of a manifold which would require modeling on an infinite dimensional Frech\'et space.  However, from the perspective of our view here, there is a very simple reason why ordinary forces are different from gravity and why gravity is special.  This is because it is ONLY on the tangent bundle of a manifold that a Riemann curvature tensor can be represented as a Jacobi curvature operator.  Thus, our treatment of gravitational energy momentum stress is unique to gravity and cannot be done for any ordinary forces or curvature operators on vector bundles associated to a Lie group principle bundle unless that associated bundle is in fact the tangent bundle with Lie group then being the general linear group or a Lie subgroup of the general linear group of all linear automorphisms of the tangent space at a point.

We have included possibly more detail than necessary for an expert mathematical physicist in the hope of reaching a wider audience of physicists and mathematicians.  In fact much of the mathematics applies in any semi-Riemannian manifold, or even more generally, in any semi-Riemannian vector bundle.  For convenience, we have also included much of the material in \cite{DUPRE}, \cite{DUPRE2}, \cite{DUPRE3}, and \cite{DUPRE4}, but modified to accomodate our view of the Jacobi curvature operator as the total energy operator for general relativity, as it is a major part of our consideration here.  Our view here in no way contradicts Einstein's Theory of General Relativity.  In fact, our developments again give a straight forward proof of the Einstein equation as a mathematical theorem in complete generality.  The only justification for the Einstein equation of comparable generality is the Hilbert argument using a calculus of variations argument involving the integral of scalar curvature.  The fundamental problem with such an argument is that there is no justification for the choice of scalar curvature as the integrand except that after the fact it gives the correct equation and as it seems to be the only choice of scalar function which presents itself.  For an equation as important as the Einstein equation for gravity, we need a better foundation than a variational argument based on such a choice.  Our physical arguments for the choice of the gravitational field's gravitational source energy momentum stress tensor will be simple and intuitive, so should be clear to any physicist.  However, like the Lagrangian arguments, we will make a reasonable assumption that gravity acts with minimal effort, in a sense which will be made intuitively clear.  It is possible that the argument could be made to fall into the usual Lagrangian framework, but we will not pursue that at this time. The rest will follow by pure mathematics, some of which may be unfamiliar to some physicists.  We have thus included the mathematical background to make our developments clear to such readers.  As well, in the interest of clarity for mathematicians, we have included possibly more detail in our physical arguments and motivations for our postulates than an experienced physicist would need.

Our view here justifies the Cooperstock hypothesis, \cite{COOPERSTOCK}, \cite{COOPERSTOCK2}, \cite{COOPERSTOCK3}, in the sense that the arguments show that the total gravitational source energy momentum stress tensor, $(1/4 \pi G)Ric,$ must vanish in the vacuum.  The other parts of the energy momentum stress contained in $(1/4 \pi G)\J$ other than the trace part must be the carrier of any energy momentum stress in the vacuum. Thus a gravity wave cannot be detected in the vacuum by gravitational effects-there can be no gravity wave energy which will gravitate in the vacuum. Of course, as soon as a detector is placed in the vacuum, you no longer have a vacuum.  As a consequence, a gravitational wave detector must be in the path of the gravitational wave to detect it, if a gravitational disturbance passes by the detector at a distance from the detector, there will be no disturbance detected, since there is no change in the gravitational source energy in the vacuum caused by the wave.  That is, if you could create a gravitational wave "laser beam", it would be undetected until it hit something, a near miss would not set off a detector. 

Having an invariant expression for the gravitational source energy momentum stress tensor as well gives an approach to the quasi-local mass problem in general relativity which we shall explore briefly.  In particular, as indicated in \cite{C&D}, \cite{C&DANNALS}, \cite{C&DFOOP},  we shall define an invariant integral as a symmetric multilinear mapping on pairs of vector fields which we will call the {\bf total gravitational source spacetime energy momentum stress of a spacetime region}.  In many examples, it is the ordinary Komar mass or Tolman integral multiplied by elapsed time.  This makes it clear that a major problem with defining quasi-local mass of a spatial slice is really the fact that clock rates vary from place to place on the slice in non-trivial gravitational fields, even in stationary examples.   From the Hamiltonian point of view, the energy is the conjugate variable to time, so if clock rates are varying in a given three dimensional slice (submanifold of spacetime) from point to point, then there will be no easy way to get an energy for a 3 dimensional region of this slice, and using only information on the regions boundary would be problematic as well.  It is this author's view that this is the real source of all the problems with developing concepts of quasi local mass or momentum and energy.  The definition of spacetime energy momentum stress \cite{C&DANNALS} makes this  clock rate problem immaterial.  Henceforth, I will also call this integral the {\bf Cooperstock Energy Momentum} in honor of my late co-author Fred Cooperstock (\cite{C&D}, \cite{C&DANNALS}, \cite{C&DFOOP}) who passed away not long ago.

%%%%%%%%%%%%%%%%%%%%%%%%%%%%%%%%%%%%%%%%%%%%%%%%%%%%%%%%%%%%%%%%%%%%%%%%%%%%%%%%%%%%%%%%%%%%%%%%%%%%%%%%%%%%%%%%%%%%%%%%%%
%%%%%%%%%%%%%%%%%%%%%%%%%%%%%%%%%%%%%%%%%%%%%%%%%%%%%%%%%%%%%%%%%%%%%%%%%%%%%%%%%%%%%%%%%%%%%%%%%%%%%%%%%%%%%%%%%%%%%%%%%
%%%%%%%%%%%%%%%%%%%%%%%%%%%%%%%%%%%%%%%%%%%%%%%%%%%%%%%%%%%%%%%%%%%%%%%%%%%%%%%%%%%%%%%%%%%%%%%%%%%%%%%%%%%%%%%%%%%%%%%%%%

\med

\section{THE MATHEMATICAL DESCRIPTION OF GRAVITY}

Because everything moves the same way under the influence of gravity, as Einstein realized, gravity must be a property of spacetime itself.  But the principle of relativity guarantees that there can be no action at a distance, so gravity must work locally.  To be emphatic, at each event, the state of the gravitational field has an existence at that event which determines how gravity acts at that event, there is no action at a distance.  As a consequence all the physics is either local or in boundary conditions.  In Newtonian gravity, the original action at a distance together with the law of motion makes gravity completely described by an acceleration field throughout spacetime which in the Newtonian case, is simply the cartesian product of the time line with Euclidean space together with a time varying vector field on that Euclidean space.  The spatial total (or Frech\'et) derivative of that vector field then gives the relative acceleration of nearby objects and consequently as in \cite{DUPRE4}, the tidal acceleration operator is then simply that derivative of the acceleration field, and one sees that mathematically the object is to find the tidal acceleration operator.  That is, as gravity must be a local theory due to the principle of relativity, it follows that gravity should locally simply be described by the tidal acceleration operator.  In a general spacetime, there is no preferred time coordinate, and each observer at a given event has his own tidal acceleration operator given by the equation of geodesic deviation.  This results that for a unit timelike vector at a specific event, the tidal acceleration operator is the Jacobi curvature operator $\J(u,u)=-\nabla_u^2.$  That is, the Jacobi curvature operator as a tensor is a "tensorization" of the differential operator $-\nabla_u^2$ which ostensibly only makes sense for a vector field $u.$  Thus, the equation of geodesic deviation makes precise in what sense that differential operator has been tensorized.  Now, in Newton's law of gravity formulated as the Poisson equation, we see that the divergence of the acceleration field can be viewed as the trace of the spatial Frech\'et derivative, that spatial derivative being the tidal acceleration operator, so the natural extension of Newton's theory to spacetime is to require that the {\bf trace of $\J(u,u)$} should be $4 \pi G \cdot E(u),$ where $E(u)$ is the total energy of everything that can be a source of gravity.  However, the trace of $\J(u,u)$ is in fact $Ric(u,u)$, where $Ric$ is the Ricci curvature tensor.  Moreover, in relativity, all the energy of matter and fields is described by the Energy Momentum Stress (EMS) tensor and so it is natural to assume that what ever is the source of gravity should be described by a symmetric tensor $S$ which has the property that for any observer with velocity $u,$ the energy he sees as the source of gravity should be $S(u,u).$  Thus, the equation of gravity should be simply

$$Ric=(4 \pi G) S$$
since that is equivalent to requiring that for each observer at each event we have $Ric(u,u)=(4 \pi G) S(u,u).$

This equation should thus be regarded as the Einstein equation of gravity.  In Einstein's time one could not imagine forms of energy we now know about, so his equation as originally given was always subject to revision as new forms of matter were discovered.  Thus, in this author's opinion, any theory of gravity described by this equation such as various so called modified gravity theories (mogs) are simply Einstein's theory with extended notions of source.  Now, if we denote by $T,$ the EMS tensor of all ordinary matter and fields other than gravity, then certainly $S$ must contain $T$ as a summand, that is we should simply regard $V=S-T$ as the EMS tensor of the gravitational field itself.  Thus, the obvious and controversial question immediately appears:  what is $V$?  Now, the conventional view that in free fall you do not feel gravity is erroneous, because free fall does not prevent tidal acceleration and that is the true essence of gravity.  As is well known, the tidal acceleration of gravity is felt by any object, no matter how small, in free fall as (internal) pressures.  This means that the true indicator of gravity's energy is the pressure.  At each event, each observer can spatially diagonalize the EMS tensor of whatever matter he sees, and thus, the sum of the spatial eigenvalues should thus be a reflection of the gravitational field energy being applied to cause the pressures.  That is, the work gravity is doing to try and enforce the geodesic law when matter is trying to thwart it is simply a reflection of the pressures felt due to tidal acceleration.  It is then easy to see that this energy in total being applied by gravity in its attempt to maintain the geodesic law is the sum of these eigenvalues which we call principal pressures.  This leads to the hypothesis that each observer should see the energy density of gravity as the sum of principal pressures, that is the sum of the spatial eigenvalues of $T$ as it acts on the orthogonal complement of his velocity vector in the tangent space at each event.  This immediately gives the result that

$$V=T-c(T)g,$$
where $c(T)$ denotes the contraction of $T.$  We will look at this in more detail in a moment.  But for now, we realize immediately that as 

$$S=T+V=2T-c(T)g,$$
the general Einstein equation is

$$Ric=(4 \pi G)[2T-c(T)g]=(8 \pi G)[T-(1/2)c(T)g],$$
which after trace reversal becomes

$$E=(8 \pi G)T,$$
on applying the involutory trace reversal to both sides of the equation, where $E=Ric-(1/2)c(Ric)g,$ the trace reversal of Ric, which is of course the Einstein tensor.  This last equation because of its brevity of course is the common form for the Einstein equation of gravity, and we see that it results from merely putting Newton's laws of gravity as the Poisson equation into a form where we recognize that it is about the trace of the tidal acceleration operator exactly, so that the only natural extension to spacetime is from the equation of geodesic deviation and from the hypothesis that each observer should see the energy of gravity as the sum of principal pressures.

So now, let us look in a little more detail at this hypothesis on the sum of principal pressures being the energy of gravity.  We should keep in mind that Einstein himself always maintained that his equation contained the effect of the energy of gravity so that it was always there in front of us.  The problem was the old Newtonian concept of the potential energy of a force field.  This leads to the idea that it must be the metric tensor which is the reflection of the potential energy and the connection coefficients are the resulting force, and as these can be made to vanish at an event merely by proper choice of coordinates, it seems the field can be  made to vanish at any point and therefore the erroneous conclusion that there can be no resulting energy viewed locally.  There are two problems with this argument.  First, it is the tidal acceleration and therefore the curvature which is the real essence of gravity, and the curvature cannot be made to vanish by simply choosing coordinates cleverly.  Second is the nature of the Newtonian potential function itself.  It is merely the result of the fact that the tidal acceleration operator is symmetric.  In Newtonian terms, the curl of the acceleration field of gravity vanishes, and thus, as the domain at each instant is simply connected, the acceleration field has a potential function which leads to a useful conservation law for Newtonian gravity in Euclidean space.  But this is a manifestation of the special nature of the Newtonian spacetime and cannot be taken to be general.  The existence of the potential function is really only a manifestation of the symmetry of the tidal acceleration operator together with the assumption that space is Euclidean, or at least simply connected, in other words, an unjustifiable assumption of a global topological nature.  Again, because of the upper limit on signal velocity in relativity, a priori global assumptions should be avoided to keep the theory local in nature.

Thus to get at the energy of gravity, we need to get out of the Newtonian potential "box".  We also need to get beyond the "box" imposed by Lagrangian mechanics.  In effect, we need to get out of the "box" imposed on us by \cite{CURIEL1}, as it imposes too many requirements.  Something has to be relaxed. What we should take from Lagrangian and Hamiltonian mechanics is that it is reasonable that however gravity is acting, it is doing so with minimal effort which means minimal energy.  The other thing about gravity that we tend to put too much emphasis on is the gross motion of things.  When you drop a book and it lands on your foot, the pain you feel is due to the energy gravity is applying in order to keep that book on its geodesic path at the event it hits your foot.  Likewise, the pressure you feel on your feet when standing is a reflection of the energy being exerted by gravity in its attempt to keep you falling through whatever floor is supporting you, that is, the state of gravity at each event of the location of your feet.  We must here for the moment forget about the usual idea that it is the enormous mass of the Earth which is causing gravity to attempt to enforce that motion but simply accept that whatever that motion is that gravity is attempting to enforce, it is doing it right at the location of our feet at a particular event, and as there is no action at a distance in relativity, the energy gravity is applying there must be located right there at our feet at the particular event.  Now, to get an idea of how much energy gravity is exerting at a specific event to maintain an object's geodesic motion, the reasonable approach is to see what energy would be required to counteract gravity.  For instance, the walls holding up the roof of a house are themselves counteracting gravity, so the mass energy of those walls is a reflection of the total energy gravity is exerting to try and make the roof fall toward the center of the Earth.  If we assume that gravity works with "minimal" energy, we could say the energy gravity is exerting to make an object fall is the same as the minimal energy of a support system which would counteract gravity.  Now, any such support system is made of matter, but the particles making up the matter are not really doing anything in and of themselves, it is the forces and pressures in the matter which keep it from collapsing.  This naturally leads to the idea that a system of minimal energy to counteract gravity would be massless and consist of pure energy.  That would be light which of course already has a fundamental position in the theory of relativity.  Thus, as a simple thought experiment, we could ask if we have an event inside the matter with EMS tensor $T$ where pressures are operating because the matter is not in geodesic flow, and an observer at that event with velocity $u,$ then at that event we can diagonalize the spatial part of $T$ and therefore obtain mutually perpendicular eigen-directions with eigenvalues the principal pressures observed by our observer.  Suppose he takes a small bit of this spacetime about his event and "scoops out" the matter and replaces it with mirror walls perpendicular to the eigen-directions forming a little box.  In each direction he can introduce a laser beam with enough laser pressure to exactly balance the principal pressures, so that the little laser box in effect is holding up the matter outside its walls.  Such a system would seem to have the minimal energy required to do this job, and so the total energy density of the laser beams should be the energy density gravity is exerting to create the pressures in the matter.  Now it is an elementary exercise in freshman physics to see that the energy density of each laser beam is the pressure in the beam that is to say the total energy density is the sum of the principal pressures (see also \cite{DUPRE}, \cite{DUPRE2}).  We include the details of this argument later in our development for convenience and completeness.  Now, if some of the principal pressures are negative, one could imagine having a pair of oppositely charged capacitor plates instead of a pair of mirrored walls in that direction, but again, the energy density of the electric field between a pair of capacitor plates is also again just the pressure with again applying simple freshman physics arguments, \cite{DUPRE2}.  In the case of negative pressure, obviously the material system here is working "with" gravity in the sense that the gravitational energy should also be considered as negative.  Thus it appears that a reasonable simple hypothesis on the energy of gravity is that each observer should see the energy density of the gravitational field at each event is simply the sum of principal pressures in his matter and fields other than gravity which leads immediately to the conclusion that $V=T-c(T)g.$  We state this formally as a hypothesis.

\medskip

{\bf GRAVITATIONAL SOURCE ENERGY HYPOTHESIS (GSEH).}  For each observer and at each event, the energy density of the gravitational field acting as source of gravity is the sum of principal pressures observed by the observer at that event.  Equivalently, again, with $c(T)$ denoting the contraction of $T$ so as to avoid indices and coordinates or frames,

\begin{equation}\label{GSEH}
V=T-c(T)g.
\end{equation}

Thus, to summarize, the GSEH and the equation of geodesic deviation give the Einstein equation in full generality, and now as a consequence, in particular, we have $div T=0,$ in complete generality since $div E=0,$ in complete generality, as long as spacetime has dimension 4.  If spacetime has dimension other than 4, since we still have div E=0, then assuming that div T=0 leads to the conclusion that scalar curvature is constant, which would seem to be much too restrictive to be realistic.

Since tidal acceleration is described by curvature, the true coin of the realm for gravity is curvature.  Uniquely, in the case of the tangent bundle, the curvature coin has two sides which are dual to each other, the Riemann curvature and the Jacobi curvature.  But it is the Jacobi curvature which has all the symmetry properties we expect of an EMS tensor, except that it's values are self adjoint operators on the tangent space instead of merely numbers.  It satisfies the first Bianchi identity, it is symmetric where the Riemann curvature is anti-symmetric or alternating, and the Jacobi curvature operator also satisfies the {\bf exchange identity} just as the Riemann curvature operator:

$$g(\J(w,x)y,z)=g(\J(y,z)w,x)  \mbox{ just as } g(\R(w,x)y,z)=g(\R(y,z)w,x).$$
Now the equation of geodesic deviation really applies in any semi-Riemannian manifold, and gives $\J(v,v)$ as the tensorization of $-\nabla_v^2$ for any tangent vector $v$ at any event, so it is then reasonable to apply this to get a tensorization of the d'Alembertian operator due to the connection.  That is, if $(v_{\alpha})$ is any local frame field with dual frame field $(v^{\alpha}),$ then we define the operator valued tensor field $\square_J$ by

\begin{equation}\label{dalembertian00}
{\bf \square}_J =-\sum_{\alpha, \beta} g(v^{\alpha}, v^{\beta}) \J (v_{\alpha} , v_{\beta}).
\end{equation}
We see that as the values of $\J$ at each point are self adjoint linear operator-valued second rank tensors on the tangent space, that it must be the case that $\square_J$ is itself a self-adjoint operator valued field of rank zero, that is, simply a field of self adjoint operators on the tangent bundle.
We can form a symmetric tensor field $g\square_J$ defined by

$$g\square_J (v,w)=g(\square_J w,v).$$

Now, the ordinary d'Alembertian due to the connection is the differential operator

$$\square=\sum_{\alpha,\beta}g(v^{\alpha}, v^{\beta}) \nabla_{v_{\alpha}} \nabla_{v_{\beta}}.$$
At the origin of normal coordinates we then find that (\cite{MTW}, page 286, equation (11.32))

\begin{equation}\label{g''(0)}
g_{\alpha \beta, \mu, \nu}=-\frac{2}{3} \J_{\alpha \beta \mu \nu}=-\frac{2}{3}g(\J(v_{\alpha},v_{\beta})v_{\mu}, v_{\nu}),
\end{equation}
and therefore

\begin{equation}\label{waveg}
\square g =\frac{2}{3} g\square_J, \mbox{ at the origin of normal coordinates.}
\end{equation}
On the other hand, the Einstein equation expressed directly in terms of the Jacobi curvature operator says for each timelike unit vector we have, using the summation convention, with $g^{\alpha \beta}=g(v^{\alpha}, v^{\beta}),$ and the exchange property above,

$$(4 \pi G) S(u,u)=Ric(u,u)=trace (\J(u,u))=g^{\alpha \beta}g(\J(u,u)v_{\alpha}, v_{\beta}))=g^{\alpha \beta}g(\J(v_{\alpha}, v_{\beta})u,u)= - g\square_J (u,u).$$
As this must hold for every timelike unit vector, we conclude that the Einstein equation is equivalent to

\begin{equation}\label{waveEinstein}
g\square_J= - (4 \pi G) S,
\end{equation}
which clearly now shows the wavelike nature of the Einstein equation for gravity.  Thus in the vacuum we have $\square g =0$ at the origin of every normal coordinate system.

At this point, we realize that the components of $\J/(4 \pi G)$ have units of energy, it is natural to think of $\J/(4 \pi G)$ as the real energy of everything including gravity and all matter and fields, so that the trace of $\J/(4 \pi G)=Ric/(4 \pi G)$ is merely the part acting as source of gravity.  In particular, in the vacuum, it must be that all the gravitational energy is in a balance so that the trace of $\J/(4 \pi G)$ is zero.  Now another well known general consequence of the differential geometry of the Jacobi curvature operator, which we will see in the sequel following, is that at a specific event $q,$ if $W$ is the tangent space $W=T_qM$  to the spacetime manifold $M,$ then defining $J=\J|_q: W \times W \lra L(W;W),$ and $b=g|_q:W \times W \lra \bR,$ and setting

$$G(x)=1_W-\frac{1}{3}J(x,x), \mbox{ for each } x \in W,$$
then there is $U_q$ open in $W$ containing the zero tangent vector in $W$ so that $G(x)$ is invertible for each $x \in U_q$ and therefore on $U_q$ viewed as an open submanifold of the vector space $W$, we can define a nondegenerate semi Riemannian metric $g_U$ by setting

$$g_U(v,w)=b(G(x)v,w),$$
for any vectors in $v,w$ in $W$ viewed as tangent vectors to $U_q$ at the point $x \in U_q.$  And we find that the Jacobi curvature operator of the semi Riemannian manifold $U_q,$ which we can denote by $\J_U$ has

$$\J_U|0=J=\J|q.$$
Thus, $U_q$ is osculating $M$ at the event $q \in M.$  That is, at $q$ which can be identified as the zero vector in $W=T_qM,$ the two manifolds have the same tangent space and the same curvature operators at the point of origin $q.$  This means that $\J|q$ is determining the dynamic character of the metric tensor $g$ at the event $q.$  Of course, this is also clear as a result of (\ref{g''(0)}) and (\ref{waveg}).  As a mathematical aside, if $A$ is any self-adjoint operator on a finite dimensional Euclidean space, $E$, then by the spectral theorem we can find an orthonormal basis $e_1,e_2,...,e_n$ for $E$ consisting of eigenvectors of $A,$ so there are real numbers $r_1,r_2,r_3,...,r_n$ with $Ae_k=r_ke_k$ for each $k \leq n.$  Letting $S_E$ be the set of unit length vectors in $E,$ we see that the image of $S_E$ under $A$, denoted $A(S_E)$ is an "ellipsoid" that is, if all eigenvalues are positive, then the equation of $A(S_E)$ is

$$(\frac{x_1}{r_1})^2+(\frac{x_2}{r_2})+...+(\frac{x_k}{r_k})^2=1.$$
Thus, the eigenvalues give the magnification factors along the various orthogonal directions by which $A$ distorts the unit sphere $S_E.$ The eigenvalues and eigen diirections are in fact easily found by the method of Lagrange multipliers, the eigenvalues are the Lagrange multipliers, working by induction "peeling off" the eigenspace of the highest eigenvalue, and thus this ellipsoid actually characterizes the self adjoint operator $A.$  Consider an observer with unit tangent vector $u$ along his path.  At each event $q,$ he can "view" the ellipsoid of $1+\J(u,u),$ so as curvatures are generally small, this purely spatial self adjoint operator has an ellipsoid at the event which characterizes it and in effect then characterizes $\J(u,u),$ viewed as a self adjoint operator on the Euclidean space orthogonal to $u.$  Thus he can view the dynamics of the spacetime along his path as a wobbling spatial ellipsoid, resulting from slightly distorting a sphere, the wobble due to the dynamic nature of the spacetime he is passing through.  For instance, if a disturbance passes him at the speed of light, then distances in the direction of motion of such a disturbance in the observer frame would be "squashed" to zero, so the ellipsoid would be an ellipsoid in the orthogonal complement of the direction of motion of the disturbance and would thus reduce to an ellipsoid in one less spatial dimension than the orthogonal complement of $u.$  Thus in spacetime dimension 4, this means the disturbance would be viewed as a wobbling ellipse in the spatial 2 dimensional plane orthogonal to the direction of motion of the disturbance.  We see that this mathematically leads quickly to the usual view of a gravity wave.

%%%%%%%%%%%%%%%%%%%%%%%%%%%%%%%%%%%%%%%%%%%%%%%%%%%%%%%%%%%%
%%%%%%%%%%%%%%%%%%%%%%%%%%%%%%%%%%%%%%%%%%%%%%%%%%%%%%%%%%%%
%%%%%%%%%%%%%%%%%%%%%%%%%%%%%%%%%%%%%%%%%%%%%%%%%%%%%%%%%%%%

%%%%%%%%%%%%%%%%%%%%%%%%%%%%%%%%%%%%%%%%%%%%%%%%%%%%%%%%%%%%
%%%%%%%%%%%%%%%%%%%%%%%%%%%%%%%%%%%%%%%%%%%%%%%%%%%%%%%%%%%%
%%%%%%%%%%%%%%%%%%%%%%%%%%%%%%%%%%%%%%%%%%%%%%%%%%%%%%%%%%%%

\section{MATHEMATICAL PRELIMINARIES}

\med

In this section we will establish our basic mathematical notation, terminology, and framework.  For more details on the mathematics, we refer the interested reader to Appendix I.  Our mathematical presentation is expository, and we make no claim to originality in the mathematical results, however the exposition itself is new in  many respects.

We generally assume that $M$ is a smooth manifold.  We say that $M$ is a semi-Riemannian manifold to mean $M$ as a smooth manifold with a given semi-Riemannian metric by which we mean a given non-degenerate symmetric real valued smooth bilinear tensor field, which we will usually denote by $g.$  When we speak of
spacetime we mean an $(n+1)$-dimensional semi-Riemannian manifold $M$ whose metric tensor $g,$ has signature $(-,+,+,+,...,+),$ or signature $(1,n),$ in which case $g$ is called a Lorentz metric and $M$ is also called a Lorentz manifold.   More generally then, we say that $g$ or $M$ has signature $(p,q)$ to mean that at each point of $M$ the metric $g$ is negative definite on a subspace of the tangent space of dimension $p$ and positive definite on the complement of dimension $q$ so that $p+q=dim(M).$  We use $\nabla$ for the resulting unique Levi-Civita Koszul connection
or covariant differentiation operator on $M$ which is torsion free and satisfies $\nabla g=0.$  We use $TM$ for the
tangent bundle of $M$ and $T_mM$ to denote the tangent space of
$M$ at $m \in M.$ If $f: M \lra N$ is a differentiable map of
manifolds, say of class $C^r,$ then $Tf:TM \lra TN$ is the tangent
map which is of class $C^{r-1}$ and we note here the simple
property $T(hf)=(Th)(Tf)$ as regards composition of differentiable
mappings. In particular, $Tf$ is a vector bundle map covering $f.$ In case $f(m)=n \in N,$ then $T_mf:T_mM \lra T_nN$ is a
linear map.  It is convenient in this setting to refer to $u \in
T_mM$ as a unit vector to mean merely $|g(u,u)|=1.$  We say a tangent vector $v$ is timelike if $g(v,v) < 0$ and spacelike if $g(v,v) > 0.$  Thus $u$ is a
time-like unit vector when $g(u,u)=-1.$  Of course if $M$ is a Lorentz
manifold, each of its tangent spaces is a Lorentz vector space of
dimension $n+1.$ By an {\it observer} at $m \in M$ we mean simply
a time-like unit vector in $T_mM.$ By a {\it spacetime model}, we will
always mean a triple $(M,g,T)$ where $M$ is a smooth Lorentz manifold of
dimension $n+1,$ where $g$ is a twice continuously differentiable
Lorentz metric on $M,$ and $T$ is a given second rank symmetric
covariant tensor field on $M.$

Suppose that $M$ is a smooth manifold.  By a {\bf smooth bundle} over $M$ we mean a smooth manifold $E$ with a given smooth map, called a {\bf bundle projection} in this setting, 

$$p_E:E \lra M.$$  
For each $x \in M,$ it is customary to denote $E_x=p^{-1}(x) \subset E,$ and call $E_x$ the fiber of $E$ or of $p$ over $x,$ of course here, $p=p_E.$  If $E$ is a smooth bundle over $M$ and $F$ is a smooth bundle over $N,$ then by a smooth bundle map $h: E \lra F$ covering $f: M \lra N,$ we mean that both $h$ and $f$ are smooth maps and that $p_F \circ h= f \circ p_E.$  This last condition of course guarantees that for each $x \in M,$ the map $h$ carries $E_x$ into $F_{f(x)}$ defining the map $h_x: E_x \lra F_{f(x)}.$  In case that $M=N,$ we say that $h$ is over $M$ provided that $f=id_M,$ the identity map of $M.$  Thus, the notion of isomorphism of bundles is clear.  We say $E$ and $F$ are isomorphic over $M$ if both are over $M$ and there is an isomorphism over $M$ of $E$ onto $F.$  If $U$ is a subset of $M,$ then $p_E^{-1}(U)=E|U \subset E$ is itself a bundle over $U$ which is smooth if $E|U$ and $U$ are submanifolds of $E$ and $M,$ respectively, which is clearly the case if $U$ is an open subset of $M.$  More generally, if $F$ is a submanifold of $E$ and $N$ is a submanifold of $M,$ with $p_E(F) \subset N$ and $p_F=p_E|F$ is the {\bf restriction} of $p_E$ to $F,$ then we call $F$ a subbundle of $E$ over $N,$ which is then in fact a subbundle of $E|N,$ if $E|N$ is a submanfold of $E.$  If $F$ is a given manifold, we call $E=M \times F$ the product bundle with fiber $F,$ in which case we mean $p_E(x,v)=x$ always.  If $E$ is merely isomorphic over $M$ to a {\bf product bundle} with fiber $F,$ then we say that $E$ is a {\bf trivial bundle} with fiber $F.$  We say that $E$ is locally trivial (with fiber $F$) if each $x \in M$ has on open neighborhood $U_x$ such that $E|U_x$ is trivial (with fiber $F$).  Any bundle isomorphism with a product bundle is called a {\bf trivialization}.  Clearly, if $M$ is connected, and if $E$ is locally trivial, then $E$ is locally trivial with fiber $F$ for some fixed $F.$  

Notice the simplest example of a bundle is $id_M :M \lra M,$ the identity map of $M,$ so for each $x \in M,$ the fiber $M_x$ is just the set consisting of the single point $x.$  Clearly, we can say this bundle is trivial with a one point fiber.  Any bundle map $s:M \lra E$ over $M$ is called a {\bf section} of $M.$  Thus, if $s$ is a section of $E,$ then $s(x)$ belongs to $E_x$ for each $x \in M.$

Suppose that $G$ is a group and $S$ is a set.  We call a map $\alpha:G \times S \lra S$ an {\bf action} of $G$ on $S$ provided (using the notation $gs=\alpha(g,s),$ for $g \in G$ and $s \in S$) that
$g(hs)=(gh)s$ and $es=s$ for all $s \in S,$ and all $g,h \in G,$ where $e$ here denotes the identity of the group $G.$  For $H \subset G$ and $T \subset S$ we denote by $HT$ the set of all $gs$ for $g \in H$ and $s \in T.$  In particular, $Gs$ is the {\bf orbit} or more specifically, the $G$-orbit of $s \in S$ for each $s \in S.$  Notice that if two orbits meet, then they are identical, so the set of all orbits, which is denoted $S/G$ is called the {\bf orbit set} of the action and thus is a partition of $S.$ We then have a natural {\bf orbit map} of $S$ onto $S/G$ carrying $s$ to $Gs.$ If $G$ is a Lie group, if $S$ is a manifold, and if $\alpha$ is smooth, we say $G$ acts smoothly on $S.$  In this case, the coinduced topology by the orbit map (which is the largest topology making the orbit map continuous) gives the orbit set a topology so we call the orbit set with this topology the {\bf orbit space}.  In many cases, the orbit space is itself a smooth  manifold so that the orbit map is smooth even in infinite dimensions (\cite{DUPREGLAZE1}, \cite{DUPREGLAZE2}, \cite{DUPGLAZEPREV}).  Let $Aut(S)$ denote the group of all bijections of $S$ onto itself.  Thus, we have a group homomorphism $h:G \lra Aut(S)$ for which $h(g)(s)=gs,$ for each $s \in S$ and each $g \in G.$  We say the action is {\bf effective} if $gs=s$ for all $s \in S$ implies that $g=e,$ the identity of $G.$  Thus, for the action to be effective means that $h$ is an injective homomorphism.  In this case, we see that $G$ can be thought of as a group of bijective self maps of $S.$  We say the action is {\bf free} if in fact $gs=s$ for some $s \in S$ implies that $g=e.$  in this case, if $s \in S$ and $gs=hs,$ then $g=h,$ since then $g^{-1}h=e.$  In effect, we can cancel on the right.  If $S$ is a set with a $G$-action, we often simply call $S$ a $G$-set or $G$-space if the action is continuous, or a $G$-manifold, if $S$ is a manifold, if $G$ is a Lie group, and if the action is smooth.  We say that a map $f:S \lra T$ is a $G$-map or {\bf equivariant map} of $G$-spaces, if $f(gs)=gf(s),$ for all $g \in G$ and for all $s \in S.$   We say the action is trivial if every member of $G$ acts like the identity on $S.$  Notice that the orbit map is equivariant where the orbit set is given the trivial action.  If $f: S \lra T$ is a $G$-map, then it induces a map of $S/G$ into $T/G$ which is covered by $f.$

Notice that if $G$ acts on $S$ and if $p:S \lra X$ is any map which is surjective and with the property that the fibers of the map $p$ are the orbits of the action, then we can naturally identify the orbit set with $X,$ that is, there is a unique bijective map of the orbit set onto $X$ covered by the identity map on $S.$  Any topological group action is automatically an open map, since if $U$ is an open subset of $S,$ then $gU$ is open for each $g$ in $G$ and so if $K$ is any subset of $G,$ then $KU$ is the union of the open sets $gU$ for $g \in K.$  If $p_G$ denotes the orbit map, then by definition of the coinduced topology, to see that $p_G(U)$ is open in $S/G,$ we need only observe that $p_G^{-1}(p_G(U))=GU$ which is open in $S.$  This means that if $p:S \lra X$ is a continuous open surjection whose fibers are the orbits of the action, then the natural bijection of the orbit set onto X covered by the identity of $S$ is a homeomorphism.  In the case of manifolds, this map is often a diffeomorphism of manifolds.

For any sets $R_1,R_2,S_1,S_2$ and any pair of maps $f_k: R_k \lra S_k,$ we define the map 

$$f_1 \times f_2 : R_1 \times R_2 \lra S_1 \times S_2$$
by

$$[f_1 \times f_2](r_1,r_2)=(f_1(r_1),f_2(r_2)).$$
For any set $X$ we can define the diagonal $D_X \subset X \times X$ as the set of all pairs $(x,x)$ for $x \in X.$  Likewise, we define the diagonal map $d_X : X \lra X \times X,$ by $d_X(x)=(x,x).$  Obviously $d_X$ carries $X$ bijectively onto the diagonal $D_X.$

If $f:M \lra N$ is a smooth map of manifolds and $E$ is a smooth bundle over $N,$ then we define the pullback, denoted $f^*E,$ of $E$ by $f$ to be the subset

$$f^*E=[f \times p_E]^{-1}(D_N) \subset M \times E.$$
Of course, this means that $f^*E$ is a subbundle of the product bundle over $M$ with fiber $E,$ and in case that $E$ is locally trivial with fiber $F,$ it is an easy exercise to show that $f^*E$ is a smooth bundle which is also locally trivial with fiber $F.$

If E and F are smooth bundles over $M$ and $N$ respectively, then $E \times F$ is a smooth bundle over $M \times N$ in the obvious way, and if both are locally trivial so is $E \times F.$  If $M=N,$ then letting $d_M$ denote the diagonal map of $M$ into $M \times M$ given by $d_M(x)=(x,x),$ we define the Whitney sum or fiber product, denoted $E \oplus F,$ by

$$E \oplus F=d_M^*[E \times F].$$
Thus, as a set, $E \oplus F$ is the union of the fiber products $E_x \times F_x$ for all $x \in M.$  As restrictions of locally trivial smooth bundles are again locally trivial, it follows that the Whitney sum of locally trivial bundles is again locally trivial. 

Suppose $p:E \lra M$ is a smooth orbit map of the Lie group action of $G$ on $M,$ which is free.  Then we have a unique {\bf transition map} denoted $\tau_E: E \oplus E \lra G$ defined by requiring that $\tau_E(s,t)=g $ is the unique $g \in G$ with $gs=t.$   We say that $p$ or $E$ is a {\bf principal $G$-bundle} provided that $\tau_E$ is smooth.  Notice the fiber of $E \oplus E$ over $x \in M$ is the product $Gs \times Gs$ where $s$ is any point of $E_x.$  If $s$ is a section of $E$ over $A \subset M,$ then the map sending $ t \in p^{-1}(A)$ to $(p(t), \tau(t,s(p(t))) \in A \times G$ is equivariant and gives a trivialization of $E|A.$  Thus, a principal $G$-bundle is continuously or smoothly locally trivial, respectively, provided that it has local continuous (respectively, smooth) sections about each point of the base manifold  $M.$  

If $E$ is a smooth principal $G$-bundle over $M,$ if $F$ is a smooth $G$-manifold and $f: E \lra F$ is an equivariant map, then since $E \times F$ is also is a $G$-space, then its orbit set is denoted by $E[F].$  The first factor projection map $E \times F \lra E$ is then equivariant, so induces a smooth map $p_F:E[F] \lra M,$ called the {\bf associated fiber bundle with fiber} $F$ to $E$ over $M.$  If the action is effective, then this associated bundle is actually a fiber bundle with fiber $F.$

 If $v:E \lra F$ is any $G$-equivariant map, then the map
$h:E \lra E \times F$ given by $h=(id_E,v)$ is $G$-equivariant so passing to orbit spaces gives a section of the associated bundle with fiber $F.$  Thus, the sections of the associated bundle correspond naturally to the equivariant maps of the principal bundle into the fiber.  Moreover, any local section of the principal bundle gives a local trivialization of the principal bundle and any associated bundle over the domain of the section.   If $S$ is say a semigroup acting on $F,$ with an action commuting with the $G$-action, then the action of $S$ on $F$ will give an action of $S$ on $E[F],$ and in particular, if $S=F$ and the action is simply left multiplication, then we see that $E[F]$ becomes a semigroup bundle.  Thus, in particular, if $F$ is a vector space and the action of $G$ on $F$ is linear, that is, given by a representation of $G$ on $F,$ then $E[F]$ becomes a vector bundle.

We say $E$ is a smooth vector bundle if for each $x \in M,$ the fiber $E_x=p^{-1}(x)$ of $E$ over $x$ is in fact a vector space and the operations of vector addition and multiplication by scalars are smooth maps on $E.$  The field of scalars can be taken to be either $\bR$ or $\bC.$  More precisely, we require that addition be smooth on the subset of $E \times E$ on which the addition is defined, which of course is the union of all sets $E_x \times E_x$ for $x \in M.$  On the other hand, we say that $h$ is a smooth vector bundle map covering $f$ provided that $E$ and $F$ are both smooth vector bundles and $h$ is linear on each fiber, meaning of course that for each $x \in M,$ the map $h_x$ is linear.    Of course, if $F$ is a vector space (finite dimensional), then the product bundle with fiber $F$ is a smooth vector bundle, so likewise is any trivial bundle with fiber $F.$   Of course, for vector bundles, we require the local trivializations to be smooth vector bundle maps, to define the notion of a locally trivial smooth vector bundle.  Unless otherwise stated, we always require our smooth vector bundles to be locally trivial with finite dimensional fibers.  If $E$ and $F$ are both smooth bundles over $M,$ with $E \subset F,$ so that $p_E$ is simply the restriction of the map $p_F$ to $E,$ then we call $E$ a subbundle of $F.$

For instance, if $E$ is a smooth vector bundle over $M,$ then the addition map $Add$ defines a smooth bundle map

$$Add : E \oplus E \lra E,$$
and as addition is actually itself a linear map, it follows that $Add$ is a smooth vector bundle map.  

If $E$ is a smooth (or only continuous) vector bundle and each fiber is an algebra, then the multiplication in the fibers gives a bundle map $m: E \oplus E \lra E$ which is the bilinear multiplication in each fiber.  We say $E$ is a smooth (or continuous) {\bf Bundle of Algebras} if $m$ is smooth (or continuous).

A semi-Riemannian vector bundle $E$ is a smooth vector bundle over the smooth manifold $M$ with a given non-degenerate metric tensor $g$ on $E,$ that is, to be precise, $g$ is a smooth map

$$g : E \oplus E \lra \bR,$$
which on each fiber of $E \oplus E$ defines a real-valued non-degenerate symmetric bilinear map.  For $v \in E$ we define the length of $v$ denoted $\|v\|$ by

$$\|v\|=\sqrt{|g(v,v)|}.$$

Of course, the main example of a locally trivial smooth semi-Riemannian vector bundle is the tangent bundle of any semi-Riemannian manifold.

The following theorem which is an easy consequence of the inverse function theorem is useful for quickly proving certain spaces are manifolds and easily determining their tangent bundles.  A major part of the convenience is that we only need to know certain mappings to be smooth but do not have to calculate derivatives in advance of knowing we are actually dealing with manifolds.  As most of the classical manifolds are easily submanifolds of vector spaces, their manifold character and their tangent bundles can be calculated with ordinary differentiation of simple functions without having to deal with coordinate systems and verifying consistency on coordinate changes.  A mapping $r:S \lra S$ is called a retraction provided that $r \circ r=r.$  In this case, the image $r(S) \subset S$ of $S$ 
under $r$ is the set of all fixed points of $r,$ that is

$$r(S)=\{x \in S : r(x)=x\}.$$

Thus, if $M$ is a manifold, and if $r:M \lra M$ is a smooth retraction, then $Tr : TM \lra TM$ is a retraction as well as a smooth vector bundle map.  Likewise we can apply this to $TTr$ and so on.

\begin{theorem}\label{RETRACT} {\bf RETRACTION THEOREM.} 
Suppose that  $X$ is a subset of the smooth manifold $M.$  Then $X$ is a smooth submanifold of $M$ if and only if for each $x \in X$ there is an open subset $U_x $ of $M$ containing the point $x$ and a smooth retraction $r_x:U_x \lra U_x$ such that

$$r_x(U_x)=U_x \cap X.$$
In this case, it follows that for each $y \in U_x$ the linear map $T_yr_x: T_yM \lra T_yM$ is a continuous linear retraction of $T_yM$ onto $T_yX.$  Thus, for each $x \in M$ the bundle map 

$$Tr_x : TU_x \lra TU_x$$
satisfies

$$Tr_x(TU_x)=T(U_x \cap X).$$  

\end{theorem}

\begin{proof}
This is an easy exercise in application of the inverse function theorem, but for more details see \cite{DUPREGLAZE1}.

\end{proof}

In many applications of the Retraction Theorem, we have a manifold or even simply a vector space $E$ with $U \subset E$ a single open subset and $r:U \lra U$ a smooth retraction with 
$r(U)=M \subset U.$  For instance, if $E$ is a vector space and if $g$ is a positive definite metric on $E,$ then with $U$ being the set of non-zero vectors in $E,$ we define $r:U \lra U$ by

$$r(x)=x/\|x\|,~x \in U,$$
so then $r$ is a smooth retraction and $S=r(U)$ is the unit sphere in $E,$ showing the unit sphere to be a submanifold of $E,$ without needing to compute a derivative.  It is now trivial to see that the tangent space $T_xS$ is the orthogonal complement of $x$ in $E,$ that is $TS=\{(x,v) \in E \times E : g(x,v)=0\}.$ 

\begin{corollary}\label{SPLIT}  Suppose that $E$ is a smooth locally trivial vector bundle over $M$ and that $Q:E \lra E$ is a smooth vector bundle map over $M$ with $Q \circ Q=Q,$  Then $Q(E)$ is a smooth vector subbundle of $E$ and with $R=id_E-Q$ we have $R \circ R=R$ and

$$E=Q(E) \oplus R(E).$$

\end{corollary}

If $E$ is a smooth bundle over $M,$ by a selection $s$ of $E$ we mean a map $s:M \lra E$ such that $p_E \circ s=id_M.$  A continuous selection is called a continuous section. A smooth selection is called a smooth section or often simply a section.  We denote by $\Gamma(E)$ the set of all smooth sections of $E.$  Notice that if $s$ is a smooth section of $E,$ then $r=s \circ p_E$ is a smooth retraction $r : E \lra E$ and $r(E)=s(M),$ so by the Retraction Theorem, the image of any section is a submanifold of $E,$ and then $s$ defines a diffeomorphism of $M$ onto $s(M).$

If $E$ and $F$ are smooth vector bundles over $M,$ then we denote by $L(E;F)$ the smooth vector bundle whose fiber over the point $x \in M$ is $L(E_x;F_x),$ the vector space of linear transformations, for each $x \in M.$  Each smooth vector bundle map defines in an obvious manner a section of $L(E;F)$ and likewise, each section of $L(E;F)$ defines a smooth vector bundle map of $E$ to $F$ over $M.$  This means we can apply vector bundle analysis to vector bundle maps themselves.  In particular, $A=L(E;E)$ is a smooth algebra bundle, that is it is a smooth vector bundle and the multiplication in each fiber defines a smooth vector bundle map $Mult:A  \oplus A \lra A.$  If $E$ is a semi-Riemannian vector bundle, then  for any smooth vector bundle map $T: E \lra E$ there is an adjoint, and likewise this applies to members of the fibers, so the adjoint operation defines a smooth map $Adj:A \lra A.$  Moreover, if $g$ is positive definite, then the operator norm defines a continuous map of $A \lra \bR.$  We will call $E$ Riemannian if $g$ is positive definite.  Now, in case that $E$ is a real Riemannian vector bundle, we can complexify forming $E_{\bC}=E \otimes \bC$ which is a smooth complex vector bundle, and we define a hermitian metric $g_{\bC}=g \otimes \bC$ on $E_{\bC}$ by 

$$g_{\bC}(u \otimes z, v \otimes w)=g(u,v)\bar{z}w.$$
This makes $E_{\bC}$ into a locally trivial Hilbert bundle \cite{DUPRE0}  and now $A_{\bC}=L(E_{\bC};E_{\bC})$ is a locally trivial $C^*$-algebra bundle as in \cite{DUPRE0}.  We can then form the $C^*$-algebra of bounded continuous sections of $A_{\bC}$ which we will denote simply by $B.$  Of course, if $M$ is compact, then $B$ consists of all the continuous sections of $A_{\bC}.$  We can then apply the analytic functional calculus to members of $B,$ and as point evaluations are $C^*$-algebra homomorphisms, it follows that applying an analytic function to a section gives the section that would be obtained by applying it to each of its values.  Using the correspondence between smooth bundle maps of smooth sections, allows us to apply the functional calculus to smooth vector bundle maps.  Thus, if $T:E \lra E,$  is a smooth vector bundle map over $M,$ then we can regard $T \otimes id_{\bC}  \in B,$ and then if $S$ denotes the spectrum of $T,$ then $S \subset \bC$ is compact and if $f$ is a complex analytic function defined on an open neighborhood of $S$ in $\bC,$ then we can form $f(T \otimes id_{\bC}))$ and have $f(T\otimes id_{\bC})_x=f(T_x \otimes id_{\bC})$ for each $x \in M.$  In case that $f$ is real analytic, then $f(T \otimes id_{\bC})=f(T) \otimes id_{\bC}.$  Moreover, as $f$ is analytic, if $T$ is smooth, then $f(T)$ will also be smooth.  For instance, if $E$ is a real smooth Riemannian vector bundle, and if $T:E \lra E$ is an isomorphism, then $S$ is compact and does not contain $0 \in \bC.$  
Thus, in the complexification we can form the positive self-adjoint operator field $T^*T$ and as the square root function is real analytic on the right half plane, we can form the square root of $T^*T,$ denoted $|T|.$  Thus we have back in L(E;E), that $|T|:E \lra E$ is a smooth vector bundle map which is a smooth positive self-adjoint operator field commuting with $T^*T,$ and
  
$$|T|^2=T^*T.$$
This means that the polar decomposition of $T$ is $T=V|T|,$ by taking $V=T|T|^{-1},$ with $V$ a unitary operator.
In case that $T$ is itself self-adjoint, then by spectral theory in the $C^*$-algebra $B,$ we know that $T$ and $|T|$ commute as well, so $V, T,$ and $|T|$ all commute, and $V$ is now a self-adjoint unitary.

If $E$ is a semi-Riemannian smooth vector bundle with metric tensor $g,$ then we can always put a smooth Riemannian metric $k$ on $E$ using a smooth partition of unity subordinate to an open cover by subsets of $M$ over which $E$ is trivial.  By the Riesz representation theorem we can find a smooth vector bundle isomorphism $G:E \lra E$ such that 

$$g(v,w)=k(Gv,w), \mbox{ for all } (v,w) \in E \oplus E.$$
Symmetry of $g$ means that $G$ is self-adjoint relative to $k$ and therefore we have for the case $T=G,$ that $G=V|G|=|G|V.$  Of course, now we can form $|T|^{1/2}$ as well, and it commutes with these operators.  Define a new Riemannian metric $h$ by setting 

$$h(v,w)=k(|G|v,w).$$
Then, as $V$ is a self-adjoint unitary,

$$g(v,w)=k(Gv,w)=h(Vv,w)=h(v,Vw),~V^*=V,~V^2=1.$$
Since $V^2=1,$ the only eigenvalues of $V$ are $\pm1.$  Define

$$Q_{\pm }=\frac{1 \pm V}{2}.$$ 
We see immediately that $Q_{\pm }$ is self-adjoint and equals its own square in $B.$  This means that it is a self adjoint linear retraction and a smooth bundle map of $E$ to itself.  We also see that

$$V=Q_{+}-Q_{-},  ~Q_{-}+Q_{+}=id_E=1,  \mbox{ and } Q_{+}Q_{-}=0=Q_{-}Q_{+}.$$
By Corollary \ref{SPLIT}, we see that $E=Q_{-}(E) \oplus Q_{+}(E).$  On the other hand, we have 

$$Vu=u,~u \in Q_{+}(E) \mbox{ and } Vu=-u,~ u \in Q_{-}(E).$$
Set

$$E_+=Q_+(E) \mbox{ and } E_-=Q_-(E).$$
Thus, $E_+$ and $E_-$ are orthogonal subbundles of $E$ relative to both $g$ and $h,$ and $E_+ +E_-=E,$ that is to say, $E=E_+ \oplus E_-$ is an orthogonal decomposition of $E,$ and if $u,w$ belong to the fiber of $E$ over the same point of $M,$ then with

$$u_{\pm}=Q_{\pm}(u) \mbox{ and } w_{\pm}=Q_{\pm}(w),$$
we have, since $Vu=u_+ - u_-,$

$$g(u,w)=h(Vu,w)=h(u_+  -u_-,w_+ +w_-)=-h(u_-,w_-)+h(u_+,w_+).$$
We have decomposed $E$ as an orthogonal direct sum of smooth vector subbundles such that on one summand $g$ is negative definite and on the other $g$ is positive definite.  We define the signature of $g$ as the pair $(trace(Q_-),trace(Q_+))$ which are then positive integer valued and so constant on connected components of $M.$  It seems to be customary to report the dimension of the negative definite space first, so to say $g$ has signature $(r,s)$ will mean that locally $E$ has an orthonormal frame field $(u_{\alpha})$ where $g(u_{\alpha},u_{\alpha})=-1$ if $1 \leq \alpha \leq  r$ and  $g(u_{\alpha},u_{\alpha})=1$ if $r < \alpha \leq r+s.$  Thus, with our spacetime convention, we have in case $M$ is a spacetime, that $r=1$ and $n=s,$ so $n$ always denotes the spatial dimension of the spacetime.

%\begin{theorem}\label{SIGNATURE}

If $u$ is a time-like unit vector in a Lorentz vector space, $L,$
then $u^{\perp} \subset L$ is a Euclidean space. We define the
projection operator $P_u :L \lra L$ by
$P_u(v)=v+g(v,u)u,$ for any $v \in L.$ Then $P_uL=u^{\perp}.$   If $B$ and $C$ are any
linear transformations of $L,$ we say that $C$ is the adjoint of
$B$ to mean that $g(Bv,w)=g(v,Cw)$ for all pairs of vectors $v,w
\in L.$ In this case, $C$ is uniquely determined by $B$ and we
write $C=B^{(g)}=B^*.$ It is easy to see that in general for any two
linear operators $B$ and $C$ on $L$ we have $(BC)^*=C^*B^*.$ In
particular, $P_u$ is self-adjoint, $P_u^*=P_u$ and as well
$P_u^2=P_u,$ so $P_u$ is a self adjoint idempotent in the algebra of linear
maps of $L.$

If $B$ is any self-adjoint linear transformation of $L,$ then
$P_uBP_u$ is also self-adjoint but has $u^{\perp}$ as an invariant
subspace and therefore defines a self-adjoint linear
transformation $B_u :u^{\perp} \lra u^{\perp}.$ But since
$u^{\perp}$ is a Euclidean space, this means that $B_u$ is
diagonlizable. We call the eigenvalues (also called proper values)
of $B_u$ the $u-$spatial eigenvalues of $B,$ we call the principal
axes or lines through eigenvectors of $B_u$ the $u-$spatial
principal directions of $B,$ and we call the average of the
eigenvalues of $B_u$ the $u-$isotropic eigenvalue of $B.$ Thus, if
$\lambda_u$ is the $u-$isotropic eigenvalue of $B,$ then
$trace(B_u)=n \lambda_u.$  In particular, if $u$ is also an
eigenvector of $B$ with eigenvalue $r,$ then $B$ is completely
diagonalizable and $trace(B)=r+n\lambda_u.$ But, more generally,
since $g(u,u)=-1,$ we always have,

\begin{equation}\label{trace0}
trace(B)=-g(u,Bu)+n\lambda_u,
\end{equation}
even if $u$ is not an eigenvector of $B.$ Thus, we emphasize that
even though $B_u$ is always diagonalizable, $B$ itself need not
be, but (\ref{trace0}) will still be true.

Using the time-like unit vector $u$ allows us to also define a
Euclidean metric or inner product $g_u$ on $L$ by defining
$g_u(v,w)=g(v,w)+2g(v,u)g(u,w).$ This makes $L$ a topological
vector space and in case $L$ is finite dimensional, this gives $L$
its unique vector topology. Thus even though the Euclidean inner
product on $L$ depends on the choice of $u,$ the resulting
topology does not. It is easy to see that for $B=B^*$ to be also
self-adjoint with respect to the Euclidean inner product $g_u,$ it
is necessary and sufficient that $u$ be an eigenvector of $B$ in
which case $B$ is itself then diagonalizable.

If $T$ is a second rank covariant tensor on $L,$ which here is merely to say that
$T$ is a real-valued bilinear map on $L,$ then there is a unique
linear map $B_T: L \lra L$ with $T(v,w)=g(v,B_Tw),$ for all $v,w \in
L.$ We can now invariantly define the {\it contraction} of $T,$
denoted $c(T),$ by

\begin{equation}\label{contract0}
c(T)=trace(B_T).
\end{equation}
Any question of eigenvalues, eigenvectors, or diagonalizbility for
$T$ is really the same question for $B_T.$ Clearly to say $T$ is
symmetric is the same as saying that $B_T$ is self-adjoint. Thus
for $T$ symmetric, its $u-$spatial principal directions are those
of $B_T,$ its $u-$spatial eigenvalues are those of $B_T$ and its
$u-$isotropic eigenvalue is that of $B_T.$ If $(M,g,T)$ is a
spacetime model, we call the $u-$isotropic eigenvalue of $T$ at $m
\in M$ the {\it isotropic pressure} for the observer with velocity
$u \in L=T_mM,$ and we will denote this by $p_u,$ in this case.  Also, in this situation, $T(u,u)$ is always designated as the {\it
energy density} observed by the observer with velocity $u,$ which
we also denote by $\rho_u.$
Thus for this situation we have, by (\ref{trace0}) and
(\ref{contract0}),

\begin{equation}\label{contract1}
c(T)=trace(B_T)=-\rho_u+np_u.
\end{equation}

At this point we want to remark that if $E$ is any vector space
with a positive definite inner product, $g,$ and if  $A: E \lra E$
is any linear transformation of $E,$ then $A$ can be viewed as a
vector field on $E,$ say ${\bf v}_A$ where ${\bf v}_A(x)=A(x)$ for
$x \in E,$ and as well it defines the dual 1-form $\lambda_A$ on
the Riemannian manifold $E$ defined by $\lambda_A
(x)(w)=g(A(x),w).$ We record the following result as a proposition
for future use. Its proof is an easy exercise.

\begin{proposition}\label{flatspacedivcurl}
For the linear transformation $A:E \lra E$ of the Euclidean space
$E,$ we have

\begin{equation}\label{flatspacediv}
div_E {\bf v}_A(0)=(div_E A)(0)= trace(A),
\end{equation}
where $div_E$ denotes the ordinary divergence operator on vector
fields defined on $E.$  Moreover, $\lambda_A$ is a closed 1-form
(meaning $d \lambda_A=0$) if and only if $A$ is self-adjoint as a
linear transformation of $E.$
\end{proposition}

For $M$ a Lorentz manifold, $m \in M,$ and $u \in T_mM$ a
time-like unit vector, and thus an observer at $m,$ we set
$E(u,m)=u^{\perp} \subset T_mM$ and call $E(u,m)$ the observer's
Euclidean space (at $m$). Choose an open subset $W$ of $T_mM,$
with $0 \in W$ and make the choice small enough that the
exponential map carries $W$ diffeomorphically onto a geodesically
convex (\cite{KRIELE}, page 131) open subset $W_L$ of $M$
containing $m.$ Denote the image of $W_E=W \cap E(u,m)$ under this
exponential diffeomorphism by $W_R.$  We shall call $W_R$ the
observer's Riemannian space (at $m$), whereas we refer to $W_E$ as
the observer's Euclidean neighborhood.  Thus, we should
intuitively think of $W_E$ as the Euclidean space an observer
thinks he is in if he is unaware of curvature, whereas $W_R$ is
the space the sophisticated observer thinks he is in when he is
aware of curvature.  Any linear transformation, $A$ of $E(u,m),$
can be viewed by the observer as
a vector field ${\bf w}$ on his Euclidean space, and by
(\ref{flatspacediv}), the divergence, $div_E {\bf w}(0)$ is simply
the trace of $A,$ where $E=E(u,m).$

We can think of the observation process in a spacetime as being performed by a field of observers placed at each event throughout a region $K$ of spacetime, whose measurements are all communicated to a single observer who then uses all the measurements to understand the region $K.$  Thus, the field of observers is modeled by a timelike unit vector field on $K,$ which we generally denote by $u.$  We then have $g(u,u)=-1,$ and therefore $\nabla u$ viewed as a transformation of the tangent bundle has $u^{\perp}$ as its range subspace, that is, $[\nabla v]u=\nabla_v u \in u^{\perp},$ for any $v,$ pointwise.  Thus, we have $\nabla u=P_u [\nabla u].$  Then we can split apart the self-adjoint (symmetric) and skew-adjoint (antisymmetric) parts of $[\nabla u] P_u=P_u[\nabla u]P_u,$ setting

$$B_u=\frac{1}{2}[([\nabla u] P_u)+([\nabla u] P_u)^*)],$$ and

\begin{equation}\label{rotation}
\omega_u=\frac{1}{2}[([\nabla u] P_u)-([\nabla u] P_u)^*].
\end{equation}
It is customary to refer to $\omega_u$ as the {\it rotation} of $u.$  As $trace(P_u)=n,$ it is customary to set

\begin{equation}\label{expansion}
\theta_u=trace([\nabla u] P_u)=trace(B_u),
\end{equation}
and
\begin{equation}\label{shear}
\sigma_u=B_u-\frac{1}{n}\theta_u P_u,
\end{equation}
and we refer to $\theta_u$ as the {\it expansion} of the field $u,$ whereas we call $\sigma_u$ the {\it shear}.  As a consequence we have

\begin{equation}\label{expansionshearrotation}
[\nabla u] P_u=\frac{1}{n}\theta_u P_u +\sigma_u + \omega_u.
\end{equation}
Notice that the trace defines an inner product on the vector space of linear transformations of $u^{\perp}$ at each event, known as the Hilbert-Schmidt inner product, and that in this inner product, the self-adjoint transformations are perpendicular to the skew-adjoint transformations and thus $B_u$ is orthogonal to the rotation $\omega_u.$  On the other hand, as $\sigma_u$ is obviously trace free and $P_u\sigma_u=\sigma_u,$ it follows that the shear is orthogonal to $\theta_uP_u,$ so the decomposition (\ref{expansionshearrotation}) is an orthogonal decomposition of $[\nabla u] P_u$ in the space of transformations of $u^{\perp}$ under the Hilbert-Schmidt inner product, in which the length of a transformation $T$ is called its Hilbert-Schmidt norm, denoted $\|T\|_2.$  We now have, as $\theta_u$ and $\sigma_u$ are self adjoint and $\omega_u^*=-\omega_u$,

$$([\nabla u] P_u)^*=\frac{1}{n}\theta_u P_u + \sigma_u -\omega_u,$$
and
$$\|(\nabla u) P_u\|_2^2=\frac{1}{n^2}\theta_u^2+\|\sigma_u\|_2^2+\|\omega_u\|_2^2.$$  Since $trace(AB)=trace(BA)$ and $P_u^2=P_u,$ if $T$ is any self-adjoint transformation of $TM,$ then $P_uTP_u$ is orthogonal to $\omega_u$ and therefore in the Hilbert-Schmidt inner product we have

$$trace(T [\nabla u] P_u)=\frac{1}{n}\theta_u ~trace(P_uTP_u)+trace(P_uT P_u \sigma_u).$$  The field $u$ is called {\it geodesic} provided that $\nabla_uu=0,$ in which case we have $\nabla u=P_u[\nabla u]P_u,$ so the $P_u$ factors can be dropped from the above equations from in front of the $\nabla u$ factors in the case $u$ is a geodesic field of observers.  For a discussion of more general vector fields from the standpoint of fluid dynamics on semi-Riemannian manifolds, see also \cite{DUPRE&ROSENCRANS}, where it is observed that for any vector field $u,$ it is the case that

$$\nabla_u-L_u=D(\nabla u),$$
where as usual, $L_u$ is the Lie derivative operator defined by $u,$ and where for any linear tangent bundle transformation field $B,$ we denote by $D(B)$ the unique derivation of the tensor algebra which vanishes on smooth scalar functions and coincides with $B$ on ordinary smooth vector fields.  In particular, if $C$ is another such transformation field, then 

$$D(B)C=[B,C], \mbox{ hence } \nabla_u(\nabla u)=L_u(\nabla u), \mbox{ as the algebraic commutator } [\nabla u, \nabla u]=0.$$
On the other hand, the Riemann curvature operator $\R$ is given by

\begin{equation}\label{Ray0}
\R(u,v)=[\nabla_u, ~\nabla_v]-\nabla_{[u,v]}.
\end{equation}
Then,

$$[\nabla_u(\nabla u)]v=\nabla_u(\nabla_v u)-[\nabla u] \nabla_u v =\nabla_u(\nabla_v u) -\nabla_{\nabla_uv}u,$$

$$[\nabla u] \circ [\nabla u]v=\nabla_{\nabla_v u} u,$$
and,

$$[\nabla (\nabla_u u)]v=\nabla_v(\nabla_u u) $$
Therefore we have the identity

\begin{equation}\label{RiemRay}
\R(u,v)u=-\R(v,u)u=[\nabla_u(\nabla u)]v - [\nabla (\nabla_u u)]v + [\nabla u] \circ [\nabla u]v, \mbox{ for any vector fields } u \mbox{ and } v.
\end{equation}
Anticipating the definition of the Jacobi curvature operator, $\J,$ will imply $\J(u,u)v=\R(v,u)u,$ we arrive at the identity underlying the Raychaudhuri equation

\begin{equation}\label{Ray1}
-\J(u,u) = [\nabla_u(\nabla u)] - [\nabla (\nabla_u u)] + [\nabla u] \circ [\nabla u], \mbox{ for any vector field } u,
\end{equation}
or using $[\nabla_u ,\nabla]u=[\nabla_u(\nabla u)] - [\nabla (\nabla_u u)]$
and  $B^2=B \circ B$ for the composition of transformations, for any transformation $B$ of the tangent bundle,  

\begin{equation}\label{Ray2}
-\J(u,u) = [\nabla_u,\nabla]u + [\nabla u]^2, \mbox{ for any vector field } u.
\end{equation}

The Raychaudhuri equation results from simply taking the trace of both sides of (\ref{Ray1}) after using the orthogonal decomposition above for $\nabla u$ and the fact that the trace of $\J(u,u)$ is $R(u,u),$ where $R$ denotes the Ricci tensor, as we shall soon see.

\med

%%%%%%%%%%%%%%%%%%%%%%%%%%%%%%%%%%%%%%%%%%%%%%%%%%%%%%%%%%%%%%%%%%%%%%%%%%%%%%%%%%%%%%%%%%%%%%%%%%%%%%
%%%%%%%%%%%%%%%%%%%%%%%%%%%%%%%%%%%%%%%%%%%%%%%%%%%%%%%%%%%%%%%%%%%%%%%%%%%%%%%%%%%%%%%%%%%%%%%%%%%%%%
%%%%%%%%%%%%%%%%%%%%%%%%%%%%%%%%%%%%%%%%%%%%%%%%%%%%%%%%%%%%%%%%%%%%%%%%%%%%%%%%%%%%%%%%%%%%%%%%%%%%%%
%%%%%%%%%%%%%   and in particular, $A_u^{(geo)},$  %%%%%%%%%%%%%%%%%%%%%%%%%%%%%%%%%%%%%%%%%%%%%%%%%%%%%%%%%%%%%%%%%%%%%%%%%%%%%%%%%%%%%%%%%

%%%%%%%%%%%%%%%%%%%%%%%%%%%%%%%%%%%%%%%%%%%%%%%%%%%%%%%%%%%%%%%%%%%%%%%%%%%%%%%%%%%%%%%%%%%%%%%%%%%%%%
%%%%%%%%%%%%%%%%%%%%%%%%%%%%%%%%%%%%%%%%%%%%%%%%%%%%%%%%%%%%%%%%%%%%%%%%%%%%%%%%%%%%%%%%%%%%%%%%%%%%%%%%
%%%%%%%%%%%%%%%%%%%%%%%%%%%%%%%%%%%%%%%%%%%%%%%%%%%%%%%%%%%%%%%%%%%%%%%%%%%%%%%%%%%%%%%%%%%%%%%%%
%%%%%%%%%%%%%%%%%%%%%%%%%%%%%%%%%%%%%%%%%%%%%%%%%%%%%%%%%%%%%%%%%%%%%%%%%%%%%%%%%%%%%%%%%%%%%%%%%%%%%%

\section{GENERAL BUNDLE CONNECTIONS}

\med

Suppose we have any smooth map $p:E \lra X$ of manifolds, and we think of this as a smooth bundle over $X.$  Then for $A \subset X,$  we here denote by $\Gamma(p;A)$ or $\Gamma(E|A)=H^0(A;E)=H^0(A;p)$ the set of smooth sections of $[p:E \lra X]|A$ meaning sections of $p$ with domain $A.$   It is natural that if $s \in \Gamma(E)$ and $u$ is a tangent vector field on $X,$ then we would like to be able to differentiate $s$ along $u$ and get a result which is independent of any particular coordinate choices.  We want an invariant construction.  

Notice that for each $x$ in $X,$ the bundle projection $p$ is constant on the fiber $E_x,$ with value $x,$ so for any $e \in E_x,$ we have $[Tp]_e ([T[E_x]_e)=0.$  With bundles, we think schematically of the fibers as being vertical and the base as horizontal, so naturally, we think of the tangent vectors to fibers as being vertical tangent vectors.  Thus, the vertical tangent vectors are all in the Kernel of $Tp,$ the tangent map of the bundle projection.  Unfortunately, there is in general no specifically given splitting of $TE$ along the vertical subbundle.  If $e \in E,$ and $s \in H^0(X;E),$ then as $ps=id_X,$ it follows that for $x=p(e),$ we have $[Tp]_e:[TE]_e \lra [TX]_x$ is surjective, so by the implicit function theorem, there is $V$ open in $E$ with $e \in V$ and so that $V \cap E_x$ is a submanifold of $E,$ and in fact, the tangent space at $e$ is the kernel of $[Tp]_e,$ for each point $e \in V \cap E_x.$  This gives a splitting, but each section through a given point of $E$ gives a possibly different splitting of the tangent space of the bundle at that point. 

Without assuming local triviality, let us simply assume that each point of $E$ is the value of some local section of $E,$ meaning that if $e \in E,$ then there is some open subset $U$ of $X$ with $s \in H^0(U;E)$ and $s(p(e)=e.$  In this case, we say the bundle is  {\bf smoothly full}.  Thus our previous observation tells us that now each fiber is a submanifold of $E.$  Then as $sp$ is a retraction of $E|U$ onto $s(U),$ it follows by the retraction theorem that $s(U)$ is also a submanifold of $E|U$ for each open subset $U$ of $X,$ and therefore if $s$ is any smooth section of $E,$ then $s(X)$ is a submanifold of $E,$ and $s$ is a diffeomorphic embedding of $X$ into $E$ with image $s(X).$  Moreover, by the retraction theorem, we have a local splitting given by the retraction vector bundle map $T(sp)$ of $TE$ onto $T[s(X)].$  We therefore have that for any $x \in X,$ the submanifolds $s(X)$ and $E_x$ intersect transversally.

 Now, getting back to differentiation of sections, we can of course notice that $Ts$ is a section of $Tp:TE \lra TX,$ and $u$ being a section of $TX$ is a smooth map $u:X \lra TX,$ so one might simply take $(Ts) \circ u:X \lra TE$ to be the derivative of $s$ along $u.$  However, if we do this, then we should look at the example of a product bundle with fiber $F$ over $X$ to see what this really means.  Thus, with $E = X \times F,$ our section $s$ is $s=(id_X,f):X \lra X \times F$ where $f:X \lra F$ is what we call the {\bf principal part} of the section $s.$  Then, provided that $F$ is a submanifold of a Banach vector space $V$, we can use ordinary differentiation for $Tf$ and have

$$(Ts)u(x)=(T(id_X)u(x),(Tf)u(x))=(u(x),(f(x),f'(x)u(x))).$$
Now in this case, the only thing we really have interest in is the function sending $x$ to $f'(x)u(x),$ that is to say we note that $f'u$ is the principal part of a section of the product bundle with fiber $V.$  But, in fact, $f'(x)u(x)$ is tangent to $F$ at $f(x).$   Now, we know that as $p$ is constant on $E_x$ it follows that $[Tp]f'u=0.$  Thus, returning to the general case, we want the derivative of $s$ along $u(x)$ to be in $T[E_{s(x)}]=Ker([Tp]_{s(x)}):[TE]_{s(x)} \lra [TX]_x.$  Since we think of bundle fibers as being vertical with the base horizontal, schematically, we think of tangent vectors to fibers as being vertical, so we thus want to look for the vertical component of $(Ts)u,$ pointwise.  Now, as $p$ is constant on fibers, it follows that at each point $e$ of $E,$ the space of vertical tangent vectors in 
$[TE]_e$ is $Ker([Tp]_e),$ and $Ker(Tp)$ is a vector subbundle of $TE,$ so we need a smooth vector bundle retraction $P_{vert}:TE \lra TE$ over $id_E$ which in effect provides a projection at each point of $E$ of the tangent space to $E$ at any point $e \in E$ onto the subspace of vertical tangent vectors at $e,$ that is $[P_{vert}]_e(TE_e)=T[E_{p(e)}].$  

The bottom line here, is that to have a process we want for differentiating sections which in the case of trivial bundles corresponds to differentiating principal parts of sections, we need a smooth vector bundle retraction 

$$P_{vert}: TE \ra TE \mbox{ over } id_E \mbox{ so that } P_{vert}(TE)=Ker[Tp].$$

It is then customary to set

$$P_{horz}=1-P_{vert}, \mbox{ where } 1=id_{TE},$$
so $P_{horz}$ is the smooth vector bundle retraction on a complementary smooth vector subbundle of the vertical subbundle of $TE,$ so we naturally call it the horizontal subbundle of $TE.$  Since the vertical subbundle is the kernel of $Tp,$ that is $Ker(Tp),$  it follows that if $HE$ denotes the image of $P_{horz},$ then $Tp|HE$ is a surjective smooth vector bundle map of $HE$ onto $TX,$ which on each fiber is an isomorphism of vector spaces.  Thus, if $e \in E,~p(x)=e,$ and $v \in TX_x,$ then there is a unique tangent vector $h \in HE_e \subset TE_e$ for which $[Tp]h=v,$ and we call this the {\bf horizontal lift} of $v$ to $TE_e.$

This said, if $s$ is a smooth section of $E,$ and $v$ is a tangent vector in $TX_x,$ then we simply define the {\bf COVARIANT DERIVATIVE} of $s$ along $v,$ denoted $\nabla_u s(x),$  to be 

$$\nabla_v s(x)=P_{vert}[Ts]v \in [TE]_{s(x)}.$$

Thus, if $s$ is a smooth section of $E$ over $X$ and $v$ is a tangent vector field on $X,$ then $\nabla_v s:X \lra Ker(Tp),$ where of course $[\nabla_v s](x)=\nabla_{v(x)} s(x),$ for each $x \in X.$

The most important case is the case where $E$ is a principal $G$-bundle over $X.$  Then, the fibers are the $G$-orbits of the action, so it is natural to require that the endomorphism $P_{vert}$ should respect this action, that is, it should be equivariant.  Of course, applying the tangent map to the action gives an action of $G$ on $TE,$ so we require that $P_{vert}$ is equivariant with respect to this action.  Consequently, the same will be true of $P_{horz}.$  If $F$ is a smooth manifold on which $G$ acts smoothly and effectively, then any section $s$ of the associated bundle $E[F]$ corresponds to a smooth map $f:E \lra F.$ which is equivariant.  Of course, the tangent map of the action gives a smooth action of $G$ on $TE$ covering the action of $G$ on $E,$ which viewing each $G$ as a function from $F$ to $F,$ since the action is effective, means for $v \in TE,$ we have simply $gv=[Tg]v.$   So for $v$ in $TX_x$ we choose a horizontal lift $w$ in $HE_{s(x)}$ and define $D_wf((s(x))=[Tf]w \in TF_{f(s(x))}.$ If $v$ is any tangent vector field, we can use the connection to obtain a horizontal lift $w$ to map of $X$ into $HE,$ so that $w(x)$ is in $HE_{s(x)}$ for each $x \in X.$  For $e \in E,$ then $g(e)=\tau(s(x),e)$ has $g(e)[s(x)]=e,$ so we define 

$$D_wf(e)=[Tf]g(e)w(p(e)).$$
Then, $D_wf$ is an equivariant map of $E$ into $TF$, so corresponds to a section of $E[TF]$ which we denote $$\nabla_v s=D_w f.$$

In case that $F$ is a vector space and $G$ acts linearly, then $E[F]$ is a vector bundle and we can reduce consideration of $TF$ to $F$ itself as each tangent space to $F$ is naturally identified with  $F$ itself.  In this case, we can replace $Tf$ with $f'=Df,$ where $D$ is the ordinary flat connection, that is ordinary differentiation for vector valued-functions.  We then have simply

$$\nabla_v s \mbox{ is the section corresponding to }f'w=D_wf.$$

Thus in the derivative of an equivariant map into $F$ along a horizontal lift is again an equivariant map into $F$ giving the covariant derivative of the corresponding section.  Then we have $\nabla_v s$ is again a section of $E[F].$  It then follows from linearity of the action that $\nabla$ is a Kozul connection on the vector bundle $E[F]. $

In case that $G$ is a Lie subgroup of the general linear group $G(V) \subset L(V;V),$ where $V$ is a vector space (as is the case for virtually all Lie groups), we can take $F=L(V;V),$ and then we obtain an embedding of $E$ in $E[F]$ as a subbundle of a vector bundle with fiber $F=L(V;V).$  In this case, the covariant derivative on the principal bundle in effect becomes an ordinary Kozul connection with $G$-invariance properties.  In particular, the curvature of the principal bundle is the curvature of this associated bundle.

 In case that $E$ is a vector bundle, then as each fiber is a vector space, and thus each vertical tangent space is naturally identified with the fiber it is tangent to, the connection $P_{horz}$ is just a retraction of $TE$ on $p^*(TX),$

In case that $E$ is a vector subbundle of the product bundle $X \times W,$ where $W$ is a possibly very high dimensional vector space, then we look for a projection valued function $Q:X \lra L(W;W)$ with the property that $Q(x)W=E_x \subset W,$ and then we define the covariant derivative using principal parts of sections of $E$ by simply setting $\nabla_v f=Qf'v=QD_v f,$ where $f:X \lra W$ is the principal part of a section, so $f(x) \in E_x \subset W,$ for each $x \in X.$  In fact, every smooth bundle over a paracompact base space is a pull back of such a bundle, though $W$ may be required to be infinite dimensional Hilbert space.  The map $Q$ is called the {\bf classifying map} of the bundle.

In case of any Kozul connection $\nabla$ on a smooth vector bundle $E,$ the Riemann curvature operator is $\R_{\nabla} \in L^2_{alt}(TM;L(E;E))$ given by

$$\R_{\nabla}(v,w)=[\nabla_v, \nabla_w]-\nabla_{[v,w]},$$
where the first pair brackets denote operator commutators, whereas the second is the Lie bracket of vector fields.  The exterior covariant derivative of $\R_{\nabla}$ always vanishes (the Bianchi identity), and can be regarded as a conservation law.  Put another way, the connection naturally extends to all tensor bundles over $E$ via the Leibniz or product rule for differentiation, so $Alt(\nabla \R)=0.$  This basically boils down to the general Jacobi identity of Lie algebra applied to the differential operators.

In the semi classical modeling of forces, the force field is modeled as the curvature operator on a principal bundle which in any group representation induces a Kozul connection on a smooth associated vector bundle.  Thus if we now denote the curvature operator by $F,$ since we are dealing with force, then the energy momentum stress tensor serving as gravity source is

$$\bar{T}_F=trace(T_F),$$
where 

$$T_F = FgF-\frac{1}{d}[c(FgF)]g, \mbox{ where } d=dim(X)$$
and the contraction

$$FgF(v,w)=F(v,e_{\alpha})g(e^{\alpha},e^{\beta})F(e_{\beta},w), \mbox{ for all tangent vector fields } v,w,$$
and

$$c(T)=g(e^{\alpha},e^{\beta})T(e_{\alpha},e_{\beta}),$$
for any operator valued bilinear map $T.$
Of course here, the family $(e_{\alpha})$ is a local frame field with dual frame field  $(e^{\alpha}).$  As trace commutes with the contraction $c,$ it follows that $c(\bar{T}_F)=0,$ so the energy momentum stress tensor is for a massless force field.  But, as $trace(T_F)$ gives the energy momentum stress tensor of the force field serving as gravitational source, it should be the case that we should regard $T_F$ itself as the energy momentum stress of the force field, and that in effect, in a vacuum, the force field energies are in a delicate balance.  Thus, the general Bianchi identity should be regarded as the conservation law, especially as it gives $div \bar{T}_F=0,$ the ordinary conservation law for gravitational sources.  In the case of the electromagnetic field, the vector bundle is simply a line bundle, so the trace becomes redundant, that is $\bar{T}_F=T_F.$   We will investigate what happens if we apply this formula directly to the ordinary Riemann curvature operator below.

\med
%%%%%%%%%%%%%%%%%%%%%%%%%%%%%%%%%%%%%%%%%%%%%%%%%%%%%%%%%%%%%%%%%%%%%%%%%%%%%%%%%%%%%%%%%%%%%%%%%%%%%%
%%%%%%%%%%%%%%%%%%%%%%%%%%%%%%%%%%%%%%%%%%%%%%%%%%%%%%%%%%%%%%%%%%%%%%%%%%%%%%%%%%%%%%%%%%%%%%%%%%%%%%
%%%%%%%%%%%%%%%%%%%%%%%%%%%%%%%%%%%%%%%%%%%%%%%%%%%%%%%%%%%%%%%%%%%%%%%%%%%%%%%%%%%%%%%%%%%%%%%%%%%%%%
%%%%%%%%%%%%%   and in particular, $A_u^{(geo)},$  %%%%%%%%%%%%%%%%%%%%%%%%%%%%%%%%%%%%%%%%%%%%%%%%%%%%%%%%%%%%%%%%%%%%%%%%%%%%%%%%%%%%%%%%%

%%%%%%%%%%%%%%%%%%%%%%%%%%%%%%%%%%%%%%%%%%%%%%%%%%%%%%%%%%%%%%%%%%%%%%%%%%%%%%%%%%%%%%%%%%%%%%%%%%%%%%
%%%%%%%%%%%%%%%%%%%%%%%%%%%%%%%%%%%%%%%%%%%%%%%%%%%%%%%%%%%%%%%%%%%%%%%%%%%%%%%%%%%%%%%%%%%%%%%%%%%%%%%%
%%%%%%%%%%%%%%%%%%%%%%%%%%%%%%%%%%%%%%%%%%%%%%%%%%%%%%%%%%%%%%%%%%%%%%%%%%%%%%%%%%%%%%%%%%%%%%%%%
%%%%%%%%%%%%%%%%%%%%%%%%%%%%%%%%%%%%%%%%%%%%%%%%%%%%%%%%%%%%%%%%%%%%%%%%%%%%%%%%%%%%%%%%%%%%%%%%%%%%%%

\section{THE OBSERVER PRINCIPLE}

\med

The results of this section are primarily the results of \cite{DUPRE3}, and extended to cover the general semi-Riemannian case,  and which we include here for the convenience of the reader, as it will be applied in our later developments.

As previously noted, in case $M$ is a spacetime, for fixed $m \in M,$ mathematically, $T_mM$ is a Lorentz vector space of dimension
$n+1,$ so taking any time-like unit vector, say $u,$ and
defining $g_u(v,w)=2g(u,v)g(u,w)+g(v,w)$ gives a Euclidean metric
on $T_mM$ making it in particular into a Banach space of finite
dimension. On the other hand, as observed in the previous section, any semi-Riemannian vector bundle can be decomposed as a direct sum of orthogonal subbundles each of which is definite, so given the decomposition, the metric naturally breaks up so that on one summand it is negative definite and on the other summand it is positive definite, and therefore relative to the decomposition there is a natural choice for a positive definite metric.  In particular, for any semi-Riemannian manifold $M,$ the semi-Riemannian vector space $T_mM$ is an example of a Banachable space-a
topological vector space whose topology can be defined by a norm.
This topology is actually well-known to be independent of the previous decomposition
choices in case of finite dimensions. Differential geometry
can be easily based on such spaces, \cite{UPMEIER}, and for some examples in
infinite dimension, the interested reader can see
\cite{KRIGL&MICHOR}, \cite{BELTITA}, \cite{DUPGLAZEPREV}, \cite{DUPREGLAZE1} and
\cite{DUPREGLAZE2}. In particular, the theory of analytic
functions and power series all goes through for general Banachable
spaces \cite{UPMEIER}.  We would like to point out how this can be applied to the
theory of semi-Riemannian vector spaces and Lorentz vector spaces, as well as vector spaces of vector fields on spacetimes or even vector spaces of vector fields on a semi-Riemannian vector bundle.

For the remainder of this section, when we speak of a semi-Riemannian manifold, we will assume that the signature $(s_-,s_+)$ has $s_- >0,$ so that at each point there are non-zero timelike vectors.  Of course if this is not the case, then all non-zero vectors are spacelike, and replacing $g$ by $-g$ interchanges timelike and spacelike to give non-trivial results.

In general, suppose that $E_1,E_2,...E_r,$ and $F$ are all vector spaces
(possibly infinite dimensional and not necessarily topological). Recall the function or mapping
$$A: E_1 \times E_2 \times...\times E_r \lra F$$ is a {\it multilinear
map} provided that it is linear in each variable when all others
are held fixed, and in this case, we say that $A$ is a multilinear map
of rank $r$ on $E_1 \times E_2 \times...\times E_r$ with values in $F.$  A useful notation here is just to use juxtaposition
for evaluation of multilinear maps, so we write

$$A(v_1,v_2,...,v_r)=Av_1v_2...v_r$$ whenever $v_k \in E_k$ for $1
\leq k \leq r.$  Thus, we simply treat the multilinear map $A$ as
a sort of generalized coefficient which allows us to multiply
vectors, and the multilinear condition simply becomes the
distributive law of multiplication.

In case that $E_k=E$ for all $k,$ there is really a single vector
space providing the input vectors, and $A:E^r \lra F.$ We say that
$A$ is a multilinear map of rank $r$ on $E$ in this case, even
though in reality, the domain of $A$ is the set $E^r.$ Here it is
useful to write $v^{(k)}$ for the $k-$fold juxtaposition of $v$'s.
Thus we have
$$A(v,v,...,v)=Av^{(r)}.$$ More generally, then for any positive
integer $m$ and vectors $v_1,v_2,...,v_m \in E$ and non-negative
integers $k_1,k_2,...,k_m$ satisfying $k_1+k_2+...+k_m=r,$ we have
the equation
$$Av_1^{(k_1)}v_2^{(k_2)}...v_m^{(k_m)}=A(v_1,...v_1,v_2,...,v_2,...,v_m,...v_m)$$ where
each vector is repeated the appropriate number of times, $v_1$
being repeated $k_1$ times, $v_2$ repeated $k_2$ times and so on.
If $k_i=0$ then that merely means that $v_i$ is
actually left out, so $v^{(0)}=1$ in effect.

Of course, we say that $A:E^r \lra F$ is symmetric if $Av_1v_2...v_r$ is
independent of the ordering of the $r$ input vectors.  Thus when
dealing with algebraic expressions involving symmetric multilinear
maps as coefficients, the commutative law is in effect.  We can define
the {\it monomial} function $f_A:E \lra F$ by the rule
$f_A(x)=A(x,x,x,...,x)=Ax^{(r)}.$

We denote by $L(E_1,E_2,...,E_r;F)$ the vector space of all multilinear maps of $E_1 \times E_2 \times ...\times E_r$ into $F,$ and set $L^r(E;F)=L(E_1,E_2,...,E_r;F)$ when all $E_k$ are the same vector space $E.$  Of course, $L^1(E;F)=L(E;F)$ is just the vector space of all linear maps from $E$ to $F.$  We denote the dual space of $E$ by $E^*=L(E;\bR).$  We use $L^r_{sym}(E;F)$ to denote the vector subspace of $L^r(E;F)$ consisting of the symmetric multilinear maps.  There is a natural isomorphism
$$L(E_1,E_2,...,E_r;F)\cong L(E_1,...,E_{r-1};L(E_r;F))$$ which identifies the $F-$valued rank $r$ multilinear map $A$ with the $L(E_r;F)-$valued multilinear map $B$ of rank $r-1$ given by
$$[Bv_1v_2...v_{r-1}](v_r)=Av_1v_2...v_r,$$  and notice that if $A$ is symmetric then so is $B,$ but of course the converse may not be true.  In any case, it is useful to simply denote $Bv_1v_2...v_{r-1}=Av_1v_2...v_{r-1}$ in this situation.

Suppose now that $E$ and $F$ are
any vector spaces and $A$ is a symmetric multilinear
map (tensor) on $E$ with values in $F,$ of rank $r.$  We will begin for simplicity by restricting to Banachable spaces, that is topological vector spaces whose topology is complete and comes from a norm.  In addition for simplicity, we assume that $A$ is continuous.  For in case $E$ and $F$ are Banachable spaces and $A$ is continuous, $f_A$ is an analytic
function. In fact, if $x_1,x_2,x_3,...,x_r \in E,$ then
differentiating, using proposition 3.3 and repeated application of
propositions 3.5 and 3.8 of \cite{LANG}, page 10, we find

\begin{equation}\label{analyticcontinuation}
D_{x_1}D_{x_2}D_{x_3}...D_{x_r}f_A(a)=(n!)A(x_1,x_2,x_3,...,x_r),~~a
\in E.
\end{equation}
From (\ref{analyticcontinuation}), we see very generally that if
$U$ is any open subset of $E$ on which $f_A$ is constant, then in
fact, $A=0,$ since we can choose $a \in U.$ Indeed, if $a \in U,$
since $f_A$ is constant on $U,$ it follows that the derivative on
the left side of the equation (\ref{analyticcontinuation}) is 0,
and hence the right side is 0, for every possible choice of
vectors $x_1,x_2,x_3,...x_r \in E.$ But notice that $a$ does not
appear on the right hand side of (\ref{analyticcontinuation}),
only $A(x_1,x_2,x_3,...x_r),$ and the vectors $x_1,x_2,x_3,...x_r$
can be chosen arbitrarily. Thus, $A=0$ follows. We have therefore
proven a special case of the following mathematical theorem, for which the proof in general will be given after some remarks.

\begin{theorem}\label{gen top anal cont}
Suppose $E$ is any topological vector space and $F$ is any vector
space. Suppose $A:E^n \lra F$ and $B:E^n \lra F$ are any symmetric multilinear
maps of rank $n.$ If there is a non-empty open subset of $E$ on
which $f_A-f_B:E \lra F$ is constant, then $A=B.$
\end{theorem}
We emphasize that the vector spaces here may be infinite dimensional and the multilinear map need not be continuous.

Notice that Theorem \ref{gen top anal cont} is a well known special case of
the uniqueness of general power series (there is only one term
here). For a purely algebraic proof in the case $r=2,$ which is
the case of most importance here, we refer the interested reader
to \cite{DUPRE}.  See also page 72 of \cite{SACHSWU} or page 260 of \cite{KRIELE} for a proof using
differentiation for the special case $r=2$ which is similar in form to
that given here next.  As well, the result for $r=2$ can easily be proved directly using algebra alone by the technique of polarization as in \cite{DUPRE}.  Thus, the {\bf Polarization Identity} for bilinear maps which itself is obvious, is that if $B$ is any symmetric bilinear map,  on the vector space $V$ with values in the vector space $W,$ then 

$$\mbox{POLARIZATION:  } B(x,y)=\frac{1}{4} [f_B(x+y)-f_B(x-y)]  \mbox{ for all }  x,y \in V.$$
Thus, in the special case where $A$ and $B$ are simply bilinear and $f_A=f_B,$ so the open set of agreement is all of $V,$ then obviously $A=B,$ by the polarization identity.  Notice that no topological consideration is necessary in this special case, and it would seem that there should be a way to eliminate the need to restrict to topological vector spaces, and after mentioning a corollary, that will be our aim.

\begin{corollary}\label{obsv2} {\bf OBSERVER PRINCIPLE.} If $A$
and $B$ are both symmetric tensors of rank $r$ on $T_mM$ with
values in $F,$ and if $Au^{(r)}=Bu^{(r)}$ for every time-like unit
vector in $T_mM,$ then $A=B.$
\end{corollary}

\begin{proof} Since $f_A$ and $f_B$ are homogeneous functions of
degree $r,$ it follows that the hypothesis guarantees
$Av^{(r)}=Bv^{(r)}$ for all $v$ in the light cone of $T_mM$ which
is an open subset of $T_mM.$
\end{proof}

Obviously, if $M$ is time oriented, then it would suffice to replace the phrase "every time-like unit vector" with "every future pointing time-like unit vector", and of course, we could just as well use "past" in place of "future".

In the following discussion, if $E$ is any semi-Riemannian vector space, we refer to the light cone of $E$ as the set of timelike vectors in $E$ which is of course an open subset of $E.$  Suppose now that $M$ is a semi-Riemannian manifold.  Thus, the light cone is a non-empty open subset, by our signature assumption for this section.  Of course, if we want, we can replace $g$ by $-g$ and interchange timelike with spacelike.  Suppose now that $E$ is any smooth semi-Riemannian vector bundle over the smooth manifold $M.$  Suppose $m \in M.$  If we define $U(E_m)$ to be the set of time-like unit
vectors in $E_m,$ then this set has a topology called the
relative topology as a subset of $E_m$ and we have a smooth retraction
function given, up to sign, by normalization which retracts the light cone onto
$U(E_m).$ It follows immediately that if $W$ is any (relatively)
open subset of $U(E_m),$ then the hypothesis of the observer
principle can be weakened to merely require $Au^{(r)}=Bu^{(r)}$
for each $u \in W.$ In particular, if $M$ is a spacetime with $E=TM,$ and if we choose a time orientation
on $T_mM,$ then we can merely require $Au^{(r)}=Bu^{(r)}$ for each
future time-like unit vector in $T_mM.$ This is in a sense, the
essence of the {\bf Principle of Relativity}, for instance, as
applied to second rank symmetric tensors-a law (at $m$), say
$A=B,$ should be true for all observers (at $m$) and conversely,
if true for all observers (at $m$), that is if $A(u,u)=B(u,u)$ for
all (future) time-like unit vectors $u \in T_mM,$ then it should
be a law (at $m$) that $A=B.$ It is for this reason that we call
Corollary \ref{obsv2} the observer principle.

We wish to be able to apply the observer principle to multilinear maps defined on vector spaces of sections of tensor bundles given by integration of sections, so we need the complete generality of Theorem \ref{gen top anal cont} whose proof we turn to now.  In fact, because of this need, it will be useful to be even more general, but the statement of the theorem becomes slightly more technical.  To give the general statement, we need to define the {\it star} of a subset of a vector space.  If $U \subset E,$ its star, denoted $Star(U,E)$ is the set of all points $v \in U$ having the property that for each $w \in E$ there is some positive number  $\delta_w$ so that $v+tw \in U$ for any number $t$ with $|t| \leq \delta_w.$  If $E$ is a topological vector space, then, as scalar multiplication is continuous, each open subset is equal to its star.  Set $Star^1(U,E)=Star(U,E)$ and inductively, define $$Star^r(U,E)=Star^{r-1}(Star(U,E)),$$ and call this the $r-$star of $U$ in $E.$  For instance, if $Star(U,E)=U,$ then obviously $Star^r(U,E)=U,$ for every $r\geq 1.$

We will also need to use the fact that if $v \in F,$ then there is $\lambda \in F^*$ with $\lambda(v)\neq 0.$  This fact in general vector spaces requires the existence of a spanning linearly independent set which is guaranteed by the Axiom of Choice of set theory.  As such an axiom might be objectionable in applications to physics, we circumvent this by simply defining $F$ to be {\it non-degenrate} provided that for each vector $v$ in $F$ there is a member $\lambda$ of $F^*$ with $\lambda(v) \neq 0.$  In all applications to vector spaces of tensor fields in physics, the non-degeneracy is usually obvious.  However, the topology is usually not, so the notion of the star will circumvent the need to actually deal with topological vector spaces, beyond merely noting that each open subset of a topological vector space equals its star and therefore its $r-$star for all $r \geq 1.$

\begin{theorem}\label{genanalcont} {\bf GENERAL OBSERVER PRINCIPLE.}  Suppose $F$ is a non-degenerate vector space. If $A$ and $B$ are a symmetric multilinear maps of rank $r$ on a vector space $E$ with values in the vector space $F$ and if the monomial
function $f_A-f_B:E \lra F$ is constant on a set having non-empty $r-$star, then $A=B.$
\end{theorem}

Obviously Theorem \ref{gen top anal cont} is a consequence of Theorem \ref{genanalcont}, by our previous remarks. Clearly, as $f_{A-B}=f_A-f_B,$ it suffices to prove the case with $A$ arbitrary and $B=0.$  Thus we begin by assuming that $A$ is an arbitrary symmetric multilinear map of rank $r$ on the arbitrary vector space $E$ with values in the arbitrary non-degenerate vector space $F.$

Then for any $v_0,v_1,...v_m \in E,$

\begin{equation}\label{multnomthm}
f_A(v_0+v_1+...+v_m)=\sum_{[k_0+k_1+...k_m=r]}C(r;k_0,k_1,...,k_m)Av_0^{(k_0)}v_1^{(k_1)}...v_m^{(k_m)}.
\end{equation}
Here $C(r;k_0,k_1,...k_m)$ is the multinomial coefficient:

\begin{equation}\label{multnomcoeff}
C(r;k_0,k_1,...,k_m)=\frac{r!}{k_0!k_1!...k_m!}.
\end{equation}

Now, proceeding inductively, let us notice that if $r=1,$ then the theorem is a triviality, since a linear map which is constant on any non-empty star is easily seen to be identically zero.  Assume the theorem is already proven for the case of rank $r-1.$  Suppose that $f_A$ is constant with value $C$ on $U \subset E$ with $Star^r(U,E)$ non-empty.  Choose $v_0 \in  Star(U,E) \subset U,$ take any $w \in E,$ and any $\lambda \in F^*.$  Then choose $\delta >0$ such that $v_0+tw \in U,$ whenever $|t| \leq \delta.$  We have

$$C=f_A(v_0+tw)=\sum^r_{k=0}C(r;k,r-k)t^kAv_0^{(r-k)}w^{(k)}=f_A(v_0)+\sum^r_{k=1} C(r;k,r-k)t^kAv_0^{(r-k)}w^{(k)}.$$ Thus, since also $f_A(v_0)=C,$ we must in fact have
$$\sum^r_{k=1}C(r;k,r-k)t^kAv_0^{(r-k)}w^{(k)}=0,$$ for any number $t$ with $|t| \leq \delta.$

Applying $\lambda$ to the previous vanishing equation gives a real-valued polynomial function on $\bR$ of degree $r$ which vanishes for all $t$ with $|t| \leq \delta.$  Since such a polynomial function can have at most $r$ roots, this vanishing implies all coefficients are zero.  In particular, this means that $$\lambda(Av_0^{(r-1)}w)= 0.$$

Since $F$ is non-degenerate and $\lambda \in F^*$ was arbitrary, $Av_0^{(r-1)}w = 0$  must be the case.  That is, the linear map $Av_0^{(r-1)}=0.$  But, $v_0$ is an arbitrary point of $Star(U,E).$   This means that if we define the rank $r-1$ multilinear map $B$ by $[Bv_1v_2...v_{r-1}]w=Av_1v_2...v_{r-1}w,$ then $f_B$ vanishes identically on $Star(U,E)$ which has non-empty $(r-1)-$star, and therefore by the inductive hypothesis, $B=0.$  But this obviously implies $A=0,$ and the proof is complete.  We take this opportunity to point out that the proof given in the appendix of \cite{DUPRE2} for this very general case is invalid, so the proof here, which is also the proof given in \cite{DUPRE3}, provides a correction to that appendix, as well as a further generalization.

By convention, a multilinear map from $E$ to $F$ of rank zero is
just a vector in $F.$  If $A_k$ is a symmetric multilinear map of
$E$ to $F$ of rank $k,$ for $0 \leq k \leq r,$ then the function

$$f=\sum_{k=0}^r f_{A_k}$$
is a polynomial function of degree $r.$   The above proof can easily be modified to apply to the polynomial function of any finite degree, by using induction on the degree of the polynomial.  The preceding expansion argument just gives a more complicated polynomial function, but as all coefficients must vanish,  applying it to the top degree term as above would show that the coefficient $A_r=0$ if we start with what we think is a degree $r$ polynomial.  But this means that the polynomial is really only of degree $r-1,$ and therefore by downward induction, we see $f=0.$  If $F$ is also a
topological vector space, then we can take limits in the sum and
consider power series. The general principle of analytic
continuation relies on the uniqueness of power series expressions.
In general, for Banach spaces, if two power series agree locally
as functions, then all their coefficients are the same-that is,
they are the same power series. The proof is easy using
differentiation, just use the same method used in freshman
calculus, but for Banach space valued functions. We have basically
proven this fact in case there is only one term in the power
series, and indicated how to prove it for a series with only a finite number of terms, but without using differentiation and without even having
topology on the vector spaces.

Of course, to apply the (General) Observer Principle, it helps to have an easy criterion to recognize when certain sets have non-empty $r-$star.  If $ U \subset E$ is a non-empty subset of the vector space $E,$ and $Star(U,E)=U,$ then certainly $Star^r(U,E)=U$ for every $r \geq 1.$   Here is a simple useful case.

\begin{proposition}\label{nonvanish}
Suppose that $A$ is a symmetric rank $r$ multilinear map of the vector space $E$ into the non-degenerate vector space $F,$ and $A \neq 0.$  Let $U=E \setminus f_A^{-1}(0).$  Then 

$$\emptyset \neq Star(U,E)=U.$$
\end{proposition} 

The proof is a simple variation on the previous proof.  As $A \neq 0,$ by the (general) observer principle, we know that $U \neq \emptyset.$  If $v \in U,$ and $w \in E,$ and $\lambda \in F^*$ with 
$f_A (v) \neq 0,$ and $\lambda$ chosen so that $\lambda(f_A (v)) \neq 0,$ then we can define $g: \bR \lra \bR$ to be $g(t)=\lambda (f_A [v+tw]).$  Them, $g(0) \neq 0,$ but expanding the expression for $g$ just as in the proof of the general observer principle, we see that $g$ is simply a polynomial function of degree $r,$ and as it does not vanish at $0,$ there is an open interval in $\bR$ containing $0$ on which it is non-vanishing.  As $w$ was arbitrary in $E,$ this shows that $v \in Star(U,E).$  As $v$ was an arbitrary member of $U,$ we have $U=Star(U,E).$

Some interesting simple corollaries follow now.

\begin{corollary}\label{vanish1}
Suppose that $A$ and $B$ are symmetric real valued multilinear maps of rank $r$ and $s,$ respectively, on the vector space $E.$  Then

$$Sym(A \otimes B)=0 \mbox{ implies either } A=0 \mbox{ or } B=0.$$
\end{corollary}  

We simply note that for any $v \in E,$ we have

$$[Sym(A \otimes B)]v^{(r+s)}=[Av^{(r)}][Bv^{(s)}]=f_A(v)f_B(v).$$
Therefore, if $A \neq 0,$ then $f_B$ must vanish on $U=E \setminus f_A^{-1}(0),$ which is nonempty with $Star(U,E)=U,$ and therefore $\emptyset \neq Star^s(U,E),$ so again by the general observer principle, we must have $B=0.$

\begin{corollary}\label{vanish2}
Suppose that $M$ is any $C^1$ manifold and that $A$ and $B$ are $C^1$ symmetric covariant tensor fields on $M,$ of any ranks.
If $x \in M$ and $Sym(A \otimes B)$ vanishes at $x,$ then either $A$ vanishes at $x$ or $B$ vanishes at $x.$
\end{corollary}

This is just the case $E=T_xM$ and $F=\bR,$ as in the next corollary as well.

\begin{corollary}\label{vanish3}
Suppose that $M$ is a $C^2$ manifold with Kozul connection $\nabla$ and that $A$ is a $C^2$ tensor field of rank $r$ which is symmetric.  Then, for any $x \in M,$ and any $v \in T_xM,$ if
$$\nabla_v(Sym[A \otimes A])(x)=0,$$ then either $A(x)=0$ or $\nabla_v A(x)=0.$
\end{corollary}

Here, we just note that at $x \in M,$ we can apply Corollary \ref{vanish2}, as symmetry of $A$ also gives symmetry of $\nabla_v A,$ and

$$\nabla_v (Sym[A \otimes A])=Sym[\nabla_v(A \otimes A)]=2 \cdot Sym(A \otimes [\nabla_v A]).$$

\begin{corollary}\label{vanish4}
Suppose that $M$ is a $C^2$ manifold with Kozul connection $\nabla$ and that $A$ is a $C^2$ never vanishing covariant symmetric tensor field on $M,$ of any rank.  If $\nabla(Sym[A \otimes A])$ vanishes everywhere on $M,$ then $\nabla A$ vanishes everywhere on $M.$
\end{corollary}

\begin{example}
Suppose that $M$ is a smooth semi-Riemannian manifold with metric tensor $g,$ and suppose that $u$ is a vector field  on $M$ with $g(u,u)=1.$  Let $u^g$ denote the 1-form dual to $u$ with respect to the metric tensor $g,$ and set
$$h=g-2(u^g \otimes u^g).$$
Then $h$ is a semi-Riemannian metric with signature changed to increase the negative signature.  Let $\nabla^g$ denote the Levi-Civita connection due to $g$ and $\nabla^h$ that due to $h.$  If we were to suppose that both these connections coincide, then denoting their common value as simply $\nabla,$ we would have $\nabla h=0$ and $\nabla g=0,$ and therefore $\nabla (u^g \otimes u^g)=0.$  By Corollary \ref{vanish4}, we must conclude that $\nabla u^g=0.$  But, since $[\nabla_v u^g]=[\nabla_v u]^g,$ for any vector field $v$ on $M,$ this gives $\nabla u=0.$  That is, the only way to modify the metric in this simple way so as to leave the connection unchanged is to use a parallel vector field, and obviously from these equations we see that if $u$ is a parallel vector field, then the Levi-Civita connection remains unchanged, as for example in the case of Minkowski space as formed by taking a parallel unit vector field on four dimensional Euclidean space.  If the vector field is not parallel, then the connections are different, the resulting Lorentz connection will therefore not have the same geodesics as Euclidean space.  And of course, if $M$ is any semi-Riemannian manifold, then this type of signature change produces a new Levi-Civita connection and thus changes some of the geodesics.
\end{example}

Our main application of the observer principle in infinite dimensions will be to vector spaces of smooth sections of $TM,$ so we would like to know that a symmetric multilinear map on the vector space of sections of $TM$ must vanish if its monomial form vanishes on all timelike vector fields.  Thus, if we know that the set of timelike vector fields equals its star, then as it equals its $r-$star for all $r \geq 0,$ then the observer principle applies by our Theorem
\ref{genanalcont}.  The next proposition solves this problem for the case of vector spaces of vector fields which are continuous over a given fixed compact subset of $M.$

\begin{theorem}\label{timelikestar}
Suppose that $M$ is any smooth manifold and that $E$ is a smooth semi-Riemannian vector bundle over $M.$ Suppose that $K$ is a compact subset of $M$ and $\Gamma$ is a vector subspace of the set of all continuous vector fields of the vector bundle $E|K.$  Let $U$ be the set of all vector fields $v$ in $\Gamma$ satisfying $g(v,v) < 0$ on $K.$  Then $$Star(U,\Gamma)=U.$$
\end{theorem}

The proof of this theorem is a simple application of one of the most useful, simple, and beautiful theorems in point-set topology which is due to A. D. Wallace (my first mathematical mentor).

\begin{theorem} ({\bf A. D.  Wallace.})  If $X$ and $Y$ are any topological spaces, if $A$ is a compact subset of $X$ and $B$ is a compact subset of $Y,$ and if $W$ is an open subset of $X\times Y$ which contains $A \times B,$ then there are open subsets $U$ of $X$ and $V$ of $Y,$ respectively, such that $A \subset U,~~B \subset Y,$ and $U \times V \subset W.$
\end{theorem}

It is customary to call a set of the form $A \times B$ a rectangle or to call it rectangular.  For a proof of Wallace's Theorem, we refer to \cite{JLKELLY}, but it is an elementary exercise in topology.  Also, it is elementary in topology that $U \times V$ is open if and only if both $U$ and $V$ are open, whereas only slightly less elementary is the fact that $A \times B$ is compact if and only if both $A$ and $B$ are compact.  An open set which contains $A$ is said to be an open neighborhood of $A.$  Thus Wallace's Theorem says simply every open neighborhood of a compact rectangle contains a rectangular open neighborhood of that compact rectangle.

To use Wallace's Theorem here, given $v_0 \in U,$ and any $w \in \Gamma,$ we define the real valued function $f_w:K \times \bR \lra \bR$ by

$$f_w(m,t)=g(v_0(m)+tw(m),v_0(m)+tw(m)),~~(m,t) \in K \times \bR.$$  Since all vector fields in $E$ are assumed continuous, it follows that $f$ is continuous.  Also, clearly

$$f_w(K \times \{0\}) \subset N,$$ where $N$ denotes the set of all negative real numbers.  Thus, as $N$ is an open subset of $\bR$ and $f$ is continuous, it follows that its inverse image $W=f_w^{-1}(N)$ is an open subset of $K \times \bR$ and hence is an open neighborhood of the compact rectangle $K \times \{0\} \subset K \times \bR.$  Keeping in mind that $K$ is an open subset of itself, by Wallace's Theorem, there is an open rectangle $K \times V_w$ with $$K \times \{0\} \subset K \times V_w \subset W.$$  Thus, $V_w$ is an open neighborhood of zero in $\bR$ so there is a positive number $\delta_w$ with the property that if $|t| < \delta_w,$ then $t \in V_w.$  This means that $v_0 +tw \in U,$ for any $t \in V_w,$ and in particular, for any $t$ with $|t| < \delta_w.$  As $w$ was arbitrary in $E,$ this means, $v_0 \in Star(U,\Gamma).$  As $v_0$ was an arbitrary vector field in the set $U,$ it follows that $Star(U,\Gamma)=U$ as claimed.

Combining Theorem \ref{genanalcont} and Theorem \ref{timelikestar}, we then have immediately the final result on the observer principle.

\begin{theorem}\label{timelikeobserverprinciple}
If $K$ is a compact subset of $M$ and if $S$ and $T$ are symmetric multilinear maps of rank $r$ on a vector space $E$ of continuous vector fields on $K$ with the property that $Sv^r=Tv^r$ for every timelike vector field $v \in E,$ then $S=T.$
\end{theorem}
We merely need to observe that $E$ must be non-degenerate. Indeed, if $v \in E$ with $v \neq 0,$ then there is some particular $m \in K$ with $v(m) \neq 0,$ and then we can choose any $\lambda \in (T_mM)^*$ with $\lambda(v(m)) \neq 0,$ to obtain an element $f \in E^*$ with $f(v) \neq 0,$ namely $f=\lambda [ev_m],$ where $ev_m :E \lra T_mM$ is the evaluation map, $ev_m(w)=w(m),$ for all $w \in E.$

As an application of Theorem \ref{timelikeobserverprinciple}, we will apply it to integrals of operators which are more general than tensor fields.  Let us call $S$ a {\it tensor operator} of rank $r$ on $K \subset M$ provided that it is a multilinear map of rank $r$ on the vector space of smooth vector fields on $K$ and whose values are also smooth scalar fields on $K$ which has the property that if $m \in K,$ and if $f$ is a smooth function on $K$ which is constant in an open neighborhood of $m,$ then

$$S(v_1,v_2,...,fv_l,...,v_r)(m)=f(m)S(v_1,v_2,...,v_r)(m).$$
 For instance, $S$ could be simply a tensor field of rank $r,$ but more generally, $S$ could be formed by covariant differentiation operators and tensor fields so as to be multilinear of rank $r,$ but not necessarily a tensor field of rank $r.$  Suppose that $\mu$ is a volume form on $M,$ and that $K \subset M.$   We can then define the scalar valued rank $r$ multilinear map $\hat{S}$ on the vector space $\Gamma_K$ of all smooth vector fields on $K$ by

\begin{equation}\label{integralS}
\hat{S}(w_1,w_2,...,w_r)=\int_K S(w_1,w_2,...,w_r)\mu, ~~~~w_1,w_2,...,w_r \in \Gamma_K.
\end{equation}
We will call $\hat{S}$ the integral of $S$ over $K,$ and denote it by

\begin{equation}\label{integraltensorfield}
\int_K S \mu =\hat{S},
\end{equation}
so we have

\begin{equation}\label{integraltensorfield2}
\int_K S \mu (w_1,w_2,...w_r)=\int_K S(w_1,w_2,...w_r) \mu,~~~w_1,w_2,...w_r \in \Gamma_K.
\end{equation}

Notice that the integral of such an operator is a multilinear map, which is therefore a special kind of tensor-not a tensor field.  By forming the multilinear map $\int S \mu,$ we in effect get around the problem that generally it does not make sense to integrate a tensor field itself, without choosing some kind of coordinate representation, as integrating components of a tensor does not give a tensor in the usual sense that physicists use the term.

It is clear that if $\int_U S \mu$ vanishes for every sufficiently small open subset of $U$ of $K,$ then $S$ itself must vanish.  For if $S(w_1,w_2,...,w_r)$ does not vanish at $m \in K,$ then it maintains its sign over some open neighborhood $U$ of $m,$ so the integral over $U$ would be either positive or negative but not zero.  To make the term sufficiently small precise here, we could say that if $\U$ is an open cover of $K$ and if $U$ is contained is some member of $\U,$ then $U$ is $\U-$small.  Then we say $U$ is sufficiently small if it is $\U-$small for some open cover $\U$ of $K.$

In any case, we see immediately that if $S$ is symmetric, then so is $\int S \mu.$  Thus if $S$ and $T$ are both symmetric, then to know that $S=T,$ by  Theorem \ref{timelikeobserverprinciple}, it is sufficient to know that the monomial forms of $\int S \mu$ and $\int T \mu$ agree on all sufficiently small open subsets of $K.$

Moreover, when $K$ is compact, if $S$ and $T$ are both symmetric rank $r$ tensor fields on $K,$ then to see that $\hat{S}=\hat{T},$ by Theorem \ref{timelikeobserverprinciple}, it suffices to know that $\hat{S}v^{(r)}=\hat{T}v^{(r)}$ for every timelike smooth vector field on $K.$  In particular, if $\hat{S}v^{(r)}=0$ for every timelike vector field $v$ on $K,$ then $\hat{S}=0.$  But, if $\hat{S}=0,$ then it follows that $S=0,$ since now we are assuming $S$ is an actual tensor field.  For if $S \neq 0,$  then we can choose a  function $f:K \lra [0,1] \subset \bR$ which vanishes outside a small open neighborhood $W_f$ of $m$ but with $f(m)=1.$  We can then choose smooth vector fields $w_1,w_2,...,w_r$ on $K$ with $S(m)(w_1(m),w_2(m),...,w_r(m)) \neq 0,$ so the function $S(w_1,w_2,...,w_r)$ does not vanish at $m.$  Then with $h=f^r:K \lra [0,1],$ we have $h(m)=1,$ and $h$ vanishes on $K \setminus W_f.$  Moreover,

\begin{equation}\label{bumpfunctionintegral}
0=\hat{S}(fw_1,fw_2,...,fw_r)=\int_K hS(w_1,w_2,...,w_r) \mu.
\end{equation}
If $S(w_1,w_2,...,w_r)$ has a positive value at $m,$ then after making $W_f$ smaller if necessary, we can assume $S(w_1,w_2,...,w_r)$ must have a positive value for all points of $W_f$ and this would mean the integral on the right hand side of equation (\ref{bumpfunctionintegral}) would not vanish which would be a contradiction.  If $S(w_1,w_2,..,w_r)$ is negative at $m,$ then replacing $S$ by $-S$ we reach the conclusion $S=0.$  We have therefore proven the following corollary.

\begin{corollary}\label{integraltensorfieldobserverprinciple}
Suppose that $K$ is a compact subset of $M$ and $S$ is a continuous symmetric tensor field on $K$ of rank $r.$  If the monomial form of the multlinear map $\int_K S \mu$ vanishes on all timelike vector fields, then $S=0.$
\end{corollary}

For instance, if $T$ is an energy momentum stress tensor field on the compact subset $K$ of $M$ and if $\int_K T(v,v) \mu=0$ for every continuous timelike vector field $v,$ then $T=0.$  Of course, the corollary can be proven using the above technique together with the special case of Theorem \ref{timelikeobserverprinciple} in finite dimensions, but our development makes it clear that their are more general applications possible where the purely finite dimensional argument would not suffice.  For instance, if $\lambda$ is a smooth 1-form and $S(u,v)=\lambda(\nabla_u v),$ then $S$ is a tensor operator which is not a tensor field.  If $\int S \mu$ vanishes in this case for all sufficiently small open sets, then this would mean that $\lambda(\nabla_u v)=0$ for all pairs of vector fields $u,v.$  If the monomial form of $\int S \mu$ vanishes for all sufficiently small open sets, then $Sym(S)=0,$ and therefore $\lambda(\nabla_u v)+\lambda(\nabla_v u)=0,$ for all pairs of vector fields $u,v.$

%%%%%%%%%%%%%%%%%%%%%%%%%%%%%%%%%%%%%%%%%%%%%%%%%%%%%%%%%%%%%%%%%%%%%%%%%%%%%%%%%%%%%%%%%%%%%%%%%%%%%%%%%%%%%%%%%%%%%%%%%%%%%%%%%%%%%%%%%%%%%%%%%%%%%%%%%%%%%%%%%%%%%%%%%%%%%%%%%%%%%%%%%%%%%%%%%%%%%%%%%%%%%%%%%%%%%%%%%%%%%%%%%%%%%%%%%%%%%%%%%%%%%%%%%%%%%%%%%%%%%%%%%%%%%%%%%%%%%%%%%%%%%%%%%%%%%%%%%%%%%%%%%%%%%%%%%%%%%%%%%%%%%%%%%%%%%%%%%%%%%%%%%%%%%%%%%%%%%%%%%%%%%%%%%%%%%%%%%%%%%%%%%%%%%%%%%%%%%%%%%%%%%%%%%%

%%%%%%%%%%%%%%%%%%%%%%%%%%%%%%%%%%%%%%%%%%%%%%%%%%%%%%%%%%%%%%%%%%%%%%%%%%%%%%%%%%%%%%%%%%%%%%%%%%%%%%%%%%
%%%%%%%%%%%%%%%%%%%%%%%%%%%%%%%%%%%%%%%%%%%%%%%%%%%%%%%%%%%%%%%%%%%%%%%%%%%%%%%%%%%%%%%%%%%%%%%%%%%%%%%%%%%
%%%%%%%%%%%%%%%%%%%%%%%%%%%%%%%%%%%%%%%%%%%%%%%%%%%%%%%%%%%%%%%%%%%%%%%%%%%%%%%%%%%%%%%%%%%%%%%%%%%%%%%%%%

\section{SEMIGROUPS, SEMIGROUP RINGS, GROUP RINGS AND ALGEBRAS}

It is useful to have a broader perspective on the the algebra underlying the {\bf Algebra of Curvature} in order to see how various curvature operators can be produced.  The key here is actually in a very elementary idea from algebra sometimes called {\bf Relative Inversion} or {\bf Pseudo Inversion} as in \cite{DUPREGLAZE2}, \cite{DUPREGLAZE3} and \cite{DUPGLAZEPREV}, which can be defined in any semigroup.

Recall that a {\bf Semigroup} is simply a set with a given associative multiplication.  In general, the multiplication is not commutative. If $S$ is a semigroup and there is no confusion as to the multiplication, $m: S \times S \lra S,$ then we simply write $m(x,y)=xy,$ for any $x,y \in S.$  We say $x \in S$  is {\bf Idempotent} if $x^2=x.$   We denote the set of all idempotents in $S$ by $P(S),$ although it is commonly denoted by $E)S$ in the semigroup literature.  Notice that if $x \in S$ is idempotent, then $xSx$ is a subsemigroup of $S,$ all members of $xSx$ must commute with $x,$ and  $x$ is an identity for $xSx.$  Thus, if $S=xSx,$ then $x$ is an identity for $S.$  Of course, a semigroup can have at most one identity, so we will usually denote it by 1 when no confusion results.  If $S$ has an identity then we say $g \in S$ is invertible provided that there is $h \in S$ with $gh=1=hg,$ in which case we call $h$ an {\bf Inverse} for $g.$  If $g$ is invertible, and $h, k \in S$ are both inverses for $g,$ then

$$h=h(gk)=(hg)k=k,$$
so $g$ can have at most one inverse, and it is denoted $g^{-1}.$  We denote by $G(S)$ the set of all invertible members of $S,$ which is then a {\bf Group} called the {\bf Group of Units} of $S.$  Thus we could define a group as simply a semigroup which equals its own group of units.  If $x \in S$ is any idempotent, then as it is an identity for the subsemigroup $xSx,$ so $G(xSx)$ is the group of units of $xSx,$ consisting of members $g \in xSx$ for which there is $h \in S$ with $gh=x=hg.$  For, $xhx$ belongs to $xSx$ and is an inverse for $g$ in $xSx,$ since $g=xg=gx=xgx,$ and therefore,

$$g(xhx)=x(gh)x=x^3=x=(xhx)g.$$
In any case, for $g$ invertible in $xSx,$ it is therefore natural to write $g^{-x}$ for the inverse of $g$ in $xSx.$

If $S$ is a semigroup, then $S^{op}$ denotes the {\bf Opposite Semigroup} obtained by simply reversing the order of multiplication.

Now, more generally, we say $u,v \in S$ are {\bf Relative Inverses}, or {\bf Muturally Relatively Inverse} or {\bf Pseudoinverses} or {\bf Mutually Pseudoinverse} if both equations

%%%%%%%%%%%%%%%%%%%%%%       DEFINITION relinvers0

\begin{equation}\label{relinverse0}
uvu = u \mbox{ and } vuv = v.
\end{equation}
When these two equations (\ref{relinverse0}) hold, then we say $u$ and $v$ are relatively invertible or pseudoinvertible and each is called a relative inverse or pseudoinverse of the other.  Thus, $u$ is  {\bf Relatively Invertible} or {\bf Pseudoinvertible} if there is some $v \in S$ for which $u$ and $v$ are mutually pseudoinverse.
If we start with only $u,x \in S$ such that $uxu=u,$ then $ux$ is idempotent, that is, $(ux)^2=ux,$ and setting $v=xux,$ we have $uvu= (ux)^2u=uxu=v,$ so $u$ is pseudoinvertible with pseudoinverse $v.$  In fact, it is easy to see now that 

\begin{equation}\label{relinverse1}
\mbox{For } w \in S \mbox{ to be pseudoinvertible if and only if  } w \in wSw.
\end{equation}

\medskip

We use $W(S)$ to denote the set of all pseudoinvertible members of $S.$

\medskip

If $S$ and $T$ are semigroups, then a map $h:S \lra T$ is a {\bf Semigroup Homomorphism} provided that $h(xy)=h(x)h(y),$ for all $x, y$ in $S.$  We denote the set of all semigroup homomorphisms of $S$ to $T,$ by $Hom(S,T).$  Obviously composition of semigroup homomorphisms are semigroup homomorphisms and the identity map on any semigroup is a semigroup homomorphism, for instance, $id(S),$ the identity map on $S$ is a semigroup homomorphism.   We say $h$ is a {\bf Semigroup Isomorphism} if there is $g \in Hom(T,S)$ with $gh = id(S).$ and $hg = id(T).$  We denote by 
$Iso(S,T)$ the set of isomorphisms of $S$ onto $T.$  We say $h \in Hom(S,S)$ is an {\bf Endomorphism} of $S.$  We define $Aut(S) = Iso(S,S),$ and say $h$ is an {\bf Automorphism} of $S$ if 
$h \in Aut(S).$ 

If $x \in G(S),$ then it is convenient to define $s^x = xsx^{-1}$ for any $s \in S.$
We define $Inn(x)  \in  Aut(S)$ by

$$[Inn(x)](s) = s^x,~s \in S,~x \in G(S).$$
Since

$$s^{(xy)}=[s^x]^y, ~x,y \in G(S),~ s \in S,$$
it follows that

$$Inn \in Hom(G(S),Aut(S))$$
 is a group homomorphism whose image we denote by $Inn(S) \subset Aut(S),$ called the {\bf Inner Automorphism Group} of $S.$

\medskip

Recall that a {\bf Ring} is simply an (additive) abelian group with an associative multiplication which satisfies the associative and distributive (with respect to addition) laws, but in general the multiplication need not be commutative.  Thus, the ring considered only in terms of its multiplication is a semigroup.  If $R$ and $S$ are rings, then $h:R \lra S$ is a {\bf Ring Homomorphism} provided that it is a semigroup homomorphism of both the additive and multiplicative structures.  Then, we have $Hom(R,S),~Iso(R,S),~End(S),~Aut(S),~Inn(S)$ defined where we now require the homomorphisms to preserve the additive as well as multiplicative structures.  Thus, in particular, if $x \in G(R),$ and $s \in R,$ then $s^x = xss^{-1}= [Inn(x)](s)$ is now making $Inn(x)$ automatically an inner automorphism of the ring. If $B$ is an abelian group, then the set of all group endomorphisms of $B$ is denoted $End(B),$ and is naturally a ring with the multiplication defined by composition of endomorphisms.  A {\bf Left Module} over a ring is an (additive) abelian group with a given ring homomorphism of the ring into the endomorphism ring of the abelian group.  Thus, if $R$ is a ring and $M$ is a left R-module, then $M$ is an abelian group with a given homomorphism $h: R \lra End(M).$  As long as $h$ is understood, we simply denote $rm=h(r)(m),$ for any $r \in R,$ and any $m \in M.$  If $M$ is in fact itself a ring, and $End(M)$ is the set of ring homomorphisms of $M$ to itself, then we would call $M$ a {\bf left $R-$algebra}.  If $R$ is a ring or algebra, then $R^{op}$ denotes the {\bf Opposite Ring} obtained by simply reversing the order of multiplication.  Thus, $M$ is a right $R-$module if there is a given ring homomorphism $h:R^{op} \lra End(M),$ in which case we write $mr=h(r)(m).$  If $R$ and $S$ are both rings, and if $M$ is both a left $R-$module and a right $S-$module, then we say it is {\b an $(R,S)-$bimodule} provided that $r(ms)=(rm)s,$ for any $r \in R, m \in M, s \in S.$  For instance, if $A$ is any ring and $R$ is any subring, then $A$ is naturally an $(R,R)-$bimodule, or simply, an $R-$bimodule. likewise, we can speak of a {\bf Bialgebra}. RecalI that   If $R$ has an identity, $1_R,$ then $b=h(1_R)$ is an idempotent in $End(M),$ so $b(M)=RM$ is a submodule, and with 1 now denoting the identity endomorphism of $M,$ we have $M=RM \oplus (1-b)(M),$ as $1-b$ is also an idempotent in $End(M),$ and $(1-b)(M)$ is an $R-$submodule of $M$ on which $R$ acts trivially.

In differential geometry, for any smooth manifold $X,$ the ring $R=C^{\infty}(X)$ of smooth real valued functions on $X$ acts on all the vector spaces of sections of the various tensor bundles over $X,$ so they are all $R-bimodules$ due to the commutativity of $R.$

If $S$ is any multiplicative semigroup and $R$ is any ring, we can form a new ring $R[S]$ consisting of formal linear combinations of members of $S$ with coefficients in $R$ and where we impose the condition that $rs=sr$ for any $r \in R, s \in S.$  To multiply two such linear combinations, just allow an unrestricted distributive law.  In effect, we are forming the {\bf Free module} generated by the set $S$ and using the multiplication in $S$ to naturally extend to a multiplication for $R[S].$  In this way, $R[S]$ becomes an 
$R-$bialgebra, called the {\bf Semigroup} $R-$algebra or ring.  We note here that $R[S]$ has the {\bf Universal Property} that if if $A$ is any $R$-algebra and $h:S \lra A$ is a homomorphism of $S$ into the underlying multiplicative semigroup of $A,$ then there is a unique extension of $h$ to a an $R$-algebra homomorphism $\bar{h}:R[S] \lra A.$   In particular, if $h:S \lra T$ is a homomorphism of semigroups, then it naturally extends to an $R-$bialgebra homomorphism

$$R[h]:R[S] \lra R[T], \mbox{ with } R[h](rs)=rh(s), r \in R, s \in S.$$ 
If we assume that $R$ has an identity, $1_R$, then we have an injective map of $S$ into $R[S]$ which sends each $s$ in $S$ to the element $1_Rs=s1_R$ of $R[S],$ so in this way, we can identify $S \subset R[S],$ by simply the rule $1_Rs=s=s1_R,$ for any $s \in S.$ If $S$ also has an identity, $1_S,$ then $1_R1_S$ is an identity for $R[S].$

In particular, If $S=G$ is a multiplicative group, then since $G$ and $G^{op}$ are the same as sets, the inversion map $I:G \lra G^{op}$ is a group isomorphism which induces an {\bf  Involution} of $R[G]$ called the {\bf  Adjoint or Star Operation}.  Thus, we write 

$$x^{*}=R[I](x), \mbox{ for } x \in R[G], \mbox{ so in particular, }~x^*=x^{-1} \mbox{for } x \in G \subset R[G],$$
and for $x$ in $R[G],$ we call $x^*$ the adjoint of $x.$

If $S$ is a semigroup, $V$ is a vector space, and $h: S \lra End(V)=L(V;V),$ is a homomorphism of the multiplicative semigroups, then we call $h$ a {\bf Representation} of $S$ on $V.$  By the universal property it then extends to a representation of $R[S]$ on $V,$ meaning of course an $R$-algebra homomorphism of $R[S]$ into $L(V;V).$  Then, we see for any idempotent $x \in P(S),$ that $h(x)$ is then a linear retraction on $V,$ and hence a projection onto a linear subspace of $V.$  If $u$ and $v$ are mutually pseudoinverse in $S,$ then $vu$ is called the {\bf Domain Projection} of $u,v$ and $uv$ is called the {\bf Range Projection} of $u,v.$  Setting $p=vu$ and $q=uv,$ we have 

\begin{equation}\label{relinverse2}
u=up=qu=qup \in qSp \mbox{ and } v=pu=vq=pvq \in pSq, \mbox{ thus } u \in (qSp) \cap (uSu).
\end{equation} 
Moreover, even though in general $u$ may have many pseudoinverses, the particular pseudoinverse $v$ is uniquely determined by the triple $(q,u,p)$ and we denote this by writing

\begin{equation}\label{relinverse3}
v = u^{-(p,q)} \mbox{ if } vu=p,~ uv=q, ~ u = qup, \mbox{ in which case, $v$ is a pseudoinverse for } u.
\end{equation}
For, notice that $qu=v$ and therefore $qu=v=vp=qvp.$
We then see that for the representation $h$ of $S$ on $V,$ that $h(p)V$ and $h(q)V$ are linear subspaces of $V$ and $h(u)$ carries $h(p)V$ isomorphically onto $h(q)V$ and is zero on the complementary subspace $h(1-p)V.$  For this reason, in algebra, $h(u)$ is called a partial isomorphism and likewise, $u$ can be called a {\bf Partial Isomorphism}.  Then $h(v)$ is the inverse partial isomorphism to $h(u)$ which is zero on $h(1-q)V$ and carries $h(q)V$ isomorphically onto the subspace $h(p)V$ inverse to that defined by $h(u).$

The main example of this phenomena of our interest is in the case of curvature algebra, where we will see that the constructions of Jacobi curvature operators from Riemann curvature operators is accomplished with a computable pair of mutually pseudoinverse operations in the the group ring of the appropriate permutation group.

Notice that if $V$ is a real vector space and $h:G \lra End(V)=L(V;V)$ is a representation of the group $G,$ which of course is simply a group homomorphism of $G$ into the group of units of $L(V;V),$ then it extends uniquely to a ring homomorphism of $R[G]$ into $End(V),$ thus giving a representation of $R[G]$ on $V.$  If $V$ has a non-degenerate inner product, then it gives $End(V)$ a natural adjoint operation or star operation, and if $h$ is a star-homomorphism on $G,$ then the resulting representation of $R[G]$ is a star-homomorphism and thus a star-representation of $R[G]$ on $V.$  Thus, in this case, each member of $G$ acts as an isometric linear map of $V$ relative to the inner product.

Obviously, if $H$ is a subgroup of $G,$ then $R[H]$ is a star-sub-algebra of $R[G].$  If $H$ is a finite subgroup, then $R[H]$ is finitely generated as a $R-$algebra.  As $H \subset G \subset R[G],$ we can form the finite sum

\begin{equation}\label{cycleH}
C_H=\sum_{h \in H} h \in R[H] \subset R[G].
\end{equation}
Notice then obviously, as $gH=H=Hg,$ for every $g \in H,$ we have

$$g(C_H)=C_H=(C_H)g, \mbox{ for all } g \in H,$$ 
and therfore

\begin{equation}\label{cycleHg}
g^k(C_H)=C_H=(C_H)g^k \mbox{ for all } g \in H \mbox{ for any integer }k.
\end{equation}
In addition, since $H=H^{-1},$ we must have

\begin{equation}\label{cycleHstar}
C_H=(C_H)^*,
\end{equation}
so $C_H$ is self-adjoint.  If $n$ is the cardinality of $H,$ it now follows from (\ref{cycleHg}) that

\begin{equation}\label{eigencycle}
C_H^2=nC_H, 
\end{equation}
and therefore

\begin{equation}\label{cycleproject}
Q_H=\frac{1}{n}C_H \mbox{ satisfies } Q_H^*=Q_H=Q_H^2 \mbox{ and } Q_HC_H=C_H=C_HQ_H.
\end{equation}
If $M$ is a left $R[G]-$module, then for any $x \in M,$ we have that $m=Q_Hx \in M$ satisfies $C_Hm=m.$  If $m \in M$ satisfies $C_Hm=m,$ then we say $m$ is {\bf $H$-cyclic}.  Thus for each $m \in M, ~Q_Hm$ is $H$-cyclic.  If $m \in M$ and $C_Hm=0,$ then we say that $m$ is {\bf $H$-acyclic}.   Define

$$Q^{-}_H=1-Q_H.$$
Then

$$[Q^{-}_H]^2=Q^{-}_H=[Q^{-}_H]^*,~~       Q_H+Q^{-}_H=1, ~~Q_HQ^{-}_H=0=Q^{-}_HQ_H.$$
If $x \in M,$ and if we set $m=Q^{-}_H x,$ then $m$ is $H$-acyclic, and we see that each member of $M$ decomposes uniquely as a sum of cyclic and acyclic members of $M.$  That is, $Q_H M$ is the submodule of $H$-cyclic members of $M,$ and $Q^{-}_H M$ is the submodule of $H$-acyclic members of $M,$ so we have the direct sum decomposition of $M$ itself:

\begin{equation}\label{cycledecomp}
M=Q_H M \oplus  Q^{-}_H M.
\end{equation}

Suppose that $g \in G$ has order $n,$ so $g$ generates a subgroup of $G$ of cardinality (or order) $n.$  Thus, $g^n=1, $ but $g^k \neq 1$ if $k$ is a positive integer less than 1.  Thus the subgroup $H= \langle g \rangle$  of $G$ generated by $g$ is {\bf Cyclic} and has order $n,$ so everything we have just done applies here to $C_{\langle g \rangle}.$

When the group $G$ acts on a set $X,$ for each $x \in X,$ we call $Gx$ the {\bf Orbit} of $x.$  The set of all orbits is denoted $X/G$ and called the orbit set.  If $H \subset G$ and $W \subset X$ with $HW \subset W,$ then we say $W$ is {\bf $H$-invariant}.  Notice that if $W$ is $H$-invariant, so is $X \setminus W.$  If $g \in G$ and $gx=x,$ (so $x$ is $g$-invariant), we call $x$ a {\bf fixed point} of $g$ and if the orbit of $x$ is just the one point set consisting of $x$ itself, then $x$ is a fixed point of the $G$-action.  Obviously, $X/G$ is a set of subsets of $X$ forming a {\bf Partition} of $X.$  If $x \in G,$ then we denote by $G_x$ the set of members of $G$ which fix $x,$ and easily see that $G_x$ is a subgroup of $G.$  Let $F(g,X)$ be the subset of $X$ consisting of points of $X$ that $g$ fixes.  Obviously,  $F(g,X)$ is $g$-invariant, hence so is $X \setminus F(g,X),$ and we call this the {\bf Support} of $g.$  Let $G_f \subset G,$ denote the set of all $g \in G$ for which $X \setminus F(g,X)$ is finite, that is, the set of all elements having finite support.  Notice for $g \in G,$ that $g$ and $g^{-1}$ have the same support, and for $g,h \in G,$ if the support of $g$ is $A \subset X$ and the support of $h$ is $B \subset X,$ then the support of $gh$ is contained in $A \cup B.$   Thus, $G_f$ is a subgroup of $G.$

%%%%%%%%%%%%%%%%%%%%%%%%%%%%%%%%%%%%%%%%%%%%%%%%%%%%%%%%%%%%%%%%%%%%%%%%%%%%%%%%%%%%%%
%%%%%%%%%%%%%%%%%%%%%%%%%%%%%%%%%%%%%%%%%%%%%%%%%%%%%%%%%%%%%%%%%%%%%%%%%%%%%%%%%%%%%%%

%%%%%%%%%%%%%%%%%%%%%%%%%%%%%%%%%%%%%%%%%%%%%%%%%%%%%%%%%%%%%%%%%%%%%%%%%%%%%%%%%%%%%%%%
%%%%%%%%%%%%%%%%%%%%%%%%%%%%%%%%%%%%%%%%%%%%%%%%%%%%%%%%%%%%%%%%%%%%%%%%%%%%%%%%%%%%%%%

%%%%%%%%%%%%%%%%%%%%%%%%%%%%%%%%%%%%%%%%%%%%%%%%%%%%%%%%%%%%%%%%%%%%%%%%%%%%%%%%%%%%%%%%%%%%%%%%%%%%%%%%%%%%%%%%%%%%%%%%%%%%%%%%%%%%%%%%%%%%%%%%%%%%%%%%%%%%%%%%%%%%%%%%%%%%%%%%%%%%%%%%%%%%%%%%%%%%%%%%%%%%%%%%%%%%%%%%%%%%%%%%%%%%%%%%%%%%%%%%%%%%%%%%%%%%%%%%%%%%%%%%%%%%%%%%%%%%%%%%%%%%%%%%%%%%%%%%%%%%%%%%%%%%%%%%%%%%%%%%%%%%%%%%%%%%%%%%%%%%%%%%%%%%%%%%%%%%%%%%%%%%%%%%%%%%%%%%%%%%%%%%%%%%%%%%%%%%%

\section{CURVATURE ALGEBRA}

Of particular interest for dealing with curvature type identities is the case where $G$ is a permutation group.  To define the relevant notation, it is simplest to deal with permutations of the set of all natural numbers.

As usual, $\bN$ denotes the set of positive natural whole numbers.  Let $G=P(\bN)$ denote the group of all bijective self maps of $\bN,$ so $G$ acts on $X=\bN,$ and we call these maps {\bf Permutations}.  We let $P_f(\bN)=G_f$ denote the set of permutations $\sigma$ of $\bN$ which have finite support.  Notice that $\sigma \in G,$ has finite support if and only if it has the property that for some positive $n$ it is the case that $\sigma(x)=x $ for all $x \in \bN$ with $x > n.$  We denote by $P(n)$ the subgroup of $G_f$ consisting of those permutations for which the support is contained in the subset $\{1,2,3,...,n\}.$   Thus, $P_f(\bN) \subset P(\bN)$ is a subgroup of the group $P(\bN)$, containing $P(n)$ as a subgroup for each $n \in \bN$ and $P(1) \subset P(2) \subset... \subset P(n) \subset P(n+1) \subset ...P_f(\bN)$ is a tower of subgroups whose union is $P_f(\bN).$  

Denote $\{1,2,3,...,n\}$ by $\bN_n,$ so $\bN_n$ is $P(n)$-invariant.  If $\sigma \in P(n),$ we say $\sigma$ is a {\bf Cycle} of {\bf Length} $k$ provided that there is a $k-$tuple $(x_1,x_2,x_3,...,x_k),$ of distinct positive integers in $\bN_n$ with $\sigma(x_k)=x_1$ and 
$\sigma(x_j)=x_{j+1},$ for any $j<k.$  In this case we will write

\begin{equation}\label{perm2}
\sigma=[x_1,x_2,x_3,...,x_k], \mbox{ so also } \sigma=[x_k,x_1,x_2,x_3,...,x_{k-1}], 
\end{equation}
and so on, which is the obvious fact that the entries can be cyclically permuted without changing $\sigma.$  Obviously, $\sigma^k=1$ in this case, so $\sigma$ generates a cyclic subgroup of order $k.$  A cycle of length of length 2 is called a {\bf Transposition}.  We say that $\alpha$ and $\beta$ in $G$ are {\bf Disjoint} if their supports are disjoint subsets of $\bN.$  Notice that disjoint permutations commute.  If $\sigma$ is any member of $P(n),$ then the cyclic subgroup $H$ generated by $\sigma$ acts on $\bN_n$ and the orbit set $\bN_n / H$ gives the set of disjoint cycles whose product is $\sigma,$ showing each permutation is a product of disjoint cycles.  That is to say, when $P(n)$ acts on 
$\{1,2,3,...,n\}$ for $\sigma \in P(n),$ the cyclic subgroup generated by $\sigma$ gives an orbit decomposition which then gives the unique decomposition of $\sigma$ itself as a product of disjoint cycles, the order being immaterial as disjoint cycles commute.
Also, if $\sigma^2=1,$ then each orbit is at most a two element set and therefore $\sigma$ is a product of disjoint transpositions.  In particular, we can notice that if $\sigma$ is a cycle of length $k$ and $k$ is odd, then $\sigma^2$ is also a cycle of length $k,$ whereas if $k=2m,$ then 
$[(\sigma)^m]^2=1,$ so $\sigma^m$ is a product of disjoint transpositions.  For example,

$$[1,2,3,4]^2=[1,3][2,4].$$  

For computations it is useful to note that if $\sigma$ and $\lambda$ are cycles which cyclically permute the sets $A,B \subset \bN$ and if 
$A \cap B = \{c\}$ is a singleton, then we can write $\sigma=[x_1,x_2...,c]$ and $\lambda=[c,y_2,y_3,...]$ and check that

\begin{equation}\label{perm3}
\sigma \lambda= [x_1,x_2...,c,y_2,y_3,...],
\end{equation}
which in particular shows every cycle is a product of transpositions.  on the other hand, if $A$ and $B$ are disjoint, then $\sigma$ and $\lambda$ commute.  If $\sigma, \lambda$ are arbitrary members of $P(n),$ then conjugation of $\lambda$ by $\sigma$ results in $\sigma \lambda \sigma^{-1}$ which is the permutation sending $b$ to $c$ if $\sigma(k)=b$ and $c=\sigma(\lambda(k))=c.$  Thus for conjugation of a cycle we have

\begin{equation}\label{perm4}
\sigma [x_1,x_2,x_3,...,x_k] \sigma^{-1}=[\sigma(x_1),\sigma(x_2),...,\sigma(x_k)].
\end{equation}
In particular, if $\sigma = [a,b]$ is simply a transposition, then $\sigma^2=1,$ so $\sigma^{-1}=\sigma,$ and therefore,

\begin{equation}\label{perm5}
\sigma [x_1,x_2,x_3,...,x_k]\sigma=[\sigma(x_1),\sigma(x_2),...,\sigma(x_k)], \mbox{ for } \sigma \mbox{ a transposition.}
\end{equation}
For example,

\begin{equation}\label{perm6}
[1,2][1,3][1,2]=[2,3]=[1,3][1,2][1,3].
\end{equation}
More generally, for transpositions, if $\alpha$ and $\beta$ are transpositions, then either (1) they are disjoint and commute, or (2) they are identical, or (3) their supports have exactly one element in common, in which case, we can find 3 elements $u,v,w$ of $\bN$ with

$$\alpha=[u,v],~\beta=[v,w] \mbox{ so then } \alpha \beta=[u,v,w]$$
and then by (\ref{perm5}) we get the

\begin{equation}\label{transflip}
\mbox{TRANSPOSITION FLIP:  }\alpha \beta \alpha =[u,w]= \beta \alpha \beta.
\end{equation}
Thus the three possibilities are (1) $\alpha=\beta,$ or (2) $\alpha \beta= \beta \alpha,$ or (3) $ \alpha \beta \alpha = \beta \alpha \beta,$ and of course (2) and (3) together are equivalent to (1).

Let us define

$$\A_P=\bR[P_f(\bN)] \mbox{ and } \A_P(n)=\A_{P(n)}=\bR[P(n)], ~n \mbox{ a positive integer, so  } \A_{P(n)} \subset \A_P. $$
In addition, keep in mind that $\A_P$ is a star-algebra with the adjoint operation the extension of inversion in the group to the whole algebra.  In particular, for each $n \in \bN$ the group algebra $\A_{P(n)}$ is a star-subalgebra of $\A_P.$

Let $sgn(\sigma)$ denote the sign of the permutation $\sigma,$ defined as $+1$ in case $\sigma$ is an even permutation and $-1$ otherwise, so $sgn$ is a group homomorphism of the permutation group $P(r)$ into the multiplicative group $\{-1,1\}.$ 
Many computations with these permutations acting as operators on various vector spaces of tensors can be more readily carried out in the group algebra 
$\A_P=\bR[P_f(\bN)],$ over $\bR$ generated by $P_f(\bN),$ for instance, 

\begin{equation}\label{altr}
Alt_r \in \A_P(r)  \mbox{ is } Alt_r=\frac{1}{r!}\sum_{\sigma \in P(r)} sgn(\sigma) \sigma.
\end{equation}
and

\begin{equation}\label{symr}
Sym_r \in \A_P(r)  \mbox{ is } Sym_r=\frac{1}{r!}\sum_{\sigma \in P(r)} \sigma.
\end{equation}
Thus, as $P(r)$ is a subgroup of $G$ of order $r!,$ it follows that

$$Sym_r=C_{P(r)} / (r!)=Q_{P(r)},$$
a self adjoint projection in $\A_P(r).$  Likewise, since $sgn(\sigma^{-1})=sgn(\sigma) $ for every $\sigma \in P(r),$ it follows that $Alt_r$ is self-adjoint in $\A_P(r).$  As $sgn:P(r) \lra \{-1,1\}$ is a group homomorphism, it follows that also $Alt_r$ is idempotent, and thus is a self adjoint projection in $\A_P(r).$  

The natural representations of $\A_P$ on various tensor algebras preserve the star operation, if there is a non-degenerate inner product on the underlying vector space of the tensor algebra so in particular, the image of any self adjoint idempotent is a self adjoint projection operator on the particular vector space of tensors on which the permutation group is acting via permutation of arguments in the tensors.  Since any transposition (cycle of length 2) has square equal to  1, it is therefore its own inverse, so is a  self adjoint projection, so 
it follows any real linear combination in $\A_P$ of transpositions is a self-adjoint member of $\A_P.$

To be more specific, for any permutation $\sigma$ of $\{1,2,...,r\},$ and any tensor $A$ of rank $r,$ on the vector space $E$ with values in the vector space $F,$ let $\sigma(A)=\sigma A$ denote the result of permuting the arguments of $A$ via $\sigma,$ so

\begin{equation}\label{perm0}
[\sigma(A)](v_1,v_2,...,v_r)=A(v_{\sigma(1)},v_{\sigma(2)},...,v_{\sigma(r)}).  
\end{equation}

In this way, the permutation $\sigma$ defines a linear isomorphism

\begin{equation}\label{perm1}
\sigma: L^r(E;F) \lra L^r(E;F), 
\end{equation}
and therefore we have here a linear action or representation of the permutation group $P(r)$ of all permutations of $\{1,2,3,...,r\}$  on the vector space $L^r(E;F),$ which then makes $L^r(E;F)$ a left $\A_P(r)$-module.

Thus, the {\it alternation operator}, $Alt_r=Alt,$ is defined by

\begin{equation}\label{alternationoperator}
Alt_r(A)=\frac{1}{r!}\sum_{\sigma \in P(r)} sgn(\sigma)[\sigma(A)],
\end{equation}
whereas the {\it symmetrization operator}, $Sym_r=Sym,$ is defined by

\begin{equation}\label{symmetrizationoperator}
Sym_r(A)=\frac{1}{r!}\sum_{\sigma \in P(r)}  [\sigma(A)].
\end{equation}

Then $Alt \circ Alt=Alt$ and $Sym \circ Sym=Sym,$ so $Alt$ is a projection of $L^r(E;F)$ onto the subspace $L^r_{alt}(E;F)$ of alternating $F-$valued tensors on $E$ and $Sym$ is a projection of $L^r(E;F)$ onto the subspace $L^r_{sym}(E;F)$ of symmetric 
$F-$valued tensors on $E.$

Let $Cycl(r)$ denote the cyclic subgroup of $P(r)$ generated by the "upward shift" cycle $\sigma_r$ defined by

\begin{equation}\label{cyclr1}
\sigma_r=[1,2,3,...r].  
\end{equation}
We can then define the {\it cycling operator}, $C_r=Cycl_r=Cycl,$ by

\begin{equation}\label{cycle operator}
C_r=Cycl_r=C_{Cycle(r)}=\sum_{\sigma \in Cycl(r)} \sigma,
\end{equation}
so we can note that in $\A_P(r),$ we have also

\begin{equation}\label{cycle1}
C_r=Cycl_r=\sum_{k=0}^{r-1} (\sigma_r)^k.
\end{equation}
Then by (\ref{cycleproject}) we know $C_r/r$ is a self adjoint projection onto the cyclic members of any $\A_P(r)$-module whose complementary projection projects onto the acyclic members. 

In general, any transposition $\tau$ is an odd permutation.  Then in case $r=3,$ the full permutation group consists merely of the cyclic subgroup of order 3 generated by a non-trivial cycle, say $\sigma_3$ of (\ref{cyclr1}), which is an even permutation as are its powers, and its coset of $\tau$ consisting of the three odd permutations.  Thus, we have

\begin{equation}\label{altcycle1}
Alt_3=\frac{1}{3!}[Cycl_3-(Cycl_3)\tau]. 
\end{equation}
In particular, we have then for any $\A_P(3)$-module $M$ and any $A \in M$ with $\tau A=-A,$

\begin{equation}\label{altcycle2}
Alt_3(A)=\frac{1}{3}Cycl_3(A), \mbox{ if }\tau(A)=-A. 
\end{equation}
Obviously, (\ref{altcycle1}) and (\ref{altcycle2}) are relevant to the Bianchi identities in curvature.
For members $x,y$ of the ring $R$ we use $[x,y]=xy-yx$ for the commutator of $x$ and $y.$

\bigskip

In the following computations, throughout the remainder of this section, for brevity, unless otherwise specified, we use $C=C_3,~c=\sigma_3$ and because of he repeated use of certain special permutations, we find it convenient to define

$$\tau=[1,2]=\sigma_2,$$
$$\alpha=[3,4],$$
$$\nu=[1,3],$$
$$\beta=[2,4]$$
$$\chi=\beta \nu =\nu  \beta=[1,3][2,4],$$
$$\gamma=[2,3],$$
$$\lambda=[1,4]$$
$$\sigma_3=[1,2,3],$$
$$C=C_3=Cycl_3=1+\sigma_3+\sigma_3^2=1+\sigma_3 + \sigma_3^{-1}=1+c+c^2.$$
Thus, in $P(4)$ there are only six transpositions, and we have given each a symbol.

\bigskip

Then, using (\ref{perm3}), (\ref{perm4}), (\ref{perm5},  (\ref{perm6}), and  (\ref{transflip}), the transposition flip, we find:

%%%%%%%%%%%%%%%%%%%%%%%%%%%%%%%%%%%%%%%%%%%%%%%%%%%%%%%

%%                   PROPOSITION

%%%%%%%%%%%%%%%%%%%%%%%%%%%%%%%%%%%%%%%%%%%%%%%%%%%%%%%

\begin{proposition}\label{permcalc}

\begin{equation}\label{cycl1}
 c^3=1, ~c^2=c^{-1},
\end{equation}
as
$$[1,2,3]^3=1,  \mbox{ so }[1,2,3]^2=[1,2,3]^{-1}=[3,2,1],$$
\begin{equation}\label{cycl1a}
c=\nu \tau,
\end{equation}
as
$$[1,2,3]=[3,1,2]=[3,1][1,2]=[1,3][1,2],$$
\begin{equation}\label{cycl1b}
c^2=\tau \nu
\end{equation}
as
$$[1,2,3]^2=[1,3,2]=[2,1,3]=[2,1][1,3]=[1,2][1,3],$$
\begin{equation}\label{cycl2}
\gamma \tau \gamma= \tau \gamma \tau=c \tau=\nu
\end{equation}
as
$$[1,2,3][1,2]=[3,1,2][1,2]=[3,1][1,2][1,2]=[1,3],$$

\begin{equation}\label{cycl3}
 c \nu=\gamma=\tau \nu \tau=\nu \tau \nu
\end{equation}
as
$$[1,2,3][1,3]=[3,1,2][1,3]=[1,3][1,2][1,3]=[2,3]=[1,2][1,3][1,2],$$
and
$$[1,2,3]^2[1,2]=[1,2,3][1,3]=[2,3] \mbox{ and } [1,2,3]^2[1,3]=[1,2,3][[2,3]=[1,2][2,3][2,3]=[1,2],$$
or

\begin{equation}\label{cycl4}
c^2 \tau=c \nu=\tau \gamma \nu=\gamma \mbox{ and } c^2 \nu=c \gamma= \tau.
\end{equation}

\end{proposition}

\medskip

%%%%%%%%%%%%%%%   PROPOSITION  CYCLPROP1

\begin{proposition}\label{cyclprop1}

In $\A_P,$

\begin{equation}\label{cycl5}
Cycl_3=1+[1,3][1,2]+[1,2][1,3]  \mbox{ or } C=1+\nu \tau + \tau \nu
\end{equation}

\begin{equation}\label{cycl6}
[1,2](Cycl_3)=(Cycl_3 )[1,2]=[1,2]+[1,3]+[2,3]=[1,3](Cycl_3)=(Cycl_3) [1,3]
\end{equation}
or
$$[C,\tau]=0=[C,\nu] \mbox{ and } C \tau=\tau + \nu + \gamma =C \nu.$$
In particular,

$$C \tau =\tau C=\nu C=C \nu.$$

\end{proposition}

\begin{proof}
We use (\ref{cycl1a}) and (\ref{cycl1b}) to get (\ref{cycl5}) and then use (\ref{cycl5}) and (\ref{perm6}) to get (\ref{cycl6}) or alternately, use (\ref{cycl2}),(\ref{cycl3}), and (\ref{cycl4}) to get (\ref{cycl6}).
\end{proof}

If $x \in \A_P$ and $x^2=1,$ then define

$$Q^{\pm}_x=\frac{1}{2}(1 \pm x),$$
so then,

$$x Q^{\pm}_x =  Q^{\pm}_x x= \pm  Q^{\pm}_x,$$

$$x=Q^{+}_x - Q^{-}_x, ~Q^{+}_x+Q^{-}_x=1, ~[Q^{\pm}_x]^2=Q^{\pm}_x, \mbox{ and } Q^{+}_xQ^{-}_x=0=Q^{-}_xQ^{+}_x,$$
and if $x=x^*,$ then $Q^{+}_x$  and  $Q^{-}_x$ are both self adjoint in $\A_P.$  Thus $Q^{+}_x$ and $Q^{-}_x$ give a spectral decomposition of $x$ in $\A_P.$  Also, if $x$ is self-adjoint, then so are $Q^{\pm}_x,$ and if in addition we have $y \in \A_P$ which commutes with $x,$ then $y$ also commutes with $Q^{\pm}_x.$  Thus, $Q^{\pm}_{\tau}$ and $Q^{\pm}_{\nu},$ BOTH commute with $C_3=Cycl_3.$

Also, we note in particular, for transpositions $\phi$ and $\theta,$

$$ \phi Q^{\pm}_{\theta} \phi =\theta Q^{\pm}_{\phi} \theta, \mbox{ if } \phi \theta \phi = \theta \phi \theta,$$
which is the case whenever we apply the transposition flip, as in, for example,

$$\nu Q^{\pm}_{\tau} \nu = \tau Q^{\pm}_{\nu} \tau.$$

Now, in $\A_P(3) \subset \A_P$ we can define the elements $R_1$ and $J$ by

%%%%%%%%%%%%%%       DEFINITION OF  J AND R1

\begin{definition}\label{JANDR}
$$J=Q^+_{\tau} \nu Q^{-}_{\tau} \mbox{ and } R_1=Q^{-}_{\tau} \nu Q^+_{\tau}.$$
\end{definition}
Then obviously, 

$$J^*=R_1,~R^*_1=J,~J^2=0 \mbox{  and  } R_1^2=0.$$
In addition, using the fact that $\nu C_3=C_3 \nu=C_3 \tau = \tau C_3,$ by (\ref{cycl5}), we have

\begin{equation}\label{acyclicJ}
J C_3=C_3J=Q^{+}_{\tau} C_3 \nu Q^{-}_{\tau}=Q^{+}_{\tau} \tau C_3   Q^{-}_{\tau}= Q^{+}_{\tau}  Q^{-}_{\tau} C_3=0,
\end{equation}
and

\begin{equation}\label{acyclicR}
R_1 C_3=C_3R_1=Q^{-}_{\tau} C_3 \nu Q^{+}_{\tau}=Q^{-}_{\tau}  C_3  \tau  Q^{+}_{\tau}= Q^{-}_{\tau}  Q^{+}_{\tau} C_3=0.
\end{equation}
In particular, if $M$ is a left $\A_P(3)$-module and $m \in M,$ then $Jm$ and $R_1 m$ are acyclic.  We have therefore proved:

%%%%%%%%%%%%%          PROPOSITION   ACYCLIC1

\begin{proposition}\label{acyclic1}
$$J^2=0, ~R_1^2=0,~J C_3=C_3 J =0\mbox{ and } R_1 C_3=C_3 R_1=0.$$
\end{proposition} 

\begin{proof}
Apply (\ref{acyclicJ}) and (\ref{acyclicR}).
\end{proof}

Next, we calculate $J R_1$ and $R_1 J.$  Using the transposition flip, and $\tau Q^{\pm}_{\tau}= \pm Q^{\pm}_{\tau},$

$$J R_1=Q^{+}_{\tau} \nu Q^{-}_{\tau} Q^{-}_{\tau} \nu Q^{+}_{\tau}=Q^{+}_{\tau} \nu Q^{-}_{\tau}  \nu Q^{+}_{\tau}=Q^{+}_{\tau} \tau Q^{-}_{\nu} \tau Q^{+}_{\tau}=Q^{+}_{\tau} Q^{-}_{\nu} Q^{+}_{\tau},$$
so

$$4J R_1=4 Q^{+}_{\tau} Q^{-}_{\nu} Q^{+}_{\tau}=[(1+\tau)(1-\nu)]Q^{+}_{\tau}=[1+ \tau -\nu -\tau \nu]Q^{+}_{\tau}$$$$=2Q^{+}_{\tau} - \nu Q^{+}_{\tau} - \tau \nu Q^{+}_{\tau} =2Q^{+}_{\tau} -\nu \tau Q^{+}_{\tau} -\tau \nu Q^{+}_{\tau}  $$
$$=2Q^{+}_{\tau} -[\nu \tau +\tau \nu]Q^{+}_{\tau} =2Q^{+}_{\tau} -[C_3 -1]Q^{+}_{\tau}=Q^{+}_{\tau}[3 - C_3].$$
On the other hand,

$$R_1 J=Q^{-}_{\tau} \nu Q^{+}_{\tau} Q^{+}_{\tau} \nu Q^{-}_{\tau}=Q^{-}_{\tau}\nu Q^{+}_{\tau} \nu Q^{-}_{\tau} =
 Q^{-}_{\tau} \tau Q^{+}_{\nu} \tau Q^{-}_{\tau}=Q^{-}_{\tau} Q^{+}_{\nu} Q^{-}_{\tau},$$
 so
 
 $$4R_1J=[(1-\tau)(1+\nu)]  Q^{-}_{\tau}= [1+\nu-\tau -\tau \nu]Q^{-}_{\tau}=2Q^{-}_{\tau}-[\nu \tau +\tau \nu]Q^{-}_{\tau}            $$
 $$=2Q^{-}_{\tau}-[C_3-1] Q^{-}_{\tau}=3Q^{-}_{\tau}-Q^{+}_{\tau} C_3=Q^{-}_{\tau}[3-C_3].$$
 This gives
 
 %%%%%%%%%%%%            PROPOSITION  JR1
 
 {\begin{proposition}\label{JR1}
 $$(4/3)JR_1=Q^{+}_{\tau}[1-(C_3/3)]=Q^{+}_{\tau}Q^{-}_{Cycl(3)} \mbox{ and } (4/3) R_1J=Q^{-}_{\tau}[1-(C_3/3)]=Q^{-}_{\tau}Q^{-}_{Cycl(3)}.$$
 \end{proposition}
 
 \begin{proof}
 Recall that $C_3/3$ is the self adjoint projection defined by the cyclic subgroup $Cycl(3)$ of $P(3),$ so 
 
 $$(C_3/3)=Q^{+}_{Cycl(3)} \mbox{ and therefore  } [1-(C_3/3)]=Q^{-}_{Cycl(3)}.$$
 
 \end{proof}
 
 Since $C_3 x=0$ is equivalent to $Q^{-}_{Cycl(3)}x=x,$ for any $x \in \A_P,$ combining Propositions \ref{acyclic1} and \ref{JR1} gives
 
\begin{corollary}

$$(4/3)JR_1J=J \mbox{  and  } (4/3)R_1 J R_1=R_1.$$

\end{corollary}

\begin{proof}
Just keep in mind, also, that $Q^{+}_{\tau}J=J$ and $Q^{-}_{\tau}R_1=R_1.$

\end{proof}

At this point, in order to somewhat hide the factor $(4/3),$ we make the choice to define $R \in \A_P(3)$ by

\begin{equation}\label{R}
R=(4/3)R_1=(4/3)Q^{-}_{\tau} \nu Q^{+}_{\tau}.
\end{equation} 
We then have

%%%%%%%%%%  COROLLARY JRJ

\begin{corollary}\label{JRJ}
$$J^2=0,~R^2=0,~J^*=(3/4)R,~R^*=(4/3)J,~JRJ=J \mbox{  and  } RJR=R.$$
Also,

$$C_3J=0=JC_3, ~C_3R=0=RC_3,~Q^{+}_{\tau} J = J = J Q^{-}_{\tau}, \mbox{ and } ~Q^{-}_{\tau} R = R = R Q^{+}_{\tau}.$$

\end{corollary}

\bigskip

Now define $Q_J$ and $Q_R$ in $\A_P(3)$ by

%%%%%%%%%%%%%%%%  DEFINITION QJ1

\begin{definition}\label{QJ1}
$$Q_J=JR \mbox{ and } Q_R=RJ.$$
\end{definition}

%%%%%%%%%%%%%%  COROLLARY  QJ2

\begin{corollary}\label{QJ2}

$$Q_J J=J=J Q_R,~Q_R R=R=RQ_J,~Q_J^2=Q_J=Q_J^* \mbox{  and  } Q_R^2=Q_R=Q_R^*,$$
moreover, $$Q_RQ_J=0=Q_J Q_R,$$
and

$$Q_J C_3=0=C_3 Q_J, ~ Q_R C_3 = 0 = C_3 Q_R,~ Q^{+}_{\tau} Q_J = Q_J =  Q_J Q^{+}_{\tau}, ~Q^{-}_{\tau} Q_R = Q_R = Q_R Q^{-}_{\tau}.$$

\end{corollary}

\begin{proof}
Use the preceding definitions, propositions and corollaries.

\end{proof}

Notice that in the language of algebra, that $R$ and $J$ are mutually pseudoinverse in $\A_P(3) \subset \A_P,$ with $Q_J$ the domain projection of $J$ and range projection of $R$ and $Q_R$  is the range projection of $J$ and domain projection of $R.$  Of course, in this case we also have complementarity of $Q_J$ and $Q_R$ by Corollary \ref{QJ2}.

\bigskip

If $M$ is an $\A_{P(3)}$-module, then we can think of the members of $M$ as (abstract) curvature operators.  As the Bianchi identity applies to any actual Riemann curvature operator on a manifold coming from a torsion free connection on the tangent bundle, this means that such a curvature operator is acyclic, and obviously the antisymmetry gives the relation to $ Q^{-}_{\tau}.$
Thus, in the abstract setting, for any $r \in M,$ we would call $r$ a(n) (abstract) {\bf Riemann curvature operator}, if $Q_R r = r.$  Then we have $\tau r= - r$ or equivalently, $Q^{-}_{\tau} r = r,$ since 
$Q^{-}_{\tau} Q_R=Q_R.$  Also, as $Q_R r = r$ we have $C_3 r=0,$ since $C_3 Q_R=0.$  Notice that if $m \in M,$ then $Rm$ is a Riemann curvature operator as, $Q_R R= R.$  Thus, $Q_R$ is the self-adjoint projection onto the subspace of   $M$ consisting of Riemann curvature operators in an abstract sense, and $R$ transforms any curvature operator into a Riemann curvature operator.  Likewise, we would call $j \in M$ a(n) (abstract) {\bf Jacobi curvature operator} if  $Q_J j = j.$  Then,  $\tau j = j,$ or equivalently, $Q^{+}_{\tau} j= j,$ since $Q^{+}_{\tau} Q_J = Q_J.$  Also, $C_3 j=0,$ since $Q_J j = j,$ and $C_3 Q_J = 0.$  If $m \in M,$ then $Jm$ is a Jacobi curvature  operator, since $Q_J J =J.$  Thus, the abstract Jacobi operators are symmetric and also acyclic, which of course is the Bianchi identity for Jacobi curvature operators, and $Q_J$ is the self adjoint projection onto the space of Jacobi curvature operators.  And likewise, $J$ transforms any curvature operator into a Jacobi curvature operator.  If we take any $m \in M,$ then $Q_R m$ is a Riemann curvature operator and $Q_J m$ is a Jacobi curvature operator, since  
$[Q_R]^2 =Q_R$ and $[Q_J]^2  = Q_ J.$  In particular, if $r \in M$ is a Riemann curvature operator, then $Jr$ is a Jacobi operator, associated to $r,$ whereas if $j \in M$ is a Jacobi curvature operator, then $Rj$ is a Riemann curvature operator associated to $j$.  Moreover, if we start with a Riemann curvature operator $r \in M,$ which of course means $Q_R r = r,$ and form $j = Jr,$ the associated Jacobi curvature operator, then the associated Riemann curvature operator to $j$ is $Rj = RJr = Q_R r = r,$ the original Riemann curvature operator, whereas if we start with a Jacobi curvature operator $j,$ so $j = Q_J j,$ and form the associated Riemann curvature operator, $r = R j,$ then the associated Jacobi curvature operator to $r$ is $J r = J R j = Q_J j =j,$ and again we are back to the original curvature operator.  Thus, the Riemann and Jacobi curvature operators should be regarded as being paired like the two sides of a coin, in a duality.  The operators $R$ and $J$ acting on $M$ are thus mutually inverse isomorphisms of the subspaces of Riemann curvature operators and Jacobi curvature operators.  That is, $R(M)$ is the subspace of Riemann curvature operators, $J(M)$ is the space of Jacobi curvature operators, and $J$ carries $R(M)$ isomorphically onto $J(M)$ and $R$ carries $J(M)$ isomorphically onto $R)M),$ with these two isomorphisms being mutually inverse.

\bigskip

Now, on a semi Riemannian manifold, there is a metric tensor which determines a unique Levi-Civita connection or {\bf metric connection} which is torsion free and for which the metric tensor is parallel.  Thus, for such connections, the Riemann and Jacobi curvature operators enjoy an additional symmetry which we can call the {\bf Exchange property}.  This involves the action of $P(4),$ so now we look to $A_{P(4)}$ and the exchange property which is $\chi$-invariance, with $\chi = \beta \nu =\nu \beta =[1,3][2,4].$  Thus, if $M$ is an $\A_{P(4)}$-module, then $m \in M$ is called a {\bf Metric Curvature Operator} if it is a curvature operator which satisfies $\chi$-invariance.  Specifically, $m \in M$ is a metric curvature operator if $\chi m = m,$ or equivalently, if $Q_{\chi} m = m.$  We likewise say that $m$ has the {\bf (Metric) Exchange Property} if $m$ is a metric curvature operator.  If $m$ is a Riemann curvature operator which satisfies the exchange property, then it is a metric Riemann curvature operator, whereas if it is a Jacobi curvature operator and satisfies the exchange property, then it is a metric Jacobi curvature operator.  We shall see that in fact, our duality between the Riemann and Jacobi curvature operators also will respect the exchange property, so if $r$ is a metric Riemann curvature operator, then $J r$ is a metric Jacobi curvature operator, whereas if $j$ is a metric Jacobi curvature operator, then $R j$ is a metric Riemann curvature operator.  This amounts to demonstrating that the action of $Q_{\chi}$ commutes with the actions of $Q_J$ and $Q_R$ in an appropriate manner.  So, we will next deal with these details.

\medskip

We begin our extension of calculations into $\A_{P(4)}$ by noting the actions of $Inn(\tau)$ and $Inn(\alpha).$  Of course, we note that $\tau = [1,2]$ and $\alpha = [3,4]$ commute as they are disjoint transpositions.  We have, using (\ref{perm5}),

$$\beta^{\tau} =[1,2][2,4][1,2]=[1,4]=\lambda \mbox{ or equivalently,  } \tau \beta = \lambda \tau, $$

$$\beta^{\alpha} = [3,4][2,4][3,4]=[2,3]=\gamma  \mbox{ or equivalently,  } \alpha \beta = \gamma \alpha, $$

$$\nu^{\tau} = [1,2][1,3][1,2]=[2,3]=\gamma  \mbox{ or equivalently,  } \tau \nu = \gamma \tau,      $$

$$ \nu^{\alpha} = [3,4][1,3][3,4] = [1,4] = \lambda  \mbox{ or equivalently,  } \alpha \nu = \lambda \alpha, $$

$$\lambda^{\tau} = [1,2][1,4][1,2] = [2,4] = \beta   \mbox{ or equivalently,  }  \tau \lambda = \beta \tau,  $$

$$\lambda^{\alpha} = [3,4][1,4][3,4] = [1,3] = \nu   \mbox{ or equivalently,  }  \alpha \lambda = \nu \alpha, $$

$$\gamma^{\tau} = [1,2][2,3][1,2] = [1,4] = \lambda   \mbox{ or equivalently,  }  \tau \gamma = \lambda \tau, $$

$$\gamma^{\alpha} = [3,4][2,3][3,4] = [2,4] = \beta  \mbox{ or equivalently,  }  \alpha \gamma = \beta \alpha.$$

\bigskip

Inspection of these equations shows that (keep in mind that $\tau \alpha = \alpha \tau$)

\begin{equation}\label{nu tau alpha}
\nu^{\tau \alpha} = \gamma^{\alpha} = \beta   \mbox{ or equivalently,  }  ( \tau \alpha) \nu = \beta ( \tau \alpha)
\end{equation}    
and

\begin{equation}\label{beta tau alpha}
\beta^{\tau \alpha} = \lambda^{\alpha} = \nu   \mbox{ or equivalently,  }  ( \tau \alpha) \beta = \nu (\tau \alpha).
\end{equation}   
This means that as $\nu$ and $\beta$ commute and $\chi = \nu \beta,$ we must have

$$ \chi (\tau \alpha) = \chi  \mbox{ or equivalently,  } \chi (\tau \alpha) = (\tau \alpha) \chi.$$

But also, since $\nu \beta = \chi = \beta \nu,$ we have

\begin{equation}\label{tau chi}
\tau^{\chi} = \chi [1,2] \chi = [3,4] = \alpha  \mbox{ or equivalently,  }  \tau \chi =  \chi \alpha,
\end{equation}
and

\begin{equation}\label{alpha chi}
\alpha^{\chi} = \chi [3,4] \chi = [1,2] = \tau   \mbox{ or equivalently,  }  \alpha \chi = \chi \tau,
\end{equation}
from which we conclude

$$(\tau + \alpha) \chi = \chi ( \tau + \alpha ) .$$

Thus, the subalgebra of $\A_{P(4)}$ consisting of members which commute with $\chi$ contains both $\tau + \alpha$ and $\tau \alpha,$ and therefore, it contains both $Q^{+}_{\tau} Q^{+}_{\alpha}$
and $Q^{-}_{\tau} Q^{-}_{\alpha}.$  It follows that both of these commute with $Q^{\pm}_{\chi}$ as well.

%%%%%%%%%%%%%%%%   DEFINITION   QQQ

\begin{definition}\label{QQQ}
We define

$$Q^{\pm} = Q^{+}_{\chi} Q^{\pm}_{\tau} Q^{\pm}_{\alpha}.$$
That is,

$$Q^{+} = Q^{+}_{\chi} Q^{+}_{\tau} Q^{+}_{\alpha} \mbox{ and  }  Q^{\pm} = Q^{+}_{\chi} Q^{-}_{\tau} Q^{-}_{\alpha}.$$
Notice that it is $Q^{+}_{\chi}$ which is the factor in both $Q^{\pm}.$

\end{definition}

\bigskip

%%%%%%%%%%%%%%%%    PROPOSITION  QQQ1

\begin{proposition}\label{QQQ1}

We have

$$[Q^{\pm}]^2=Q^{\pm} = [Q^{\pm}]^{*},$$
and $Q^{+}_{\tau},~Q^{-}_{\tau},~Q^{+}_{\alpha},~Q^{-}_{\alpha}, $ and $Q^{+}_{\chi}$ all commute with $Q^{\pm}.$
Also

$$Q^{\pm}_{\tau} Q^{\pm } = Q^{\pm}, ~~Q^{\pm}_{\alpha} Q^{\pm } = Q^{\pm},$$
and

$$Q^{\pm}_{\chi} Q^{\pm} = Q^{\pm}.$$

\end{proposition}

\begin{proof}

Since $\tau$ and $\alpha$ commute, it follows that $Q^{\pm}_{\tau}$ and $Q^{\pm}_{\alpha}$ commute, and their product commutes with $Q^{+}_{\chi},$ as noted in the preceding discussion.  This means that the product that forms $Q^{\pm}$ can be arranged with either of these three as the first factor or the last factor, as wished, from which we conclude these equations hold and that $Q^{\pm}$ is self adjoint as all three factors are.

\end{proof}

\medskip

%%%%%%%%%%%%%%%%%   COROLLARY  QQQ2

\begin{corollary}\label{QQQ2}

We have

$$\tau Q^{\pm} = Q^{\pm} \tau = \pm Q^{\pm},~\alpha Q^{\pm} = Q^{\pm} \alpha = \pm Q^{\pm} \mbox{ and } \chi Q^{\pm} = Q^{\pm} \chi = Q^{\pm}.$$
In particular,

$$(\alpha \tau )Q^{\pm} = Q^{\pm} (\alpha \tau) = Q^{\pm}.$$
In addition,

\begin{equation}\label{nuQ}
(\tau \alpha) \nu Q^{\pm}  = \nu Q^{\pm}      \mbox{ and }  \tau \nu Q^{\pm} = \alpha \nu Q^{\pm}.
\end{equation}

\end{corollary}

\begin{proof}
The first two equations follow from the preceding remarks and (\ref{nuQ}) follows from the first two equations because $Q^{\pm}$ absorbs both $\chi$ and $\tau \alpha$ so (\ref{nu tau alpha}) applies to give the result, since $\chi \beta = \beta \chi =\nu$ and $\tau \alpha = \alpha \tau.$

\end{proof}

\medskip
%%%%%%%%%%%%%%%  COROLLARY  QQQ3

\begin{definition}\label{QQQ3}

Define $J_0 = JQ^{-}$ and $R_0 = R Q^{+}.$

\end{definition}

If $M$ is an $\A_{P(4)}$-module and $m$ belongs to $M,$ then we have that $m$ is a metric Riemann curvature operator if $Q_R m = m$ and $\chi m = m,$ or equivalently, if $Q_R m = m$ and 
$Q^{+}_{\chi} m = m.$   If so, then $\tau m = -m$ and $Cm=0$ and $\chi m = m.$  Then, $\alpha m = \alpha \chi m = \chi \tau m = - \chi m = - m,$ therefore $\alpha m = - m$ and therefore,

$$Q^{-}_{\tau} m = m,~Q^{-}_{\alpha} m = m, \mbox{ and } Q^{+}_{\chi} m = m.$$
It follows that $Q^{-} m = m.$  On the other hand, if $Q^{-} m = m, $ then it follows immediately from Proposition \ref{QQQ1} that the three previous equations hold so that if $Q_R m = m,$ then $m$ is a metric Riemann curvature operator.   Replacing all the minus superscripts with plus superscripts and $R$ by $J,$ we see that if $m$ is a Jacobi curvature operator, then if it is in fact a metric Jacobi curvature operator which means also $\chi m = m,$ then we have

$$Q^{+}_{\tau} m = m,~Q^{+}_{\alpha} m = m, \mbox{ and } Q^{+}_{\chi} m = m,$$
from which it likewise follows that $Q^{+} m = m.$  Again, by Proposition \ref{QQQ1}, if this last equation holds, then all three of the previous three equations hold.  Thus, if $m$ is a metric Riemann curvature operator, then $J_0 m = Jm $ and therefore $J_0 m$ is a Jacobi curvature operator.  Likewise, if $m$ is a metric Jacobi curvature operator, then $R_0 m = R m$ is a Riemann curvature operator.  We want to show that in fact in either case, the result is also a metric curvature operator.

Thus, assuming that $\chi m = m$ and $Q_R m = m,$ we want to show that $\chi J m = Jm,$ showing that the Jacobi curvature operator associated to $m$ is also a metric curvature operator and likewise, we want to show that if $m$ is a metric Jacobi curvature operator, then $Rm$ the associated Riemann curvature operator is in fact a metric curvature operator.

\medskip

Thus, we want to show that $\chi J_0 = J_0$ and $\chi R_0 = R_0.$

\medskip

Of course, $R=(4/3)Q^{-}_{\tau} \nu Q^{+}_{\tau},$ so by Proposition \ref{QQQ1} we have $R_0=(4/3)Q^{-}_{\tau} \nu Q^{+}.$  Now, we can apply Corollary \ref{QQQ2} and the fact that 

$$\chi \nu = \beta = \nu \chi.$$
Thus, using (\ref{nu tau alpha}), (\ref{beta tau alpha}),   (\ref{tau chi}), and (\ref{alpha chi}),

$$(3/4)R_0 = Q^{-}_{\tau} \nu Q^{+} =  Q^{-}_{\tau} \nu  (\alpha \tau) \chi Q^{+} = Q^{-}_{\tau}  (\alpha \tau) \beta \chi Q^{+}= (- \alpha ) Q^{-}_{\tau} \nu Q^{+} = - (3/4) \alpha R_0,$$
from which we conclude that

$$\alpha R_0 = - R_0 \mbox{ and } Q^{-}_{\alpha} R_0 =R_0,$$
so as $\chi$ commutes with $Q^{-}_{\alpha}Q^{-}_{\tau},$ we have 

$$(3/4) \chi R_0 = \chi Q^{-}_{\alpha}Q^{-}_{\tau} \nu Q^{+} = Q^{-}_{\alpha}Q^{-}_{\tau} \chi \nu Q^{+} = Q^{-}_{\alpha}Q^{-}_{\tau} \nu \chi Q^{+} = Q^{-}_{\alpha}Q^{-}_{\tau} \nu Q^{+}=(3/4)R_0.$$

Thus, finally, we do indeed have

$$\chi R_0 = R_0 = Q^{+}_{\chi}R_0 = Q^{+}_{\chi} Q^{-}_{\alpha}Q^{-}_{\tau} R_0 = Q^{-} R_0.$$

\medskip

Next, we interchange the plusses and minuses to deal with $J_0,$  keeping in mind that $Q^{-} \alpha \tau =Q^{-},$ so again using $\beta \chi = \nu,$

$$J_0 = J Q^{-} =  Q^{+}_{\tau}  \nu Q^{-}_{\tau} Q^{-}_{\tau}Q^{-}_{\alpha}  Q^{+}_{\chi} =  Q^{+}_{\tau} \nu Q^{-}=  Q^{+}_{\tau} \nu (\tau \alpha) \chi Q^{-} =  \alpha  Q^{+}_{\tau} \beta  \chi Q^{-} = \alpha J_0.$$
Thus,

$$\alpha J_0 = J_0 \mbox{ and }  Q^{+}_{\alpha} J_0 =J_0.$$
Thus, as $\chi$ commutes with $ Q^{-}_{\alpha}Q^{-}_{\tau}$

$$\chi J_0 =  Q^{+}_{\alpha}Q^{-}_{\tau} \chi \nu Q^{+} =  Q^{+}_{\alpha}Q^{+}_{\tau} \nu \chi Q^{+} =  Q^{+}_{\alpha}Q^{+}_{\tau} \nu Q^-{-} =  Q^{-}_{\alpha} J_0 =J_0.$$
This finally gives

$$\chi J_0 =J_0 \mbox{ and }  Q^{+}_{\chi} J_0 =J_0 = Q^{+} J_0.$$

%%%%%%%%%%%%%%%%%   PROPOSITION  QQQ4

\begin{proposition}\label{QQQ4}

Regarding $J_0 = J Q^{-} $ and $R_0 =R Q^{+},$ we have

$$J_0=Q^{+} J_0 Q^{-} = Q^{+} J Q^{-} =  Q^{+} \nu Q^{-}  \mbox{  and  }   R_0=Q^{-} R_0 Q^{+} = Q^{-} R Q^{+} =   Q^{-} \nu Q^{+}.$$

\end{proposition}

\begin{proof}

The result is now obvious from the preceding equations.

\end{proof}

%%%%%%%%%%%%%%%%   DEFINITION QQQ5

\begin{definition}\label{QQQ5}

Define $Q^0_J$ and $Q^0_R$ in $\A_{P(4)}$ by

$$Q^0_J=J_0R_0 \mbox{ and } Q^0_R = R_0 J_0.$$

\end{definition}

We can note from Proposition \ref{JR1}, Corollary \ref{JRJ}, and Corollary \ref{QJ2} that in $\A_P$ we have

$$J^{*} = \frac{3}{4} R \mbox{ and } R^{*} =\frac{4}{3} J $$
and all the $Q$'s are self adjoint idempotents, which is to say they are orthogonal projections in $\A_P.$  Applying adjoints to the equations in Proposition \ref{QQQ4}, Proposition \ref{JR1}, and its Corollaries, we have

$$J_0^{*} = [JQ^{-}]^* = Q^{-} J^* = \frac{3}{4} Q^{-} R \mbox{ and } J_0^* =[Q^{+} J Q^{-}]^* = Q^{-} J^* Q^{+} = \frac{3}{4} Q^{-} R Q^{+}= \frac{3}{4} R_0$$
from which we conclude that

$$R_0 =  Q^{-} R  \mbox{  and  }  J_0^* = \frac{3}{4} R_0.$$
Likewise,

$$R_0^* =[RQ^{+}]^* =Q^{+} R^* = \frac{4}{3} Q^{+} J  \mbox{ and } R_0^* = [Q^{-} R Q^{+}]^* =Q^{+} R^* Q^{-} = \frac{4}{3} Q^{+} J Q^{-} = \frac{4}{3} J_0, $$
so

$$J_0 = Q^{+} J \mbox{ and } R_0^* = \frac{4}{3} J_0.$$
Now,

$$Q_J^0  J_0 = J_0 R_0 J_0 = [Q^{+}J][ R Q^{+}][ J Q^{-}]= Q^{+}J R[Q^{+}JQ^{-}]= Q^{+} JR[ J Q^{-}] =Q^{+} [JRJ] Q^{-} = Q^{+} J Q^{-} =J_0,$$
and either reversing the plus and minus signs and interchanging $R$ with $J$ or else taking adjoints in 

\begin{equation}\label{relinvJ0}
J_0 R_0 J_0 = Q_J^0 J_0 = J_0,
\end{equation}
gives 

\begin{equation}\label{relinvR0}
R_0 J_0 R_0 = Q_R^0 R_0 = R_0.
\end{equation}

We can now easily see that

$$Q^{+} Q_J = Q^{+} J R = [Q^{+} J Q^{-}] R = [Q^{+} J][Q^{-} R]= J_0 R_0 = Q_J^0,$$
$$Q_J Q^{+} = J R Q^{+} = J[Q^{-} R Q^{+}] = [J Q^{-}][R Q^{+}] = J_0 R_0 = Q_J^0 , $$
$$Q^{-} Q_R = Q^{-} R J = [Q^{-} R Q^{+}] J =  [Q^{-} R] [Q^{+} j] = R_0 J_0 = Q_R^0,$$
and

$$ Q_R Q^{-} = R J Q^{-} =R [Q^{+} J Q^{-}] = [R Q^{+} ][ J Q^{-}] = R_0 J_0 =Q_R^0.       $$

\bigskip

Finally, we summarize these results in

%%%%%%%%%%%%%%%%%%%%   THEOREM  QQQ6

\begin{theorem}\label{QQQ6}

In $\A_{P(4)}$

$$J_0 = Q^{+} J = Q^{+} J Q^{-} = J Q^{-} \and R_0 = Q^{-} R = Q^{-} R Q^{+},$$

$$J_0^{*} = \frac{3}{4} R_0 \mbox{ and } R_0^{*} = \frac{4}{3} J_0$$

$$J_0R_0J_0 = Q^0_J  J_0 = J_0  \mbox{ and } R_0 J_0 R_0 =  Q^0_R R_0 = R_0.$$

Moreover, $Q^0_J$ and $Q^0_R$ are self adjoint idempotents, which is to say they are orthogonal projections in $\A_P.$  In fact,

$$Q^0_J = Q^{+} Q_J =Q_J Q^{+} = Q^{+} Q_J Q^{+} \mbox{ and } ~Q_R^0 = Q^{-} Q_R =Q_R Q^{-} = Q^{-} Q_R Q^{-}.$$

\end{theorem}

\begin{proof}

\end{proof}

\medskip

We note that the main upshot here in algebraic terms is that $J_0$ and $R_0$ are mutually pseudoinverse with $Q_J^0$ the domain projection of $J_0$ and $Q_R^0$ the domain projection of $R_0.$

Thus,

$$[J_0]^{-(Q_J^0,Q_R^0)} =R_0 \mbox{ and likewise } [R_0]^{-(Q_R^0,Q_J^0)} =J_0.$$

\bigskip

If $M$ is a $\A_{P(4)}$-module, and $m \in M,$ then by definition, $m$ is a metric curvature operator if and only if $\chi m = m,$ so by the preceding propositions and theorems, if $m$ is in fact a metric Riemann curvature operator, then $Q_R m = m $ so $\tau m = - m$ and $\chi m = m,$ and therefore as $\alpha  \chi = \chi \tau,$ it follows that $\alpha m =  \alpha \chi m = \chi \tau m = - \chi m =-m$, so  as already noted above, we have $Q^{-}m = m$ and therefore $Q^{-} Q_R m = m.$  Likewise, if $m$ is a metric Jacobi curvature operator, the fact that $\tau m = m$ and $\chi m =m,$ immediately gives $\alpha m = m,$ so we have $Q^{+} m = m.$  Thus, the orthogonal projections $Q_J^0$ and $Q_R^0$ are projecting $M$ onto the subspaces of metric Jacobi curvature operators and metric Riemann curvature operators, respectively.  

We note that $JM \subset M,$ so 

$$Q_J JM \subset Q_J M =JR M \subset JM = Q_J JM \subset Q_J M,$$ 
and therefore

$$Q_J M = JM.$$
Likewise, reversing the roles of $R$ and $J,$ we find that

$$Q_R M = RM.$$

Finally, we note that we can simply take $M = \A_{P(4)},$ as a left $\A_{P(4)}$-module.  In this case, the equation $Q_J J =J,$ says that $J$ itself is a Jacobi curvature operator, so is the {\bf Universal Jacobi Curvature Operator} and likewise, the equation $Q_R R = R$ makes $R$ the {\bf Universal Riemann Curvature Operator}.

%%%%%%%%%%%%%%%%%%%%%%%%%%%%%%%%
%%%%%%%%%%%%%%%%%%%%%%%%%%%%%%%%%

\vfill
\break

%%%%%%%%%%%%%%%%%%%%%%%%%%%%%%%%%%%%%%%%%%%%%%%%%%%%%%%%%%%%%%%%%%%%%%%%%%%%%%%%%%%%%%%%%%%%%%%%%%%%%%%%%
%%%%%%%%%%%%%%%%%%%%%%%%%%%%%%%%%%%%%%%%%%%%%%%%%%%%%%%%%%%%%%%%%%%%%%%%%%%%%%%%%%%%%%%%%%%%%%%%%%%%%%%%%

\section{ALGEBRAIC CURVATURE OPERATORS}

%%%%%%%%%%%%%%%%%

%%%$\gt,,\gb,\gth,\gx$

%%%\medskip

%%%%%%%%%%%%%%%%%%%%%%%%%

The various identities satisfied by the Riemann curvature operator have been usefully and extensively studied algebraically by mathematicians and physicists with applications to semi-Riemannian geometry.  The monograph by Peter Gilkey, \cite{GILKEY1}, is a very useful introduction to this literature, however, most of the discussion is limited to the quadratic form of the Jacobi curvature operator.  We are interested in the full Jacobi curvature operator, which in \cite{GILKEY1} is rarely used and when used it is through polarization. The treatment here is a useful way to derive relations between the full Jacobi and Riemann curvature operators, as we do not restrict attention to the quadratic forms so we do not have to polarize to get the full tensor.  It is a little more involved than the quadratic form, but it will be helpful in our development.  Some of the results here can also be found in \cite{MTW}, pages 286-287, in index notation, where the Jacobi curvature operator field on a semi-Riemannian manifold is shown to have a particularly useful description in Riemann normal coordinates.  The results of this section are valid for any semi-Riemannian vector space, that is, for any finite dimensional vector space, $V,$ with a given non-degenerate symmetric bilinear form, $g.$  Even some of the basic facts concerning the Riemann curvature tensor can be more efficiently seen through the Jacobi curvature operator.  But it will be useful to consider curvature operators more generally than is usually done, in order to make the most efficient use of them, and this will require more detailed permutation computations.

%%%%%%%%%%%%%%%%%%%%%%%%%%%%%%%%%%%%%%%%%%%%%%%%%%%%%%%%%%%%%%%%
%%%%%%%%%%%%%%%%%%%%%%%%%%%%%%%%%%%%%%%%%%%%%%%%%%%%%%%%%%%%%%%%
%%%%%%%%%%%%%%%%%%%%%%%%%%%%%%%%%%%%%%%%%%%%%%%%%%%%%%%%%%%%%%%%%

In addition, it is useful to note that $\A_P$ is a star-algebra with the adjoint operation the extension of inversion in the group to the whole algebra. The natural representations of $\A_P$ on these various tensor algebras preserve the star operation, so in particular, the image of any self adjoint idempotent is a self adjoint projection operator on the particular vector space of tensors on which the permutation group is acting via permutation of arguments in the tensors.  Since any transposition (cycle of length 2) is idempotent, it is therefore its own inverse, so is self adjoint, so the decomposition in $\A_P$ as a sum of self adjoint projections produces self adjoint projection operators acting on tensors.

the decomposition in $\A_P$ as a sum of self adjoint projections produces self adjoint projection operators acting on tensors. That is to say, that if $B$ is a transposition, then $Q^{\pm}_B$ are self adjoint idempotents, that is, they are orthogonal projection operators.

%%%%%%%%%%%%%%%%%%%%%%%%%%%%%%%%%%%%%%%%%%%%%%%%%%%%%%%%%%%%%%%%%%%%%%%%%%%%%%%%%%%%%%%%%%%%%%%%%%%%%%%%%%%%%%%%%%%%%%%%%%%%%%%%%%%%%%%%%%%%%%%%%%%%%%%%%%%%%%%%%%%%%%%%%%%%%%%%%%%%%%%%%%%%%%%%%%%%%%%%%%%%%%%%%%%%%%%%%%%%%%%%%%%%%%%%%%%%%%%%

Algebraically, curvature operators modeling actual curvature operators on a tangent bundle at a point can be modeled as members of $PC(V)=L^2(V;L(V;V))$ with various properties, where $V$ is a vector space with given non-degenerate indefinite inner product $g,$ so for simplicity, we refer to the members of $PC=PC(V)$ as {\it precurvature operators} or simply as {\bf Curvature Operators}.   Of course, the properties of interest are expressed in terms of the action of the permutation group $P(4)$ and the resulting action of $\A_{P(4)},$ so we will identify $L^3(V;V)=PC(V)$ through evaluation, so $T(u,v,w)=[T(u,v)]w $ for any $u,v,w \in V,$ and through the metric $g$ we identify $PC(V)$ with $L^4(V, \bR),$ so $T(w,x,y,z)=g(T(w,x)y,z)$ for any $w,x,y,z \in V.$  Thus, $\A_{P(3)}$ and $\A_{P(4)}$ act on $PC(V)$ in the obvious way. 

More generally, in case of general bundles such as principal bundles and vector bundles, the objects are still more general, so we would have $V$ as above, but replace $L(V;V)$ by any vector space $A$ with a bilinear map
$m:A \times A \lra A$ specified which is not assumed commutative or associative, but defines what we think of as a specified multiplication denoted as juxtaposition, so for $a,b \in A$ we use 
$ab=m(a,b).$  Then, any member of $L^2(V;A)$ is a {\bf curvature operator} or {\bf algebraic curvature operator}.  Thus, in case of a vector bundle over a manifold with fiber $W,$ we would have $A=L(W;W)$ is the algebra of linear transformations of $W.$  Since the semi classical modeling of forces uses curvature operators in this general sense, and as we are concerned with energy, we can point out here, that the energy momentum stress tensor of such a field $F$ is

$$T_F=trace(F^{<2>} - \frac{1}{d}[c(F^{<2>})]g), \mbox{ where } c(F^{<2>}) \in A \mbox{ is a contraction of } F^{<2>}, \mbox{ and } d=dim(V).$$
Here we are using $F^{<2>}$ for the contraction

$$F^{<2>}=FgF, \mbox{ where in any frame at a point, } [FgF](u,v) = F(u,e_{\mu})F(e^{\mu},v), \mbox{ for any tangent vectors } u,v.$$

In fact, the contraction used in the equation commutes with the trace and therefore, $c(T_F)=0,$ as $c(g)=d=dim(V).$  Thus, these semi classical models are for massless force fields.  But, given our idea that the trace of an operator giving energy momentum stress density should itself be considered as a generalized energy momentum stress tensor, it seems reasonable that in the general semi classical setting of modeling massless forces, that the full energy momentum stress of the field should be

$$T_F= F^{<2>}-\frac{1}{d}[c(F^{<2>})]g$$
and then we regard $trace(T_F)$ as the gravitational source energy momentum tensor of the force field, thus maybe we should as well consider that force fields have other forms of energy that are in balance so as not to be a source of gravity, just like with the gravitational energy momentum stress tensor itself.  After all, the Bianchi identity applies to all vector bundle connections, so is really a general form of conservation of energy.

 However, in these more general cases, there is not much to the algebra so, for now, we stick with the case that $A=L(V;V).$

Extending the notation of the preceding section on group rings, we find it convenient to introduce specific symbols for the operations that we will need to use over and over.  Except for the computations with transpositions which are referred to, the following treatment is fairly self-contained and will not require the preceding section for any reader not interested in that level of generality.  Of course, we have the {\bf Natural Isomorphisms} of $PC(V)$ onto $L^3(V;V)$ and onto $L^4(V; \bR)$ which will allow us to think of curvature operators as being in either of these places for convenience, where the discussion makes the clear, as for instance in applying the $Cycl_3$ to a curvature operator or $Cycl_4$ as the case may be.

$$\tau,\alpha, \beta, \gamma, \lambda, \chi,\nu,Cycl : PC(V) \lra PC(V), \mbox{ so for any }T \in PC(V),$$

%  $$\gamma(T)=\hat{T} \in L^3(V;V),$$
$$[1,2]T=\tau(T)=\tilde{T} \in PC(V),$$
$$[3,4]T=\alpha(T)=T^* \in PC(V),$$
$$[2,4]T=\beta T,$$
$$[2,3]T= \gamma T,$$
$$[1,4]T = \lambda T,$$
$$[1,3][2,4]T=\chi(T)=T^{\dag} \in PC(V),$$
$$[1,3]T=\nu(T)=T^{\#} \in PC(V),$$ 
$$Cycl(T) \in PC(V),$$
by requiring

%\begin{equation}\label{pc1}
%[\gamma(T)](u,v,w) =\hat{T}(u,v,w)=[T(u,v)]w,
%\end{equation}

\begin{equation}\label{pc2}
[\tau(T)](u,v)=\tilde{T}(u,v)=T(v,u),
\end{equation}

\begin{equation}\label{pc3}
[\alpha(T)](u,v)=T^*(u,v)=[T(u,v)]^*,
\end{equation}

\begin{equation}\label{pc4}
g([[\chi(T)](u,v)]w,x)=g([T^{\dag}(u,v)]w,x)=g(T(w,x)u,v),
\end{equation}

\begin{equation}\label{pc5}
[[\nu(T)](u,v)]w=[T^{\#}(u,v)]w=[T(w,v)]u,
\end{equation}

\begin{equation}\label{pc6a}
Cycl(T)=Cycl_3(T),
\end{equation}
hold for all vectors $u,v,w,x \in V.$
%

%More generally, if $T \in PC(V),$ and $\sigma \in P(3),$ then $\sigma T =\sigma(T)$ is defined by

%\begin{equation}\label{3perm}
%\gamma (\sigma (T))= \sigma (\gamma(T))
%\end{equation}
Think $\tau$ for twisting the first two arguments, $\alpha$ for the adjoint which then interchanges the last two arguments.  We can think of $\chi$ as exchanging the first pair with the last pair maintaining order which effects an exchange operation, and the symbol $\chi$ then looks like an "ex".  The action of $\nu$ performs the "hash" operation which will be fundamental.

\medskip

%%%%%%%%%%%%%%    PROPOSITION

\begin{proposition}\label{involutions}
(1)  As linear maps on $PC(V)$  the permutations $\tau, \alpha, \beta, \chi, \nu$ are linear involutions, that is, with $id$ denoting the identity map of $PC(V),$

$$\tau  \tau=id,~~\alpha  \alpha=id,~~ \beta \beta = id, ~~\chi \chi=id, ~~\nu  \nu=id,$$
so for any $T \in PC(V),$

$$\tilde{\tilde{T}}=T,~~(T^*)^*=T,~~(T^{\dag})^{\dag}=T,~~(T^{\#})^{\#}=T.$$

\end{proposition}

\begin{proof}
Obvious because, except for $\chi,$ they are defined by transpositions.  If $\sigma$ is a transposition in any permutation group, then $\sigma^2=1.$  But, $\chi=[1,3][2,4]$ and as disjoint cycles commute, $\chi^2=1,$ as well.
\end{proof}
\medskip

We say that $T \in PC(V)$ is an {\it acyclic curvature operator} provided that $~Cycl(T)=0.$  If $T$ is an acyclic curvature operator with $\tilde{T}=T,$ we call $T$ a(n algebraic) {\it Jacobi} curvature operator, whereas if $\tilde{T}=-T,$ we call $T$ a(n algebraic) {\it Riemann} curvature operator.   We say that $T$ is a metric-$g$ curvature operator if it is a curvature operator which also satisfies $T^{\dag}=T.$  If $T \in PC(V)$ satisfies $Cycl(T)=0,$ we say it satisfies the {\bf First Bianchi Identity}, whereas if $T$ satisfies $T^{\dag}=T,$ then we say $T$ has the {\bf (metric) Exchange Property}.  Thus all acyclic curvature operators satisfy the First Bianchi identity and all metric acyclic curvature operators in addition have the Exchange Property.

If $W$ is any vector space and $B \in L(W;W),$ we say $B$ is an involution if $B \circ B =id_W,$ in which case, if $w \in W,$ then we define 

$$Q_B^+(w)=Sym_B(w)=\frac{w+B(w)}{2}, \mbox{ so } B(Sym_B(w))=Sym_B(w),$$
and likewise

$$Q_B^-(w)=Alt_B(w)=\frac{w-B(w)}{2} \mbox{ so } B(Alt_B)(w)=-Alt_B(w).$$ Thus, with $1_W=id_W,$ the identity map on $W,$ we have simply

$$Q_B^+=\frac{1_W+B}{2} \mbox{ and } Q_B^-=\frac{1_W-B}{2},$$
and these are the idempotent operators projecting onto the $\pm1$ eigenspaces of $B,$ giving the trivial spectral decomposition $B=Q_B^+ -Q_B^-,~Q_B^+Q_B^-=0=Q_B^-Q_B^+,~Q_B^+ + Q_B^-=1_W.$  In particular, for $W=PC(V),$ the liner maps $\tau, \alpha, \chi, \nu$ are all involutions and thus are linear isomorphisms.

\medskip

For any $T \in PC(V),$ since $\tau$ is an involution, define the new curvature operators

$$J[[T]]=J_T, ~R[[T]]=R_T \in PC(V)$$ by

\begin{equation}\label{pc6b}
J[[T]]=J_T=Sym_{\tau}(T^{\#})=Q^+_{\gt}(T^{\#})
\end{equation}
and

\begin{equation}\label{pc7}
R[[T]]=R_T=\frac{4}{3}Alt_{\tau}(T^{\#})=\frac{4}{3}Q^-_{\gt}(T^{\#}).
\end{equation}

%%%%%%%%%%%%%%%%%%           CAUTION ABOUT DEFINITIONS

{\bf CAUTION: The notation here is slightly different from that in Definitions \ref{JANDR} and \ref{QQQ3}, because we are dealing with the specific $\A_P$-module $PC(V),$ so it will be convenient to modifiy the operators $J$ and $R,$  that is, in terms of those definitions, we have, with the $PC(V)$ as a left $\A_P$-module, with $J, R \in \A_P,$ defined as in Definition \ref{JANDR},

$$JT=J[[Q^-T]]  \mbox{ and } RT=R[[Q^+T],$$
so

$$JT=J_S \mbox{ where } S=Q^-T, \mbox{ and } RT=R_S \mbox{ where } S=Q^+T.$$}

%%%%%%%%%%%%%%%%%%%%%%%%%%%%%%%%%%%%%%%%  END CAUTION

The relations between these constructs are summarized in the following propositions whose proofs  are an elementary exercise, but in the proofs below, we give each step a description to make it clear how to proceed.  It is also useful to note that the quadratic form of $J_T$ is easily seen to be

\begin{equation}\label{jquad}
J_T(u,u)v=T(v,u)u,
\end{equation} 
so by the observer principle, we see that $J_T$ is determined readily from the curvature operator $T$.  In physical applications it is usually the quadratic form of the Jacobi curvature operator which is needed, but we will find use for the full Jacobi curvature operator.

%%%%%%%%%      PROPOSITION  cyclprop2

\begin{proposition}\label{cyclprop2}
For any $T \in PC(V),$ we have

\begin{equation}\label{cycl7}
Cycl([1,2]T)=Cycl([1,3]T) \mbox{ in other words, } Cycl(\tilde{T})=Cycl(T^{\#}),
\end{equation}

\begin{equation}\label{cycl8}
Cycl(T)=T+[1,2](T^{\#}) +([1,2]T)^{\#}=T+[T^{\#}]\tilde ~+ [\tilde{T}]^{\#}
\end{equation}

\begin{equation}\label{cycl8c}
Cycl(Q_{\tau}^{\pm}(T^{\#})=Cycl(Q_{\tau}^{\pm}(\tilde{T}))
\end{equation}

\begin{equation}\label{cycl9}
T^{\#}=Q_{\tau}^+(T^{\#}) +Q_{\tau}^-(T^{\#}) =J[[T]]+\frac{3}{4}R[[T]] =J_T+\frac{3}{4}R_T.
\end{equation}
\begin{equation}\label{cycl9a}
T=[J_T]^{\#}+\frac{3}{4}[R_T]^{\#}.
\end{equation}

\begin{equation}\label{cycl9b}
Q^+_{\tau} ( T) =J[[J_T]] + \frac{3}{4} J[[R_T]] \mbox{ and } Q^-_{\tau} (T)= \frac{3}{4} R[[J_T]] + \frac{9}{16} R[[R_T]].
\end{equation}

\begin{equation}\label{cycl10}
\mbox{If } \tilde{T}=T, \mbox{ then } Cycl(T)=T+2J_T,
\end{equation}
whereas,

\begin{equation}\label{cycl11}
\mbox{if } \tilde{T}=-T, \mbox{ then } Cycl(T)=T-\frac{3}{2}R_T.
\end{equation}

\end{proposition}

\begin{proof}
We have (\ref{cycl7}) and (\ref{cycl8}) as immediate consequences of (\ref{cyclprop1}).  Then (\ref{cycl8c}) follows immediately from  (\ref{cycl7}) using $(T^{\#})^{\#}=T.$ Of course, (\ref{cycl9}) is completely obvious, as 

$$Q_{\tau}^+ + Q_{\tau}^-=1,$$
and (\ref{cycl9a}) follows from (\ref{cycl9}) again using $T[^{\#}]^{\#}=T.$
Finally, (\ref{cycl10}) and (\ref{cycl11}) follow from (\ref{cycl8}).

\end{proof}

%%%%%%%%   PROPOSITION

\begin{proposition}\label{cyclprop3}
Suppose that $T \in PC(V)$ with $Cycl(T)=0.$  

\begin{equation}\label{cycl12J}
\mbox{If } \tilde{T}=T, \mbox{ then } J_T=-\frac{1}{2}T \mbox{ and therefore also } J[[J_T]]=\frac{1}{4}T.
\end{equation}

\begin{equation}\label{cycl12R}
\mbox{If } \tilde{T}=-T, \mbox{ then } R_T=\frac{2}{3}T \mbox{ and therefore also } R[[R_T]]=\frac{4}{9}T.
\end{equation}

\end{proposition}

\begin{proof}
The statement (\ref{cycl12J}) is an immediate consequence of (\ref{cycl10}) and statement (\ref{cycl12R}) is an immediate consequence of (\ref{cycl11}).

\end{proof}

\begin{corollary}\label{RRJJ}

If $J$ is a Jacobi curvature operator, then  $J[[J]] = J_J = (-1/2)J,$ whereas, if $R$ is a Riemann curvature operator, then $R[[R]] = R_R = (2/3)R,$  hence,

\begin{equation}
Sym(J^{\#}) = - \frac{1}{2} J \mbox{ and }  Alt(R^{\#}) =\frac{3}{4} R[[R]] = \frac{1}{2} R,
\end{equation}
and therefore,

\begin{equation}
J^{\#} = Alt(J^{\#}) + Sym(J^{\#}) = \frac{3}{4} R[[J]] - \frac{1}{2} J  \mbox{ and } R^{\#} =  Sym(R^{\#}) + Alt(R^{\#})  = J[[R]] +  \frac{1}{2} R.
\end{equation}

\end{corollary}

%%%%%%%%%%%%%%%   PROPOSITION

\begin{proposition}\label{pc8}
(1)  As linear maps $\tau, \alpha, \chi, \nu$ are linear involutions of $PC(V)$, that is, with $id$ denoting the identity map of $PC(V),$

\begin{equation}\label{involution1}
\tau^2=id,~~\alpha^2=id,~~\beta^2=id, ~~ \chi \circ \chi=id, ~~\nu \circ \nu=id, \mbox{ and, } \chi \nu=\nu \chi =\beta
\end{equation}
so for any $T \in PC(V),$

\begin{equation}\label{involution2}
\tilde{\tilde{T}}=T,~~(T^*)^*=T,~~(T^{\dag})^{\dag}=T,~~(T^{\#})^{\#}=T.
\end{equation}
Moreover,

\begin{equation}\label{commutation1}
\a  \gt=\gt   \alpha,~~\x  \nu =\nu  \x,
\end{equation}
whereas

\begin{equation}\label{commutation2}
\gt  \x=\x  \a,~ \chi  \tau=\alpha  \chi,
\end{equation}
and
$$\nu \tau=\sigma_3,~~\tau \nu =\sigma_3^{-1}=\sigma_3^2  \mbox{ so } \tau  \nu=\sigma_3  \nu \tau=\nu \gt \nu \gt,$$
and therefore

$$\nu \gt \nu =\gamma= \gt \nu \gt,$$
%
%where $\sigma_3=[1,2,3]$ is the upward shift genenerator of $Cycl(3) \subset P(3).$
%
$$Cycl=1+\nu \gt + \gt \nu,$$
consequently

$$\nu[Cycl]= [Cycl] \nu = \gt [Cycl]= [Cycl] \gt.$$

\medskip

(2)  Suppose $T \in PC(V).$  Then, 

$$\tilde{J_T}=J_T,~~\tilde{R}_T=-R_T,$$

$$(\tilde{T})^{\dag}=(T^{\dag})^*,~~(T^{\#})^{\dag}=(T^{\dag})^{\#}.$$

(3)  If any two of the three equations

$$\tilde{T}=T,~~T^*=T,~~T^*=\tilde{T},$$
hold, then all three hold.  If any two of the three equations

$$\tilde{T}=-T,~~T^*=-T,~~T^*=\tilde{T},$$
hold, then all three hold.

\medskip

(4)  If $Cycl(T)=0,$ then both

$$Cycl(T^{\#})=0 \mbox{ and } Cycl(\tilde{T})=0,$$
%
%whereas if both  $$Cycl(T)=0 \mbox{ and } Cycl(\tilde{T})=0,$$ then
%
and moreover,
$$Cycl(J_T)=0,~and~Cycl(R_T)=0.$$

\medskip

(5)  If $$\tilde{T}=T,$$ then

$$Cycl(R_T)=0,$$
whereas if $$\tilde{T}=-T,$$
then $$Cycl(J_T)=0.$$

\medskip
 
(6)  If $$T^{\dag}=T=\tilde{T},$$ then $$T^*=T ~~and ~~(R_T)^{\dag}=R_T,$$ whereas, if

$$T^{\dag}=T=-\tilde{T},$$ then $$T^*=\tilde{T}~~and ~~(J_T)^{\dag}=J_T.$$  

\medskip

(7)  Given that $$Cycl(T)=0,$$ 
if 

$$T=\tilde{T},~~then~~J[[R_T]]=T,$$
whereas if

$$\tilde{T}=-T,~~then~~R[[J_T]]=T.$$

\end{proposition} 

\begin{proof}

Parts (1)-(4) are almost obvious and what is not is straight forward to check.  In particular, the commutation relations in (1), (2) and (3) follow from Proposition \ref{permcalc}, whereas (4) follows from (1) and Proposition \ref{cyclprop2}.  Likewise, while (5) maybe surprising, it is also straight forward to check.  However, it is also an immediate consequence of (\ref{cycl8c}) because if $\tilde{T}=T,$ then $Q^-_{\tau}(\tilde{T})=0,$ whereas if $\tilde{T}=-T,$ then $Q^+_{\tau}(\tilde{T})=0.$  Likewise (6) is straightforward to check, either directly or using the relations in (1) and (2) and the easily checked fact that in the permutation group

\begin{equation}\label{experm}
\chi \tau \nu \tau = [2,1,3,4] = \tau \nu \alpha.
\end{equation}
Since $T^{\dag}=T$ and $\tilde{T}=\pm T$ give $\alpha \tau (T)=T,$ and as $\chi $ and $\nu $ commute, (\ref{experm}) takes care of the non-obvious term of $Q^{\pm}(T^{\#}).$
For (7), just unravel the left side of the equation evaluated on three vectors and note that the result is $(1/3)$(sum of four terms) of which two are identical immediately from either $T=\tilde{T}$ or $\tilde{T}=-T,$ as the case may be, whereas the sum of the remaining two terms is also the same because of $Cycl(T)=0.$  Alternately, we can note that (7) is an immediate consequence of proposition \ref{cyclprop3} together with (\ref{cycl9b}) as the result is to give

$$\tilde{T}=T \mbox{ implies } T=\frac{1}{4} T+ \frac{3}{4} J[[R_T]],$$
whereas

$$\tilde{T}=-T \mbox{ implies } T=\frac{1}{4} T+ \frac{3}{4} R[[J_T]],$$

\end{proof}

\begin{corollary}\label{cpc8}
The involutions $\a$ and $\gt$ commute, so the composition $\a   \gt $ is an involution, and both $\a+\gt$ and $\a  \gt$ commute with $\x,$ so $\x$ and hence also $Q^{\pm}_{\x}$ commute with both $Q^+_{\a}Q^+_{\gt}=Q^+_{\gt}Q^+_{\a}$ and $Q^-_{\a}Q^-_{\gt}=Q^-_{\gt}Q^-_{\a}.$
\end{corollary}

\begin{proof}
This is an obvious consequence of the proposition \ref{pc8}.
\end{proof}

%%%%

As an immediate consequence of Proposition \ref{pc8} parts (2) and (4), if $T \in PC(V)$ is an acyclic curvature operator, then $T^{\#}$ is also an acyclic curvature operator which is a metric-$g$ algebraic curvature operator if $T$ is metric-$g$, and by parts (5) and (6), if $T \in PC(V)$ is an acyclic Jacobi curvature operator (respectively, a metric-$g$ acyclic Jacobi curvature operator), then  $R_T$ is an acyclic Riemann curvature operator (respectively, a metric-$g$ acyclic Riemann curvature operator), whereas if $T$ is an acyclic Riemann curvature operator (respectively, a metric-$g$ acyclic Riemann curvature operator), then $J_T$ is an acyclic Jacobi curvature operator (respectively, a metric-$g$ acyclicic Jacobi curvature operator).  

Let $AC(V) \subset PC(V)$ be the set of acyclic curvature operators and $AC(V,g) \subset AC(V)$ be the set of metric-$g$ acyclic curvature operators.  Thus the involution $\nu$ of $PC(V)$ carries $AC(V)$ onto $AC(V)$ and $AC(V,g)$ onto $C(V,g).$

\begin{proposition}
The involution $\nu$ of $PC(V)$ carries the linear subspace $AC(V)$ onto itself and carries $AC(V,g)$ onto itself.  In addition,

$$(\tau \alpha )\nu (\tau \alpha) = \beta \mbox{ so } \tau \alpha \nu = \beta \tau \alpha \mbox{ and } \tau \alpha \beta = \nu \tau \alpha.$$

\end{proposition}

\begin{proof}
This is an immediate consequence of Proposition \ref{pc8} parts (2) and (4).
\end{proof}

Let $JC(V) \subset AC(V)$ denote the set of all acyclic Jacobi curvature operators on $V$ and let $JC(V,g) \subset JC(V)$ denote the subset of all metric-$g$ acyclic Jacobi curvature operators. Likewise, let $RC(V) \subset AC(V)$ denote the set of all acyclic Riemann curvature operators on $V,$ and let $RC(V,g) \subset RC(V)$ denote the subset of all metric-$g$ acyclic Riemann curvature operators.  Obviously, $JC(V),~JC(V,g),~RC(V),$ and $RC(V,g)$ are linear subspaces of $AC(V).$   Define linear maps

\begin{equation}\label{pc9}
{\bf \R}: JC(V) \lra RC(V) \mbox{ by } {\bf \R}(J)=R_J \mbox{ for any } J \in JC(V).
\end{equation}
and

\begin{equation}\label{pc10} 
{\bf \J}: RC(V) \lra JC(V) \mbox{ by } {\bf \J}(R)=J_R, \mbox{ for any } R \in RC(V).
\end{equation}

\begin{theorem}\label{pc11}  The linear maps

$${\bf \R}: JC(V) \lra RC(V) \mbox{ and } {\bf \J}:RC(V) \lra JC(V)$$
are mutually inverse vector space isomorphisms and ${\bf \R}$ carries $JC(V,g)$ isomorphically onto $RC(V,g)$ and likewise, ${\bf \J}$ carries $RC(V,g)$ isomorphically onto $JC(V,g).$

\end{theorem}

\begin{proof}
This is an immediate consequence of Proposition \ref{pc8} parts (5),(6) and (7).
\end{proof}

We note here that at each point of a semi-Riemannian manifold, the Riemann curvature operator is an acyclicic Riemann curvature operator on the tangent space at that point, and if $\R$ is the Riemann curvature operator (field) on such a manifold, then $J_{\R}$ is the Jacobi curvature operator (field), \cite{MTW}.  Of course, the inner product $g$ on $V$ gives naturally an inner product on all associated tensor products of $V$ and its dual and therefore on $PC(V),$ as well, but these mutually inverse isomorphisms are not isometric because they are compositions of isometric isomorphisms (due to permutations) and orthogonal projections, and the orthogonal projection factors prevent the isomorphisms from being isometric.  However, these mutually inverse isomorphisms mean that we can consider that the Riemann curvature operator and its associated Jacobi curvature operator both contain the exact same information and are thus sort of dual ways of keeping track of curvature.  In a given situation, one may be more useful than the other as far as a particular need is concerned.

\begin{corollary}\label{pc12}
If $U$ and $W$ are open subsets of $V$ and if $J_1$ and $J_2$ are acyclic Jacobi curvature operators with $J_k^*=J_k,$ for $k=1,2,$ and if $g(J_1u,u)w,w)=g(J_2(u,u)w,w)$ for every $(u,w) \in U \times W,$ then $J_1=J_2.$
\end{corollary}

\begin{proof}
We simply apply the observer principle in two stages.  For each choice of  fixed $u \in U,$ the hypothesis guarantees that the quadratic forms viewed as functions of $w \in W$ agree and therefore it follows from the observer principle that $J_1(u,u)=J_2(u,u),$ for each $u \in U,$ and then again by the observer principle applied directly now gives $J_1=J_2.$
\end{proof}

\begin{corollary}\label{pc13}

If $U$ and $W$ are open subsets of $V$ and if $J_1$ and $J_2$ are metric-$g$ acyclic Jacobi curvature operators with  $g(J_1u,u)w,w)=g(J_2(u,u)w,w)$ for every $(u,w) \in U \times W,$ then $J_1=J_2.$

\end{corollary}

\begin{proof}

If $J$ is a metric-$g$ acyclic Jacobi curvature operator, then $J$ is symmetric and $J^{\dag}=J,$ so therefore, $J=J^*.$  Now apply the previous corollary.
\end{proof}

\begin{corollary}\label{pc14}

If $U$ and $W$ are open subsets of $V$ and if $R_1$ and $R_2$ are metric-$g$ acyclic Riemann curvature operators with  $g(R_1(w,u)u,w)=g(R_2(w,u)u,w)$ for every $(u,w) \in U \times W,$ then $R_1=R_2.$

\end{corollary}

\begin{proof}

Apply \ref{jquad}, the theorem, and the previous corollary.

\end{proof}

From the preceding corollary, we see that metric acyclic Riemann curvature operators are characterized by their sectional curvatures because of their antisymmetry properties.

\begin{corollary}\label{pc15}
If $J$ is an acyclic Jacobi curvature operator, then 

$$Sym_3(J)=0,$$
or equivalently, $J(u,u)u=0,$ for all $u \in V.$
\end{corollary}

\begin{proof}

By the theorem, there is an algebraic Riemann curvature operator $R$ with $J=J_R,$ namely $R=R_J.$  Then, by  (\ref{jquad}), with the usual abuse of notation,

$$J(u,u,w)=J(u,u)w=R(w,u)u,$$
so by antisymmetry of $R$ it follows that $R(u,u)=0,$ so $J(u,u,u)=0,$ for all $u \in V,$ so by the observer principle, $Sym_3(J)=0.$

\end{proof}

\begin{corollary}\label{RRJJSHARPS}
If $J$ is a Jacocbi curvature operator and $R$ is a Riemann curvature operator with $J = J[[R]]$ and $R = R[[J]],$ then

$$J^{\#} + \frac{1}{2}R^{\#} = R \mbox{ and } J = \frac{3}{4} R^{\#} - \frac{1}{2} J^{\#}.$$

 \end{corollary}
 
 \begin{proof}
 By Corollary \ref{RRJJ} we have
 
 $$J^{\#} = \frac{3}{4} R - \frac{1}{2} J \mbox{ and } R^{\#} = J + \frac{1}{2} R,$$
 and the result is an immediate consequence of these two equations, by either adding them or simply using the fact that $\nu$ is a linear involution.
 
 \end{proof}

As a consequence of Corollary \ref{pc15}, if $J$ is an acyclic Jacobi curvature operator and $u \in V,$ then $u$ is in the kernel of $J(u,u)$ so if $J$ is a metric-$g$ acyclic Jacobi curvature operator, then as $J(u,u)$ is self-adjoint, it therefore defines a self-adjoint transformation of $u^{\perp},$ the orthogonal complement of $u$ in $V.$  In particular, if $V$ is a Lorentz vector space and $u$ is timelike, then $u^{\perp}$ is a Euclidean space and $J(u,u)$ being self-adjoint is thus diagonalizable by the spectral theorem.
Of course, the semi-Riemannian inner product $g$ on $V$ induces semi-Riemannian inner
products on all the various tensor products of $V$ and its dual $V^*=L(V; \bR),$ which we will also simply denote by $g$.  In fact, all the various natural ( in the sense of category theory) isomorphisms of these tensor products are all isometric.  Thus, in particular, the natural isomorphism of $V^* \otimes V$ onto $L(V;V)$  which sends $\mu \otimes v$ to the linear transformation whose value on $w$ is $\mu(w)v,$  is isometric.  For instance, it is easy to check that

$$trace(\mu \otimes v)=\mu(v), \mbox{ for any } \mu \in V^*, ~v \in V.$$

It is convenient to use the metric $g$ to explicitly define an isomorphism $\theta : V \lra V^*$ given by $\theta(x)=x^*,$ where $x^*(y)=g(x,y),$ for every $y \in V.$  In that case, the previous formula for trace gives

$$trace(x^* \otimes y)= g(x,y), \mbox{ for all } x,y \in V.$$

If $S^*$ denotes the adjoint with respect to $g$ of $S \in L(V;V),$ then the inner product on $L(V;V)$ due to $g$ on $V$ is easily seen to be

$$g(S,T)=trace(S^*T).$$

Likewise, we have $PC(V)$ naturally isomorphic with $[V^* \otimes V^*] \times [V^* \otimes V]$ which is in turn naturally isomorphic with $V^* \otimes V^* \otimes V^* \otimes V,$ and all these isomorphisms are isometric with respect to $g.$  Of course in a tensor product, the actual tensor products of various vectors are called elementary tensors, the full tensor product being the linear span of the elementary tensors.  Thus to check the equality of multi linear maps, it is enough to check on elementary tensors.  Also, as a special case, the natural isomorphism $PC(V) \lra L^3(V;V),$ is an isometric isomorphism.

\begin{proposition}\label{pc16}
The linear map $Cycl:PC(V) \lra PC(V),$  and the linear involutions $\tau, ~\alpha,~\chi, ~\nu$ of $PC(V)$ are all self adjoint.  
Consequently, $Q^{\pm}_{\a},Q^{\pm}_{\gt},Q^{\pm}_{\x},Q^{\pm}_{\nu}$ are all self adjoint idempotent operators and hence are orthogonal projection operators on $PC(V).$ Moreover, if $\sigma$ is a permutation in the permutation group $P(3),$ it induces a self adjoint isomorphism of $PC(V),$ as it permutes the three dual space factors of $PC(V).$
\end{proposition}

\begin{proof}
It is a simple matter to simply check on pairs of elementary tensors in $V^* \otimes V^* \otimes V^* \otimes V.$  Alternately, note that in $\A_P,$ all transpositions are yeir own inverse so are self adjoint, and by (\ref{cycl5}), as then $([1,2][1,3])^*=[1,3][1,2]$ it follows that $Cycl$ is self adjoint.
\end{proof} 

Now, everything that we have done so far for acyclicic curvature operators makes sense for semi-Riemannian vector bundles over any smooth manifold $M.$  Thus, if $E$ is any semi-Riemannian vector bundle over $M,$ then we have the semi-Riemannian vector bundle $PC(E)$ and the smooth self adjoint bundle involutions $\tau, ~\alpha,~\chi, ~\nu$ of $PC(E),$ over $M,$ and the previous definitions and propositions extend obviously to give results about algebraic curvature operator fields.  The specific case of when $E=TM,$ and we have the actual Riemann curvature operator of a connection on $M$ is what we will deal with after we make some remarks about how we can use these results to construct examples of curvature operators. 

\medskip

By Proposition \ref{pc8} and its Corollary together with Proposition \ref{pc16}, it follows that the operator $Q^+_{\a}Q^+_{\gt}$ is an orthogonal projection operator as the two factors are and they commute, and in turn, this commutes with $Q^+_{\x}$ so defining

\begin{equation}\label{pc17}
Q^+:=Q^+_{\x}Q^+_{\a}Q^+_{\gt}=Q^+_{\gt}Q^+_{\a}Q^+_{\x},
\end{equation}
results that $Q^+$ is also an orthogonal projection operator and

\begin{equation}\label{pc18}
Q^+Q^+_{\a}=Q^+=Q^+_{\a}Q^+,~Q^+Q^+_{\gt}=Q^+=Q^+_{\gt}Q^+, \mbox{ and } Q^+_{\x}Q^+=Q^+=Q^+Q^+_{\x}.
\end{equation}
Therefore, if $U \in PC(V)$ and $T=Q^+(U),$ then $\tilde{T}=T,~T^*=T,$ and $T^{\dag}=T.$  Thus, by Proposition \ref{pc8}, $R_T=(4/3)Q^-_{\gt}(T^{\#})=(4/3)Q^-_{\gt}(\nu(T))$ is a metric$-g$ algebraic Riemann curvature operator, so by Theorem \ref{pc11}, ${\bf \J}(R_T)$ is a metric$-g$ acyclic Jacobi curvature operator.  Since, $J_U=Q^+_{\gt}(U^{\dag})=Q^+_{\gt}(\nu(U)),$ this means finally, defining

\begin{equation}\label{pc19}
Q_J=\tilde{\bf \J}:PC(V) \lra PC(V), ~\tilde{\bf \J}=Q^+_{\gt} \nu (4/3)Q^-_{\gt} \nu Q ^+=\frac{4}{3} Q^+_{\gt} \nu Q^-_{\gt} \nu Q ^+,
\end{equation}
we have $\tilde{\bf \J }(T)=T,$ if $T \in JC(V,g).$ Indeed, if $T\in JC(V,g)$ is already a metric$-g$ acyclic curvature operator, then $Q^+(T)=T,$ so by Theorem \ref{pc11},

\begin{equation}\label{pc20}
\tilde{\bf \J}(T)={\bf \J} \circ {\bf \R}(T)=T.
\end{equation}

This means that $\tilde{\bf \J}$ is an idempotent operator.  On the other hand, using (\ref{pc18}), we have, since $Q^+(T)=T,$ for any $T \in JC(V,g)$

\begin{equation}\label{pc21}
\tilde{\bf \J}=\frac{4}{3}Q^+Q^+_{\gt} \nu Q^-_{\gt} \nu Q^+=\frac{4}{3}Q^+ \nu Q^-_{\gt} \nu Q^+,
\end{equation}
The last expression is obviously self adjoint as all the factors are and the symbols are a palindrome (other than the scalar factor), so this means that $\tilde{\bf \J}$ is a self adjoint idempotent, in other words, it is an orthogonal projection on $PC(V)$ with range $JC(V,g).$   Again appealing to Corollary \ref{cpc8}, we see that we could replace $Q^+_{\a}$ and $Q^+_{\gt}$ by $Q^-_{\a}$ and $Q^-_{\gt}$, respectively, and define likewise

\begin{equation}\label{pc22}
Q^-=Q^+_{\x} Q^-_{\a} Q^-_{\gt} \mbox{ and } Q_R=\tilde{\bf \R}=\frac{4}{3}Q^- \nu Q^+_{\gt} \nu Q^-,
\end{equation}
which is likewise the orthogonal projection on $PC(V)$ with range $RC(V,g).$

\begin{theorem}\label{pc23}
The linear operators

$$Q_J=\tilde{\bf \J}=\frac{4}{3}Q^+\nu Q^-_{\gt} \nu Q^+ \mbox{ and }Q_R= \tilde{\bf \R}=\frac{4}{3}Q^- \nu Q^+_{\gt} \nu Q^-$$
are self adjoint idempotent operators giving orthogonal projections of $PC(V)$ onto $JC(V,g)$ and $RC(V,g),$ respectively.  Moreover, $JC(V,g)$ and $RC((V,g)$ are orthogonal in $PC(V).$
\end{theorem}

\begin{proof}
The proof is in the remarks above, the orthogonality following form $Q^+Q^-=0,$ which causes 

$$Q_JQ_R=\tilde{\bf \J}  \tilde{\bf \R}=0.$$
\end{proof}

Using (\ref{pc19}) and the fact observed above resulting from Proposition \ref{pc8}, we saw that if $T \in PC(V),$ then $R=(4/3)Q^-_{\gt} \nu Q^+(T)$ is a metric$-g$ acyclic curvature operator and
$\tilde{\bf \J} (T)=J_R$ is the metric$-g$ acyclic Jacobi operator resulting from projecting $T$ on $JC(V,g)$ via $\tilde{\bf \J}.$  If we assume $T^*=T,$ then we can simply replace $Q^+$ in front with $Q^+_{\x}Q^+_{\gt},$ and calculating in $\A_P,$ we have

\begin{equation}\label{pc24}
\frac{4}{3} Q^+_{\gt} \nu Q^-_{\gt} \nu Q^+_{\x} Q^+_{\gt}=\frac{1}{6}[1+\gt - \nu \gt \nu-\gt \nu \gt \nu +\x +\gt \x -\nu \gt \nu \x-\gt \nu \gt \nu \x]Q^+_{\gt}
\end{equation}
If in addition we have $\tilde{T}=T,$ then we can drop the last factor of $Q^+_{\gt}$ and we can directly calculate with the permutations that

$$\gt=[1,2]$$

$$\nu \gt \nu=[2,3]$$

$$\gt \nu \gt \nu=[1,2,3]$$

$$\gt \nu \gt=[2,3]$$

$$\x=[1,3][2,4]$$

$$\nu \gt \nu \x=[1,2,4,3]$$

$$\gt \nu \gt \nu \x=[2,4,3].$$
Thus in case $[1,2]T=T$ and $T^*=T,$ we have

\begin{equation}\label{pc25}
%g(\tilde{\bf \J}(T)(u,v)w,x)=\frac{1}{6}[g(T(u,v)w,x)+g(T(v,u)w,x)-g(T(u,w)v,x)-g(T(v,w)u,x)
%+g(T(w,x)u,v)+g(T(w,x)v,u)-g(T(v,x)u,w)-g(T(u,x)v,w)]
g(\tilde{\bf \J}(T)(u,v)w,x)= \left\{  \begin{array}{rcl}
          (1/6)[g(T(u,v)w,x) & +g(T(v,u)w,x)\\ -g(T(u,w)v,x)  & -g(T(v,w)u,x) \\
        +g(T(w,x)u,v) &+g(T(w,x)v,u) \\-g(T(v,x)u,w) & -g(T(u,x)v,w)]
                \end{array}\right\} 
\end{equation}

A case of interest is the case where $T=h \otimes 1_V \in PC(V),$ with $h \in L^2(V;\bR),$ so

$$T(u,v)w=h(u,v)w.$$  Assuming that $h$ is symmetric, then $\tilde{T}=T$ and $T^*=T,$ obviously.  Let $H \in L(V;V)$ be the unique linear operator with $g(Hx,y)=h(x,y),$ for all $x,y \in V.$  Then

\begin{equation}\label{pc26}
\tilde{\bf \J}(T)(u,v)w=\left\{   \begin{array}{rcl}
           (1/6)[h(u,v)w & +h(v,u)w \\
           -h(u,w)v  & -h(v,w)u  \\
          + g(u,v)Hw  & +g(v,u)Hw  \\
           -g(u,w)Hv  &  -g(v,w)Hu].
           \end{array}\right\}
  \end{equation}
  \bigskip
  Moreover, $h(x,w)y=g(Hx,w)y=[(Hx)^*w]y=[(Hx)^* \otimes y ](w),$ so 
  
  $$h(x,w)y = [(Hx)^* \otimes y ](w) \mbox{ and } g(x,w)Hy = [x^* \otimes Hy](w).$$
  Therefore, \ref{pc26} becomes
  
  \begin{equation}\label{pcJ}
  \tilde{\bf \J}(T)(u,v)=\left\{   \begin{array}{rcl}
           (1/6)[h(u,v) id & +h(v,u)id \\
           -[H(u)^* \otimes v  & -[Hv]^* \otimes u  \\
          + g(u,v)H  & +g(v,u)H  \\
           -[u^* \otimes Hv]  &  -[v^* \otimes Hu].
           \end{array}\right\}
  \end{equation}
Letting $d=dim(V),$ we have $trace(h \otimes id)=d \cdot h.$  As before, $c(h)=trace(H)$ is the contraction of $h,$ and

$$trace([Hx]^* y)=[Hx]^*y=g(Hx,y)=h(x,y),$$
whereas
$$trace([x^* \otimes Hy]=x^*(Hy)=g(x,Hy)=g(Hy,x)=h(y,x) = h(x,y).$$
Thus, the trace of $\J(T)$ is now easy to calculate when $T= h \otimes id =h \otimes 1_V.$

   In general, we see that for any symmetric $h$ above, we have, still using $c(h)$ for the contraction of $h,$ and letting $d=dim(V),$
  
\begin{equation}\label{pc27}
\mbox{if } T=h \otimes 1_V, \mbox{ then }trace \circ  [\tilde{\bf \J}(T)]=\frac{1}{3}[(d-2)h+c(h) g].
\end{equation}
Since $c(g)=d,$ taking the case $h=g$ gives

$$trace \circ  [\tilde{\bf \J}(g \otimes 1_V)]=\frac{2}{3}[(d-1)g].$$
and therefore

$$\mbox{if } T_0=\frac{c(h)}{2(d-1)}[g \otimes 1_V], \mbox{ then } trace \circ \tilde{\bf \J}(T_0)=\frac{1}{3}c(h)g,$$
and of course, $R[[Q^+_{\x}(T_0)]]$ is an algebraic curvature tensor as is $R_T$, since $\tilde{T}=T$ and $T^*=T.$  Then

$$trace \circ \tilde{\bf \J}(T-T_0)=\frac{1}{3}(d-2)h.$$
Now set

\begin{equation}\label{pc28}
T_h=\frac{3}{d-2}[T-T_0]=\frac{3}{d-2}[h \otimes 1_V -\frac{c(h)}{2(d-1)}g \otimes 1_V].
\end{equation}
Then

\begin{equation}\label{pc29}
trace \circ \tilde{\bf \J}(T_h)=h, \mbox{ for any symmetric } h \in L^2(V;\bR).
\end{equation}
\bigskip

The relation of the Ricci tensor to the Riemann curvature operator becomes more direct through the Jacobi curvature operator, because $$Ricci=trace(J_R),$$ if $R$ is an algebraic Riemann curvature tensor, and in particular, we see that the Ricci tensor has no dependence on the metric tensor.  It is useful to be able to break away the essential part of the curvature operator which has the Ricci curvature.  The simplest way to do this is to take $h=Ricci,$ in the above calculations and then apply the trace. 
But, for any symmetric $h \in L^2(V;\bR),$ we know that $R[[Q^+_{\x}(T_h)]]$ is a metric$-g$ acyclic curvature operator whose Ricci tensor is now $$h=trace \circ {\bf \J}(T_h).$$  If we begin with any metric$-g$ acycliic curvature operator $R,$ then its Ricci tensor is $Ric_R=trace \circ J_R,$ so with $h=Ric_R,$ we have $trace \circ T_h=Ric_R,$ and therefore, $R-R[[Q^+_{\x}(T_h)]]$ has vanishing Ricci tensor.  Thus, $W_R=R-R[[Q^+_{\x}(T)]],$ where $T=T_{Ric_R}$ is the Weyl curvature operator for the metric$-g$ Riemann curvature operator $R.$  So, $W_R$ is a metric$-g$ acyclic Riemann curvature operator in $RC(V,g)$ whose Ricci tensor vanishes.  

Now, if $T=T^*$ and $\tilde{T}=T,$ then as we observed from Proposition \ref{pc8}, 

$$R=(4/3)Q^-_{\gt} \nu Q^+_{\x}(T)$$
is a metric$-g$ acyclic Riemann curvature operator and $J_R=\tilde{\bf \J}(T).$
Thus, to get $W_R,$ we need to compute

$$(4/3)Q^-_{\gt} \nu Q^+_{\x}=\frac{1}{3}[(1-\gt) \nu (1+ \x)]$$
$$=\frac{1}{3}[\nu - \gt \nu + \nu \x -\gt \nu \x].$$ 

Now,

$$\nu=[1,3],$$
$$\gt \nu=[1,2][1,3]=[2,1,3]$$
$$    \nu  \x=[1,3][1,3][2,4]=[2,4]$$
$$\gt \nu \x= [1,2][2,4]=[1,2,4].$$

Therefore,

$$g(R(u,v)w,x)=\frac{1}{3}[g(T(w,v)u,x)-g(T(w,u)v,x) +g(T(u,x)w,v)-g(T(v,x)w,u)].    $$
Taking the case where $T=h \otimes 1_V,$ where $h \in L^2(V;\bR)$ is symmetric, we have

$$g(R(u,v)w,x)=\frac{1}{3}[h(w,v)g(u,x)-h(w,u)g(v,x)+h(u,x)g(w,v)-h(v,x)g(w,u)].$$

Applying this to $T=T_h$ given by (\ref{pc28}), we get the difference $R_1-R_0,$ where $R_0$ simply results from the case of $h=[3c(h)/[2(d-2)(d-1)]g $ and $R_1$ results from the case where $h$ is replaced by $[3/(d-2)]h.$  Thus,

$$g(R_0(u,v)w,x)=\frac{c(h)}{2(d-2)(d-1)}[g(w,v)g(u,x)-g(w,u)g(v,x)+g(u,x)g(w,v)-g(v,x)g(w,u)]$$
and
$$g(R_1(u,v)w,x)=\frac{1}{d-2}[h(w,v)g(u,x)-h(w,u)g(v,x)+h(u,x)g(w,v)-h(v,x)g(w,u)].$$
Finally, this means that the Weyl curvature operator for the metric$-g$ acyclic curvature operator $R$ must be given by

\begin{equation}
W_R= R-R_1+R_0, \mbox{ with } h=Ricci \mbox{ so } c(h) \mbox { is the scalar curvature.}
\end{equation}

\begin{proposition}\label{nuQT}

If $T \in PC(V)$ with $Q^{\pm}T=T,$ by which we mean that either $Q^+T=T$ {\bf or} $Q^-T=T,$ then

$$(\tau \alpha) T^{\#} = T^{\#}  \mbox{ and } \tau [T^{\#}] = ( T^{\#}\tilde{)}  = \alpha [T^{\#}] = (T^{\#})^*,$$
and moreover, $T^{\#}$ has the metric exchange property,

$$(T^{\#})^{\dag} = T^{\#}.$$

\end{proposition}

\begin{proof}
This is an immediate consequence of (\ref{nuQ}) and the fact that $\chi$ and $\nu$ commute.

\end{proof}

We remark here, that any metric Jacobi or metric Riemann curvature operator satisfies the hypothesis of Proposition \ref{nuQT}.

%%%%%%%%%%%%%%%%%%%%%%%%%%%%%%%%%%%%%%%%%%%%%%%%%%%%%%%%%%%%%%%%%%%%%%%%%%%%%%%%%%%%%%%%%%%%%%%%%%%%%%%%%%%%%%%%%%%%%%%%%%%%%%%%%%%%%%%%%%%%%%%%%%%%%%%%%%%%%%%%%%%%%%%%%%%%%%%%%%%%%%%%%%%%%%%%%%%

%%%%%%%%%%%%%%%%%%%%%%%%%%%%%%%%%%%%%%%%%%%%%%%%%%%%%%%%%%%%%%%%%%%%%%%%%%%%%%%%%%%%%%%%%%%%%%%%%%%%%%%%%%
%%%%%%%%%%%%%%%%%%%%%%%%%%%%%%%%%%%%%%%%%%%%%%%%%%%%%%%%%%%%%%%%%%%%%%%%%%%%%%%%%%%%%%%%%%%%%%%%%%%%%%%%%%
%%%%%%%%%%%%%%%%%%%%%%%%%%%%%%%%%%%%%%%%%%%%%%%%%%%%%%%%%%%%%%%%%%%%%%%%%%%%%%%%%%%%%%%%%%%%%%%%%%%%%%%%%%

\section{THE RIEMANN AND JACOBI CURVATURE OPERATOR FIELDS}

\med

In this section we give a fairly self-contained discussion of the Riemann and Jacobi operator fields on a semi-Riemannian manifold for the reader who does not want to read the previous section on algebraic curvature operators.  However, some of the results of the previous section do make it easier to manipulate curvature operators.

For $v \in T_mM,$ we denote by $\nabla_v$ the covariant
differentiation operator with respect to $v$ at $m.$ We have then
the Riemann curvature operator field, $\R,$ given by

\begin{equation}\label{curvatureoperator}
\R(u,v)=[\nabla_u,\nabla_v]-\nabla_{[u,v]},
\end{equation}
where $u$ and $v$ are any tangent vector fields on an open subset
$U$ of $M.$ We note that $\R(u,v)$ actually defines a vector
bundle map of the tangent bundle $TM|U$ to itself covering the
identity map of $U,$ and it as well then determines the Riemann
curvature tensor, $Riemann,$ of fourth rank, which means that $\R$
is itself an alternating second rank tensor field on $M$ which at
each point $m \in M$ gives a linear transformation valued tensor
on $T_mM.$ In particular, this means that $\R(u,v)$ is defined,
giving a linear transformation of $T_mM$ for any pair of tangent
vectors $u,v \in T_mM.$  Thus, $\R(u,v)$ is a smooth section of the vector bundle $L(\tau_M;\tau_M),$ and therefore $\R$ itself is an $L(\tau_M;\tau_M)-$valued 2-form on $M.$ At each point $q \in M,$ the Riemann curvature operator field defines a metric algebraic Riemann curvature operator $\R[q]$ on the tangent space $T_qM$ with respect to the inner product $g[q].$  If we denote $L=L(\tau_M;\tau_M)$ for short, for the moment, then the second Bianchi identity is simply
\begin{equation}\label{2Bianchi}
d_L \R=0,
\end{equation}
where $d_L$ is the exterior covariant derivative operator on alternating $L-valued$ forms.  For more on this, see Appendix I.  In any case, the vanishing of an exterior covariant derivative should be viewed as a generalized conservation law, so we view the second Bianchi identity as a generalized conservation law for the Riemann curvature operator, $\R.$  In fact, the fact that $d\R=0$ is the basis for the characteristic cohomology classes which are actual closed differential forms and can therefore be integrated.  

One of our main concerns is the certain
contraction of $Riemann$ known as the Ricci tensor, $Ric.$ In fact
in any frame at $m \in M$ with basis $(e_{\alpha})$ for $T_mM$ and
dual basis $(\omega^{\alpha}),$ we have, using the summation
convention,

\begin{equation}\label{ricci1}
Ric(u,v)=\omega^{\alpha}(\R(e_{\alpha},u)v),~~u,v \in T_mM.
\end{equation}
Our notation is chosen to emphasize we are not restricting
ourselves to coordinate frames nor to orthonormal frames unless
explicitly stated. Thus, we will refrain from using the abstract
index notation, as it is too often restricted to imply coordinate
framing. In particular, for any pair of tangent vectors $v,w \in
T_mM,$ the curvature operator defines another linear
transformation $\K(v,w)$ of $T_mM$ defined by

\begin{equation}\label{curv1}
\K(v,w)z=\R(z,v)w=R^{\#}(w,v)z,\,\ z \in T_mM, \mbox{ so } \K(v,w)=R^{\#}(w,v).
\end{equation}
Then, using the basic symmetries of the curvature tensor, one
finds easily

\begin{equation}\label{curv2}
\K(v,w)^*=\K(w,v),\, v,w \in T_mM.
\end{equation}

But, in terms of our previous work on algebraic curvature operators and Proposition \ref{pc8}, we know that as 

$$Cycl(\R)=0,~~ \tilde{\R}=-\R, \mbox{ and } \R^{\dag}=\R,$$ 
it follows that 

$$\J=Sym(\R^{\#})$$
is an algebraic Jacobi operator (field).  Moreover, in view of Proposition \ref{pc8} and Theorem \ref{pc11}, we see right away that $\K^*=\tilde{\K}=\R^{\#},$ so that

\begin{equation}\label{algcurvop}
\J=Sym(\tilde{\K})=Sym(\K) \mbox{ and }  \R=[\R]_{\J}=(4/3)Alt([\J]^{\#}).
\end{equation}

We also now have

\begin{equation}\label{curv3}
Ric(v,w)=trace [\,\K(v,w)],
\end{equation}
which means in terms of $\K$ we simply find $Ric=trace \circ \K.$
The symmetry of $Ric$ then follows immediately from (\ref{curv2}).
Moreover, if $u$ is any tangent vector, then $\K(u,u)$ is
self-adjoint or symmetric, clearly vanishes on the line through
$u,$ and therefore has $u^{\perp}$ as an invariant subspace. Thus
$\K(u,u)$ really "lives" on $u^{\perp},$ the orthogonal complement
of $u$ in $T_mM.$ We shall denote by $A_u^{(geo)}$ the restriction
of $-\K(u,u)$ to $u^{\perp}.$  Thus, in the Lorentz case, if $u$ is timelike, then $A_u^{(geo)}:u^{\perp} \lra
u^{\perp}$ is a self-adjoint linear transformation of the
Euclidean space $u^{\perp}.$ Thus for this case that $u$ is a time-like unit
vector, the metric tensor is positive definite on this orthogonal
complement, and it follows that $Ric(u,u)$ is simply the sum of
the eigenvalues of $-A_u$ or of $\K(u,u).$ In general, in the Lorentz case, if
$(u,e_1,e_2,...,e_n)$ is an orthonormal frame with $u$ a time-like
unit vector, then we note that

\begin{equation}\label{sectcurv1}
g(e_k,\R(e_k,u)u)=g(e_k,\K(u,u)e_k)
\end{equation}
is the negative of the Riemann sectional curvature of the span of
$u$ and $e_k,$ in $T_mM,$ because $g(u,u)=-1.$ Thus, the
eigenvalues of $A_u^{(geo)}$ are the principal Riemann sectional
curvatures through $u.$  We can now symmetrize and define the {\it Jacobi curvature operator field}, $\J=Sym(\K),$ so

\begin{equation}\label{symcurv}
\J(v,w)=Sym(\K(v,w)=\frac{1}{2}[\K(v,w)+\K(w,v)],\, v,w \in T_mM.
\end{equation}
We see immediately from (\ref{curv2}) that $\J$ is a symmetric
$L(\tau_M;\tau_M)-$valued tensor whose values are themselves
self-adjoint transformations of $T_mM.$  In particular, at each point $q \in M,$ the Jacobi operator field evaluated at that point, $\J[q],$ defines the algebraic Jacobi curvature operator due to the algebraic Riemann curvature operator $\R[q].$  Moreover, we also have
$\J(v,v)=\K(v,v)$ for each $v \in T_mM,$ whereas, $Ric(v,w)=trace
\, \J(v,w),$ for any $v,w \in T_mM.$  Consequently, in terms of the Jacobi operator we have

\begin{equation}\label{RiceqtraceJacobi}
Ric=trac \circ \J.
\end{equation}

Notice that the observer principle can be applied to $\J,$  the Jacobi curvature operator (\ref{symcurv}), as well as to $Ric,$ as tensors on $T_mM.$ Thus,
the observer principle says in a sense that these symmetric
tensors are {\it observable}, in the sense that they are
completely determined at a given event by knowing how all
observers at the event see their monomial forms.

As before, for any permutation $\sigma$ of $\{1,2,...,r\},$ and any tensor $A$ of rank $r,$ let $\sigma(A)$ denote the result of permuting the arguments of $A$ via $\sigma,$ so

$$[\sigma(A)](v_1,v_2,...,v_r)=A(v_{\sigma(1)},v_{\sigma(2)},...,v_{\sigma(r)}).$$ 
Let $sgn(\sigma)$ denote the sign of the permutation $\sigma,$ defined as $+1$ in case $\sigma$ is an even permutation and $-1$ otherwise, so $sgn$ is a group homomorphism of the permutation group into the group $\{-1,1\}.$  Thus, the alternation operator $Alt$ is defined by

$$Alt(A)=\frac{1}{r!}\sum sgn(\sigma)[\sigma(A)].$$
Let $\sigma_1$ denote the permutation of $\{1,2,3\}$ which simply interchanges 2 with 3, fixing 1, so $\sigma_1=[2,3].$  Then in case $r=3,$ the full permutation group consists merely of the cyclic subgroup of order 3 generated by a non-trivial cycle and its coset of $\sigma_1$ consisting of the three odd permutations.  Thus, for any third rank tensor $A,$ we have

$$Alt(A)=\frac{1}{3!}[Cycl(A)-Cycl(\sigma_1(A))].$$  

From $\R$ and $\J$ we form the third rank $TM-$valued covariant tensors ${\bf R}$ and ${\bf J}$ by evaluation, so

\begin{equation}\label{Riemannthirdrank}
{\bf R}(w,u,v)=\R(u,v)w,
\end{equation}
and

\begin{equation}\label{Jacobithirdrank}
{\bf J}(w,u,v)=\J(u,v)w.
\end{equation}

Thus,

\begin{equation}\label{thirdrank2}
([1,3,2]{\bf R})(u,v,w)=\R(u,v)w \mbox{ and } ([1,3,2]{\bf J})(u,v,w)=\J(u,v)w.
\end{equation}

In particular, as $\R$ is anti-symmetric we have $\sigma_1({\bf R})=-{\bf R}$ and therefore

$$Alt({\bf R})=\frac{1}{3}Cycl({\bf R}).$$

The first Bianchi identity is usually stated as $Cycl({\bf R})=0,$ so we see it is equivalent to

\begin{equation}\label{first Bianchi}
Alt({\bf R})=0.
\end{equation}
Similarly, the second Bianchi identity $d_L\R=0$ is now seen equivalent to $Cycl(\nabla \R)=0,$ which is the form in which it is usually stated.

The {\it Jacobi curvature tensor}, $J,$ is customarily defined as the fourth rank tensor given by

\begin{equation}\label{jacobitensor}
J(\omega,u,v,w)=\omega(\J(v,w)u)=\omega ({\bf J}(u,v,w)).
\end{equation}
Clearly knowing any one of $\R,~{\bf R},$ and $Riemann$ determines all the others and likewise for $\J,~{\bf J},$ and $J.$  But it is useful to keep in mind that $\R, ~{\bf R},$ and $Riemann$ are all different mathematical objects, and likewise for $\J,~{\bf J},$ and $J.$  As $\R$ determines $\K$ which determines $\J,$ it is clear that $\R$ determines $\J.$  However, by Theorem \ref{pc11}, $\R$ has enough symmetries to in turn be determined by $\J.$  Since  $\J=Sym(\K),$ for those who skipped the preceding section,
$$\J(u,v)w=\frac{1}{2}[\R(w,v)u+\R(w,u)v],$$ and therefore
$$2\J(w,v)u-2\J(w,u)v=\R(u,w)v+\R(u,v)w-\R(v,w)u-\R(v,u)w$$
$$=\R(u,w)v+2\R(u,v)w-\R(v,w)u,$$  where we have used the anti-symmetry of $\R$ to combine the second and fourth terms. Using the first Bianchi identity, the third term can be replaced by $\R(u,v)w+\R(w,u)v$ giving
$$2[\J(w,v)u-\J(w,u)v]=\R(u,w)v+2\R(u,v)w+\R(u,v)w+\R(w,u)v=3\R(u,v)w.$$
Thus (see \cite{MTW}, (11.37), page 287)

\begin{equation}\label{R from J}
\R(u,v)w=\frac{2}{3}[\J(w,v)u-\J(w,u)v]=\frac{2}{3}[\J(v,w)u-\J(u,w)v]=\frac{2}{3}[\J(v,w)u-\J(w,u)v].
\end{equation}
The last expression can be viewed as saying to get the Riemann curvature operator from the Jacobi curvature operator acting on a vector, the result is two thirds the result of cycling left minus cycling right.
Designating $\sigma_3=[1,2]$ as the transposition which merely interchanges 1 an 2 and with $\lambda=[1,2,3]$ as the appropriate cycle,
$${\bf R}=\frac{2}{3}\lambda[{\bf J}-\sigma_3({\bf J})].$$
%%%
Indeed, direct calculation gives, using the symmetry of $\J,$ and either (\ref{R from J}) or (\ref{algcurvop}),

$$\frac{2}{3} \lambda [{\bf J} -\sigma_3({\bf J})](u,v,w)=\frac{2}{3}[{\bf J}-\sigma_3({\bf J})](v,w,u)=\frac{2}{3} [{\bf J}(v,w,u)-{\bf J}(w,v,u)]$$

$$=\frac{2}{3}[\J(w,u)v-\J(v,u)w]=\frac{2}{3}[\J(u,w)v-\J(u,v)w]=\frac{2}{3} [\J^{\#}(v,w)u-\J^{\#}(w,v)u] $$

$$=\frac{4}{3} [Alt(\J^{\#})](v,w)u=\R(v,w)u={\bf R}(u,v,w).$$

We need to emphasize here that so far, other than or discussion of the relation to sectional curvatures, this correspondence between Riemann curvature operators and tensor fields and Jacobi curvature operator and tensor fields is completely independent of any metric and only requires the connection be torsion free so as to obtain the first Bianchi identity.  However, in case of the Levi Civita connection on a semi Riemannian manifold, the Riemann curvature operator also satisfies the metric exchange property

$$g(\R(u,v)w,x)=g(\R(w,x)u,v),$$
so that $\R$ is a field of metric-$g$ algebraic Riemann curvature operator fields, and consequently $\J$ is also a metric-$g$ Jacobi curvature operator field and thus also satisfies the metric exchange property

$$g(\J(u,v)w,x)=g(\J(w,x)u,v),$$
valid for all vector fields $u,v,w,x$ on the semi Riemannian manifold $M,$ and these vector fields do not even need to be continuous.

\bigskip

%%%%%%%%%%%%%%%%%%%%%%%%%%%%%%%%%%%%%%%%%%%%%%%%%%%%%%%%%%%%%%%%%%%%%%%%%%%%%%%%%%%%%%%%%%%%%%%%%%%%%%%%%%%%%%%%%%%%%%%%%%%%%%%%%%%%%%%%%%%%%%%%%%%%%%%%%%%%%%%%%%%%%%%%%%%%%%%%%%%%%%%%%%%%%%%%%%%%%%%%%%%%%%%%%%%%%%%%%%%%%%%%%%%%%%%%%%%%%%%%%%%%%%%%%%%%%%%%%%%%%%%%%%%%%%%%%%%%%%%%%%%%%%%%%%%%%%

%%%%%%%%%%%%%%%%%%%%%%%%%%%%%%%%%%%%%%%%%%%%%%%%%%%%%%%%%%%%%%%%%%%%%%%%%%%%%%%%%%%%%%%
%%%%%%%%%%%%%%%%%%%%%%%%%%%%%%%%%%%%%%%%%%%%%%%%%%%%%%%%%%%%%%%%%%%%%%%%%%%%%%%%%%%%%%%
%%%%%%%%%%%%%%%%%%%%%%%%%%%%%%%%%%%%%%%%%%%%%%%%%%%%%%%%%%%%%%%%%%%%%%%%%%%%%%%%%%%%%%%

\section{COMPARISON OF KOZUL CONNECTIONS AND CURVATURE OPERATORS}

Suppose that $B$  is a connection on a smooth manifold $M$ and we are presented with a new Kozul connection $\nabla$ on $M$.  Their difference $H$ is actually a tensor field on $M.$  That is, with

$$H=\nabla - B, \mbox{ so } H_u v=\nabla_u v- B_u v, \mbox{ for all smooth tangent vector fields } u \mbox{ and } v \mbox{ on } M,$$

$$H(u,v)=H_u v$$
is bilinear over the ring of smooth functions on $M$ and therefore defines a tensor field on $M.$  Indeed, if $f$ is any smooth function on $M,$ then obviously

$$H(fu,v)=fH(u,v), \mbox{ as } \nabla_{fu} v= f\nabla_u v  \mbox{ and } B_{fu} v=f B_u v,$$
but also,
$$H(u,fv)=[D_u f]v + f\nabla_u v -[D_uf]v -fB_u v= f [\nabla_u v -B_u v].$$
Thus, we have

$$\nabla=B +H, \mbox{ and } H \mbox{ is a smooth section of } L(TM;L(TM;TM))$$
To compare the curvature operators of $\nabla$ and $B$, we can write

$$\R_{\nabla}(u,v)=[\nabla_u,\nabla_v]-\nabla_{[u,v]}$$

$$\R_B(u,v)=[B_u,B_v]-B_{[u,v]}, \mbox{ and } \R_H(u,v)=[H_u,H_v]-H_{[u,v]}.$$

Thus,

$$\R_{\nabla}(u,v)=[B_u+H_u,B_v+H_v]-B_{[u,v]} -H_{[u,v]}$$

$$\mbox{so, }\R_{\nabla}(u,v)=\R_B(u,v)+\R_H(u,v)+[B_u,H_v]+[H_u,B_v]$$

$$\mbox{and thus, } \R_{\nabla}(u,v)=\R_B(u,v)+\R_H(u,v)+[B_u,H_v]-[B_v,H_u]$$
Now, as $B$ is a connection, it naturally extends to all tensor fields on $M$ via the product rule for differentiation, so for any smooth tangent vector fields $x,y,z,$

$$(B_x[H_y])z=B_x [H_y z] -B_x [H_y z]=[B_x,H_y]z.$$
The result is that

\begin{equation}\label{cnxcurv1}
\R_{\nabla}(u,v)=\R_B(u,v)+\R_H(u,v)+B_u[H_v]-B_v[H_u].
\end{equation}
Of course, the usual application of (\ref{cnxcurv1}) is to the case where $M$ is replaced by one of its open subsets, say $U,$ diffeomorphic to an open subset of a Euclidean space and $B=D$ is the flat connection resulting from the transfer of the standard flat connection on the Eucllidean space to $U.$  Then $\R_B=\R_D=0,$ and the result becomes simply

$$\R_{\nabla}(u,v)=\R_H(u,v) +D_u[H_v]-D_v[H_u].$$
In particular,

\begin{equation}\label{cnxcurv2}
\R_{\nabla}(u,v)=[H_u,H_v]+D_u[H_v]-D_v[H_u], \mbox{ if } [u,v]=0,
\end{equation}
as is the case when $u$ and $v$ are chosen as coordinate vector fields from the diffeomorphism of $U$ with an open subset of Euclidean space.  Notice that in a local coordinate system, the components of $H$ will be the usual Christoffel symbols which are generally thought of as not defining a true tensor.  That is clearly not the case as we see here.  The reason for a problem here is that if we change the diffeomorphism of $U$ into an open subset of a Euclidean space, then the transferred flat connection is different, it is no longer $D,$ but rather a new flat connection say $D'$ and then the difference $\nabla-D'$ is no longer equal to $H$ but gives rather a new difference tensor $H'.$  It is interesting here to note that some mathematicians, physicists, and philosophers had objected to Einstein's criterion of general covariance for the expression of physical laws on the grounds that once a specific coordinate system is chosen, any specification of components defines a tensor in that coordinate system and therefore in all coordinate systems, so that consequently, the covariance criterion was meaningless.  In fact, the criterion was far from meaningless, it is just that it was not sufficiently spelled out.  If you define a tensor by choosing a specific coordinate system, it is then defined in all coordinate systems, but the real idea of covariance is that the tensor has invariant natural meaning independent of any choice of coordinates.  Thus, choosing a coordinate system to produce a flat connection $D$ does not produce a natural $H$ in any real sense, it is completely dependent on the coordinate system used to construct $D.$  Thus, we see that true physical laws should be formulated without coordinates at all.  In order to be able to do that, one needs to be proficient in understanding differential geometry without recourse to coordinates.  Moreover, it is only by working without coordinates that we begin to see the mathematical structures with which we are dealing which enables us to properly formulate laws and strategies for using them.  A student of differential geometry limited to index manipulations seeing tensors as merely indexed collections of numbers cannot properly appreciate the structures he is dealing with and is analogous to an art student who only sees a painting as a collection of colors and who thus cannot tell the significance of the difference between a Rembrandt and a Jackson Pollock.

\bigskip

%%%%%%%%%%%%%%%%%%%%%%%%%%%%%%%%%%%%%%%%%%%%%%%%%%%%%%%%%%%%%%%%%%%%%%%%%%%%%%%%%%%
%%%%%%%%%%%%%%%%%%%%%%%%%%%%%%%%%%%%%%%%%%%%%%%%%%%%%%%%%%%%%%%%%%%%%%%%%%%%%%%%%%%%
%%%%%%%%%%%%%%%%%%%%%%%%%%%%%%%%%%%%%%%%%%%%%%%%%%%%%%%%%%%%%%%%%%%%%%%%%%%%%%%%%%%

%%%%%%%%%%%%%%%%%%%%%%%%%%%%%%%%%%%%%%%%%%%%%%%%%%%%%%%%%%%%%%%%%%%%%%%%%%%%%%%%%%%
%%%%%%%%%%%%%%%%%%%%%%%%%%%%%%%%%%%%%%%%%%%%%%%%%%%%%%%%%%%%%%%%%%%%%%%%%%%%%%%%%%%%
%%%%%%%%%%%%%%%%%%%%%%%%%%%%%%%%%%%%%%%%%%%%%%%%%%%%%%%%%%%%%%%%%%%%%%%%%%%%%%%%%%%

\section{COMPARISON OF CONNECTIONS FROM TWO METRIC TENSORS ON A MANIFOLD}

Here we consider the general problem of relating two metric tensors $g$ and $h$ and their resulting connections and curvature operators, on any semi Riemannian manifold $M.$  In fact, the two metric tensors need not have the same signature, as in the case of vector bundles where it was useful to consider a new positive definite metric tensor on the bundle.  But here, we just have two metric tensors $g$ and $h.$  Let us think of $h$ as the original metric tensor and $g$ as some new metric tensor, both are semi Riemannian.  Now, since both are non-degenerate, this means there is a unique isomorphism $G$ of the tangent bundle of $M$ so that 

$$g(u,v)=h(Gu,v), \mbox{ for all smooth tangent vector fields } u,v \mbox{ on } M.$$
Thus, as we are assuming our metric tensors are smooth, this means that $G$ is a smooth section of $L(TM;TM)$ whose values are invertible linear transformations at each point of $M.$  Moreover, symmetry of $g$ means that $G$ is self adjoint relative to $h,$ so $h(Gu,v)=h(u,Gv),$ for all vector fields $u$ and $v$ on $M.$

Suppose now that $\nabla^g$ is a connection on $M$ with $\nabla^g g=0.$  Let $T$ be the torsion tensor for $\nabla^g.$  Then of course we have for any smooth tangent vector fields $u,v,w,$

\begin{equation}\label{geo6}
g(\nabla^g_u v,w)=\frac{1}{2} [ D_u [g(v,w)] +D_v [g(u,w)] - D_w [g(u,v)] +g([w,u],v)+g([w,v],u)+g(w,[u,v]))
\end{equation} 
$$+g(T(w,u),v)+g(T(w,v),u)+g(w,T(u,v)) ].$$ 
and (\ref{geo6}) then determines $\nabla^g$ if $T$ is given.  And the same equation works for $h,$ if $\nabla^h=\nabla$ is the Levi Civita connection for $h,$ so in that case, torsion vanishes.  So let us assume that the torsion $T=0$ as well, so that $\nabla^g$ is the Levi Civita connection for $g$ as well.   We then have, expressing $g$ in terms of $h$ and $G,$

\begin{equation}\label{geo7}
h(\nabla^g_u v,Gw)=\frac{1}{2} [D_u [h(Gv,w)] + D_v [h(Gu,w)]-D_w [h(Gu,v)]+h([w,u],Gv)+h([w,v],Gu)+h(Gw,[u,v]).
\end{equation}
and likewise

\begin{equation}\label{geo8}
h(\nabla_u v,Gw)=\frac{1}{2} [D_u [h(v,Gw)]+D_v[h(u,Gw)-D_{Gw}[h(u,v)]+h([Gw,u],v)+h([Gw,v],u)+h(Gw,[u,v])]
\end{equation}
Since $\nabla h=0$ and again, $\nabla$ is torsion free, it follows for any tangent vector fields $x,y,z$ we have 

\begin{equation}\label{geo9}
D_x [h(Gy,z)]=h([\nabla_xG]y,z)+h(G \nabla_x y,z) + h(Gx, \nabla_x z) \mbox{ and } [x,y]=\nabla_x y - \nabla_y x,
\end{equation}
$$\mbox{and therfore } \nabla_{Gw} x=\nabla_x [Gw] +[Gw,x]=[\nabla_xG]w+G\nabla_x w +[Gw,x].$$
Applying the relations (\ref{geo9}), we find

\begin{equation}\label{geo10}
D_{Gw}[h(u,v)]=h([\nabla_uG]w,v)+h(G\nabla_u w,v)+h([Gw,u],v)
\end{equation}
$$  +h(u,[\nabla_vG]w)+h(u,G\nabla_v w)+ h(u,[Gw,v]). $$
When we substitute (\ref{geo10}) into (\ref{geo8}) and cancelling, we then have

\begin{equation}\label{geo11}
h(\nabla_u v,Gw)=\frac{1}{2} [D_u [h(Gv,w)] + D_v [h(Gu,w)]+h(Gw,[u,v])
\end{equation}
$$-h([\nabla_uG]w,v)-h(G\nabla_u w,v)-h(u,[\nabla_vG]w)-h(u,G\nabla_v w)].$$
We can now subtract (\ref{geo11}) from (\ref{geo7}) and cancel to find (remember $G$ is self adjoint)

\begin{equation}\label{geo12}
h(\nabla^g_u v -\nabla_u v, Gw)=\frac{1}{2} [ h([w,u],Gv)+h([w,v],Gu)-D_w [h(Gu,v)]
\end{equation}
$$+h([\nabla_uG]w,v)+h(\nabla_u w,Gv)+h(u,[\nabla_vG]w)+h(Gu,\nabla_v w)].$$
Now, as $[w,x]+ \nabla_x w=\nabla_w x$ for any tangent vector field $x,$ we find (\ref{geo12}) simplifies to

\begin{equation}\label{geo13}
h(\nabla^g_u v -\nabla_u v, Gw)=\frac{1}{2} [h(\nabla_w u,Gv)+h(\nabla_w v,Gu)-D_w[h(Gu,v)]+h([\nabla_u G]w,v)+h(u,[\nabla_v G]w)]
\end{equation}
Again using $\nabla h=0,$

\begin{equation}\label{geo14}
D_w [h(Gu,v)]=h([\nabla_w G]u,v)+h(G\nabla_w u,v)+h(Gu, \nabla_w v),
\end{equation}
so with this final cancellation in (\ref{geo13}), the result is
\begin{equation}\label{geo15}
h(\nabla^g_u v -\nabla_u v, Gw)=\frac{1}{2} [h([\nabla_u G]w,v)+h(u,[\nabla_v G]w)-h([\nabla_w G]u,v)]
\end{equation}
Now, as $G$ is self adjoint, it follows that $\nabla_x G$ is also self adjoint for any tangent vector field $x$ on $M,$  so

\begin{equation}\label{geo16}
h(G[\nabla^g_u v -\nabla_u v], w)=\frac{1}{2} [h([\nabla_u G]v,w)+h([\nabla_v G]u,w)-h([\nabla_w G]u,v)]
\end{equation}
If we express $\nabla G$ as a section of $L(TM,TM;TM)$ with $[\nabla G](x,y)=[\nabla_y G]x,$ then view $$[\nabla G] \in L(TM;L(TM;TM)),$$ so $[\nabla G]x \in L(TM;TM)$ and with only a slight violation of standard notational practice, we write 

$$([\nabla G]x)y=[\nabla_y G]x,$$
so then we can write $$h([\nabla_w G]u,v)=h([[\nabla G]u]w,v)=h([[\nabla G]u]^*v,w)$$ 
and  then (\ref{geo16}) becomes 

\begin{equation}\label{geo17}
h(G[\nabla^g_u v -\nabla_u v], w)=\frac{1}{2} [h([\nabla_u G]v,w)+h([\nabla_v G]u,w)-h([[\nabla G]u]^*v,w)].
\end{equation}
As this is valid for all tangent vector fields $u,v,w,$ it follows that

\begin{equation}\label{geo18}
G[\nabla^g_u v -\nabla_u v] =\frac{1}{2} [[\nabla_u G]v+[\nabla_v G]u-[[\nabla G]u]^*v].
\end{equation}
and finally,

\begin{equation}\label{geo18}
\nabla^g_u v -\nabla_u v=\frac{1}{2} G^{-1} [[\nabla_u G]v+[\nabla_v G]u-[[\nabla G]u]^*v].
\end{equation}

\bigskip

%%%%%%%%%%%%%%%%%%%%%%%%%%%%%%%%%%%%%%%%%%%%%%%%%%%%%%%%%%%%%%%%%%%%%%%%%%%%%%%%%%%%%%%%%%%
%%%%%%%%%%%%%%%%%%%%%%%%%%%%%%%%%%%%%%%%%%%%%%%%%%%%%%%%%%%%%%%%%%%%%%%%%%%%%%%%%%%%%%%%%%%
%%%%%%%%%%%%%%%%%%%%%%%%%%%%%%%%%%%%%%%%%%%%%%%%%%%%%%%%%%%%%%%%%%%%%%%%%%%%%%%%%%%%%%%%%%%
%%%%%%%%%%%%%%%%%%%%%%%%%%%%%%%%%%%%%%%%%%%%%%%%%%%%%%%%%%%%%%%%%%%%%%%%%%%%%%%%%%%%%%%%%%%

\section{GEOMETRIZATION OF ALGEBRAIC CURVATURE OPERATORS}

If $V$ is a vector space and $B$ is an algebraic curvature operator on $V,$ in either Jacobi or Riemann case, then it is natural to ask if there is a manifold $M$ with Kozul connection $\nabla$ such that the curvature operator due to $\nabla$ of the corresponding type is $B.$  According to our result on algebraic curvature operators, we can begin with an algebraic Jacobi curvature operator $J,$ for if $B$ is a Riemannian algebraic curvature operator and if $\R_{\nabla}(q)=B,$  then $\J_{\nabla}(q)=J[[B]],$ so if we used $B=R[[J]],$ then $J[[B]]=J$ and we have found $\J_{\nabla}(q)=J.$  In this situation we would say we have a geometric model of $J,$ or in short, we have geometrized $J.$  If we wish to geometrize a Riemann algebraic curvature operator $B,$ and if we find $\nabla$ with  $\J_{\nabla}(q)=J[[B]],$ then $\R_{\nabla}(q)=R[[J[[B]]]]=B,$ so we have geometrized $B.$  That is, geometrizing a curvature operator in either the Jacobi or Riemann type automatically geometrizes both corresponding types because of Theorem \ref{pc11}.  So what we will find convenient is to begin with $J$ and simply find $\nabla$ with $\R_{\nabla}(q)=R[[J]].$  When there is no metric to deal with, it is easy to do.

To begin then, let us simply assume that $J$ is a symmetric precurvature operator on $V,$ so 

$$J \in L_{sym}^2(V;L(V;V)).$$
We will find $\nabla$ which geometrizes $R[[J]].$

We take $V$ itself as the manifold with its trivial tangent bundle, so we can identify vector fields with their principal parts which are then simply $V$ valued functions on open subsets of $V.$  We will simply take $q=0.$  Thus, to define $\nabla$ we have to define $\nabla_v w$ for any function $v$ on an open subset $U$ of $V$  having values in $V,$ and $w$ any smooth function on $U$ with values in $V.$  We denote by $D$ the ordinary differentiation operator, so $D$ is the standard flat connection.  We can then define $1_V$ to be the identity function on $V$ itself which is clearly smooth and define $\nabla$ by

\begin{equation}\label{geo1}
\nabla_v w=D_v w + \frac{2}{3}J(v,w)1_V, 
\end{equation} 
Here we will follow the convention that if $v$ and $w$ and $z$ are functions with different domains contained in $V$ and with values in $V,$ then the domain of $J(v,w)z$ is the intersection of those domains, the domain they have in common, which if that is $U \subset V,$

\begin{equation}\label{geo2}
\mbox{then }  [\nabla_v w](x)= [D_v w](x) + \frac{2}{3} J(v(x),w(x))x, \mbox{ for any } x \in U, \mbox{ with } U \mbox{ open in } V.
\end{equation}

Now since $J$ is symmetric, and $D$ is torsion free, it follows obviously that $\nabla$ is torsion free, and therefore that $\R_{\nabla}$ satisfies the first Bianchi identity, so $\R_{\nabla}$ is a field of algebraic Riemann curvature operators on the tangent bundle of $V.$  We claim that in fact, $\R_{\nabla}(0)=R[[J]]$ from which $\J_{\nabla}(0)=J$ is then immediate, in case that $J$ satisfies the first Bianchi identity making it an algebraic Jacobi curvature operator.

First note that obviously,

\begin{equation}\label{geo3}
[\nabla_v w ](0)= [D_v w ](0).
\end{equation}

If $u,v,w$ are all smooth $V-$valued functions on open subsets of $V$, then, as $D_u 1_V=u,$

$$D_u\nabla_v w= D_uD_v w +\frac{2}{3} [J(D_u v,w)1_V+J(v,D_u w)1_V+J(v,w)u],$$
and

$$\nabla_u \nabla_v w=D_u \nabla_v w + \frac{2}{3}J(u,\nabla_v w)1_V,$$
so

$$\nabla_u \nabla_v w=D_u D_v w +\frac{2}{3} [J(D_u v,w)1_V + J(v, D_u w)1_V + J(v,w)u] +\frac{2}{3}J(u,D_v w +[2/3]J(v,w)1_V)1_V$$
and likewise

$$\nabla_v \nabla_u w=D_v D_u w +\frac{2}{3}[J(D_v u, w)1_V + J(u, D_v w)1_V +J(u,w)v] + \frac{2}{3}J(v, D_u w +[2/3]J(u,w)1_V)1_V.$$
and therefore cancellations give

$$[\nabla_u,\nabla_v]w=[D_u,D_v] w +\frac{2}{3}[J([u,v],w)1_V+ J(v,w)u-J(u,w)v]    + \frac{4}{9}[J(u,J(v,w)1_V)1_V -J(v,J(u,w))1_V)1_V].$$
Also,

$$\nabla_{[u,v]} w=D_{[u,v]} w + \frac{2}{3} J([u,v], w)1_V.$$
So as 

$$\R_{\nabla}(u,v)=[\nabla_u, \nabla_v]-\nabla_{[u,v]}, \mbox{ and } \R_D(u,v)=[D_u,D_v]-D_{[u,v]},$$
we find
$$\R_{\nabla}(u,v w)=\R_D(u,v) w+\frac{2}{3}[J(v,w)u-J(u,w)v] +\frac{4}{9}[J(u,J(v,w)1_V)1_V-J(v,J(u,w)1_V)1_V ].$$
Making use of the symmetry of  $J$ and the fact that $D$ is flat so $\R_D=0,$ we now have

$$\R_{\nabla}(u,v) w=\frac{2}{3}[J(w,v)u-J(w,u,)v] +\frac{4}{9}[J(u,J(v,w)1_V)1_V-J(v,J(u,w)1_V)1_V ].$$
But, $J(w,v)u=J^{\#}((u,v)w$ and likewise, $J(w,u)v=J^{\#}(v,u)w,$ so

$$\R_{\nabla}(u,v) w = \frac{2}{3} [ J^{\#}(u,v)w - J^{\#}(v,u)w] +\frac{4}{9}[J(u,J(v,w)1_V)1_V-J(v,J(u,w)1_V)1_V ]$$$$=\frac{4}{3} [Alt [J^{\#}]](u,v)w +\frac{4}{9}[J(u,J(v,w)1_V)1_V-J(v,J(u,w)1_V)1_V ]$$
$$=\frac{4}{3} [Alt [J^{\#}]](u,v)w +\frac{4}{9}[J(J(v,w)1_V,u)1_V-J(J(u,w)1_V,v)1_V ]$$
$$=\frac{4}{3} [Alt [J^{\#}]](u,v)w +\frac{4}{9}[J^{\#}(1_V,u)J(v,w)1_V-J^{\#}(1_V,v) J(u,w)1_V]$$
$$=\frac{4}{3} [Alt [J^{\#}]](u,v)w +\frac{4}{9}[J^{\#}(1_V,u)J(w,v)1_V-J^{\#}(1_V,v) J(w,u)1_V ]$$
$$=\frac{4}{3} [Alt [J^{\#}]](u,v)w +\frac{4}{9}[J^{\#}(1_V,u)J^{\#}(1_V,v)w-J^{\#}(1_V,v) J^{\#}(1_V,u)w ]$$
$$\mbox{so }  \R_{\nabla}(u,v) w =\frac{4}{3} [Alt [J^{\#}]](u,v)w +\frac{4}{9}[J^{\#}(1_V,u),J^{\#}(1_V,v)]w, \mbox{ for every }w,$$
where the bracket in the second term is indicating the commutator of operators, and at last we find then,

\begin{equation}\label{geo4}
\R_{\nabla}(u,v)  =\frac{4}{3} [Alt [J^{\#}]](u,v) +\frac{4}{9}[J^{\#}(1_V,u),J^{\#}(1_V,v)]=R[[J]](u,v)+\frac{4}{9}[J^{\#}(1_V,u),J^{\#}(1_V,v)].
\end{equation}
Finally, since $1_V(0)=0$ means the second term vanishes at $0 \in V,$

$$\R_{\nabla}(0)=\frac{4}{3}Alt[J^{\#}]=R[[J]].$$
This means that with the connection $\nabla=\nabla^J,$ defined by (\ref{geo1}), we find the curvature tensor on $V$ with value $R[[J]]$ at $0 \in V,$ In particular then, $\R_{\nabla}(0)=R[[J]].$  Now if $J$ is an algebraic Jacobi curvature operator, meaning here that it also satisfies the first Bianchi identity, then $\J_{\nabla}=J[[\R_{\nabla}]]=J[[R[[J]]]]=J$ so we have geometrized $J$ as well.

To summarize, we have

\begin{theorem}\label{geo5}
If $J$ is a symmetric algebraic precurvature operator on the vector space $V,$ then the Kozul connection $\nabla^J=\nabla$ on $V$ considered as a manifold given by (\ref{geo1}) has Riemannian curvature $\R_{\nabla}$ given by (\ref{geo4}) and in particular, $\R_{\nabla}(0)=R[[J]],$ and if in addition $J$ is an algebraic Jacobi curvature operator, then $\J_{\nabla}(0)=J.$
\end{theorem}

\begin{proof}
This is an immediate consequence of (\ref{geo4}).
\end{proof}

\bigskip

%%%%%%%%%%%%%%%%%%%%%%%%%%%%%%%%%%%%%%%%%%%%%%%%%%%%%%%%%%%%%%%%%%%%%%%%%%%%%%%%%%%%%%%%%%
%%%%%%%%%%%%%%%%%%%%%%%%%%%%%%%%%%%%%%%%%%%%%%%%%%%%%%%%%%%%%%%%%%%%%%%%%%%%%%%%%%%%%%%%%%
%%%%%%%%%%%%%%%%%%%%%%%%%%%%%%%%%%%%%%%%%%%%%%%%%%%%%%%%%%%%%%%%%%%%%%%%%%%%%%%%%%%%%%%%%%
%%%%%%%%%%%%%%%%%%%%%%%%%%%%%%%%%%%%%%%%%%%%%%%%%%%%%%%%%%%%%%%%%%%%%%%%%%%%%%%%%%%%%%%%%%

%%%%%%%%%%%%%%%%%%%%%%%%%%%%%%%%%%%%%%%%%%%%%%%%%%%%%%%%%%%%%%%%%%%%%%%%%%%%%%%%%%%%%%%%%%
%%%%%%%%%%%%%%%%%%%%%%%%%%%%%%%%%%%%%%%%%%%%%%%%%%%%%%%%%%%%%%%%%%%%%%%%%%%%%%%%%%%%%%%%%%
%%%%%%%%%%%%%%%%%%%%%%%%%%%%%%%%%%%%%%%%%%%%%%%%%%%%%%%%%%%%%%%%%%%%%%%%%%%%%%%%%%%%%%%%%%
%%%%%%%%%%%%%%%%%%%%%%%%%%%%%%%%%%%%%%%%%%%%%%%%%%%%%%%%%%%%%%%%%%%%%%%%%%%%%%%%%%%%%%%%%%

\section{GEOMETRIZATION OF METRIC ALGEBRAIC CURVATURE OPERATORS}

The geometrization of a metric algebraic curvature operator presents a slightly more complicated problem than in the purely algebraic case.   Given a semi Riemannian inner product $b$ on the vector space $V,$ we need to 
find an open subset $U$ of $V$ with $0 \in U$ and a semi Riemannian metric tensor $g$ on the manifold $U$ whose Levi Civita connection has Riemannian curvature at $0$ equal to a prescribed given metric$-b$ algebraic curvature operator.  So we begin with $J,$ a metric$-b$ algebraic Jacobi curvature operator, and we have to use it to produce a metric tensor, $g$ on a suitable open neighborhood $U$ of $0 \in V$ so that for $\nabla=\nabla^g$ the Levi Civita connection on $U$ determined by $g$ we have $\R_{\nabla}(0)=R[[J]],$ and  of course we want $g|_0=b,$ that is the metric tensor specifes $b$ as the inner product on the tangent space $T_0U=V.$

Part of the problem is a special case of the general problem of relating two metric tensors $g$ and $h$ and their resulting connections and curvature operators, on any semi Riemannian manifold $M.$  In fact, the two metric tensors need not have the same signature, as in the case of vector bundles where it was useful to consider a new positive definite metric tensor on the bundle.   But here we will make use of (\ref{cnxcurv1}) and (\ref{geo18}) to produce a new connection and calculate its curvature by using $J$ to produce $G.$  First, we begin by taking $V$ as a manifold, and look for an appropriate smooth $G: V \lra L(V;V)$ which then defines a smooth vector bundle map of $TV$ into itself over $V.$  The map which works is given by

\begin{equation}\label{mgeo1}
G(x)=1_V-\frac{1}{3}J(x,x), \mbox{ for all } x \in V.
\end{equation}

We then observe that $G(0)=1_V$ where $1_V$ is still denoting the identity transformation on $V,$ and therefore $G(0)$ is invertible.  Since $G$ is obviously smooth and therefore continuous, and as the set of invertible linear transformation is an open subset of $V,$ we define $U$ to be the subset of $V$ consisting of points $x \in V$ for which $G(x)$ is invertible, and therefore $U$ is an open subset of $V$ and then $G$ defines an isomorphism of $TU$ with itself covering the identity map $1_U$ of $U.$  Thus, we define our new metric tensor on $U$ to be $g$ where

\begin{equation}\label{mgeo2}
g(u(x),v(x))=b(G(x)u(x),v(x))=b(u(x),v(x))-\frac{1}{3}b(J(x,x)u(x),v(x)), \mbox{ for any } u,v:U \lra V.
\end{equation}
Of course, just as in the algebraic case, since $U$ is an open subset of $V,$ the tangent bundle $TU$ is trivial with fiber $V$ and so identifying each tangent vector field with its principal part, the tangent vector fields are simply functions from $U$ to $V.$  It is now convenient to use $1_U$ to denote the identity map on the open subset $U$ of $V.$  Then we can write (\ref{mgeo2}) as 

\begin{equation}\label{mgeo3}
g(u,v)=b(Gu,v)=b(u,v)-\frac{1}{3}b(J(1_U,1_U)u,v), \mbox{ for any } u,v:U \lra V.
\end{equation}

Notice as $J$ is a metric$-b$ algebraic Jacobi curvature operator we can apply the metric exchange property and rewrite (\ref{mgeo3}) as

\begin{equation}\label{mgeo4}
g(u,v)=b(u,v)-\frac{1}{3}b(J(u,v)1_U,1_U),  \mbox{ for any } u,v:U \lra V.
\end{equation} 
We again let $D$ denote the standard flat connection on $U$ given by ordinary differentiation, which is of course the Levi Civita connection on $U$ prescribed by the metric tensor $b$ on $U.$  Applying (\ref{geo18}) with $h=b,~D=\nabla,$ and $\nabla^g$ the Levi Civita connection on $U$ given by $g,$

\begin{equation}\label{mgeo5}
\nabla^g_u v -D_u v=\frac{1}{2} G^{-1} [[D_u G]v+[D_v G]u-[[DG]u]^*v].
\end{equation}
We need to compute $[[DG]u]^*v]$ in terms of $J$ which will simplify by the metric exchange property.  If $w$ is also a vector field on $U,$ then

$$b([[DG]u]^*v,w)=b(v,[DG]uw)=b(v,[D_wG]u)=b([D_wG]u,v).$$
\medskip

As $G(x)=1_V-(1/3)J(x,x),$ it follows from symmetry of $J$ that

$$[D_wG](x)=-\frac{2}{3}J(x,w(x)), \mbox{ for all } x \in U,$$
so

\begin{equation}\label{DG}
D_wG=-\frac{2}{3}J(1_U,w), \mbox{ for any } w:U \lra V,
\end{equation}
and therefore, in particular, for constant vectors,

\begin{equation}\label{DDG}
[D_vD_wG](x)=-\frac{2}{3}J(v,w),  \mbox{ for any } x \in U, v,w \in V.
\end{equation}
As an aside, this means that viewing $g$ as a function from $U$ to $L(V,V;\bR),$

\begin{equation}\label{DDg}
g''(x)vw=[D_vD_w g](x)=-\frac{2}{3}[bJ](v,w),
\end{equation}
where $bJ$ denotes the bilinear map in $L(V,V; (L(V,V;\bR))$ given by $bJ(u,v)(w,x)=b(J(u,v)w,x),$ for any $u,v,w,x \in V.$
Returning to our present general vector fields, using the metric exchange property for $J,$

$$b([[DG]u]^*v,w)=b([D_wG]u,v)=-\frac{2}{3}b(J(1_U,w)u,v)=-\frac{2}{3}b(J(u,v)1_U,w).$$
The result is

\begin{equation}\label{mgeo6}
[[DG]u]^*v]=-\frac{2}{3}J(u,v)1_U,
\end{equation}
so from (\ref{mgeo5}) we obtain

\begin{equation}\label{mgeo7}
\nabla^g_u v -D_u v=\frac{1}{2} G^{-1} [[D_u G]v+[D_v G]u+\frac{2}{3}J(u,v)1_U].
\end{equation}
Keep in mind that $G(0)=1_V$ and as $D_wG=-(2/3)J(1_U,w),$ for any $w:U\lra V,$ it follows from $1_U(0)=0,$ that

$$[\nabla^g_u v -D_u v](0)=0, \mbox{ that is } [\nabla^g_u v](0) =[D_u v](0).$$
Writing

$$H=\nabla^g-D$$

\begin{equation}\label{mH}
H_u v=\nabla_u^g v-D_u v=\frac{1}{2} G^{-1} [[D_u G]v+[D_v G]u+\frac{2}{3}J(u,v)1_U],
\end{equation}
we therefore have $H(0)=0$ by (\ref{DG}) and by (\ref{cnxcurv1}),

\begin{equation}\label{mgeo8}
\R_{\nabla^g}(u,v)=\R_D(u,v)+\R_H(u,v)+D_u[H_v]-D_v[H_u].
\end{equation}
But, $H(0)=0$ implies $\R_H(0)=0,$ and $\R_D=0$ because $D$ is flat, so finally

\begin{equation}\label{mgeo8}
\R_{\nabla^g}(u,v)(0)=[D_u[H_v]-D_v[H_u]](0).
\end{equation}
Now, again applying (\ref{DG}), we have from (\ref{mH}) that

\begin{equation}\label{mH2}
H_u v=\frac{1}{3}G^{-1}[J(u,v)1_U-J(1_U,u)v-J(1_U,v)u]=\frac{1}{3}G^{-1}[J(u,v)1_U-J^{\#}(v,u)1_U-J^{\#}(u,v)1_U.
\end{equation}
$$=\frac{1}{3}G^{-1}[J(u,v)1_U-2J_J(u,v)1_U].$$
By (\ref{cycl12J}) in Proposition \ref{cyclprop3}, $2J_J=-J,$ therefore we finally get

\begin{equation}\label{mH3}
H_u v=\frac{2}{3}G^{-1}J(u,v)1_U= \frac{2}{3} G^{-1}J(v,u)1_U=\frac{2}{3} G^{-1}J^{\#}(1_U,u)v, \mbox{ so } H_u=\frac{2}{3} G^{-1}J^{\#}(1_U,u).
\end{equation}
Notice how by comparison with (\ref{geo1}) the result just has the extra factor of $G^{-1}$ which was not needed in the non-metric case.

We only need next to compute the derivatives $D_u H_v,$ which is now easy using the final form of (\ref{mH3}), the result being immediately now,

\begin{equation}\label{mH4}
D_uH_v=\frac{2}{3}[D_u [G^{-1}]]J^{\#}(1_U,v)+\frac{2}{3} G^{-1}J^{\#}(u,v)+\frac{2}{3} G^{-1}J^{\#}(1_U,D_u v)
\end{equation}

Of course, from $G^{-1}G=1_V$ being a constant function into $L(V;V),$ we differentiate both sides to get

\begin{equation}\label{mH5}
D_u[G^{-1}]=-G^{-1}[D_uG]G^{-1}=\frac{2}{3}G^{-1}J(1_U,u)G^{-1}.
\end{equation}
Since $1_U(0)=0$ and $G(0)=1_V,$ we now clearly have

\begin{equation}\label{mH6}
D_uH_v|_{x=0}= \frac{2}{3}J^{\#}(u|_{x=0},v|_{x=0})=\frac{2}{3}J^{\#}(u,v)|_{x=0}
\end{equation}
so from (\ref{mgeo8}), we get

\begin{equation}\label{mH7}
\R_{\nabla^g}(u,v)(0)=[D_u[H_v]-D_v[H_u]](0)=\frac{4}{3}Alt[J^{\#}](u,v)|_{x=0}=R_J(u,v)|_{x=0},
\end{equation}
and finally therefore,

\begin{equation}\label{mH8}
\R_{\nabla^g}(0)=R_J  \mbox{ and } \J_{\nabla^g}(0)=J[[R_J]]=J.
\end{equation}
We can also use the preceding equations to obtain expressions for $\R_{\nabla^g}$ on all of $U,$ since the preceding results make the calculation of $\R_H$  and of $D_uH_v-D_vH_u$ a simple matter.  From (\ref{mH3}) and (\ref{mH4}) with 

$$2Q_{\nu}^+(J)=J+J^{\#},$$
we find due to the cancelation of the last terms (\ref{mH4}) of $D_uH_v-D_vH_u$ with the bracket term of $\R_H,$  and the middle terms giving $G^{-1}R_J(u,v),$ that

\begin{equation}\label{mH9}
\R_{\nabla^g}=G^{-1}R_J +\frac{16}{9}Alt (K_J),
\end{equation}
where

\begin{equation}\label{mHK}
K(u,v)=\frac{8}{9} G^{-1}[Q_{\nu}^+(J)](1_U,u)G^{-1}J^{\#}(1_U,v), \mbox{ for any } u,v: U \lra V.
\end{equation}
Of course as $1_U(0)=0$ and $G(0)=1_V,$ we see immediately, that (\ref{mH9}) reduces to (\ref{mH8}) at $x=0$ in $U.$

\bigskip

\begin{theorem}\label{mgeothm}
If $V$ is a vector space and $b$ is a non degenerate inner product on $V$ with $J$ a metric$-b$ algebraic Jacobi curvature operator on $V,$ then defining $G(x):V \lra L(V;V)$ by (\ref{mgeo1}) and defining $U_J \subset V$ to be the set of all $x \in V$ for which $G(x)$ is invertible, we define the the pseudo Riemannian metric tensor field $g$ on $U_J$ by (\ref{mgeo2}) and find that the Levi Civita connection $\nabla^g$ of $U_J$ with metric tensor $g$ is given by (\ref{mgeo5}).  Moreover, the curvature tensor $\R_{\nabla^g}$ is given by (\ref{mH9}) and (\ref{mHK}) so 

$$\R_{\nabla^g}(0)=R_J \mbox{ and } \J_{\nabla^g}(0)=J.$$

\end{theorem}

\begin{proof}
Of course, the proof is contained in the discussion preceding the statement of the theorem.
\end{proof}

\begin{corollary}\label{osculate}
If $M$ is a semi-Riemannian manifold with metric tensor $h,$ Jacobi curvature operator $\J_h,$ and Riemann curvature operator $\R_h,$ we can take $V=T_qM,$ the tangent space of $M$ at $q$, take $b=h(q),$ and set $U_M$ to be the inverse image of $U_q:=U$ under the exponential map with origin $q$.  As the derivative of the exponential map at $q$ is the identity map on $V=T_qM,$ we have that the semi Riemannian manifolds $U_q$ and $U_M$ agree to second order at $q \in M$ as their Riemann and Jacobi curvature operators agree at $q,$ that is,

\begin{equation}\label{osculatecurvature}
\R_h(q)=\R_{\nabla^g}(0) \mbox{  and  }  \J_h(q)=  \J_{\nabla^g}(0).
\end{equation}

Thus we can call $U_q$ the {\bf Osculating Tangent Manifold} to $M$ at $q \in M.$

\end{corollary}

\begin{proof}
This is an immediate consequence of Theorem \ref{mgeothm}.
\end{proof}

For any semi Riemannian manifold $M$ we can form the {\bf Osculating Tangent Bundle} denoted $O(M)$ by requiring that $O(M)$ be the subbundle of the tangent bundle $TM$ having fiber $U_q$ over $q \in M$ and given the structure of semi Riemannian manifold with metric tensor determined by (\ref{mgeo1}) and (\ref{mgeo2}) on $U_q,$ making $O(M)$ into a bundle of semi Riemannian manifolds.  Thus, $O(M)$ is an open subset of $TM$ so is a submanifold of $TM$ foliated by semi Riemannian manifolds each osculating with $M$ at its point of origin.

\begin{corollary}\label{normalcoord}
If $M$ is a semi-Riemannian manifold with metric tensor $h$ and Jacobi curvature operator field $\J=\J_h,$ and if $q \in M,$ in Riemann normal coordinates with origin $q$ we have (see also \cite{MTW} page 286 equation (11.32))

\begin{equation}\label{secderivmetricnormcoord}
h_{\alpha  \beta , \gamma \lambda}(q)=-\frac{2}{3}[\J_{\alpha  \beta  \gamma \lambda}(q)]
\end{equation}
giving the second partial derivatives of the metric tensor components in terms of the components of the Jacobi curvature operator components.  In particular, by the metric exchange property of the Jacobi curvature operator we have therefore, in Riemann normal coordinates with origin $q,$

\begin{equation}\label{metricderivexchange}
h_{\alpha  \beta , \gamma \lambda}(q)=h_{\gamma \lambda ,  \alpha  \beta}(q).
\end{equation}

\end{corollary}

\begin{proof}
The second derivative at zero in $V$ of the metric tensor $g$ given by (\ref{mgeo1}), (\ref{mgeo2}), (\ref{DDG}), and (\ref{DDg}), is obviously given by $g''(0)=-(2/3)bJ,$ as a bilinear map on $V,$ so now apply Corollary \ref{osculate}.

\end{proof}

We can remark that using components, the proof of Theorem \ref{mgeothm} can be given rather briefly as for instance in \cite{GILKEY1}, by making use of formulas for the Riemann curvature components in normal coordinates, but, expressing the metric in terms of the algebraic Jacobi operator instead of the Riemann curvature operator and making use of the resulting metric exchange property and the obvious $g''(0)=-(2/3)bJ,$ even makes that proof in components quite a bit more efficient for those of you who are addicted to index manipulations.

\bigskip

%%%%%%%%%%%%%%%%%%%%%%%%%%%%%%%%%%%%%%%%%%%%%%%%%%%%%%%%%%%%%%%%%%%%%%%%%%%%%%%%%%%%%%%%%%%%%%%%%%%%%%%%%%%%%%%%%%%%%%%%%%%%%%%%%%%%%%%%%%%%%%%%%%%%%%%%%%%%%%%%%%%%%%%%%%%%%%%%%%%%%%%%%%%%%%%%%%%%%%%%%%%%%%%%%%%%%%%%

%%%%%%%%%%%%%%%%%%%%%%%%%%%%%%%%%%%%%%%%%%%%%%%%%%%%%%%%%%%%%%%%%%%%%%%%%%%%%%%%%%%%%%%%%%
%%%%%%%%%%%%%%%%%%%%%%%%%%%%%%%%%%%%%%%%%%%%%%%%%%%%%%%%%%%%%%%%%%%%%%%%%%%%%%%%%%%%%%%%%%
%%%%%%%%%%%%%%%%%%%%%%%%%%%%%%%%%%%%%%%%%%%%%%%%%%%%%%%%%%%%%%%%%%%%%%%%%%%%%%%%%%%%%%%%%%
%%%%%%%%%%%%%%%%%%%%%%%%%%%%%%%%%%%%%%%%%%%%%%%%%%%%%%%%%%%%%%%%%%%%%%%%%%%%%%%%%%%%%%%%%%

\section{THE JACOBI CURVATURE OPERATOR AND THE LEVI CIVITA CONNECTION}

It is useful to have an expression for the Jacobi curvature operator $\J$ expressed in terms of the Levi Civita connection $\nabla^g=\nabla$ on any pseudo Riemannian manifold $M.$  To do this we begin with the Riemann Curvature operator $\R$ and calculate $\R^{\#}$ in terms of $\nabla.$  For any tangent vector fields $u,v,w$ on $M,$ since

$$\R(u,v)=[\nabla(u,\nabla_v]-\nabla_{[u,v]},$$
we have

\begin{equation}\label{JL1}
\R^{\#}(u,v)w=\R(w,v)u=[\nabla_w,\nabla_v]u-\nabla_{[w,v]} u.
\end{equation}
Now, as $\nabla u$ is a section of $L(TM;TM),$ we have $\nabla_v u=[\nabla u]v.$  Moreover, by the product rule for covariant differentiation, keeping in mind that compositions of linear transformation tensor fields are really contractions of tensor products and covariant differentiation commutes with all contractions, we have

\begin{equation}\label{JL2}
\R^{\#}(u,v)w=\nabla_w(\nabla_v u)-\nabla_v (\nabla_w u)-[\nabla u][w,v]=[\nabla (\nabla_v u) ]w- \nabla_v([\nabla u]w)-  [\nabla u][\nabla_w v -\nabla_v w]
\end{equation}
$$=[\nabla( \nabla_v u)]w - [\nabla_v(\nabla u)]w-[\nabla u][\nabla_v w]-[\nabla u][\nabla_w v] + [\nabla u][\nabla_v w]$$
$$=[\nabla( \nabla_v u)]w - [\nabla_v(\nabla u)]w-[\nabla u][\nabla_w v]=[\nabla( \nabla_v u)]w - [\nabla_v(\nabla u)]w-[\nabla u][\nabla v]w,$$
from which we conclude

\begin{equation}\label{JL3}
\R^{\#}(u,v)=[\nabla ( \nabla_v u)] - [\nabla_v (\nabla u)]-[\nabla u][\nabla v].
\end{equation}
At this point we take the liberty of defining the operator

\begin{equation}\label{JL4}
[\nabla, \nabla_v]=-[\nabla_v, \nabla] \mbox{ by } [\nabla, \nabla_v]u=\nabla(\nabla_v u)-\nabla_v(\nabla u),
\end{equation}
resulting in the expression

\begin{equation}\label{JL5}
\R^{\#}(u,v)=[\nabla, \nabla_v]u-[\nabla u][\nabla v].
\end{equation}
Since $\J=Sym(\R^{\#}),$ this gives finally

\begin{equation}\label{JL6}
\J(u,v)=\frac{1}{2} ( [\nabla, \nabla_u ]v +[\nabla , \nabla_v] u -[\nabla u][\nabla v] -[\nabla v][\nabla u]). 
\end{equation}
In particular,

\begin{equation}\label{JL7}
\J(u,u)=[\nabla , \nabla_u] u -[\nabla u]^2,
\end{equation}
and of course, (\ref{JL7}) completely determines $\J$ by the observer principle, but even more simply here we can just apply the polarization identity for symmetric bilinear mappings.  Obviously, (\ref{JL7}) is equivalent to (\ref{Ray1}) and (\ref{Ray2}).

\bigskip

%%%%%%%%%%%%%%%%%%%%%%%%%%%%%%%%%%%%%%%%%%%%%%%%%%%%%%%%%%%%%%%%%%%%%%%%%%%%%%%%%%%%%%%%%%%%%%%%%%%%%%%%%%%%%%%%%%%%%%%%%%%%%%%%%%%%%%%%%%%%%%%%%%%%%%%%%%%%%%%%%%%%%%%%%%%%%%%%%%%%%%%%%%%%%%%%%%%%%%%%%%%%%%%%%%%%%%%%

%%%%%%%%%%%%%%%%%%%%%%%%%%%%%%%%%%%%%%%%%%%%%%%%%%%%%%%%%%%%%%%%%%%%%%%%%%%%%%%%%%%%%%%%%%
%%%%%%%%%%%%%%%%%%%%%%%%%%%%%%%%%%%%%%%%%%%%%%%%%%%%%%%%%%%%%%%%%%%%%%%%%%%%%%%%%%%%%%%%%%
%%%%%%%%%%%%%%%%%%%%%%%%%%%%%%%%%%%%%%%%%%%%%%%%%%%%%%%%%%%%%%%%%%%%%%%%%%%%%%%%%%%%%%%%%%
%%%%%%%%%%%%%%%%%%%%%%%%%%%%%%%%%%%%%%%%%%%%%%%%%%%%%%%%%%%%%%%%%%%%%%%%%%%%%%%%%%%%%%%%%%

\section{THE EQUATION OF GEODESIC DEVIATION}

If $M$ is any smooth manifold and if $\nabla$ is any torsion free connection on $M,$ then consider a smooth vector field $v$ is geodesic if $\nabla_vv=0.$  Suppose that $s$ is another smooth vector field on $M$ and that the Lie bracket $[v,s]=0.$  Then as $\nabla$ is torsion free, we have

$$\nabla_vs- \nabla_sv=[v,s]=0, \mbox{ so } \nabla_vs=\nabla_sv.$$
But also, as $[v,s]=0,$ it follows that $\R(v,s)=[\nabla_v, \nabla_s].$  Thus, we have

$$\J(v,v)s=\R(s,v)v=\nabla_s \nabla_v v - \nabla_v \nabla_s v=-\nabla_v^2s.$$
It follows that $\J(v,v)$ acts as $-\nabla_v^2 $ on any smooth vector field which commutes with $v$ in the sense that its Lie bracket with $v$ vanishes.  In a real sense, we can think of $-\J(v,v)$ as a {\bf tensorization} of the differential operator $\nabla_v^2$ which results from restricting its domain from smooth vector fields to merely the vector space of smooth vector fields commuting with $v.$  Notice also that as $\J$ is symmetric, by the observer principle, $\J$ is completely determined by the quadratic forms $\J(v,v)$ for all vector fields $v.$  As $\J$ is a tensor, this means that for any point $q \in M,$ we can determine the algebraic Jacobi curvature operator $\J|_q$ on $T_qM$ by determining $\J(v,v)$ for each tangent vector $v \in T_qM.$  We can extend $v \in T_qM$ to a geodesic vector field and then determine $\J(v,v)s$ with any vector field extending $s \in T_qM$ in a way which commutes with $v$ in a neighborhood of $q.$ 

In case that $M$ is a semi Riemannian manifold with metric tensor $g,$ and $\nabla$ is the Levi Civita connection, so $\nabla$ is completely determined by the fact that it is torsion free and $\nabla g=0$,  if we have a local orthonormal frame field $(e_{\alpha})$, then it now seems reasonable to consider the "tensorization" of the d'Alembertian operator

$$\square=\sum_{\alpha} g(e_{\alpha}, e_{\alpha}) \nabla_{e_{\alpha}}^2$$
as the tensor operator

$${\bf \square}_J= - \sum_{\alpha}  g(e_{\alpha}, e_{\alpha}) \J(e_{\alpha},e_{\alpha}),$$
which is then a local smooth section of $L(TM;TM),$ and a smooth field of self adjoint linear operators.  More generally, then for any local smooth frame field $(v_{\alpha})$, with dual frame field $(v^{\alpha})$ we can set

\begin{equation}\label{dalembertian1}
{\bf \square}_J =-\sum_{\alpha, \beta} g(v^{\alpha}, v^{\beta}) \J (v_{\alpha} , v_{\beta}).
\end{equation}
Thus, ${\bf \square}_J$ is simply an operator valued contraction of $\J,$ which we can think of as the tensorized d'Alembertian operator, which would seem to be some sort of wave operator.  Since the values of $\J$ are all self adjoint operators, it follows that ${\bf \square}_J$ is a well defined smooth field of self adjoint operators on $M$ which is a smooth section of $L(TM;TM).$  Next, we define the associated $L(TM,TM;\epsilon(M;\bR))-$valued operator field $g{\bf \square}_J$ via

\begin{equation}\label{gblockJ}
[g{\bf \square}_J](u,v)=g([{\bf \square}_J ]u,v).
\end{equation}
Then by Theorem \ref{mgeothm} and its corollaries and (\ref{DDg}) we have

\begin{proposition}\label{dalembertian2}
If $M$ is a semi Riemannian manifold with metric tensor $g,$ and if $q \in M,$ in normal coordinates with origin $q,$ we have at $q,$ 

\begin{equation}\label{dalembertian3}
[\square g]|_q = \frac{2}{3} [g{\bf \square}_J]|_q,
\end{equation}
where of course,  $\square$ on the left side of the equation is the ordinary d'Alembertian in normal coordinates being applied to the metric tensor matrix in normal coordinates.

\end{proposition}

In view of \ref{dalembertian3}, we can think of $ \frac{2}{3} [g{\bf \square}_J]|_q$ as a "tensorization" of the ordinary d'Alembertian operator as it applies to the metric tensor.

%%%%%%%%%%%%%%%%%%%%%%%%%%%%%%%%%%%%%%%%%%%%%%%%%%%%%%%%%%%%%%%%%%%%%%%%%%%
%%%%%%%%%%%%%%%%%%%%%%%%%%%%%%%%%%%%%%%%%%%%%%%%%%%%%%%%%%%%%%%%%%%%%%%%%%
%%%%%%%%%%%%%%%%%%%%%%%%%%%%%%%%%%%%%%%%%%%%%%%%%%%%%%%%%%%%%%%%%%%%%%%%%%%%%

\med

\section{THE RICCI TENSOR AND SPATIAL DIVERGENCE}

In order to see in more detail how the Ricci tensor enters into the theory of
gravity, we should recall the equation of geodesic deviation again in more detail which is useful on any semi-Riemannian manifold. If
$[-a,a]$ and $[-b,b]$ is a pair of intervals in $\bR,$ then a
Jacobi field is a smooth map $\mathbb{J}:[-a,a] \times [-b,b] \lra M$ such
that for each fixed $\sigma \in [-b,b]$ the map $\mathbb{J}_{\sigma} ;
[-a,a] \lra M,$ given by $\mathbb{ J}_{\sigma}(\tau)=\mathbb{ J}(\tau,\sigma),$ is a
unit speed geodesic in $M.$  We can then form local vector fields
$u,s$ on an open neighborhood of the image of $\mathbb{ J}$ in $M,$ denoted
$Im~\mathbb{ J},$ so that

\begin{equation}
u(\mathbb{J}(\tau,\sigma))=\del_{\tau}\mathbb{ J}(\tau,\sigma),~~~s(\mathbb{ J}(\tau,\sigma))=\del_{\sigma}\mathbb{ J}(\tau,\sigma).
\end{equation}
Thus we must have $[s,u]=0$ and $\nabla_uu=0,$ on $Im~\mathbb{J},$ so we
find

\begin{equation}\label{Jacobi1}
\R(s,u)u=-\nabla_u\nabla_su.
\end{equation}
We will call $s$ in this situation a tangent Jacobi field along
$\mathbb{ J}_0,$ and at each point it gives the infinitesimal separation
vector. In fact, given $m$ a point on $\mathbb{J}_0$ and any unit vector
$s_m \in T_m$ which is orthogonal to $u(m),$ we can arrange that
$s(m)=s_m.$

Since our connection is assumed to be the unique torsion free
metric connection, we have $$[s,u]=\nabla_su-\nabla_us,$$ so the
condition that $[s,u]=0$ gives $\nabla_su=\nabla_us$ in our
present case. In view of (\ref{Jacobi1}), we then find the
equation of geodesic deviation on $Im~\mathbb{J},$

\begin{equation}\label{Jacobi2}
\K(u,u)s+\nabla_u \nabla_u s=\R(s,u)u+\nabla_u \nabla_u s=0.
\end{equation}
But, $\K(u,u)=\J(u,u),$ so we can and should write (\ref{Jacobi2}) as

\begin{equation}\label{Jacobi3}
\J(u,u)s+\nabla_u \nabla_u s=\R(s,u)u+\nabla_u \nabla_u s=0.
\end{equation}

In other words, $\nabla_u^2$ "is" the quadratic form of $-\J,$ the negative Jacobi curvature operator,
applied to $u,$ which is $-\J(u(m),u(m)).$  Restrict now to the Lorentz case.  Thus, if $u$ is timelike, we see
$A_{u(m)}^{(geo)}=-\J(u(m),u(m))$ is the linear transformation of $E(u,m)$ giving
the infinitesimal tidal acceleration field, ${\bf
a}_{u(m)}=A_{u(m)}^{(geo)}$ at $m,$ a vector field on $E(u,m)$
defined by ${\bf a}_{u(m)}(s)=A_{u(m)}^{(geo)}s,$ for any
separation vector $s \in E(u,m).$

Of interest to operator theorists here (see \cite{DUPRE0} for
spectral theory and functional calculus of operator fields) could
be the observation that in some sense we have found a relationship
between $\nabla_{u(m)}$ and $(A_{u(m)}^{(geo)})^{1/2}.$

Now, we simply combine the little proposition
\ref{flatspacedivcurl} together with (\ref{curv3}), and find at
$m,$ with $0_m$ denoting the zero vector of $T_mM,$

\begin{equation}\label{ricci2}
Ric(u(m),u(m))=trace(\J(u(m),u(m))=-trace(A_{u(m)}^{(geo)})=-div_{E(u,m)}
A_{u(m)}^{(geo)}(0_m).
\end{equation}

We are interpreting this result as relating to the observer's flat
Euclidean space divergence of his flat Euclidean space
infinitesimal tidal acceleration field.  Moreover, by
(\ref{curv2}) $A_{u(m)}^{(geo)}$ is a self-adjoint linear
transformation of the Euclidean space $E(u,m).$

For comparison, we point out that our discussion above for
(\ref{ricci2}) is also the content of results in \cite{ONEIL},
pages 225-219 and 8.9, page 219, as well as \cite{SACHSWU}, 4.2.2,
page 114. We can notice here that by the observer principle, Theorem \ref{gen top anal cont},
knowledge of ${\bf a}_{u(m)}=A_{u(m)}^{(geo)}$ for every possible
(even just future pointing) time-like unit vector $u \in T_mM$
would, by (\ref{Jacobi3}) determine $\J$ and hence $\R$ and the full Riemann curvature tensor by (\ref{R from J}) and Theorem \ref{pc11}.

\med

%%%%%%%%%%%%%%%%%%%%%%%%%%%%%%%%%%%%%%%%%%%%%%%%%%%%%%%%%%%%%%%%%%%%%%%%%%%%%%%%%%%%%%%%%%%%%%%%%%%%%%%%%%%%%%%%%%%
%%%%%%%%%%%%%%%%%%%%%%%%%%%%%%%%%%%%%%%%%%%%%%%%%%%%%%%%%%%%%%%%%%%%%%%%%%%%%%%%%%%%%%%%%%%%%%%%%%%%%%%%%%%%%%%%
%%%%%%%%%%%%%%%%%%%%%%%%%%%%%%%%%%%%%%%%%%%%%%%%%%%%%%%%%%%%%%%%%%%%%%%%%%%%%%%%%%%%%%%%%%%%%%%%%%%%%%%%%%%%%%%

\section{NEWTON'S LAW OF INFINITESIMAL \\ TIDAL ACCELERATION}

\med

Let us briefly review how Newton's Law of gravity is formulated
and how it can be recast in terms of the infinitesimal tidal
acceleration field.  Here we have a Euclidean space, $E,$ of
dimension $n$ and a smooth time dependent gravitational vector
field ${\bf f}$ defined (at time $t$) on an open subset $U_t$ of $E,$ where $U_t$ is
just $E$ with possibly a finite set of points removed which
represent point masses.  Thus the energy density $\rho$ is a
smooth function on $U_t$ and, in case $n=3,$ Newton's Law says that
$div_E  {\bf f}=-4\pi G \rho$ and ${\bf f}=-\nabla \Phi.$  Now this
last condition is easily equivalent to $curl {\bf f}=0,$ since $U_t$
is simply connected.  On the other hand, keeping in mind that
${\bf f}$ represents an acceleration field (or force per unit
mass), if $s \in T_mE=E$ is a separation vector, then $D_s{\bf
f}(m)$ is the infinitesimal gravitational tidal acceleration with
separation vector $s$ as in \cite{MTW}, pages 272-273 (see also
\cite{OHANIAN&RUFFINI}, pages 38-42). This means that the
infinitesimal gravitational tidal acceleration field is just
$A_m^{(grv)}=T_m{\bf f},$ the tangent map of ${\bf f}$ at $m,$ viewed as a vector field on $E.$ From
 Proposition \ref{flatspacedivcurl} concerning (\ref{flatspacediv}) and the
$curl,$ we see that Newton's Law for gravity in terms of the
infinitesimal gravitational tidal acceleration field $A_m^{(grv)}$
at $m$ says simply both 

$$trace(A_m^{(grv)})=-k_n \rho(m)$$ 
and
$$A_m^{(grv)}=[A_m^{(grv)}]^*.$$
Of course, this also makes sense
for any spatial dimension $n,$ as in (\ref{flatspacedivcurl}).
Here, $k_n$ is a constant which only depends on the spatial
dimension $n.$ Since $A_u^{(geo)}$ is self-adjoint (\ref{curv2}),
the obvious way to geometrize gravity in any spacetime dimension is simply to identify (in other words replace)
$A^{(grv)}$ with $A^{(geo)}.$

Thus on any Lorentz manifold $M,$ we say that the Newton-Einstein
Law for infinitesimal tidal acceleration holds at $m$ for the
observer $u \in T_mM$ provided that $trace(A_u^{(geo)})=-k_n
\rho_u(m),$ where we now view $A_u^{(geo)}$ as the observer's
infinitesimal gravitational tidal acceleration and where
$\rho_u(m)$ is the energy density of matter and fields other than
gravity $u$ observes at the event $m.$ Now, one relativistic
problem with Newton's Law for gravity is the fact that it amounts
to instantaneous action at a distance which conflicts with
relativity. We will assume that this problem is surmounted by {\it
only} requiring the equation to hold at the observer's event $m.$
He can say nothing about events other than his location event as
far as the law of gravity is concerned. This means henceforth, by
definition, and in view of (\ref{ricci2}) that the Newton-Einstein
Law of infinitesimal tidal acceleration at $m \in M$ for the
time-like unit vector $u \in T_mM$ simply states

\begin{equation}\label{newtonlawtides1}
Ric(u,u)=trace(\J(u,u))=-trace(A_u^{(geo)})=k_n \rho_u(m),
\end{equation}
which for short we refer to simply as NEIL at $m$ for observer $u
\in T_mM.$

We say that NEIL holds at $m \in M$ provided that it holds for
each observer $u \in T_mM.$ In this way, we overcome the
relativists claim that there should be no preferred observer. We
say that NEIL holds on $M$ provided that NEIL holds at each point
of $M.$ In this way we make the NEIL into a relativistic universal
law of gravity, which we think of as the Newton-Einstein Law of
gravity. Notice that each observer only claims
(\ref{newtonlawtides1}) to hold at his own event.

We have in fact almost arrived at the correct law, but relativists
could claim we have failed to include all the source energy on the
right hand side of the equation. That is to say, NEIL should be
{\it corrected by requiring that the part of the gravitational field energy
density which serves as gravity source as observed by each observer is also included in the source
term energy density}. We shall call this the corrected NEIL the
Einstein-Hilbert-Newton Law of infinitesimal tidal acceleration or
EHNIL.  Anticipating the ability of the observer $u$ at $m$ to find
the gravity source energy density of the gravitational field,
$\rho_u^{(grv)}(m),$ we say that the EHNIL holds at $m$ for
observer $u$ provided that

\begin{equation}\label{crrtdnewtnlawtides1}
Ric(u,u)=k_n[\rho_u(m)+\rho_u^{(grv)}(m)].
\end{equation}
Naturally, we then say the EHNIL holds at $m \in M$ provided it
holds for each observer $u \in T_mM,$ and say that the EHNIL holds
for $M$ provided that it holds at each event $m \in M.$  This then
is our {\it universal law of gravitation},  the
Einstein-Hilbert-Newton Infinitesimal Law of Gravity. Of course,
this naturally leads to the question as to the gravitational source energy density of
the gravitational field which an observer sees at his location
event, which we turn to in the next section.

Finally here, let us look at Newton's Law for gravity in an alternate way. If $\phi$ is the potential function for the Newtonian gravitational field, then Newton's Law says that $\nabla^2 \phi=4 \pi \rho.$   We have two gravitational energy terms in this equation, one on each side.  Why is it not the case that each side is itself an energy?  Why do we view the $\phi$ as an energy and $\rho$ as an energy?  It seems that if $\rho$ is an energy, then also $(1/4 \pi G)\nabla^2 \phi=-(1/4 \pi G)div_E{\bf f}=-(1/4 \pi G)trace(A_m^{(grv)})$ should also be an energy, and if so, why should just the diagonal components of $-A_m^{(grv)}/(4 \pi G)$ contribute to energy?  It seems on purely mathematical grounds here that Newton's Law is begging for us to interpret all of $-(1/4 \pi G)A_m^{(grv)}$ as energy.  In that light, Newton's Law of gravity says that it is $-trace(A_m^{(grv)}/(4 \pi G))$ that gives the part of $-A_m^{(grv)}/(4 \pi G)$ acting as the source for gravity.  Here, when we replace $-A_m^{(grv)}$ by $-A_u^{(geo)},$ we see that we must have $(1/4 \pi G)\J(u,u)$ as the total energy in a spacetime model of gravity.  Its trace then gives the gravitational source energy.  This points to the idea that, by the observer principle,  $(1/4 \pi G)\J$ should be viewed as the total energy momentum stress tensor of all matter and fields including gravity, in any spacetime model.  In case we do take such a view, then the total conservation of energy should be simply the second Bianchi identity, which in view of the fact that 

$$R_{\J}=\R,$$
and with $d=d_L,$ denoting exterior differentiation of $L$ valued forms, where $L=L(\tau_M;\tau_M),$ says simply

$$0=d\R=dR_{\J}=0.$$

\med

%%%%%%%%%%%%%%%%%%%%%%%%%%%%%%%%%%%%%%%%%%%%%%%%%%%%%%%%%%%%%%%%%%%%%%%%%%%%%%%%%%%%%%%%%%%%%%%%%%%%%%%%%%%%%%%%%%%
%%%%%%%%%%%%%%%%%%%%%%%%%%%%%%%%%%%%%%%%%%%%%%%%%%%%%%%%%%%%%%%%%%%%%%%%%%%%%%%%%%%%%%%%%%%%%%%%%%%%%%%%%%%%%%%%%%%%
%%%%%%%%%%%%%%%%%%%%%%%%%%%%%%%%%%%%%%%%%%%%%%%%%%%%%%%%%%%%%%%%%%%%%%%%%%%%%%%%%%%%%%%%%%%%%%%%%%%%%%%%%%%%%%%%%%%

\section{THE GRAVITATIONAL SOURCE ENERGY DENSITY \\
 OF THE GRAVITATIONAL FIELD}

\med

Again, the discussions here are those of \cite{DUPRE} and \cite{DUPRE2}, but modified for the point of view that we distinguish between the overall energy of the gravitational field and the energy of the gravitational field which can act as source of gravity.

In order to deal with the gravity source energy density of the gravitational
field, we must think of how the tension energy of the gravitational field can be thought of as a mass. First think in terms of the basic assumption of the
geometric notion of gravity which is that "free test" particles
must follow geodesics. To partially paraphrase J. A. Wheeler,
spacetime tells matter how to move.  That is its job. But
spacetime is the physical manifestation of the gravitational
field, so it is really the job of the gravitational field to tell
matter how to move.  That is, the geodesic hypothesis, that
gravity is responsible for making free particles follow geodesics
means that when particles fail to follow geodesics, the
gravitational field is put in tension which should result in a type of
energy in the gravitational field which is equivalent to a form of mass.  But in general relativity, the
only thing that can cause particles to fail to follow geodesic
motion is forces or, in density terms, pressures.  Specifically, we should regard the pressure required to cause deviation from geodesic motion as a reflection of the energy that the gravitational field is exerting to try and maintain geodesic motion.  Thus, the
principal $u-$spatial directions and $u-$spatial eigenvalues of
$T$ in the spacetime model $(M,g,T)$ should tell us the the energy
density of the gravitational field observed by $u \in T_mM.$  Thus, in particular, dust particles follow geodesic motion since there is no pressure in dust, by definition.

For instance, the pressure you feel on your bottom when sitting in
a chair is a manifestation of the energy density of the
gravitational field at those points on your chair. In a sense
then, we could say that if the surface of your chair were replaced
by an infinitesimally thin slab sitting on top of an
infinitesimally lower chair, then the mass energy of the slab
required to hold you in place divided by the volume of the slab is
a reflection of the gravitational source energy density of the gravitational field
there. What is the minimum mass which can take care of this job?

In fact, the material the chair is made of in some sense is a
reflection of the gravity source energy density of the gravitational field right
where your chair is located. Even primitive people have an
intuitive idea of the strength of material needed to make a chair,
and thus have a working idea of the energy density of the
gravitational field. We should therefore think of the least mass
energy of material required to make a chair as a rough measure of
the gravity source energy density of the gravitational field where the chair is
to be used.  More generally, in the immediate vicinity of a specific event, strength of building material requirements there are reflections of the energy of the gravitational field there.  For instance, strength requirements of building materials on the surface of the Earth are clearly different than on the surface of the Moon as evidenced by the construction of the Apollo 11 Moon Lander seen worldwide on TV.

More generally, imagine an observer located at $m \in M,$
ghost-like inside a medium with energy momentum stress tensor $T$
and suppose that his velocity at $m$ is $u.$ Then $u^{\perp}
\subset T_mM,$ is the orthogonal complement of $u$ in $T_mM,$ and
the $u-$spatial principal pressures are given by the restriction of $T$ to
$u^{\perp} \times u^{\perp}.$ Thus, there is an orthonormal frame
$(e_1,e_2,...,e_n)$ for $u^{\perp}$ with the property that
$T(e_a,e_b)=p_a \delta^a_b,$ for $a,b \in \{1,2,...,n\}.$ These
observed $u-$spatial eigenvalues of $T$ are, of course, the
principal pressures observed by $u,$ and their average is referred
to as the observed {\it isotropic pressure}, $p_u.$ Thus, $np_u$
is the sum of the principal pressures as seen by the observer with
velocity $u.$ Imagine scooping out a tiny infinitesimal box in $M$
at $m$ whose edges are parallel to these $u-$spatial principal
axes of this spatial part of the energy momentum stress tensor.  What is the minimum mass-energy required to maintian the scooped out space.  We
can imagine putting infinitesimally thin $(n-1)-$dimensional
reflecting mirrors for walls of the box and filling the box with
laser beams reflecting back and forth in directions parallel to
the edges of the box with enough light pressure in each direction
to balance the force from outside on these reflecting walls.

In a sense, we have standardized a system to balance the pressures
acting to disturb the gravitational field, so we {\it define} the gravitational source
energy density of the gravitational field as seen by our observer
to be the energy density of the light in this little box. The fact
that a photon has zero rest mass should mean that the light energy
constitutes a minimum amount of energy to accomplish this task of
balancing the gravitational energy. However, for the case $n=3,$ it is a problem in elementary physics to see that the energy density of the light
along a given axis is exactly the pressure in that principal
direction.  That is, if we call one of these principal directions
the $x-$direction, then $p_x,$ the principal pressure or
$u-$spatial eigenvalue of $T$ in that direction is the energy
density of the laser beams parallel to that direction.  For instance, see \cite{DUPRE} or \cite{DUPRE2}, for further discussion and details.

If $p_x$ is negative, a similar argument using opposite charge
distributions on the opposite walls of the box along the
$x-$direction would have the opposite walls behaving like a
capacitor and again, elementary calculations (the freshman physics
"pillbox" argument using Gauss' Law for electric flux) easily lead
to the conclusion that the energy density of the electric field of
the capacitor has the same absolute value as the negative
stretching pressure of the medium, and here it seems that the
energy due to this stretching pressure (like the pressure in a
stretched rubber band) should count as negative energy. Thus, when
pressure is negative, the pressure is serving to reduce the energy
of the gravitational field as it is "working with" the
gravitational field. We could imagine for instance in case the
capacitor is overcharged, that allowing this scooped out capacitor
to "snap shut" then supplies the capacitor energy to the
gravitational field and also lowers the energy of the
gravitational field. So in this case of negative pressure, it must
be that the energy density should be negative just as is the
pressure.

We can therefore take it to be the case that the principal
pressure in the $x-$direction gives the gravitational source energy density
contribution for that direction in any case. Likewise for the
other two axes, consequently we see that the total gravitational source energy density
in the box is the sum of the pressures that the beams and fields
are balancing, that is the trace of the observer's spatial part of
the energy-stress tensor, $p_x+p_y+p_z,$ for the case spatial
dimension $n=3.$ For a more extended and detailed discussion of
these physical motivations, we refer the interested reader to
\cite{DUPRE2}.

More generally, in light of the preceding heuristic arguments, for
any $n,$ we {\it define} the gravitational source energy density seen by
the observer $u \in T_mM$ to be the sum of the principal
pressures, $\rho_u^{(grv)}(m)=np_u.$ We now state this as a formal
postulate.

\begin{postulate}\label{postltgrvenrgdnst} {\bf GRAVITATIONAL SOURCE
ENERGY DENSITY POSTULATE.} For any spacetime model $(M,g,T),$ at
each event $m \in M,$ each observer $u \in T_mM$ observes the gravitational source
energy density of the gravitational field as being
$\rho_u^{(grv)}(m)=np_u,$ the sum of the principal pressures of
$T$ at event $m$ in $M.$
\end{postulate}

At this point we can notice that we are already dealing with a
physically intuitive description of the gravitational source energy of the gravitational field which
implies the Cooperstock hypothesis \cite{COOPERSTOCK}, \cite{COOPERSTOCK2}, \cite{COOPERSTOCK3} which says the gravitational
field has no energy in the vacuum, because indeed, in the vacuum
there is certainly no pressure. That is obviously postulate
\ref{postltgrvenrgdnst} implies the Cooperstock hypothesis.  However, here for the Cooperstock hypothesis, we must interpret the energy in the vacuum as refering to gravitational {\it source} energy.  Thus, the Cooperstock hypothesis implies that gravity waves have no gravitating energy in the vacuum.  To be more clear here, the simple point is that in the vacuum there is certainly no source gravitational energy for otherwise you clearly do not have a vacuum, and consequently, the Cooperstock Hypothesis takes on the character of a simple tautology regarding the meaning of the word vacuum.

Consider now the example of a system of isolated gravitating masses interacting purely gravitatationally.  The attraction of the masses for each other due to gravity can be thought of as a form of pressure, so considered as a fluid, the pressure is negative, so the gravitational field in this fluid has negative energy density.  Thus any form of integral of the total energy density of all the energy of matter and fields and gravity should result in a value less than the integral over the same domain for purely the masses and non-gravitational fields.

\medskip

%%%%%%%%%%%%%%%%%%%%%%%%%%%%%%%%%%%%%%%%%%%%%%%%%%%%%%%%%%%%%%%%%%%%%%%%%%%%%%%%%%%%%%%%%%%%%%%%%%%%%%%%%%%%%%%%%%%%
%%%%%%%%%%%%%%%%%%%%%%%%%%%%%%%%%%%%%%%%%%%%%%%%%%%%%%%%%%%%%%%%%%%%%%%%%%%%%%%%%%%%%%%%%%%%%%%%%%%%%%%%%%%%%%%%%%%
%%%%%%%%%%%%%%%%%%%%%%%%%%%%%%%%%%%%%%%%%%%%%%%%%%%%%%%%%%%%%%%%%%%%%%%%%%%%%%%%%%%%%%%%%%%%%%%%%%%%%%%%%%%%%%%%%

\section{THE GRAVITATIONAL SOURCE \\ ENERGY MOMENTUM STRESS TENSOR\\ OF THE GRAVITATIONAL FIELD}

\med

In view of the results of the preceding section we can now prove
our elementary theorem on the gravitational source energy density of the gravitational
field.  Again, this section is the argument in \cite{DUPRE2} modified to distinguish between overall gravitational energy momentum stress and that which merely serves as the source of gravity.

\begin{theorem}\label{enrgmmntmstrsgravtheorem}
If $(M,g,T)$ is a spacetime model, then assuming the Gravitational Source 
Energy Density Postulate \ref{postltgrvenrgdnst} is equivalent to
assuming the covariant symmetric tensor

\begin{equation}\label{gravenergydensity}
T_g=T-c(T)g
\end{equation}
is the unique symmetric tensor giving the gravitational source energy momentum stress
tensor of the gravitational field.
\end{theorem}

For the proof, suppose that $m \in M$ is any event and $u \in
T_mM$ is an observer at $m \in M.$ Suppose that $T$ is the second
rank covariant energy momentum stress tensor for the matter and
fields other than gravity. Our task is to find the covariant
second rank symmetric tensor $T_g,$ which gives the source energy
momentum stress of the gravitational field from our previous
physical argument that every observer should see it as the sum of
the principal pressures.  Thus, by the observer principle and the
gravitational  source energy density postulate \ref{postltgrvenrgdnst},
$T_g$ is uniquely determined by the requirement that
$T_g(u,u)=np_u$ for each time-like unit vector $u \in T_mM$ no
matter which $m \in M.$

Now, applying (\ref{contract1}) to compute $c(T)$  we find that
$$c(T)=-T(u,u)+np_u,$$ and therefore,
$$np_u=T(u,u)+c(T)=T(u,u)-c(T)g(u,u),$$ which is to say finally that
the second rank symmetric covariant tensor $$T_g=T-c(T)g$$ does
indeed do the job, by The Observer Principle.

We point out here, that in general, such uniqueness does not imply
existence, but here we have existence of the required tensor we
seek from equation (\ref{gravenergydensity}) itself. That is
really the assumption that there is a covariant symmetric tensor
$T$ giving the energy momentum stress tensor of all matter and
fields other than gravity is also giving the existence of the
tensor $T_g$ through equation (\ref{gravenergydensity}).

In the reverse direction, by (\ref{contract1}), if we assume that
(\ref{gravenergydensity}) is the gravitational source energy momentum stress tensor for
the gravitational field, then the gravitational source energy density
postulate \ref{postltgrvenrgdnst} is an immediate consequence.
This completes the proof of the theorem
\ref{enrgmmntmstrsgravtheorem}.

In view of (\ref{gravenergydensity}) we define the {\it total gravitational source
energy momentum stress tensor} of all matter and fields including
gravity to be

\begin{equation}\label{totalenrgmmntmstrss}
H=T+T_g=2T-c(T)g=2[T-(1/2)c(T)g].
\end{equation}
Thus, by Theorem 6.1 and the observer principle, we know $H$ must
be the symmetric tensor which should serve as the source term for
the gravitation equation, since for every $m \in M$ and observer
$u \in T_mM$ we have

\begin{equation}\label{totalenrgdnsty}
H(u,u)=\rho_u(m)+\rho_u^{(grv)}(m).
\end{equation}

\med

%%%%%%%%%%%%%%%%%%%%%%%%%%%%%%%%%%%%%%%%%%%%%%%%%%%%%%%%%%%%%%%%%%%%%%%%%%%%%%%%%%%%%%%%%%%%%%%%%%%%%%%%%%%%%%%
%%%%%%%%%%%%%%%%%%%%%%%%%%%%%%%%%%%%%%%%%%%%%%%%%%%%%%%%%%%%%%%%%%%%%%%%%%%%%%%%%%%%%%%%%%%%%%%%%%%%%%%%%%%%%%
%%%%%%%%%%%%%%%%%%%%%%%%%%%%%%%%%%%%%%%%%%%%%%%%%%%%%%%%%%%%%%%%%%%%%%%%%%%%%%%%%%%%%%%%%%%%%%%%%%%%%%%%%%%%%%%

\section{THE DERIVATION AND PROOF OF THE EINSTEIN EQUATION}

\med

We are now in a position to state and prove our theorem as regards
the Einstein equation.

\begin{theorem}\label{theoremeinstein}
If $(M,g,T)$ is spacetime model, then the
EHNIL(\ref{crrtdnewtnlawtides1}) together with the Gravitational Source
Energy Density Postulate \ref{postltgrvenrgdnst} is equivalent to
the assumption that the equations

\begin{equation}\label{einsteinequation2}
 Ric=k_nH=k_n[2T-c(T)g]=2k_n[T-(1/2)c(T)g]
 \end{equation}
 hold on $M$ with $H=(1/k_n)Ric$ being the total source energy momentum stress
 tensor of gravity and all matter and fields for the model. In particular, if
 $n=3$ so spacetime is four dimensional, then automatically
 $div(T)=0$ as a consequence of these assumptions.  If $n$ is not
 3, then these assumptions and the assumption that $div(T)=0$
 imply that $dR=0,$ where $R$ is the scalar curvature of $M.$
 \end{theorem}

 To prove the theorem \ref{theoremeinstein} use theorem  \ref{enrgmmntmstrsgravtheorem}. Assuming
 the EHNIL(\ref{crrtdnewtnlawtides1}) holds for all observers
 everywhere,  we now have by (\ref{totalenrgdnsty}),

 \begin{equation}\label{einsteinequation1}
 Ric(u,u)=k_nH(u,u),
 \end{equation}
 for every
 observer $u \in T_mM$ at $m \in M,$
 where $k_n$ is a universal constant depending only on $n.$ As an aside,
 beyond \ref{enrgmmntmstrsgravtheorem} and the EHNIL,
 the real reason behind everything here
 is the fact that $Ric(u,u)$ is the
 negative Euclidean divergence of the tidal acceleration
 (\ref{ricci2}), so in physical terms we are
 using the tracial identification of the gravitational tidal acceleration with the geometric tidal
 acceleration. But let us return to the proof.
 Thus, (\ref{einsteinequation1}) merely
 says that any observer $u$ at any $m \in M$ sees the EHNIL to hold.
 Since (\ref{einsteinequation1}) is
 true for $u$ being any time-like unit
 vector, by the observer principle (corollary \ref{obsv2}), we
 must have (\ref{einsteinequation2}) as an immediate consequence.

 For the reverse direction, if we assume that the equation (\ref{einsteinequation2}) holds
 with $H$ giving the total gravitational source energy momentum stress tensor of all
 gravity all matter and all fields, then as $T$ is the energy
 momentum stress tensor of all matter and fields other than
 gravity, and as $H=2T-c(T)g,$ we must have $$T_g=H-T=T-c(T)g$$
 which by Theorem \ref{enrgmmntmstrsgravtheorem}
 then implies the gravitational source energy density postulate \ref{postltgrvenrgdnst} and then
 (\ref{einsteinequation1}) holds for any observer $u$ which now by (\ref{totalenrgdnsty})
 says the EHNIL(\ref{crrtdnewtnlawtides1}) holds for all observers.

In case $n+1=4,$ it is customary to write $k_3=4\pi G,$ so then

 \begin{equation}\label{einsteinequation3}
 Ric=4\pi G H=8\pi G[T - (1/2)c(T)g],\,\ n=3,
 \end{equation}
 which is a well-known form of Einstein's equation. In general, as $c(g)=n+1$
 and $c(Ric)=R,$ where as usual, $R$ is the scalar curvature, we
 find that $$R=(1-n)k_n[c(T)],$$ so the equation can be also written
 as $$Ric=k_n[2T]+(1/(n-1))Rg,$$ and this results in

\begin{equation}\label{einsteineqforgendim=n}
Ric-\frac{1}{n-1}Rg=2k_nT.
\end{equation}

The gravitational source energy momentum stress tensor $T_g=T-c(T)g$ can be expressed in terms
of the Ricci tensor and scalar curvature using
(\ref{einsteineqforgendim=n}) and the result is

\begin{equation}\label{gravengdensgendim=n}
T_g=(\frac{1}{2k_n})[Ric+(\frac{1}{n-1})Rg].
\end{equation}

As usual, we define the {\it Einstein tensor} by
$$Einstein=Ric-(1/2)Rg,$$ which has the property that
$$div(Einstein)=0,$$ no matter the value of $n.$ But in case $n=3,$
we find that the left hand side of (\ref{einsteineqforgendim=n})
is the Einstein tensor. In this case,
 with $k_3=4\pi G,$ we find the most
 familiar form of the Einstein equation

 \begin{equation}\label{einsteinequation4}
 Einstein=Ric-(1/2)Rg=8\pi G~T,\,\ n=3.
 \end{equation}

 Notice that we have not used local conservation of energy,
 $div(T)=0.$  Since the left side of (\ref{einsteinequation4}), the Einstein
 tensor, $Einstein,$ is divergence free, we find $div(T)=0$ as a consequence of our
 derivation, in the case where $n=3.$   On the other hand, it appears that
 for $n$ not equal to 3 we would have that $div(T)$ is in general not zero.  That is, it is only in
 spacetime dimension 4 that the energy momentum stress tensor of matter and
 fields other than gravity can be infinitesimally conserved
 without automatically putting severe restrictions on spacetime. Specifically,
 in case $n+1$ is not equal to 4, we find immediately that $div (T)=0$ implies $dR=0,$ and
 hence the scalar curvature of spacetime must be constant if the energy stress tensor
 of matter and fields other than gravity has zero divergence.

 To see this, in more
 detail, a simple calculation shows that for any smooth scalar
 function $f$ we have $div(fg)=df,$ since $\nabla g=0.$
 The Einstein tensor has vanishing divergence in any dimension
 (due to the second Bianchi identity), and this is clearly
 equivalent to $div(Ric)=(1/2)dR.$  Therefore, taking the
 divergence of both sides of (\ref{einsteineqforgendim=n}) gives
 $$(\frac{1}{2}-\frac{1}{n-1})dR=2k_n div(T).$$  Thus, if we
 assume that $div(T)=0,$ then either $dR=0$ or $n=3,$ and on the
 other hand, if we assume $n=3,$ then as the Einstein tensor has
 vanishing divergence, then so must $T.$ This completes the proof
 of our theorem \ref{theoremeinstein}.

 Before proceeding further, let us rewrite (\ref{einsteineqforgendim=n}) using (\ref{ricci2}).  The result is that

 \begin{equation}\label{JacobiEinstein1}
 trace(\J)=k_nH.
 \end{equation}
 This means that as $H$ is the total gravitational source energy momentum stress, that part of $\J,$ namely the "diagonal components" are gravitational source energy momentum stress.  Why should only the diagonal be considered as comprising energy momentum stress.  It seems that the mathematics is asking us to think of the whole Jacobi operator as proportional to energy momentum stress.  We thus will define $(1/k_n)\J$ to be the total energy momentum stress in any spacetime model and define $trace((1/k_n)\J)=(1/k_n)Ric$ to be the total gravitational source energy momentum stress tensor for the spacetime model.

 There is another reason why all the components of $\J$ should be considered as energy momentum stress.  This is because by the observer principle, $\J$ is determined on $T_mM$ by knowing $\J(u,u)$ for each timelike unit vector $u \in T_mM.$ But $\J(u,u)$ is a self-adjoint linear transformation of $u^{\perp},$ the Euclidean space orthogonal to $u$ in $T_mM,$ so is diagonalizable, and thus the principal axes are physically real manifestations of curvature at $m$ and observer $u$ can choose an orthonormal frame for $u^{\perp}$ in which $\J(u,u)$ is diagonal.  Thus, the total gravitational source energy observed by $u,$ being proportional to the trace of $\J(u,u)$ is really the sum of all the non-zero components of $\J(u,u)$ in this frame, so $\J(u,u)$ is all gravitational source energy in this frame.  Thus, if $u$ uses a frame in which $\J(u,u)$ on $u^{\perp}$ is not diagonal, the off-diagonal components are mixtures of the diagonal components from the diagonal frame with coefficients simply determined by the linear transformation to the non-diagonal frame, and thus should still be considered to contain energy.  It is simply that the trace is an invariant, so the total of the energy in the trace is the total source energy for gravity.  But, the whole transformation $\J(u,u)$ is the proper representation of the energy as seen by $u,$ because the break up of the trace as a sum of terms is brought about by the physical curvatures.  For instance, if we have a system consisting of a gas with energy momentum stress $T_{gas}$ together with an electromagnetic field with energy momentum stress tensor $T_{EM},$ the total energy momentum stress tensor is the sum of these two tensors plus that due to the gravitational field.  But, we still refer to $T_{EM}$ alone as energy momentum stress, even though there is other energy momentum stress present.  Gravity only sees the total energy momentum stress as source, but physically, $\J(u,u)$ is telling observer $u$ that this total energy in light of curvature also has a natural structure as a sum of individual energies in the orthonormal frame.  Of course, the orthonormal frame which diagonalizes $\J(u,u)$ could be different from that which diagonalizes $Ric$ or $T.$

At this point, we should remark that it is often
thought that as the Weyl curvature need not vanish in the vacuum,
that it should enter into the expression for $T_g.$ However, we
have a physical expression for the gravitational source energy density of the field as
seen by any observer, so that determines what the expression for
$T_g$ will be. If the Weyl curvature is not part of the result, we
must accept the fact that Weyl curvature cannot generate
gravitational source energy momentum stress, on this view.  However, and this is important for considering the Weyl tensor, it is the case that the Weyl tensor and the Ricci tensor together determine the full Riemann curvature tensor and thus, equivalently, the Jacobi tensor, and as well, the Jacobi curvature operator determines the full Riemann curvature tensor and hence the Ricci tensor and Weyl tensor.  Since in our view, the complete total energy momentum stress of everything is proportional to the Jacobi tensor, this means that the Weyl tensor contains the gravitational non-source energy of the gravitational field.  However, in our opinion, the fact that as the Jacobi tensor is directly related to the tidal acceleration operator on spacetime and as it obeys appropriate symmetry for an energy momentum stress tensor, it is better to think of the Jacobi operator as being proportional to the total energy momentum stress tensor of everything, that is $(1/ k_n)\J$ is the total energy momentum stress of gravity and all matter and fields, and its trace, $(1/k_n)Ric$  gives the gravitational source energy momentum stress tensor.

To put the Weyl tensor in perspective here, let $\W=\W_{\R}$ denote the Weyl tensor due to the Riemann curvature operator $\R.$  Put $RIC=\R-\W,$ so we know $\R=RIC+ \W$ is an orthogonal decomposition of $\R,$ and our isometric isomorphism of metric Riemann curvature operators onto the metric Jacobi curvature operators gives the orthogonal decomposition

$$\J=J[[RIC]] + J[[\W]],  \mbox{ with } trace(J[[RIC]])=Ric.$$
Thus, we should now think of $(k_n^{-1}J[[\W]]$ as the gravitational energy momentum stress tensor of the vacuum itself and $(k_n^{-1})J[[RIC]]$ as the energy of the ordinary matter and fields whose trace is the source of gravity.  In this view, the Weyl tensor is the real embodiment of vacuum energy, it is just that it needs to be in the form of an energy operator, $J[[\W]].$

Additionally, we should keep in mind that including appropriate
boundary conditions, when $n=3,$ the Einstein equation determines
the full Riemann curvature tensor, therefore including the Weyl
curvature tensor. In particular, the reader should note equations
(4.28) and (4.29) on page 85 of \cite{HAWKINGELLIS} for the Weyl
curvature tensor, which follow from the Bianchi identities and are
similar in form to the Maxwell equations for the electromagnetic
field tensor. Thus, the Weyl curvature which gives the curvature
in the vacuum, as the Ricci curvature vanishes in the vacuum, is
contained in the boundary conditions under the Einstein equation.

As in our view here, the whole Jacobi curvature operator should be viewed as consisting of energy, momentum and stress, it is reasonable to notice that for the observer with timelike unit vector $u$ at event $q \in M,$ as $\J(u,u)=\J(u,u)^*$ has components of energy, and as $trace(\J(u,u)$ vanishes in the vacuum, we should really think of

$$|\J(u,u)|=\sqrt{\J(u,u)^*\J(u,u)}=\sqrt{(\J(u,u)^2)}$$ 
as the total energy including potential energy available which is not in the gravitational source energy.  This is because in a diagonal frame for $\J(u,u),$ the curvatures along the principal axes are all contributing to the energy of deformation of any virtual test object placed there in the spacetime.  For instance, in the vacuum near a very large spherically symmetric object, the curvatures in the two principle directions orthogonal to the radial direction must be the same and their total added to that along the radial direction must be zero, hence the radial principal curvature is twice the absolute value of that in the direction orthogonal to the radial direction.  An object placed in the vacuum at an event near such an object would receive an enormous energy of deformation as the principal directions of curvature would operate on the virtual object independently in the three principal directions, for all practical purposes.  Thus it is the sum of the absolute values of the eigenvalues of $\J(u,u)$ that really should be considered as the true total energy density $E_u$ in the spacetime at event $q \in M,$ including potential energy which does not gravitate, and that is

\begin{equation}\label{JacobiEnergy}
E_u=trace(|\J(u,u)|).
\end{equation}
If $E_u=trace(|\J(u,u)|)$ vanishes throughout a region of spacetime for all timelike unit vectors, then $|\J(u,u)|$ vanishes in that region for all timelike unit vectors, and thus by the observer principle, it follows that $\J=0$ and hence that region is flat.  Of course, by the observer principle, it is enough that $\J(u,u)$ itself vanishes for an open set of timelike unit vectors in order to conclude that a region of spacetime is flat.

We will until further notice now
restrict to the case $n+1=4,$ so we have ordinary spacetime, and
therefore $div(T)=0.$  In the case of $n+1=4,$ if the condition that $div(T)=0$ is dropped in the usual derivation where one equates $T$ to a linear combination of the metric tensor, the Ricci tensor and the product of the scalar curvature with the metric tensor, and only requires the time components give Newtonian gravity in the Newtonian limit, the result is a generalization of the Einstein equation with a new free parameter which when equal to 1 gives the usual Einstein equation.  This has been investigated as to its ramifications for cosmology \cite{ALRAWAF}. But, this does mean that the assumption $div(T)=0$ is necessary for this type of derivation of the Einstein equation, since without it the free parameter may be other than unity.

From (\ref{gravengdensgendim=n}) we now have

 \begin{equation}\label{gravenrgydensitygeoform}
 T_g=(1/8\pi G)[Ric+(1/2)Rg]=(1/8\pi G)(Einstein+Rg).
 \end{equation}
 From the last expression on the right, we see, as $T$ and the Einstein
 tensor, $Einstein=Ric-(1/2)Rg$ both have zero divergence, that

 \begin{equation}\label{energydivergence}
 div(H)=div(T_g)=(1/8\pi G)div(Rg)=(1/8\pi G)dR=-d[c(T)].
 \end{equation}
 So even though the total gravitational source energy momentum stress and gravitational source energy momentum stress due to gravity itself are not
 infinitesimally conserved, the divergence is simply proportional to the exterior
 derivative of  the scalar curvature. Of
 course,  $div(T_g)=-d[c(T)]$ is obvious
 from the definition, (\ref{gravenergydensity}), once we accept $div(T)=0.$
 In particular, as $d^2=0,$ this means that

 \begin{equation}\label{exteriorderivative of grav divergence}
 d[div(H)]=d[div(T_g)]=0,
 \end{equation}
but (\ref{energydivergence}) is even better as it shows $div(T_g)$
is an exact 1-form on $M.$

On the other hand, the equation (\ref{energydivergence}), when
written

\begin{equation}\label{divgravenergy1}
div(T_g)+d[c(T)]=0
\end{equation}
has another interpretation.  In classical continuum mechanics
written in four dimensional form of space plus time, the
divergence of the energy stress tensor equals the density of
external forces.  Of course in relativity, the energy momentum
stress tensor $T$ contains everything and there are no external
forces, as gravity is not a force.  But, we can view
(\ref{divgravenergy1}) as saying that from the point of view of
the gravitational field, the matter and fields represented by $T$
are acting on the gravitational field as an external force density
of $-d[c(T)].$  In classical continuum mechanics, the external
force density has zero time component, but relativistically such
is not the case, the force only has zero time component in the
instantaneous rest frame of the object acted on. We can therefore
view (\ref{divgravenergy1}) as saying that the divergence of the
the gravitational field's  source energy momentum stress tensor is being balanced
by the rate of increase of $-c(T).$ If $p_x,p_y,p_z$ are the
principal pressures in the frame of an observer with velocity $u,$
where $g(u,u)=-1,$ then $\rho_u=T(u,u)$ is the energy density
observed, and $div(T_g)(u)$ is then the power loss density of the
gravitational field.

Now $c(T)=-\rho_u+p_x+p_y+p_z=-\rho_u+3p_u,$ where $p_u$ is the
isotropic pressure, so (\ref{divgravenergy1}) becomes

\begin{equation}\label{divgravenergy2}
div(T_g)(u)=D_u\rho-3D_up_u.
\end{equation}
Thus, the observer sees the divergence of source energy of the
gravitational field is exactly the rate of increase of energy
density of the matter and fields less the rate of increase of
principal pressures.  In particular, in any dust model of the
universe (pressure zero), the source gravitational energy dissipation is
exactly balanced by the rate of increase of energy density of the
matter and fields. If $T$ is purely the electromagnetic stress
tensor in a region where there are only electromagnetic fields,
then $c(T)=0,$ and the gravitational source energy momentum stress tensor has zero
divergence, so is then infinitesimally conserved.

To compare our gravitational source energy momentum stress tensor of the gravitational
field with the various gravitational energy pseudo-tensors, keep
in mind that all examples of such gravitational energy
pseudo-tensors can be made to vanish by appropriate choice of
coordinates and therefore cannot represent real energy of any kind
in relativity.    Such pseudo-tensors generally obey coordinate
conservation laws in appropriately chosen coordinates making them
useful in certain calculations, but they cannot represent real
energy as in relativity, real energy cannot just be transformed
away by some choice of coordinates.  By contrast, our $T_g$ is
fully covariant and represents localizable gravitating energy
density as seen by each observer, but in general is not conserved,
as $div (T_g)$ may not vanish in general.

Looking back at the derivation and proof, one can now see that if
there is a distinction between active gravitational mass and
inertial gravitational mass, then in equation
(\ref{crrtdnewtnlawtides1}), the first of the two terms is the
energy density due to active gravitational mass and the second of
the two terms being the sum of principal pressures is therefore an
inertial mass, as it is inertial mass not following geodesic
motion which creates pressure.  This would mean that in the
equation $Ric=4 \pi G [T+T_g]$ the second term on the right is the
tensor which has the inertial mass whereas the first term is the
term with the active gravitational mass.  But, this would seem to
lead to a violation of the principle of relativity, as the
pressures would be different in different reference frames leading
to conversion between active gravitational and inertial masses
depending on the observer.  Thus, this derivation seems to
indicate the equality of active gravitational mass with inertial
mass, a point which is not addressed in the usual derivations of
Einstein's equation. Possibly an improved version of this
derivation might derive the equality of inertial and active
gravitational mass.  Of course, the geodesic hypothesis itself
makes the passive gravitational mass equal to the inertial mass,
which seems to be the reason why the problem of equality of active
gravitational and inertial mass is often overlooked in elementary
treatments of general relativity.

Finally here, we should point out that Einstein's original
Equivalence Principle is often misconstrued to say that
gravitational fields can be transformed away by choice of
coordinates, and this is certainly not the case, as Frank Tipler
has stated on many occasions.  This is well known to experts in
general relativity. Gravity in general relativity is curvature of
spacetime, and curvature cannot be transformed away. If we view
connection coefficients as "gravitational forces", then using
normal coordinates at a point (that is in a neighborhood of a given point) makes the connection coefficients vanish at that point (event), but not necessarily in a neighborhood of that point.  And this
merely reflects the fact that gravitational forces do not exist in
general relativity, virtually by definition. Einstein used the
example of an accelerating coordinate system to effectively
transform away a uniform gravitational field in which there is no
actual curvature of spacetime and therefore no real gravity.  Of course, using continuity, any function vanishing a a given point can be made arbitrarily small in a sufficiently small neighborhood of that point.  Therefore, at
each event, given a specified limit in level of measurement
accuracy, there is a neighborhood in which curvature effects
cannot then be measured, and in such neighborhoods of an event,
the equivalence principle may be used effectively.  One must be
careful of subtle pitfalls, since the curvature may actually be non-zero. For instance, when Einstein used the
elevator thought experiment to reason that light would bend in a
gravitational field, he was using the fact that in the accelerated
reference frame the null geodesics appear curved and then
generalizing to arbitrary gravitational fields.

In fact, the elevator thought experiment merely gives the result
for the bending of light that Newtonian gravity in flat Euclidean
space would give under the assumption that photons have inertial
mass. It takes the full Schwarzschild solution to arrive at the
correct answer for the bending of light, which Einstein
fortunately realized before the experimental measurements were
made.

Another way to look at this light bending problem would be that in
NEIL we have out in the near vacuum of space that for the light
beam the law of gravity is $Ric=4 \pi G T_{EM}$ where $T_{EM}$
denotes the energy momentum stress tensor of the electromagnetic
field. But, as $c(T_{EM})=0,$ we have $(T_g)_{EM}=T_{EM},$ so
Einstein's equation, EHNIL, becomes $Ric=4 \pi G [T_{EM}+T_{EM}]=4
\pi G [2T_{EM}]$ which means that the photon's electromagnetic
field gives twice the curvature of the spacetime at points along
its track as would be the case in NEIL, which should reasonably
lead to the doubling of the bending angle. Of course this is a
nonsense argument, since the light bending has to do with tracks
of null geodesics in the gravitational field of a large
gravitating object and not the gravitational field of an
electromagnetic wave itself. But, if we think of a photon passing
a planet, theoretically, we are allowed to think of the planet as
following a path in the gravitational field of the photon, and the
preceding analysis says the planet's path should be bent twice as
much in Einstein's theory as in Newton's theory, so reciprocally,
the photon's track should be bent twice as much. Maybe the
argument is not so specious, and should be examined further. On
the other hand, this does tell us that the effective gravitational
mass of pure electromagnetic radiation or laser light is double
its inertial mass, which possibly could be detected using powerful
lasers in an inertial confinement fusion laboratory, thus leading
to another test of Einstein's theory. Tolman (\cite{TOLMAN},
Chapter VIII) has noticed the prevalence of this doubling effect
for electromagnetic radiation in many examples, all calculated
using the weak field approximation. But, now using Theorems
\ref{enrgmmntmstrsgravtheorem} and \ref{theoremeinstein}, we see
that the EHNIL is telling us the effective gravitational
mass-energy of electromagnetic fields is very generally double the
inertial mass-energy.

More generally, our conclusion here is that $\rho_u+3p_u$ is the
effective  gravitational source mass-energy density observed by an
observer with velocity $u.$ For that is what is dictated by
Einstein's equation, since it is equivalent to the EHNIL and the
gravitational energy density postulate, by Theorem
\ref{theoremeinstein}.

\med

%%%%%%%%%%%%%%%%%%%%%%%%%%%%%%%%%%%%%%%%%%%%%%%%%%%%%%%%%%%%%%%%%%%%%%%%%%%%%%%%%%%%%%%%%%%%%%%%%%%%%%%%%%%%%%%%%%%%
%%%%%%%%%%%%%%%%%%%%%%%%%%%%%%%%%%%%%%%%%%%%%%%%%%%%%%%%%%%%%%%%%%%%%%%%%%%%%%%%%%%%%%%%%%%%%%%%%%%%%%%%%%%%%%%%%%%
%%%%%%%%%%%%%%%%%%%%%%%%%%%%%%%%%%%%%%%%%%%%%%%%%%%%%%%%%%%%%%%%%%%%%%%%%%%%%%%%%%%%%%%%%%%%%%%%%%%%%%%%%%%%%%%%%%%

\section{THE GRAVITATION CONSTANT $G$}

\med

Our development here is the same as that in \cite{DUPRE} and \cite{DUPRE2}, except modified to reflect the required modifications for gravitational energy which serves as source of gravity.

So far, we have not said anything about the determination of the
gravitation constant $G.$  To evaluate this, we merely need to
check the results of experiments with attractive "forces" between
masses. But it is much simpler to just use Newtonian gravity in an
easy example where the results should be obviously approximately
the same.  Consider an observer situated at the center of a
spherical dust cloud of uniform density $\rho,$ and calculate the
tidal or separation acceleration field using Newton's law of
gravitation. We can observe here that the matter energy momentum stress
tensor satisfies $T(v,w)=\rho g(v,u)g(v,w)$ where $u$ is the
velocity field of the dust cloud.  Thus we calculate easily that
$T_g(u,u)=0$ meaning that a co-moving observer sees the
gravitational field as having gravitational source energy density zero.  In this case,
the NEIL and EHNIL coincide for $u$ and thus as it seems
reasonable that the NEIL should have the Newton gravitation
constant as its constant, then that means $G=G_N.$

It is easy to give a more elementary argument here. At distance
$r$ from the center, but inside the cloud, the mass acting on test
particles at radial distance $r$ is simply the mass inside that
radius, $M(r),$ by spherical symmetry, as is well-known in
Newtonian gravitation. Here, we have $M(r)=(4/3)\pi r^3 \rho.$

But Newton's Law says the acceleration of a test mass near the
center of the dust cloud is radially inward, and if $r$ is the
distance from the center, then the radial component of
acceleration is given by

\begin{equation}\label{Newton1}
a_r(r)=-G_N\frac{M(r)}{r^2}=-G_N \frac{4 \pi \rho r}{3}.
\end{equation}
Here, $G_N$ is the Newtonian gravitation constant.

On the other hand, considering an angular separation of $\theta,$
the spatial separation is $s=r\theta,$ so the relative
acceleration of nearby test particles in the $s-$ direction
perpendicular to the radial direction is therefore

\begin{equation}\label{Newton2}
a_s(r)=\theta a_r(r)=-G_N\frac{4 \pi \rho r \theta}{3}=-G_N\frac{4
\pi \rho s}{3}.
\end{equation}
Thus the rate of change of separation acceleration of nearby
radially separated test particles in the radial direction at given
$r$ is by (\ref{Newton2}),

\begin{equation}\label{Newton3}
\frac{da_r}{dr}=-G_N \frac{ 4 \pi \rho}{3},
\end{equation}
whereas in the $s$ direction we have the rate of change of
separation acceleration is

\begin{equation}\label{Newton4}
\frac{da_s}{ds}=-G_N \frac{ 4 \pi \rho}{3},
\end{equation}
the same result again.  But there are two orthogonal directions
perpendicular to the radial, so now we see that if ${\bf a}_u$
denotes the spatial or tidal acceleration field around our
observer at the center of the dust cloud, then

\begin{equation}\label{Newton5}
div_u({\bf a}_u)=-G_N 4 \pi \rho.
\end{equation}

As we are dealing with dust, the pressures are zero, so there is
no gravitational source energy density of gravity, and thus $\rho$ is now the total source
energy density seen by our observer.  Thus, we have by
(\ref{ricci2}), that $Ric(u,u)=G_N 4 \pi \rho=4 \pi G_N ~H(u,u).$
But now comparing this result with (\ref{einsteinequation1}), with
$k_3=4 \pi G,$ we see that we must have $G=G_N.$

At this point, we can simply choose units such that $G=1$ and we
henceforth drop this factor from the equation for simplicity.

\med

%%%%%%%%%%%%%%%%%%%%%%%%%%%%%%%%%%%%%%%%%%%%%%%%%%%%%%%%%%%%%%%%%%%%%%%%%%%%%%%%%%%%%%%%%%%%%%%%%%%%%%%%%%%%%%%%%%%%
%%%%%%%%%%%%%%%%%%%%%%%%%%%%%%%%%%%%%%%%%%%%%%%%%%%%%%%%%%%%%%%%%%%%%%%%%%%%%%%%%%%%%%%%%%%%%%%%%%%%%%%%%%%%%%%%%%%%
%%%%%%%%%%%%%%%%%%%%%%%%%%%%%%%%%%%%%%%%%%%%%%%%%%%%%%%%%%%%%%%%%%%%%%%%%%%%%%%%%%%%%%%%%%%%%%%%%%%%%%%%%%%%%%%%%%%%

\section{THE EINSTEIN DERIVATION}

\med

It is interesting that Einstein realized fairly early in his
search for the gravitation equation that the vacuum equation
should be $Ric=0.$  This lead him to try the equation $Ric=4\pi
G\,\ T,$ as the general gravitation equation when matter is
present. In fact, this equation obviously results from the
observer principle if we assume spacetime satisfies NEIL instead
of EHNIL, that is if our observers neglect the gravitational source energy density of
the gravitational field. He soon rejected this as not being
compatible with reality, partly due to the fact that it would
require that $div\,Ric=0.$

He also knew that the energy density of the gravitational field
should be included in the source, so if he had found the
expression we have for the gravitational source energy momentum stress tensor of the
gravitational field, he would have surely arrived at the final
equation at this time.

As it was, in summary, he finally \cite{EINSTEIN2} took the
already accepted vacuum equation $Ric=0$ and for special
coordinate frames, he was able to rewrite the vacuum equation in
the form $s=k(t-(1/2)c(t)g)$ where $t$ is his pseudo-tensor whose
coordinate divergence is zero, and where $s$ itself is a
coordinate divergence of a third rank pseudo-tensor. Let us call a
coordinate system {\it isotropic} provided that in these
coordinates we have $det(g_{\alpha \beta})=-1.$  Thus Einstein
found it useful to restrict to isotropic coordinates.

In fact, he found $Ric=s-k(t-(1/2)c(t)g),$ to be true in any
isotropic coordinate system. He therefore interpreted $t$ as the
energy density of the vacuum gravitational field and interpreted
the new form of the vacuum equation as making $t$ the source.  He
then merely guesses that in the presence of matter with energy
momentum stress tensor $T$ the source should be $t+T$ instead of
merely $t.$ Thus, when $t$ is replaced by $t+T$ in the new form of
the vacuum equation we have $s=k[(T+t)-(1/2)c(T+t)g]$ as the
candidate for the general non-vacuum equation.  We then see that
moving the terms involving the pseudo tensor back to the left side
of the equation results in $Ric=k[T-(1/2)c(T)g],$ true in any
isotropic coordinates.  But this last equation is a fully a
covariant equation.

In a sense, his derivation begins with and is based on the pseudo
tensor for the energy density of the gravitational field.
Technically his equation was $-Ric=k[T-(1/2)c(T)g],$ because he
used a metric with signature $(+---).$ Of course, our summary has
left out the Hamiltonian method he used to arrive at his pseudo
tensor and the considerable technical calculations required to
arrive at the vacuum equation in terms of the pseudo tensor. But,
in outline, it is really quite a nice derivation, and in many ways
superior to most of the modern derivations.

The fact that the coordinate divergence of the pseudo tensor
vanishes means that the general Stokes' theorem (sometimes in this
particular setting called the divergence theorem or Gauss'
theorem) can be applied to give macroscopic conservation of
gravitational energy. On the other hand, once the source $t$ is
replaced by $t+T,$ it is no longer the case that this latter
gravitational pseudo tensor has vanishing coordinate divergence,
so the gravitational pseudo tensor loses its conservation law in
the presence of matter.  It is rather Einstein's total energy
stress pseudo tensor, $T+t,$ material and gravitational, whose
coordinate divergence vanishes. It seems this lead Einstein to
question the need and even the validity for general covariance in
his formulation, since the vanishing of the covariant divergence
could not be integrated to give any macroscopic conservation law.

In our opinion, the real major weakness in the argument is the
reliance on a variational argument using a Hamiltonian to obtain
the pseudo tensor, since there is no apparent way to justify this,
other than picking something that seems simple out of thin air.
Specifically, he chose the integrand to be $g^{\mu
\nu}\Gamma^{\alpha}_{\mu \beta} \Gamma^{\beta}_{\nu \alpha}$ for
his variational integral, where here we can take
$\Gamma^{\alpha}_{\beta
\gamma}=\omega^{\alpha}(\nabla_{e_{\gamma}}e_{\beta}),$ with
$(e_{\alpha})$ the coordinate frame basis and with
$(\omega^{\alpha})$ the corresponding dual frame basis.

For instance, in the Hilbert argument using the scalar curvature
as the integrand, it is certainly simple to write down and after
the fact, it does give the correct equation. But, what is the
physical basis for choosing the scalar curvature for the
variational argument? Without any physical justification, we have
to admit it is just a lucky guess based on trying the simplest
thing, which, of course, is always a good idea when you have
nothing else to go on. Just because you try something simple and
it happens to work does not mean you understand why it works.
After nearly a century of general relativity, we are quite
confident of the results of action principles in general
relativity, but for deriving the equation, it is unsatisfactory.
For instance, the Einstein derivation evolved out of Einstein's
consideration of various physical problems and possibilities and
he happened to arrive at the result at almost the same time as
Hilbert.  Now if Hilbert had proposed his action method two years
earlier, would Einstein have believed the equation was the correct
equation?  Maybe and maybe not. It is putting the cart before the
horse. In fact, setting $Einstein=E,$ if we simply want a simple
derivation, as $E(u,u)$ is half the scalar curvature of $W_R,$ the
exponential Riemannian space orthogonal to $u,$ for any time-like
unit vector $u,$ the simplest derivation is just to guess each
observer sees his spatial curvature proportional to his observed
mass density with a universal constant of proportionality.  This
immediately gives $E(u,u)=2k_3T(u,u),$ for each time-like unit
vector $u$ from which we conclude $E=2k_3T,$ for some constant
$k_3,$ by the observer principle. Instant derivation of the
Einstein equation. But why should we have spatial curvature
proportional to energy density? If you are aware of the observer
principle, it is the obvious guess, but you have no way to know
you are correct, since there is no physics in the argument-it is
just mathematics.

We are not claiming that the Einstein Hilbert Lagrangian method
has no value. It surely has value for certain calculations,
especially since the Lagrangian terms for many fields are known
and can be added in to the calculations. We are simply pointing
out, that if you did not know the equation before such a
derivation, you still might not be convinced.  The derivation we
have presented here seems to have the convincing property that it
is the only way to very generally and naturally "push" Newton's
Law of gravity into a general relativistic framework which
includes the gravitational source energy density of the gravitational field as seen by
all observers. It would seem to us that the fact that $div(T)=0$
is an immediate consequence of this derivation makes it all the
more attractive and convincing.

\med

%%%%%%%%%%%%%%%%%%%%%%%%%%%%%%%%%%%%%%%%%%%%%%%%%%%%%%%%%%%%%%%%%%%%%%%%%%%%%%%%%%%%%%%%%%%%%%%%%%%%%%%%%%%%%%%%%%%%%%
%%%%%%%%%%%%%%%%%%%%%%%%%%%%%%%%%%%%%%%%%%%%%%%%%%%%%%%%%%%%%%%%%%%%%%%%%%%%%%%%%%%%%%%%%%%%%%%%%%%%%%%%%%%%%%%%%%%%%%
%%%%%%%%%%%%%%%%%%%%%%%%%%%%%%%%%%%%%%%%%%%%%%%%%%%%%%%%%%%%%%%%%%%%%%%%%%%%%%%%%%%%%%%%%%%%%%%%%%%%%%%%%%%%%%%%%%%%%%

\section{THE COSMOLOGICAL CONSTANT}

\med

If we include the cosmological constant $\Lambda$ in the Einstein
equation, it becomes

\begin{equation}\label{cosmo1}
Ric -(1/2)Rg+\Lambda g=8\pi T,
\end{equation}
which is of course the same as

\begin{equation}\label{cosmo2}
Ric-(1/2)Rg=8\pi[T-(1/8\pi)\Lambda g],
\end{equation}
which means we view the equation here as having a modified source energy
momentum stress tensor

\begin{equation}\label{cosmo3}
T_{\Lambda}=T-(1/8\pi)\Lambda g.
\end{equation}

We then have $c(T_{\Lambda})=c(T)-(1/2\pi)\Lambda,$ so the
effective  source energy momentum stress tensor of the gravitational field
is

\begin{equation}\label{cosmo4}
(T_{\Lambda})_g=T-c(T)g+(3/8\pi)\Lambda g=T_g+(3/8\pi)\Lambda g,
\end{equation}
and the effective total energy momentum stress tensor serving as
source is

\begin{equation}\label{cosmo5}
H_{\Lambda}=2T-c(T)g+(1/4\pi)\Lambda g=H+(1/4\pi)\Lambda g.
\end{equation}
In any case, as $div~g=0,$ it follows that our conclusions about
the gravitational source energy-momentum flow of the gravitational field from
(\ref{divgravenergy1}) and (\ref{divgravenergy2}) remain valid,
even in the presence of a cosmological constant.  Equations
(\ref{cosmo4}) and (\ref{cosmo5}) are corrections of equations
(8.4) and (8.5) of \cite{DUPRE} where the numerical coefficients
of the $\Lambda g$ terms were incorrectly given as $1/2 \pi,$ in
both cases.

\med

%%%%%%%%%%%%%%%%%%%%%%%%%%%%%%%%%%%%%%%%%%%%%%%%%%%%%%%%%%%%%%%%%%%%%%%%%%%%%%%%%%%%%%%%%%%%%%%%%%%%%%%%%%%%%%%%%%%%%%
%%%%%%%%%%%%%%%%%%%%%%%%%%%%%%%%%%%%%%%%%%%%%%%%%%%%%%%%%%%%%%%%%%%%%%%%%%%%%%%%%%%%%%%%%%%%%%%%%%%%%%%%%%%%%%%%%%%%%%
%%%%%%%%%%%%%%%%%%%%%%%%%%%%%%%%%%%%%%%%%%%%%%%%%%%%%%%%%%%%%%%%%%%%%%%%%%%%%%%%%%%%%%%%%%%%%%%%%%%%%%%%%%%%%%%%%%%%%%

\section{QUASI LOCAL MASS}

\med

The problem of defining the energy contained in a space-like
hyper-surface has led to many different definitions of the mass
enclosed by a closed space-like surface contained in an arbitrary
spacetime manifold, and these go by the general name quasi-local
mass. Typically, they are defined by some kind of surface integral
and give an indication of the mass enclosed by the space-like
surface. One of the oldest is known as the Tolman integral and is
advocated by Fred Cooperstock \cite{COOPERSTOCK}, \cite{C&D}, \cite{C&DANNALS},
\cite{TOLMAN} (see also \cite{L&L}, equation (100.19), as well as
\cite{MITRA1}, \cite{MITRA2}). For an extensive survey of these we
refer the interested reader to \cite{SZ}. In particular, the
results of \cite{TIPLER} on the Penrose quasi-local mass show that
the results can be interesting when the space-like surface is not
the boundary of a space-like hyper-surface.

A list of desirable properties of any definition of quasi-local
mass is given in \cite{YAU3}, where in particular it is shown that
for their definition, the quasi-local mass enclosed by a
space-like surface $S$ is non-negative provided that the dominant
energy condition holds and the surface $S$ is the boundary of a
hyper-surface, $\Omega.$  It is further assumed that the boundary
surface $S$ has positive Gauss curvature and space-like mean
curvature vector, and consists of finitely many connected
components. The local energy condition assumed (equivalent to the
dominant energy condition) is framed in terms of the second
fundamental form of the hyper-surface, and in particular, we can
see that for a geodesic hyper-surface it reduces to the condition
that the scalar curvature of the hyper-surface, $\Omega,$ is
non-negative, since in that case the second fundamental form
vanishes (extrinsic curvature zero). But, in this case, the scalar
curvature of the space-like hyper-surface $\Omega$ is
$2Einstein(u,u)=16 \pi T(u,u),$ where $u$ is a time-like future
pointing unit normal field on $\Omega.$  So if the energy momentum
stress tensor satisfies the weak energy condition in this case,
then the energy density as seen by observers riding the
hyper-surface is non-negative, and we would simply integrate $(1/8
\pi)Einstein(u,u)$ over the hyper-surface to find the energy
inside, which is clearly non-negative.

The amazing result in \cite{YAU3} is that the quasi-local mass
defined there, which is defined in terms of integrals over the
boundary $S,$ is non-negative under the dominant energy condition.
For instance, their results show if the energy inside any one
component of $S$ vanishes, then $S$ is connected and $\Omega$ is
flat (\cite{YAU3}, Theorem 1, page 183), and thus the result shows
that the energy in $\Omega$ is in some sense determined by the
geometry of the boundary and its mean curvature vector under the
assumptions stated above.

The small scale and large scale asymptotic properties are analyzed
in \cite{YU}, and in particular, in the vacuum the result is that
to fifth order the quasi-local mass for small spheres is
asymptotic to the Bel-Robinson tensor whereas in general to third
order it is asymptotic to the energy momentum stress tensor of
matter times volume. Unfortunately, there are drawbacks to this
definition of quasi-local mass, as pointed out in \cite{MURCH},
and it seems the situation is improved with the later treatments
of Wang and Yau, in \cite{WY1}, \cite{WY2}, \cite{WY3}.

Let us use our total energy momentum stress tensor to formulate an
invariant approach to quasi-local energy momentum stress.  If we have an open
subset $U$ of $M$ and vector fields $v$ and $w$ defined on
$U,$ we can think of the integral curves of $u$ as being the
histories of a field of observers.  We can then form $H(v,w)$ as a
function on $U$ and assuming orientability of $U$ we can choose a
normalized volume form $\mu_U$ so that $\mu_U(u,e_1,e_2,e_3)=1,$
for $(u,e_1,e_2,e_3)$ any local positively oriented orthonormal
frame. The natural way to proceed here seems to be to form a type
of action integral which we can call, following \cite{C&D} and \cite{C&DANNALS}, the {\it spacetime source total energy momentum stress integral (including that due to the gravitational field itself)}, as in (\ref{integraltensorfield})

\begin{equation}\label{action1}
A(v,w,U)=(\int_U H \mu )(v,w)=\int_U H(v,w) \mu_U=\frac{1}{4 \pi G} \int_U Ric(v,w)
\mu_U.
\end{equation}
The strong energy condition says $Ric(v,v) \geq 0$ for any
time-like vector $v,$ and thus if this condition is satisfied,
then clearly the only way that the action integral for $v=u=w$ can vanish is
for $Ric(u,u)$ to vanish on $U.$ But, this does not seem to
obviously allow us to conclude that $Ric=0$ on $U.$  However, on
physical grounds, it should allow us to conclude $Ric=0$ on $U.$
That is, if we fill spacetime with observers everywhere, then if
nobody observes any gravitating energy, there should be none.  On the other hand, if $A(v,v,V)=0$ for every timelike vector field on any open subset $V$ of $U,$ then $Ric=0$ on $U$ by the observer principle (corollary \ref{integraltensorfieldobserverprinciple}) and we conclude that there is no gravitational source in $U.$

In order to make use of the total gravitational source energy momentum stress tensor, $H,$ in a
setting similar to that of Liu and Yau, \cite{YAU3} or Wang and
Yau \cite{WY1},\cite{WY2}, one would assume an appropriate energy
condition, and then for a space-like hyper-surface $K$ with future
time-like unit normal field $u,$ it is natural to consider
$H(u,u)\mu_K$ where $\mu_K$ is the volume form due to the
Riemannian metric induced on $K.$ The integral of $H(u,u)\mu_K$
over all of $K$ should be the total energy inside $K.$

More generally, if we assume that $H$ is dominantly non-negative,
that is, it satisfies the analogue of the dominant energy
condition for $T,$ then given another reference future pointing
time-like vector field $k,$ one might then integrate $H(u,k)
\mu_K$ over $K.$ If a 2-form $\alpha$ can be found on $K$
satisfying $d \alpha = H(u,k)\mu_K,$ and if $K$ is a 3-submanifold
with boundary $B,$ then by Stokes' theorem, the total energy
inside $K$ is related to the integral of $\alpha$ over the
boundary $B$ of $K.$

In particular, we say that $K$ is {\it instantaneously static} if
there is an open set $U \subset M$ containing $K$ and a vector
field $k$ on $U$ which is future pointing and orthogonal to $K$
and which satisfies Killing's equation, at each point of $K.$ If
$\omega=k^*$ is the dual 1-form to $k,$ so $\omega(v)=g(k,v)$ for
all vectors $v,$ then this is equivalent to requiring $Sym(\nabla
\omega)|K=0$ or equivalently that $(d\omega)|K=2\nabla \omega|K,$
which to be perfectly clear means that the difference
$d\omega-2\nabla \omega$ as calculated on $U$ in fact is zero at
each point of $K.$ Then as in the Komar \cite{KOMAR} integral (see
\cite{PINTONETOSOARES}, \cite{WALD}, pages 287-289 or
\cite{POISSON}, pages 149-151) it follows that

\begin{equation}\label{quasilocalmass0}
(-1/8 \pi)d*d\omega=(1/4 \pi)Ric(u,k)\mu_K=H(u,k)\mu_K.
\end{equation}
Here, $*$ denotes the Hodge star operator on $M.$ Thus, $(-1/8
\pi)*d\omega$ is a potential for the total energy on $K.$ For any
closed 2-submanifold $S$ of $K$ we define the quasi-local total
energy $H(S,k)$ by

\begin{equation}\label{quasilocalmass1}
H(S,k)=-\frac{1}{8 \pi}\int_S *d(k^*).
\end{equation}
Thus, if $K_0 \subset K$ is a submanifold with boundary $S=\del
K_0,$ then by Stokes' Theorem, (\ref{quasilocalmass1}) becomes

\begin{equation}\label{quasilocalmass2}
H(S,k)=-\frac{1}{8 \pi}\int_{K_0} d*d(k^*)=\int_{K_0} H(u,k)
\mu_K,
\end{equation}
which is then non-negative if the strong energy condition holds.
Thus, if $H(S,k)=0,$ with $S=\del K_0,$ then by
(\ref{quasilocalmass2}), under the assumption that the strong
energy condition holds, we would conclude that $Ric(u,k)=0$ on
$K_0.$ But, this means that $Ric(u,u)=0$ on $K_0,$ which means
that none of the observers in the field detect any energy.

Notice that if we have an asymptotically flat spacetime with a
global time-like Killing vector field orthogonal to a spacelike
slice, normalized to be a unit vector at spatial infinity, then
our definition of the quasi-local total energy would be exactly
the Komar mass which is well known in the literature \cite{SZ}.
Thus in the expression $H(S,k),$ the normalization for $k$ is
determined by requiring that it be of unit length at the event at
which the observer is located.  If the observer is located so that
$S$ is in the observer's causal past, then it would seem we must
assume that the domain of $k$ contains this past light cone.

In general, if $k$ is a Killing field on all of the open set $U,$
then being orthogonal to $K$ means (\cite{WALD}, page 119,
(6.1.1)) that also $\omega \wedge d\omega =0,$ where $\omega=k^*.$
Then (see \cite{WALD}, page 443, (C.3.12)) we find, using
$f=ln(|g(k,k)|),$

\begin{equation}\label{quasilocalmass4}
d\omega=-\omega \wedge df,
\end{equation}
and using the fact that here $*[\omega \wedge df]=-(e^{f/2}D_nf)
\mu_S,$ where $n$ is the outward unit normal to $S=\del K_0,$ and
$\mu_S=dA$ is the area 2-form on $S,$ we obtain finally,

\begin{equation}\label{quasilocalmass5}
H(S,k)=-\frac{1}{8 \pi} \int_S e^{f/2}D_nf dA.
\end{equation}
In particular, for the vacuum Schwarzschild solution with mass
parameter $\M,$ taking the Killing field $k=\del_t,$ we see easily
that the mass calculated using the integral
(\ref{quasilocalmass5}) gives the value $\M$ for the mass enclosed
by any sphere centered at the "origin" when we normalize the
Killing field to be a unit vector at infinity.  On the other hand,
if we calculate that value of the integral by normalizing to make
the Killing vector a unit at radial coordinate $r_0,$ as $H(S,k)$
is homogeneous in $k,$ the normalizing constant comes out
resulting in

\begin{equation}\label{quasilocalmass3}
\M_{r_0}=\frac{\M}{[1-\frac{2 \M}{r_0}]^{1/2}}.
\end{equation}
Keeping in mind this is now the total gravitational source energy, gravitational and
massive, this indicates a problem develops as $r_0 \rightarrow 2
\M,$ even though we know it is not a real problem for the
spacetime. The problem is probably due to the normalization
involving the Schwarzschild radial coordinate which obviously
breaks down at $r_0=2 \M.$ After all, what we are integrating is
equivalent by Stokes' Theorem to integrating $H(u,k)\mu_K,$ when
$S=\del K_0,$ and we really want to be integrating $H(u,u)\mu_K.$
We do not have the actual potential. On the other hand, this does
seem to reflect correctly the fact that as one approaches the
horizon of a black hole it takes infinite force to keep from
falling in.

Let us now use these results to compute the spacetime gravitational source integral
(\ref{action1}). To do this, let us assume that $U$ is foliated by
spacelike submanifolds determined by the Killing parameter $t$ on
$U,$ so the leaves are the level manifolds of $t,$ and that $u$ is
orthogonal to each leaf of this foliation. We assume that $k=hu$
is the Killing vector field on all of $U,$ where $h$ is the
redshift factor. Assume now that $K$ is a compact 4-submanifold of
$U$ with boundary $\del K$ and that the intersection of $K$ with
the leaf at time $t$ is $K_t$ with boundary $\del K_t$ which is
the intersection of $\del K$ with the leaf at time $t,$ for $t_1
\leq t \leq t_2.$ Then $k^*=hu^*$ and we see the volume form on
$K$ can be expressed as $\mu_K=\mu_{K_t}hdt.$ This means that the
action integral can be expressed as

\begin{equation}\label{action2}
A(u,u,K)=\int_{t_1}^{t_2} \int_{K_t}H(u,u)h
\mu_{K_t}dt=\int_{t_1}^{t_2}\int_{K_t}H(u,k)\mu_{K_t}dt=\int_{t_1}^{t_2}H(K_t,k)dt.
\end{equation}
In the particular case of the Schwarzschild solution, this leads
immediately to

\begin{equation}\label{action3}
A(u,u,K)=\M \Delta t,
\end{equation}
with the Killing vector normalized so the redshift factor is 1 at
infinity.

We can now see that the real quasi-local mass problem is the fact that in
integrating over a spatial slice, the proper time is elapsing at
different rates at different parts of space, so that in general,
the quasi-local mass definitions have to contend with this problem
whether they like it or not \cite{MITRA1}, \cite{MITRA2}.  Thus,
in general, if we have no Killing vector field, if $t$ is an
arbitrary "time" function on $U,$ and if $K(t_1,t_2)$ is the
submanifold of $U$ given by $t_1 \leq t \leq t_2,$ then we should
simply think of $A(u,u,K(t_1,t_2))=\M_{av}\Delta t,$ where now
$\M_{av}$ is the average quasi-local mass over the given time
interval. This naturally leads to taking

\begin{equation}\label{masst}
\M(t)=\frac{d}{dt}A(u,u,K(t_1,t)),
\end{equation}
as the mass at time $t.$ Thus for the Schwarzschild solution we
now find that the mass is $\M,$ the mass parameter, which
indicates that the mass parameter is the total gravitating source mass including that
due to gravitational field energy.

As pointed out above in the discussion of the Jacobi curvature operator as the energy momentum stress of everything including gravity, the total potential of all the energy as seen by $u$ appears to be $trace(|\J(u,u)|)$ which only vanishes where spacetime is flat.  We should therefore regard the integral

$$\int_K trace(|\J|) \mu$$ as the total energy momentum stress of everything including gravity and {\it all its potential} in the region $K$ of spacetime.  For if the quadratic form of this multilinear map vanishes, then by Corollary \ref{integraltensorfieldobserverprinciple}, we would have $trace(|\J|)=0$ and hence $\J=0,$ so spacetime is flat.

\med

%%%%%%%%%%%%%%%%%%%%%%%%%%%%%%%%%%%%%%%%%%%%%%%%%%%%%%%%%%%%%%%%%%%%%%%%%%%%%%%%%%%%%%%%%%%%%%%%%%%%%%%%%%%%%%%%%%%
%%%%%%%%%%%%%%%%%%%%%%%%%%%%%%%%%%%%%%%%%%%%%%%%%%%%%%%%%%%%%%%%%%%%%%%%%%%%%%%%%%%%%%%%%%%%%%%%%%%%%%%%%%%%%%%%%%%
%%%%%%%%%%%%%%%%%%%%%%%%%%%%%%%%%%%%%%%%%%%%%%%%%%%%%%%%%%%%%%%%%%%%%%%%%%%%%%%%%%%%%%%%%%%%%%%%%%%%%%%%%%%%%%%%%%%

\section{GRAVITATIONAL RADIATION}

\med

Frank Tipler has pointed out that due to the definition of the gravitational source
energy momentum stress tensor of the gravitational field in terms
of equation (\ref{gravenergydensity}), it follows that the speed
of sound in the gravitational field equals the speed of light, for
a vacuum electromagnetic field. More specifically, he points out
that the gravitational source gravitational field energy density tensor equals the
electromagnetic energy density tensor exactly, according to
(\ref{gravenergydensity}) as the contraction of the latter is
zero. Thus the speed of sound in the gravitational field due to a
vacuum electromagnetic field is equal to the speed of sound in a
vacuum electromagnetic field, which is of course the speed of
light. This certainly seems reasonable given our way of viewing
the gravitational source energy density in terms of electromagnetic
fields (the laser light box). On the other hand, in the vacuum there is no gravitational source energy of the
gravitational field, and consequently from this point of view, a
gravitational wave carries no gravitational source field energy through
the vacuum.  That is there is no part of a vacuum gravitational wave solution which could appear to gravitate.

This point of view has been elaborated previously
\cite{COOPERSTOCK}, in what has become known in the literature as
the Cooperstock hypothesis, purely on mathematical and somewhat
philosophical grounds that the equations for the various pseudo
tensors have no content in the vacuum, and as well, on the basis
of his detailed computation \cite{COOPERSTOCK} involving an
example of a capacitor in a gravitational wave. On the other hand,
the gravitational source energy momentum stress tensor of the gravitational field is
not divergence free which means it can dissipate in one place and
appear in another. That is, the time varying matter tensor causes
gravitational source energy to disappear into the vacuum and then
reappear elsewhere where there is matter.  Since it is really $(1/4 \pi)\J$ that should be viewed as the total energy momentum stress, we view this as the effect of the tracial part of the Jacobi curvature going over into off diagonal components in the vacuum, with diagonal components balancing to zero in the vacuum, and then in the presence of matter, becoming unbalanced with possibly some off diagonal components feeding back over into the diagonal components. This of course is a computational
difficulty for analyzing gravitational radiation, and it means
that using a coordinate conserved pseudo-tensor or any other
device which is conserved in some useful sense is certainly
justified if it aids in calculation.

For instance, Hayward \cite{HAYWARD}, in analyzing gravitational
radiation in a quasi-spherical approximation defines an energy
density tensor for the gravitational radiation which carries
positive energy and in the second approximation reacts on the
solution when included in the source of the truncated Einstein
equation. This means that one can in special circumstances use
special definitions in a way that can be usefully interpreted
physically, even if it is technically a fiction. On the other
hand, to quote \cite{BORISSOVA}, "At the present time there are
many solutions of the gravitational wave problem, but none of them
are satisfactory...another difficulty: there is no general
covariant d'Alembertian, which being in its clear form, could be
included into the Einstein equations."

However, in answer to this last statement from \cite{BORISSOVA}, we can now look to our Jacobi curvature operator and recall (\ref{dalembertian1}), (\ref{dalembertian2}), and Proposition \ref{dalembertian3}.  To review,  according to our remarks immediately following (\ref{Jacobi3}), it is the case that 
$-\J(u,u)=\nabla_u^2,$ the proof preceding (\ref{Jacobi3}) actually being valid on any semi-Riemannian manifold.  Notice that $\J(u,u)$ is a tensor field as is $\J,$ but $\nabla_u^2$ is not a tensor field.  We can view this as saying that $\J(u,u)$ is a "tensorized" version of $\nabla_u^2.$  Thus, if $(u=e_0,e_1,e_2,...,e_n )$ is a local orthonormal frame field, then we can apply this to each of the frame fields and we see that

\begin{equation}\label{blockJ}
-g^{\alpha \beta}\J(e_{\alpha},e_{\beta})={\bf \square}_J
\end{equation} 
can be thought of as a tensorized version of the d'Alembertian operator, where we have given it a subscript in order to distinguish it from the actual d'Alembertian operator, $\square.$  But this last equation makes sense in any local frame field.  In fact, according to Proposition \ref{dalembertian3} or equation (11.32) on page 286 of \cite{MTW}, it is immediately obvious that in Riemann normal coordinates, at the point at the coordinate system origin,

\begin{equation}\label{blockg}
(3/2) \square g_{\mu \nu}=g({\bf \square}_J e_{\mu}, e_{\nu}   )=g(-g^{\alpha \beta} \J(e_{\alpha}, e_{\beta})e_{\mu}, e_{\nu})=-g^{\alpha \beta} g(\J(e_{\alpha}, e_{\beta})e_{\mu}, e_{\nu})
\end{equation}
But, $\J$ satisfies the exchange property $\J^{\dagger}=\J,$ by Proposition \ref{pc8}, which means that just like for the Riemann curvature tensor, we have

\begin{equation}\label{exchange}
g(\J(w,x)y,z)=g(\J(y,z)w,x), \mbox{ for all tangent vectors } w,x,y,z \mbox{ at any give point.}
\end{equation}

This means that in any local frame field $(e_{\alpha})$ on $M,$ we have

\begin{equation}
-g({\bf \square}_J e_{\mu}, e_{\nu}   )=g^{\alpha \beta} g(\J(e_{\alpha}, e_{\beta})e_{\mu}, e_{\nu})=g^{\alpha \beta}g(\J(e_{\mu},e_{\nu})e_{\alpha}, e_{\beta} )=trace [\J(e_{\mu},e_{\nu})]=R_{\mu \nu}
\end{equation}
where as usual, in index notation, as in \cite{MTW},

$$R_{\mu \nu}=Ric(e_{\mu}, e_{\nu}).$$
Suppose now that $(M,g,T)$ is a spacetime model.  The Einstein equation is then

$$Ric=4 \pi G[2T-c(T)g].$$
Put $S= [2T-c(T)g],$ for convenience.  Thus, $S$ is the source of curvature through the Einstein equation $Ric= [4 \pi G] S.$  Thus, we can chose a smooth tensor transformation $P: TM \lra TM$ which is a tensor field with

$$g(Pv,w)=-S(v,w),\mbox{ for all } (v,w) \in TM \oplus TM.$$
It is thus natural to think of $P$ as the total source momentum transformation field including that due to gravity.

We then have,
$$-g({\bf \square}_J e_{\mu},e_{\nu})=Ric(e_{\mu},e_{\nu})=-[4 \pi G] g(Pe_{\mu},e_{\nu}),$$
hence finally

\begin{equation}\label{GRVWAVE}
{\bf  \square}_J=[4 \pi G] P,
\end{equation}
which is again the Einstein equation, but now in the form of a wave equation with source $P$ which is effectively the curvature source as momentum density of all matter and fields including gravity itself.  In normal coordinates, using (\ref{blockg}) or (\ref{dalembertian3}), this last equation becomes

$$\square g_{\mu \nu}=(2/3)[4 \pi G] S_{\mu \nu}, \mbox{ at the point of origin of the normal coordinates.}$$
{\bf Thus the Einstein equation for gravity is equivalent to requiring that at each point, in Riemann normal coordinates with that point as origin, the wave equation  for the metric tensor with source $(2/3)[4 \pi G]S$ is satisfied at that point.}

We must keep in mind here, that our view of the gravitational source energy density of
the gravitational field is in complete agreement with the Einstein
equation, so it cannot contradict any of its results and likewise,
no result of solving the Einstein equation can possibly contradict
our view of the gravitational source energy density of the gravitational field. In
particular, both Carl Brans and Frank Tipler (in personal
communication) have expressed concerns about how the view
expressed here on the gravitational source gravitational energy momentum stress tensor
relates to the analysis of the energy dissipation from binary
pulsars, an issue also addressed in \cite{COOPERSTOCK} in relation
to the Cooperstock hypothesis. Particularly relevant here are the
calculations in \cite{COOPERSTOCK2} and \cite{COOPERSTOCK3} of the
gravitational radiation due to a rotating rod, showing the general
relativistic calculation to be consistent with the Cooperstock
Hypothesis. The idea that gravitational radiation carries energy
away may be a useful idea for keeping track of the various
"energies", or conserved quantities, in the system, but the
calculations always involve a choice of reference background
metric which produces the apparent "energy".  However, our view here is that in the vacuum there is no gravitational source gravitational energy momentum stress, as  the Ricci tensor vanishes, but the full jacobi curvature operator can be non-vanishing in the vacuum, so the gravity wave would be carrying gravitational energy which does not serve as gravitational source.  Putting a detector in the path of the wave destroys the vacuum and therefore converts the non-source energy into source energy which can be measured.  Alternately, it seems
that there is no mathematical vacuum in realistic models of the
universe, because of background radiation and possibly dark
energy, so there is background matter to carry even source gravitational
energy.

On the other hand, there certainly must be something going on in the vacuum, and as our previous discussion of the Jacobi curvature operator indicated, we can only have $trace(|\J|)=0$ in flat spacetime, so any model in which there is nonvanishing curvature, which would certainly include any non-static model with any gravity waves of any kind would definitely have a non-vanishing of $trace(|\J|).$

It would seem from this view of the Einstein equation as a tensorized wave equation, that the general theoretical treatment of gravity waves should follow, but we will leave this topic for future research.  One can notice that a number of the formulas in the subject relating curvature to gravitational radiation in \cite{MTW}  are rather immediate from the Jacobi curvature operator and the geometrization formula for algebraic Jacobi curvature operators presented in the preceding sections.

\med

%%%%%%%%%%%%%%%%%%%%%%%%%%%%%%%%%%%%%%%%%%%%%%%%%%%%%%%%%%%%%%%%%%%%%%%%%%%%%%%%%%%%%%%%%%%%%%%%%%%%%%%%%%%%%%%%%%%%%
%%%%%%%%%%%%%%%%%%%%%%%%%%%%%%%%%%%%%%%%%%%%%%%%%%%%%%%%%%%%%%%%%%%%%%%%%%%%%%%%%%%%%%%%%%%%%%%%%%%%%%%%%%%%%%%%%%%%%
%%%%%%%%%%%%%%%%%%%%%%%%%%%%%%%%%%%%%%%%%%%%%%%%%%%%%%%%%%%%%%%%%%%%%%%%%%%%%%%%%%%%%%%%%%%%%%%%%%%%%%%%%%%%%%%%%%%%%

\section{BLACK HOLES}

\med

Since the vacuum has no gravitational source energy density, it follows that the
assignment of mass to black holes or to cosmological solutions is
heavily influenced by boundary conditions assumed for the solution
to the Einstein field equations (see e.g. \cite{NEST}). For
instance, in the case of a Schwarzschild black hole, if we try to
integrate over a region enclosed by a sphere, we find that it is
not the boundary of any compact spacelike slice.  The preceding
analysis leading to (\ref{quasilocalmass3}) would have to be
modified to include also an inner boundary as a cutoff so that the
region bounded is compact.  On the other hand, as it stands, for
the Schwarzschild case, the mass is $\M$ no matter what matter
resides in the interior as long as the matter is not all inside
the Schwarzschild radius, which indicates that it is reasonable to
assign the artificial mass $\M$ to the Schwarzschild black hole
with mass parameter $\M,$ as a reflection of a boundary condition,
the boundary being the black hole horizon.  Thus, in the general
black hole case, one of the various definitions of quasi-local
mass must be adopted. As far as we can see, the actual gravitational source energy
momentum stress tensor of the gravitational field cannot help
here. On the other hand, as the curvature operator is non-vanishing, this means that $trace(|\J|)$ is non-vanishing inside the black hole horizon, which would seem to indicate again that we should think of this as the real energy momentum stress of everything including any potential energy of the gravitational field beyond mere gravitational source energy.

\med

%%%%%%%%%%%%%%%%%%%%%%%%%%%%%%%%%%%%%%%%%%%%%%%%%%%%%%%%%%%%%%%%%%%%%%%%%%%%%%%%%%%%%%%%%%%%%%%%%%%%%%%%%%%%%%%%%%%%%%
%%%%%%%%%%%%%%%%%%%%%%%%%%%%%%%%%%%%%%%%%%%%%%%%%%%%%%%%%%%%%%%%%%%%%%%%%%%%%%%%%%%%%%%%%%%%%%%%%%%%%%%%%%%%%%%%%%%%%%
%%%%%%%%%%%%%%%%%%%%%%%%%%%%%%%%%%%%%%%%%%%%%%%%%%%%%%%%%%%%%%%%%%%%%%%%%%%%%%%%%%%%%%%%%%%%%%%%%%%%%%%%%%%%%%%%%%%%%%

\section{SUPERENERGY}

The main tensors, excluding pseudo-tensors and teleparallel gravity frame dependent tensors, which have been considered as important for the description of energy density in General Relativity are the Bel and Bel Robinson tensors which have been generalized to what are called {\bf Superenergy Tensors} in very general settings.  For some of the relevant literature, we refer the reader to the work of Bergqvist \cite{BERGQVIST1}, \cite{BERGQVIST2}, Bonilla and Senovilla \cite{BONILLA1},  Dereli and Tucker \cite{DERELI1}, Deser \cite{DESER1}, Robinson \cite{ROBINSON1}, Pozo \cite{POZO1}, Senovilla \cite{SENOVILLA1}, and Teyssandier \cite{TEYSSANDIER1}.   The most important aspect of these tensors is their positivity property.  In the main case of interest, they are fourth rank tensors which are positive when all four vector inputs are in the forward light cone.  In fact, if such a tensor vanishes on putting in a single timelike vector for all four inputs, then the whole tensor vanishes.  This is highly desirable for any energy tensor, as it says in effect if any single observer can see no energy, then there is none, and spactime is flat.  A treatment of these positivity properties is most efficiently done with Clifford Algebra and spinors, but this would take us too far afield.  However, we will look to see the basic properties of the most important example which is the Bel Tensor, and look at how the various terms relate to the Jacobi and Riemann curvature tensors.

 It is clear from the formulas for these tensors as expressed in terms of the Riemann curvature tensor that as they are quadratic functions of the Riemann operator, they are likewise quadratic functions of the Jacobi operator.  However, the main drawback of these tensors is that in the natural units, they have the units of the square of the energy, and thus some effort has gone into developing a square root of these tensors, \cite{BONILLA1}.  But, in our opinion, since the Bel tensor for instance is a quadratic function of the Jacobi curvature operator, it would seem to lend more weight to our contention that the Jacobi curvature operator should be considered the true carrier of all the energy, momentum, and stress.  We can take for the basic definition of the Bel tensor, $\B =\B_{\R}$, in its purely covariant form, the expression

\begin{equation}\label{BEL TENSOR DEF}
\begin{aligned}
\B(w,x,y,z) = {} &  \B_0(w,x,y,z) + \B_0(w,x, z,y)     \\
                        &  - \frac{1}{2} [g(w,x) \B_0(e^{\gamma}, e_{\gamma}, y,z) +g(y,z) \B_0(w,x, e^{\gamma}, e_{\gamma}) ] \\
                        & +\frac{1}{8} g(w,x)g(y,z) \B_0(e^{\gamma},e_{\gamma}, e^{\lambda}, e_{\lambda}),
\end{aligned}
\end{equation}
where

\begin{equation}\label{BEL 0}
\B_0(w, x, y, z) = g( w,  \R( y, e_{\sigma}) e_{\rho}) g( x,  \R(z, e^{\sigma}) e^{\rho}).
\end{equation}
Here we are using the MTW \cite{MTW} conventions, but because of the symmetries of the Riemann tensor and the quadratic nature of the expression for the Bel tensor, the result is really independent of the conventions chosen for the Riemann tensor and its index notation.  Here of course, $(e_{\alpha})$ is any frame for $V$ and $(e^{\alpha})$ is the corresponding dual frame under the natural identification of $V$ and its dual $V^*=L(V;\bR)$ via the non-degenerate inner product $g$ on $V.$  We take $\R$ to be any (algebraic) metric Riemann curvature operator in $PC(V) = L(V,V; L(V;V)).$  In particular, if we take the Weyl tensor in place of the Riemann curvature tensor, the result is the Bel-Robinson tensor.

Because $\R^* = -\R = \tilde{\R},$ we can write

$$\B_0(w, x, y, z) = g(  \R(  e_{\sigma}, y)w, e_{\rho}) g(e^{\rho}, \R(e^{\sigma}, z) x ) =     g(  \R(  e_{\sigma}, y)w, \R(e^{\sigma}, z) x )  = g(\R^{\#}(w, y) e_{\sigma}, \R^{\#}(x, z) e^{\sigma}) .   $$
Thus,

\begin{equation}\label{B00}
\B_0(w, x, y, z) =  trace( \R^{\#}(x, z)^* \R^{\#}(w,y)).
\end{equation}

\medskip

Now, it is natural to define a bilinear product $m:PC(V) \times PC(V) \lra PC(V),$ which we will denote as $m(A,B) = \langle A | B \rangle,$ with

$$m(A,B) = \langle A | B \rangle = trace(A^*B), \mbox{ where } g(m(A,B)(w,x) y ,z) =trace(A(w,x)^* B(y,z)).$$
The last term of course is the trace of the composition of the operators, which of course does not depend on the order. 
Moreover, the trace of the adjoint of a linear operator is the same as the trace of the operator and adjoint reverses multiplicative order.  Explicitly, the ordinary Hilbert Schmidt  non degenerate inner product on $L = L(V;V)$ is given by

$$\langle A | B \rangle = trace(A^*B) = \langle B | A \rangle \mbox{ since } trace(AB) = trace(BA) \mbox{ and } trace(A^*) = trace(A), \mbox{ for all } A,B \in L.$$
It follows that in $PC(V)$ we have

$$\langle A | B \rangle^{\dag} = \langle B | A \rangle.$$
In detail, we have

$$\langle A | B \rangle^{\dag} (w,x,y,z) = \langle A | B \rangle (y,z,w,x) = trace( A(y,z)^* B(w,x) ) = trace( B(w,x)^* A(y,z)) = \langle B | A \rangle (w,x,y,z). $$

Thus,

$$\B_0 (w,x,y,z) = \langle \R^{\#} | \\R^{\#} \rangle (x,z,w,y) = \langle \R^{\#} | \\R^{\#} \rangle (w,y,x,z) = \B_0(x    , w  , z  , y)   $$ 
that is to say, $\B_0$ has the property

\begin{equation}\label{BEL1}
\B_0(w,x,y,z)= \B_0(x,w,z,y).
\end{equation}

With $[2,3]$ denoting the transposition permutation of the integers which interchanges 2 and 3, we have now simply

$$\B_0 = [2,3] \langle \R^{\#} | \\R^{\#} \rangle.$$

\medskip

From Proposition \ref{nuQT} we know that $\R^{\#}$ has the metric exchange property and also $\R^{\#}(x,y)^*= \R^{\#}(y,x),$ for all $x,y \in V.$

For $v \in V,$ let $v^*$ denote the member of $V^*$ given by $v^*(x)=g(v,x),~x \in V.$  Then going back to (\ref{BEL 0}) we have

$$\B_0(w, x, y, z) = w^*(\R^{\#}(e_{\rho}, e_{\sigma}) y) \cdot x^*(\R^{\#}( e^{\rho}, e^{\sigma}) z) = [w^* \otimes x^*][\R^{\#}(e_{\rho}, e_{\sigma}) \otimes \R^{\#}(e^{\rho}, e^{\sigma})] (y \otimes z).$$
This means that considering $\B_0 \in PC(V)$ means $\B_0(w,x)$ is the composition of linear operators

$$\B_0(w,x) = [w^* \otimes x^*] [\R^{\#}(e_{\rho}, e_{\sigma}) \otimes \R^{\#}(e^{\rho}, e^{\sigma})],$$
and using $g$ for the inner product  induced on $V \otimes V$ by $g$ on $V,$ we have

$$g(\B_0(w,x)y,z) = g(w \otimes x , [\R^{\#}(e_{\rho}, e_{\sigma}) \otimes \R^{\#}(e^{\rho}, e^{\sigma})] (y \otimes z)).$$
But, as $\R^{\#}(a,b)^* = \R^{\#}(b,a)$, it follows from the relation $[A \otimes B]^*=A^* \otimes B^*$ for operators, that if

$$T = T_{\R} = [\R^{\#}(e_{\rho}, e_{\sigma}) \otimes \R^{\#}(e^{\rho}, e^{\sigma})], \mbox{ then } T = T^*,$$
and

$$\B_0(w, x, y, z) = g(w \otimes x , [T_{\R}][ y \otimes z]).$$

This proves that $\B_0$ has the metric exchange property,

$$\B_0^{\dag} = \B_0.$$

\medskip

As the Jacobi curvature operator is

$$\J = Sym( \R^{\#} ),$$
it follows that

$$ \B_0(w,w) = \langle \J | \J \rangle (w,w),  \mbox{ so } Sym_2 (\B_0) = \langle \J | \J \rangle.$$
This also means that

$$\B_0(w,w,w,w) = trace([ \J(w,w)]^2).$$
But, if $w$ is timelike, then as  $\J(w,w)w = 0$ and $\J(w,w)$  is self-adjoint, it is then completely determined as a self-adjoint operator on the Euclidean orthogonal complement of $w$ if $g$ is a Lorentz metric.  Thus:

\begin{equation}\label{B0 POS}
\mbox{If } g \mbox{ is Lorentz and $w$ is timelike,  then } \B_0(w,w,w,w) \geq 0 \mbox{ and } \mbox{ if }  \B_0(w,w,w,w)=0 \mbox{ then } \J(w,w) = 0.
\end{equation}

\bigskip

Now, returning to the definition of the Bel tensor (\ref{BEL TENSOR DEF})  , we see that the second term is $\B_0^*(w,x,y,z),$ and therefore the sum of the first two terms is self-adjoint.  Let

$$\B_1 = \B_0 + \B_0^*.$$
Using property given in (\ref{BEL1}), we have

$$\B_1(x,w,y,z) = \B_0(x,w, y, z) + \B_0(x, w, z, y) = \B_0( w, x, z, y) + \B_0(w,x,y,z) = \B_1(w,x,y,z).$$  
This means that in $PC(V)$ the curvature operator $\B_1$ is symmetric as well as self-adjoint.  Since $\B_0$ has the metric exchange property, now, using that and again using (\ref{BEL1}),

$$\B_1(w,x,y,z) = \B_0(w,x,y,z) + \B_0(w, x,z, y) = \B_0(y, z,w, x) + \B_0(z, y, w, x)$$$$ = \B_0(y,z,w,x) + \B_0( y,z,x,w) = \B_1(y,z,w,x),$$
and therefore $\B_1$ has the metric exchange property.  Thus, $\B_1$ is symmetric and self-adjoint and has the metric exchange property.

\medskip

Now, considering the two terms of (\ref{BEL TENSOR DEF}) inside the bracket in the middle, the first of these is $g(w,x)\B_0(e^{\gamma}, e_{\gamma}, y,z).$  But we now know that $\B_0$ has the metric exchange property, so this term can be expressed as

$$g(w,x)\B_0(e^{\gamma}, e_{\gamma}, y,z) = g(w,x)\B_0(y, z, e^{\gamma}, e_{\gamma}) = g(w,x) trace( \B_0(y,z)) = [g \otimes trace(\B_0)](w,x,y,z).$$
We now see that the second term inside the brackets in (\ref{BEL TENSOR DEF}) is simply $[g \otimes trace(\B_0)]^{\dag}$ which means that we have 

\begin{equation}\label{BEL TENSOR 1}
\B = \B_1 - \frac{1}{2}[ [ g \otimes  trace(\B_0)] + [g \otimes  trace(\B_0)]^{\dag}] +\frac{1}{8} [ g \otimes g] [ \B_0(e^{\gamma},e_{\gamma}, e^{\lambda}, e_{\lambda})].
\end{equation}
Since $\B_0$ is symmetric, it follows that $trace(\B_0)$ is also symmetric, and therefore that $\B$ is symmetric as a member of $PC(V),$ meaning that

$$\B(w,x)=\B(x,w),$$
and as well now, $\B$ has the metric exchange property as the sum of the two terms in brackets do, and already we know $\B_1$ has the metric exchange property, and as all terms are self-adjoint, it follows that $\B$ has all the properties of an abstract Jacobi curvature operator except that it need not be acyclic (that is it need not satisfy the first Bianchi identity).  In particular, as $\R$ is a linear function of $\J,$ we see that $\B$ is a quadratic function of $\J$ as $\B_0$ is.

While the positivity properties of the Bel tensor are most efficiently analyzed with Clifford algebra, we can see directly some aspects of this directly from the properties of $\J$ and $\R.$

If $g$ is a Lorentz metric and $u$ is timelike, then $g(u,u) < 0,$ so that the sign of the second term is the sign of $trace(\B_0)(u,u)=trace[\B_0(u,u)] =\B_0(u,u, e^{\rho} , e_{\rho}).$
Since $\R^{\#}(x,y)^* = \R^{\#}(y,x),$ from (\ref{B00}) we see that

$$trace[\B_0(u,u)]=\B_0(u,u, e^{\rho} , e_{\rho}) = trace(\R^{\#}( u, e_{\rho})\R^{\#}(u, e^{\rho})).$$

\bigskip

Now, in general,  by (\ref{B00}) we have 

$$\B_0(w, x, y, z) = \langle \R^{\#}(x, z) | \R^{\#}(w, y) \rangle = \langle \R^{\#}(w, y) |  \R^{\#}(x, z) \rangle.$$
By  Corollary \ref{RRJJ}, we have

\begin{equation}\label{RhashJhash}
\R^{\#} = \J + \frac{1}{2} \R \mbox{ and } \J^{\#} = \frac{3}{4} \R -\frac{1}{2} \J. 
\end{equation}
However, in $L(V;V)$ with the Hilbert Schmidt inner product defined by the trace, the self adjoint operators are orthogonal to the anti-adjoint operators, and the values of $\J$ are all self adjoint, whereas all the values of $\R$ are anti-adjoint.  Thus,

\begin{equation}\label{B00 DECOMP}
\B_0( w, x, y, z ) = \langle \J(w,y) | \J(x, z) \rangle  +  \frac{1}{4} \langle \R( w, y ) | \R(x, z ) \rangle,
\end{equation}
and

\begin{equation}\label{TRACEB00}
trace[ \B_0(w,w) ] =  \B_0(w,w, e_{\rho} , e^{\rho}) =       \langle \J ( w, e_{\rho} ) | \J( w, e^{\rho} ) \rangle + \frac{1}{4} \langle \R( w, e_{\rho} ) | \R( w, e^{\rho} ) \rangle.
\end{equation}

We need to look more closely at $\J(x,y)$ as a self adjoint operator on $V,$ especially in case that $V$ is Lorentz.  First, we note that

$$\J(u,x)u = \J(x,u)u = Sym(\R^{\#})(x,u)u =(1/2)  [ \R^{\#}(x,u)u +  \R^{\#}(u,x)u]  $$$$ = (1/2)   [ \R(u,u)x + \R(u,x)u] = (1/2) \R(u,x)u,$$
for any vectors $u, x \in V.$   But, $\R(u,x)u = - \R(x, u )u = - \R^{\#}(u,u)x = - \J(u, u)x.$  Thus,

\begin{equation}\label{JSWITCH}
\J(u,x)u = (1/2) \R(u,x) u = - (1/2) \J(u,u)x,
\end{equation}
and in particular, since $\R(u,u) = 0,$ we have $\J(u,u)u = 0,$ for all $u \in V.$  As $\R$ is anti symmetric and has the exchange property, it follows that

$$g(\R(x,u)u,u) = g(\R(u,u)x,u) = 0,$$
and therefore always $\R(u, x)u = 2 \J(u,x)u$ is orthogonal to $u.$   Let us restrict our attention to the case of unit vectors, timelike or spacelike.   Thus, assume $|g(u,u)| = 1.$   Let $Q_u$ be the orthogonal projection in $L(V;V)$ with $Q_u V = \bR u \subset V,$ so $Q_u$ is the self adjoint projection operator projecting onto the line through $u.$  Let $P_u$ be the orthogonal complementary projection to $Q_u,$ so $P_u = 1_V - Q_u.$  Then, as $u$ cannot be null, being as $|g(u,u)| = 1,$

$$Q_u x = g(u,u) g(x,u)u, \mbox{ for all } x \in V.$$
Since $\J(u,x)u$ is orthogonal to $u,$ it follows that

$$Q_u \J(u,x) Q_u = 0 \mbox{ and } \J(u,x) Q_u = P_u \J(u,x) Q_u,$$
and therefore,

\begin{equation}\label{EXPAND JX}
\J(u,x) = P_u \J(u,x) P_u + P_u \J(u,x) Q_u + Q_u \J(u,x) P_u = P_u \J(u,x) P_u + \J(u,x) Q_u + Q_u \J(u,x) P_u.
\end{equation}
Consequently, putting $J_x = \J(u,x), ~~ P = P_u,$ and $Q=Q_u,$ as $PQ = 0 = QP,$

\begin{equation}\label{EXPAND JX SQR}
J_x^2 = (PJ_x P)^2 +P J_x P J_x Q     +  P J_x Q J_x P + Q J_x P J_x P + Q J_x P J_x Q.  
\end{equation}
In particular, since already $0 = \J(u,u)u = J_u u,$   we have $J_u Q = 0,$ so now from (\ref{EXPAND JX}), it follows that

$$J_u = PJ_u P + Q J_u P .$$
But, $J_u =\J(u,u)$ is self adjoint, so as the first term on the right side of the previous equation is self adjoint, the second term must also be, and thus again using $ J_u Q = 0,$

$$ Q J_u P =  [ Q J_u P]^*   =  P J_u Q =0.$$
Alternately, the fact that $J_u Q = 0$ means that $J_u Q$ is self adjoint, but as $J_u$ is self adjoint, taking adjoints gives $Q J_u = 0$ as well, which means that $Q$ commutes with $J_u,$  and therefore so must $P = P_u$ also commute with $J_u$ which immediately gives $P J_u Q = 0,$ since $PQ = 0.$  
Thus, finally, 

\begin{equation}\label{JU}
\J(u,u) = P_u \J(u,u) P_u, \mbox{ for all non null } u \in V,
\end{equation}
and therefore also

\begin{equation}\label{JUSQR}
[\J(u,u)]^2 =  [ P_u \J(u,u) P_u   ]^2, \mbox{ for all non null } u \in V.
\end{equation}

In particular, again, as $\B_0(u,u,u,u) = trace([\J(u,u)]^2),$ if $P_u V$ is spacelike so $g$ on $P_uV$ is positive definite, then $\B_0(u,u,u,u) \geq 0$ and if $\B_0(u,u,u,u) $ vanishes, then $\J(u,u) = 0.$

Returning now to the more general case of $J_x = \J(u,x),$ we have by ( \ref{EXPAND JX SQR}) and the fact that $trace(AB) = trace(BA),$ again, as $PQ = 0,$

$$trace(J_x^2) = trace( (PJ_x P)^2     +  P J_x Q J_x P  + Q J_x P J_x Q) = \langle PJ_x P | P J_x P \rangle + \langle Q J_x P | Q J_x P \rangle + \langle P J_x Q | P J_x Q \rangle,$$
and also

$$trace(J_x^2) = trace( (PJ_x P)^2     +  P J_x Q J_x P  + Q J_x P J_x Q)  =  trace( (PJ_x P)^2     +  P J_x Q J_x   +  J_x P J_x Q) .          $$
For any $y \in V,$ since $PJ_x Q = J_x Q,$ which means $J_x u = P J_x u,$

$$PJ_xQJ_x y = g(u,u) g(J_x y , u) PJ_x u = g(u,u) g(J_x u, y) J_x u = g(u,u) g(PJ_x u, y) J_x u = g(u,u) g(PJ_x u, Py) PJ_x u.$$
It follows that if the frame is chosen so that $u$ is one of the frame vectors in an orthonormal frame, then 

$$trace( PJ_x Q J_x ) = g(u,u) g(PJ_x u, Pe_{\rho} ) g(P J_x u , e^{\rho}) =g(u,u) g^{\sigma \rho} g( PJ_x u, Pe_{\sigma}) g( PJ_x u, Pe_{\rho}),$$
so if $g$ on $PV$ is positive definite, then 

$$trace(P J_x Q J_x) = g(u,u) g(PJ_x u, PJ_x u), \mbox{ and }   g(PJ_x u, PJ_x u)  \geq 0. $$

On the other hand, since $P$ and $J_x$ are self adjoint and $P^2 = P,$

$$Q J_x P J_x u = g(u,u) g(J_xP J_x u, u) u = g(u,u) g(PJ_x u, J_x u)u = g(u,u) g(P J_x u, P J_x u)u$$
and if $g$ on $PV$ is positive definite, then we have also

$$trace(J_xPJ_xQ) = [g(u,u)]^2 g(PJ_x u, J_x u) =  g(P J_x u, P J_x u).$$
Thus, if $g$ on $PV$ is positive definite, then $trace(J_x^2) \geq 0.$

%%%%%%%%%%%%%%%%%%%%%%%%%%%%%%%%%%%%%%%%%%%%%%

%\vskip 1in

%%%%%%%%%%%%%%%%%%%%%%%%%%%%%%%%%%%%%%%%%%%%%%%  

%PROBLEMS BELOW

\vskip .2in

Let us first concentrate on the first term on the right hand side of (\ref{TRACEB00}).  Suppose that we take $w=u$ a timelike or spacelike vector in $V.$  Then, $\J(u,u)u = 0,$ so $u$ is an eigenvector with eigenvalue zero, and as $\J$ has all self adjoint values, it follows that $\J(u,u)$ acts as a self-adjoint operator on the orthogonal complement of $u.$  But, if $e$ is any vector, then by the exchange property, $g(\J(u,e)u, u) = g(\J(u, u) u, e) = 0.$  Thus, in case $u$ is timelike, $\J(u,e)u$ is spacelike for all $e \in V,$ spacelike or not.  Let $P_u$ be the orthogonal projection In $L(V;V)$ which projects $V$ onto the orthogonal complement of $u.$  Then $P_u$ is self adjoint and 

$$P_u \J(u,x)u = \J(u,x)u, \mbox{ for all } x \in V, \mbox{ if } |g(u,u)| = 1.$$
 In what follows, $u$ will always denote a unit vector in $V,$ meaning that $|g(u,u)| = 1.$   We will try to remind the reader of this fact but in case of an ommision, please keep this in mind.  In general, for any unit vector $u \in V,$ spacelike or timelike, using again $P_u$ to denote the orthogonal projection operator in $L(V;V)$ of $V$ onto the orthogonal complement of $u,$ we have 

$$trace(A) = g(u,u)g(Au,u) + trace( P_u A P_u ), \mbox{ for all } A \in L(V;V),$$
and therefore if $A , B \in L(V;V).$

$$\langle A | B \rangle = trace(A^*B) =   g(u, u) g(Au, Bu) + trace(   [AP_u]^*  B P_u) =  g(u, u) g(Au, Bu) + \langle   AP_u  | B P_u \rangle.$$

Now,  $[P_u \J(x,y)]^* = \J(x,y) P_u$, and $P_u^2 = P_u.$  To simplify notation, in what follows now set $P_u=P,~Q_u=Q,$ and $\J(x,y)=J.$

%\vskip .2in

%%%%%%%%%%%%%%%%%%%%%%%%%%%%%    

%PROB;EM BELOW TO FIX AS CANNOT GET PJP HERE, only JP

%\vskip .5in

Notice that if $A,B \in L(V;V),$ then $PA$ and $QB$ are orthogonal in $L(V;V).$  Indeed, as $PQ=0=QP,$

$$\langle PA | QB \rangle = trace([PA]^*QB)=trace((APQB)=0.$$

Thus, interchanging roles of $A$ and $B,$ we see that likewise

$$\langle QA | PB \rangle =0.$$

So,

$$\langle J | J \rangle = \langle JP |JP \rangle +g(u,u) g( Ju,Ju) $$$$  =  \langle (P+Q)JP |(P+Q)JP \rangle   +g(u,u) g( Ju,Ju) $$

$$   = \langle PJP | PJP \rangle   + \langle QJP | QJP \rangle +  g(u,u)g( Ju,Ju). $$

But,

$$  \langle QJP | QJP \rangle = trace( PJQJP) = trace(P^2JQJ)=trace(PJQJ)=trace( QJPJ)=trace(JPJQ)$$
and $JPJQ$ is a rank one operator, so, assuming that $u$ is one of the frame vectors in the orthonormal frame,

$$trace(JPJQ)= g(e^{\rho},e^{\sigma})g(JPJQe_{\rho},e_{\sigma}) = g(u,u)g(JPJu,u)=g(u,u)g(PJu,PJu) . $$
On the other hand,

$$ g(Ju,Ju)=g([P+Q]Ju,Ju)=g(PJu,Ju)+g(QJu,Ju)=g(PJu,PJu)+g(QJu,QJu),  $$
and
$$QJu=g(u,u)g(Ju,u)u,$$
from which we conclude that

$$g(QJu,QJu)=[g(u,u)]^2[g(Ju,u)]^2g(u,u)=g(u,u)[g(Ju,u)]^2.   $$

Thus, for all $x,y \in V,$

\begin{equation}\label{JANGLE EQ 1}
\langle \J(x,y) | \J(x,y) \rangle  = \langle  P_u  \J(x,y) P_u | P_u \J(x,y) P_u  \rangle   + 2 g(u,u) g(P_u  \J(x,y)u,P_u \J(x,y)u)  +[g(u,u) g( \J(x,y)u,u)]^2
\end{equation}
$$  = \langle  P_u  \J(x,y) P_u | P_u \J(x,y) P_u  \rangle   + 2 g(u,u) g(P_u  \J(x,y)u,P_u \J(x,y)u)  +[g( \J(x,y)u,u)]^2.   $$
We can note here, that if either $x$ or $y$ coincides with u, then by the exchange property, the last square term actually vanishes, as $\J(u,u)u=0$ and $\J^*=\J.$
Thus,

$$\langle \J(x,u) | \J(x,u) \rangle  = \langle  P_u  \J(x,u) P_u | P_u \J(x,u) P_u  \rangle   + 2 g(u,u) g(P_u  \J(x,u)u,P_u \J(x,u)u), ~x \in V. $$

\bigskip

Thus, if $V$ is Lorentz and $u$ is a timelike unit vector, since then $P_u \J(x,y) P_u$ is a self adjoint operator on the Euclidean space $P_u V \subset V,$ it follows that 

$$  \langle  P_u \J(x,y) P_u | P_u \J(x,y) P_u  \rangle \geq 0, $$
and therefore, as the last term in \ref{JANGLE EQ 1} is a square,

\begin{equation}\label{JANGLE INEQ0}
\langle \J(x,y) | \J(x,y) \rangle   \geq   2g(u,u) g(P_u  \J(x,y)u,P_u \J(x,y)u)   \mbox{ for all } x,y  \in V. 
\end{equation}
\vskip .3in

A consequence of Corollary \ref{RRJJ}, is that

$$\J^{\#} = \frac{3}{4} \R - \frac{1}{2} \J  \mbox{ and } Sym( \J^{\#} ) = - \frac{1}{2} \J,$$
and therefore,  

$$\J(u,x)u = \J(x,u ) u = \J^{\#}(u,u)x = Sym( \J^{\#} ) (u, u ) x = - \frac{1}{2} \J(u,u) x,$$
We therefore conclude from \ref{JANGLE INEQ0} that

\begin{equation}\label{JANGLE INEQ1} 
 \langle \J(u,x) | \J(u,x) \rangle   \geq  \frac{ g(u,u)}{2} g( \J(u,u)x,  \J(u,u)x) \mbox{ for all } x \in V.
\end{equation}
\medskip

As a consequence we have for $u$ a timelike unit vector in a Lorentz vector space $V,$

$$  \langle \J(u, e_{\rho} ) | \J(u, e^{\rho} ) \rangle   \geq   \frac{g(u,u)}{2} g( \J(u,u)e_{\rho}), \J(u,u) e^{\rho} )) = \frac{g(u,u)}{2} trace( [\J(u,u)]^2), $$
so

\begin{equation}\label{JANGLE INEQ2}
\langle \J(u, e_{\rho} ) | \J(u, e^{\rho} ) \rangle   \geq   \frac{g(u,u)}{2} trace( [\J(u,u)]^2), \mbox{ for any timelike unit vector } u \in V,
\end{equation}
and for any vector  unit vector $u \in V,$

\begin{equation}\label{JANGLE EQ2}
 \langle \J(u, e_{\rho} ) | \J(u, e^{\rho} ) \rangle = \langle  P_u  \J(u,e_{\rho}) P_u | P_u \J(u,e^{\rho}) P_u  \rangle +  \frac{g(u,u)}{2} trace( [\J(u,u)]^2).
\end{equation}

\bigskip

In the case of $\langle \R( x, y ) | \R( x, y ) \rangle, $ we can instead use the antisymmetry of $\R$ and the exchange property.  But here, $\R( x, y ) y = \R^{\#} (y, y )x = \J(y,y)x.$  Thus, as $\R(y, y) = 0,$ again, we have using the exchange property of $\R$ that $P_u \R(u,x ) u = \R(u,x)u,$  and therefore 

$$P_u \R(u,x) Q_u = \R(u,x)Q_u.$$

Thus, since $(-1)^2 = 1,$ the same arguments as above apply here on replacing $\J$ by $\R$ throughout giving for all $x,y \in V$

\begin{equation}\label{RANGLE EQ1}
\langle \R( x, y ) | \R( x, y ) \rangle  = \langle   P_u \R( x, y ) P_u | P_u \R( x, y ) P_u  \rangle  + 2g(u,u) g(P_u \R( x, y)u, P_u \R(x, y)u) + [g(\R(x,y)u,u)]^2.
\end{equation}
But now, applying the exchange property of $\R$ to the last square term in this expression gives zero as $\R$ is antisymmetric, so we have

\begin{equation}\label{RANGLE EQ0}
\langle \R( x, y ) | \R( x, y ) \rangle  = \langle   P_u \R( x, y ) P_u | P_u \R( x, y ) P_u  \rangle  + 2g(u,u) g(P_u \R( x, y)u, P_u \R(x, y)u).
\end{equation}

$$\langle \R( u, x ) | \R( u, x ) \rangle  = \langle   P_u \R(  u, x ) P_u | P_u  \R( u, x ) P_u  \rangle  +2g(u,u) g(\J( u, u)x, \J(u, u)x)  \mbox{ for all } x  \in V,$$

\begin{equation}\label{RANGLE INEQ1}
\langle \R(u,x) | \R(u,x) \rangle   \geq         2  g(u,u) g( \J(u,u)x, \J(u,u)x)  \mbox{ for all } x \in V, \mbox{ for a timelike unit vector } u \in V,
\end{equation}

\begin{equation}\label{RANGLE INEQ2}
\langle \R(u, e_{\rho} ) | \R(u, e^{\rho} ) \rangle   \geq   2 g(u,u) trace( [\J(u,u)]^2),
\end{equation}
and

\begin{equation}\label{RANGLE EQ2}
 \langle \R(u, e_{\rho} ) | \R(u, e^{\rho} ) \rangle = \langle  P_u  \R(u,e_{\rho}) P_u | P_u \R(u,e^{\rho}) P_u  \rangle +  2 g(u,u) trace( [\J(u,u)]^2), \mbox{ for all unit vectors } u \in V.
\end{equation}

\bigskip

Combining these results with (\ref{B00 DECOMP}) and (\ref{TRACEB00}) we have for all $u,~x \in V,$ with $u$ a unit vector,

\begin{equation}\label{B00 DECOMP1}
\B_0(u,u,x,x) =  \langle  P_u  \J(u,x) P_u | P_u \J(u,x) P_u  \rangle + \frac{1}{4} \langle  P_u  \R(u,x) P_u | P_u \R(u,x) P_u  \rangle  + g(u,u) g(\J(u,u)x, \J(u,u)x),
\end{equation}
and if $u$ is also timelike, then choosing the frame to be orthonormal with $u$ as one of the frame vectors, since $\R(u,u) = 0,$ and $P_u \J(u,u) P_u = \J(u,u),$ and since $u^*=-u,$ we then have

$$\langle P_u \J(u,u) P_u | P_u \J(u,u) P_u \rangle = g(u,u) trace([\J(u,u)]^2),$$
it follows that

\begin{equation}\label{TRACEB01}
trace( \B_0(u,u)) =  \langle  P_u  \J(u,e_{\rho}) P_u | P_u \J(u,e^{\rho}) P_u  \rangle + \frac{1}{4} \langle  P_u  \R(u,e_{\rho}) P_u | P_u \R(u,e^{\rho}) P_u  \rangle +  g(u,u) trace( [\J(u,u)]^2)
\end{equation}

$$ = 2g(u,u) trace( [\J(u,u)]^2 ) + \langle P_u \J( u, P_u e_{\rho}) P_u  |   P_u \J( u, P_u e^{\rho}) P_u     \rangle  +  \frac{1}{4} \langle  P_u  \R(u,e_{\rho}) P_u | P_u \R(u,e^{\rho}) P_u  \rangle.$$

\bigskip

Finally, the last term of the Bel tensor, $\B,$ is simply $g \otimes g \cdot trace(\B_0(e_{\gamma}, e^{\gamma})).$  Thus, if $u$ is a timelike unit vector, the last term in the expression for $\B(u,u,u,u)$ is then $trace(\B_0(e_{\gamma}, e^{\gamma})).$  Thus, if $V$ is a Lorentz vector space and $u$ is a timelike unit vector, then we can choose an orhtonormal frame so that $u$ is one of the frame vectors, and then, since $P_u u = 0, ~~ e^{\gamma} = - e_{\gamma}$ if $u = e_{\gamma},$ and $g^{\alpha \beta} = 0, $ if $\alpha \neq \beta,$

$$trace(\B_0(e_{\gamma}, e^{\gamma})) = -trace(\B_0(u,u) ) +   g^{\alpha \beta} trace(\B_0(P_u e_{\alpha}, P_u e_{\beta})).$$
In view of the factor $1/8$ in front of the last term in $\B,$ this means that for $u$ a timelike unit vector, so then $g(u,u) = -1$ and therefore we have

$$\B(u,u,u,u) = 2 \cdot trace( [\J(u,u)]^2) + \frac{7}{8} trace(\B_0(u,u)) + \frac{1}{8}g^{\alpha \beta} trace(\B_0(P_u e_{\alpha}, P_u e_{\beta}))$$

$$= 4 \cdot trace([\J(u,u)]^2)   +  \frac{7}{8} \langle P_u \J(u, P_u e_{\rho} ) P_u |  P_u \J(u, P_u e_{\rho} ) P_u  \rangle        $$

$$ + \frac{7}{32} \langle  P_u  \R(u,e_{\rho}) P_u | P_u \R(u,e^{\rho}) P_u  \rangle  +     g^{\alpha \beta} trace(\B_0(P_u e_{\alpha}, P_u e_{\beta})),    $$
and in the last expression, every term is non-negative except for possibly the last term $ g^{\alpha \beta} trace(\B_0(P_u e_{\alpha}, P_u e_{\beta})).$  We therefore only have to examine the case of 
$trace(\B_0(e,e))$ when $e$ is a spacelike unit vector, specifically, one of the basis vectors.  This is then the sum of terms of the form $\B_0(e,e,v,v)$ where $v$ is another spacelike vector, specifically, one of the basis vectors.  But then, 

$$\B_0(e,e,v,v)= \langle \J(e,v) | \J(e,v) \rangle +\frac{1}{4} \langle \R(e,v) | \R(e,v) \rangle,$$

$$= trace( [\J(e,v)]^*[\J(e,v)] )  + \frac{1}{4} trace( [\R(e,v)]^*[ \R(e,v)] ) ,                   $$
so

$$\B_0(e,e,v,v) = trace(P_u  [\J(e,v)]^*[\J(e,v)] P_u ) +g(u,u) g(\J(e,v)u, \J(e,v)u)$$$$ + \frac{1}{4} trace(P_u  [\R(e,v)]^*[ \R(e,v)] P_u ) + \frac{1}{4}g(u,u)g(\R(e,v)u, \R(e,v)u).$$

$$=   trace(P_u  [\J(e,v)][\J(e,v)] P_u ) +g(u,u) g(\J(e,v)u, \J(e,v)u)$$$$ + \frac{1}{4} trace(P_u  [\R(e,v)]^*[ \R(e,v)] P_u ) + \frac{1}{4}g(u,u)g(\R(e,v)u, \R(e,v)u) .$$
Now, let $Q_u$ be the orthogonal projection on the line through $u,$ so $Q_u w = g(u,u) g(u,w) u,$ for all $w \in V,$ and $P_u + Q_u =id_V =1.$  Also, both $P_u$ and $Q_u$ are self-adjoint idempotents.

Thus,

$$trace(P_u \J(e,v) \J(e,v) P_u) = trace( P_u \J(e,v) [P_u +Q_u] \J(e,v) P_u) $$$$= trace(P_u \J(e,v) P_u \J(e,v) P_u )  +    trace(P_u \J(e,v) Q_u \J(e,v) P_u )     $$

$$= \langle [ P_u \J(e,v) P_u ]^*  | P_u \J(e,v) P_u   \rangle  + \langle  [P_u \J(e,v) Q_u]^* |   Q_u \J(e,v) P_u   \rangle           $$

$$=  \langle  P_u \J(e,v) P_u  | P_u \J(e,v) P_u   \rangle  + \langle  [Q_u \J(e,v) P_u |   Q_u \J(e,v) P_u   \rangle.              $$

For any $A \in L(V;V),$ 

$$ Q_u A x = g(u,u) g(Ax,u)u        $$
so
$$Q_u A P_u x = g(u,u )g(A P_u x, u) u =g(u,u) g(P_u x, A^*u) u.$$
Therefore

$$  Q_u \J(e,v)  P_u x = g(u,u) g(P_u x, \J(e,v)u) = g(u,u) g(x, P_u \J(e,v) u)u,          $$
but since $P_u^2 = P_u$ and $[\J(e,v)]^* = \J(e,v),$
 
$$  Q_u \J(e,v)  P_u x = Q_u \J(u,v) P_u P_u x = g(u,u) g(P_u x, P_u \J(e,v) u ) u = g(u,u) g( \J(e,v) P_u x, u) u       $$

$$= g(u,u) g(\J(P_u x, u) e, v) u = g(u,u) g( \J(u, P_u x) e,v) u= g(u,u) g(\J(e,v) u, P_u x) u.  $$
Since $(g(u,u))^2=1,$ this means that

$$[Q_u \J(e,v)P_u]^*[Q_u \J(e,v)P_u]x = P_u \J(e,v) Q_u (g(u,u) g(\J(e,v)u, P_u x) u$$$$ = g(u,u) g(\J(e,v)u, P_u x) P_u \J(e,v)u,$$
and this means that $[Q_u \J(e,v)P_u]^*[Q_u \J(e,v)P_u]$ is only a rank one operator mapping all of $V$ into $Q_uV,$ and therefore

$$[Q_u \J(e,v)P_u]^*[Q_u \J(e,v)P_u] = Q_u [Q_u \J(e,v)P_u]^*[Q_u \J(e,v)P_u] = Q_u P_u \J(e,v) Q_u Q_u \J(e,v) P_u = 0,$$
as $Q_u P_u = 0.$

We therefore have simply

$$trace(P_u \J(e,v) \J(e,v) P_u) = trace( P_u \J(e,v) P_u)^2 ) =  \langle  P_u \J(e,v) P_u  | P_u \J(e,v) P_u   \rangle \geq 0.$$
Of course, the exact same argument applies to give

$$ trace(P_u  [\R(e,v)]^*[ \R(e,v)] P_u ) =   \langle  P_u \R(e,v) | P_u  | P_u \R(e,v) P_u   \rangle \geq 0,$$
since the only change is $[R(e,v)]^* = - \R(e,v),$ so as $-0 = 0 = Q_u P_u,$ the outcome is the same.

Thus, now the last two  terms to deal with are simply the terms which combine to give

$$ g(u,u) [ g(\J(e,v)u, \J(e,v)u) + \frac{1}{4} g(\R(e,v)u, \R(e,v)u)] $$

$$= g(u,u) [g(u, [\J(e,v)]^*\J(e,v)u) +  \frac{1}{4} g(u, [\R(e,v)]^*\R(e,v)u)] $$

$$ =  g(u,u)[g(Q_u u, [\J(e,v)]^*\J(e,v)u) + \frac{1}{4} g(Q_u u, [\R(e,v)]^*\R(e,v)u)] $$

$$ =  g(u,u) g(Q_u u, [\J(e,v)]^*\J(e,v) + \frac{1}{4} [ \R(e,v)]^*\R(e,v))]u) $$

$$= g(u,u) g(Q_u u, Q_u [\J(e,v)]^*\J(e,v) +  \frac{1}{4} [\R(e,v)]^*\R(e,v))] Q_u u)$$

$$=  g(u,u) trace (Q_u [\J(e,v)]^*\J(e,v) +  \frac{1}{4} [\R(e,v)]^*\R(e,v))]Q_u).$$
This results in the fact that

\begin{equation}\label{B0SIMP1}
\B_0(e,e,v,v) = \langle  P_u \J(e,v) P_u  | P_u \J(e,v) P_u  \rangle +  \langle  P_u [\R(e,v)]^* P_u  | P_u \R(e,v) P_u   \rangle
\end{equation}
$$+  g(u,u)[ trace (Q_u [\J(e,v)]^*\J(e,v)Q_u)] +  \frac{1}{4} g(u,u) [ trace(     Q_u   [ \R(e,v)]^*\R(e,v))]Q_u).$$
Now. 

$$Q_u [\J(e,v)]^*\J(e,v)Q_u = Q_u [\J(e,v)]^*[P_u + Q_u]\J(e,v)Q_u $$$$ = Q_u [\J(e,v)]^*P_u \J(e,v)Q_u + Q_u [\J(e,v)]^*Q_u\J(e,v)Q_u,$$

whose trace is

$$=\langle P_u \J(e,v)Q_u | P_u \J(e,v)Q_u \rangle  +  \langle  Q_u\J(e,v)Q_u |  Q_u\J(e,v)Q_u \rangle.$$
and similarly,

$$Q_u [\R(e,v)]^*\R(e,v)Q_u = Q_u [\R(e,v)]^*[P_u + Q_u]\R(e,v)Q_u $$$$ = Q_u [\R(e,v)]^*P_u \R(e,v)Q_u + Q_u [\R(e,v)]^*Q_u\R(e,v)Q_u,$$ 
whose trace is

$$=\langle P_u \R(e,v)Q_u | P_u \R(e,v)Q_u \rangle  +  \langle  Q_u\R(e,v)Q_u |  Q_u\R(e,v)Q_u \rangle.$$

This gives finally,

\begin{equation}\label{B0SIMP2}
\B_0(e,e,v,v) = \langle P_u \J(e,v)P_u | P_u \J(e,v)P_u \rangle  +  \langle  Q_u\J(e,v)Q_u |  Q_u\J(e,v)Q_u \rangle
\end{equation}
$$~~+ g(u,u) [\langle P_u \J(e,v)Q_u | P_u \J(e,v)Q_u \rangle  +  \langle  Q_u\J(e,v)Q_u |  Q_u\J(e,v)Q_u \rangle$$
$$~~+\frac{1}{4} g(u,u)[\langle P_u \R(e,v)Q_u | P_u \R(e,v)Q_u \rangle  +  \langle  Q_u\R(e,v)Q_u |  Q_u\R(e,v)Q_u \rangle].$$

As $$Q_u A Q_u$$ only has timelike values, for any $A \in L(V;V),$ it follows that both

$$g(u,u)  \langle  Q_u\J(e,v)Q_u |  Q_u\J(e,v)Q_u \rangle \geq 0 \mbox{ and } g(u,u)  \langle  Q_u\R(e,v)Q_u |  Q_u\R(e,v)Q_u \rangle]  \geq 0.  $$

Now supposing that $e$ and $v$ are vectors in an orthonormal basis of eigenvectors of $J(u,u),$ we have

$$ \langle  Q_u\J(e,v)Q_u |  Q_u\J(e,v)Q_u \rangle= g(u,u)g(Q_u\J(e,v)u,Q_u\J(e,v)u),$$
and
$$ \langle  Q_u\J(e,v)Q_u |  Q_u\J(e,v)Q_u \rangle= g(u,u)g(Q_u\J(e,v)u,Q_u\J(e,v)u),$$

but

$$g(Q_u\J(e,v)u,u)=g(\J(e,v)u,u)=g(\J(u,u)e,v)=0, \mbox{ if } e \neq v,$$
and

$$g(Q_u \R(e,v)u,u)=g(\R(e,v)u,u)=g(\R(u,u)e,v)=0.$$
This shows that in this situation of an orthonormal basis with $u$ timelike and eigenvectors of $\J(u,u),$ that the terms of the form $Q_uAQ_u$
 in the expression for $\B_0(e,e,v,v)$ will vanish unless $e=v$ in which case, only the $Q_u \J(e,e)Q_u$ terms survive.  At this point, the positivity of the Bel tensor seems very plausible, but we will leave this to later work to better explicate the relation between the Bel tensor and the Jacobi curvature operator.

\vfill
\break

%%%%%%%%%%%%%%%%%%%%%%%%%%%%%%%%%%%%%%%%%%%%%%%%%%%%%%%%%%%%%%%%%%%%%%
%%%%%%%%%%%%%%%%%%%%%%%%%%%%%%%%%%%%%%%%%%%%%%%%%%%%%%%%%%%%%%%%%%%%%%

%%%%%%%%%%%%%%%%%%%%%%%%%%%%%%%%%%%%%%%%%%%%%%%%%%%%%%%%%%%%%%%%%%%%%%
%%%%%%%%%%%%%%%%%%%%%%%%%%%%%%%%%%%%%%%%%%%%%%%%%%%%%%%%%%%%%%%%%%%%%%%

\section{APPENDIX I:  MATHEMATICAL PRELIMINARIES}

In this section we will establish our basic mathematical notation and framework and point out some mathematical results which may be unfamiliar to the reader.  Our treatment here covers material which is in many references, the most comprehensive being \cite{GREUB}, but our treatment is possibly simpler in some ways, since we avoid connections on principal bundles.  If $V_{\alpha},~\alpha \in A$ and $W$ are vector spaces, then we denote by $$\bigoplus_{\alpha \in A}V_{\alpha}$$ the direct sum of the family of vector spaces $(V_{\alpha})_{\alpha \in A}$ whereas, $$\Pi_{\alpha \in A}V_{\alpha}$$ denotes the cartesian product, just as for general sets.  These are actually the same if the index set $A$ is a finite set, that is, for any finite direct sum or cartesian product.  A direct sum of infinitely many vector spaces would be a vector subspace of the cartesian product consisting of those vectors in the cartesian product having only finitely many non-zero entries.  In case $A=\{1,2,...,m\},$ we denote by $$V_1 \otimes V_2 \otimes ... \otimes V_m$$ the {\it tensor product} of the vector spaces.  There is a multilinear map
$$\tau:V_1 \oplus V_2 \oplus ... \oplus V_m \lra V_1 \otimes V_2 \otimes ... \otimes V_m$$ with the {\it universal property} that for any vector space $W$ and any multilinear map
$$f:V_1 \oplus V_2 \oplus ... \oplus V_n \lra W$$ there is a unique linear map $$g:V_1 \otimes V_2 \otimes ... \otimes V_n \ra W$$ such that $$f=g \circ \tau.$$  If $v_k \in V_k,$ for each $k \in A,$ then it is customary to denote $$\tau(v_1,v_2,...,v_n)=v_1 \otimes v_2 \otimes ... \otimes v_n,$$ and call this vector the {\it tensor product} of the vectors $(v_k)_{k \in A}.$  All vectors in $V_1 \otimes V_2 \otimes ... \otimes V_n$ are called {\it tensors}, whereas those that are tensor products of families of vectors $(v_k)$ are called {\it elementary tensors}.  Thus, all the tensors in the tensor product are sums of elementary tensors.

We denote by $$L(V_1,V_2,...,V_n;W)$$ the vector space of all multilinear maps $V_1 \times V_2 \times ... \times V_n \lra W,$ and thus by the universal property of the tensor product, we have a unique natural vector space isomorphism
$$\Psi:L(V_1,V_2,...,V_n;W) \cong L(V_1 \otimes V_2 \otimes ... \otimes V_n;W).$$

If $V$ is any vector space, then we denote its {\it dual} by $V^*,$ so $$V^*=L(V;\bR).$$  There is a natural isomorphism $\Theta $ of $V$ onto a vector subspace of $V^{**}$ such that $$[\Theta (v)](f)=f(v),~ v \in V,~f \in V^*.$$  In case $V$ is finite dimensional, this natural isomorphism is an isomorphism of $V$ onto $V^{**}.$  In particular, for a family of finite dimensional vector spaces, this gives isomorphisms
$$V_1 \otimes V_2 \otimes ...\otimes V_n \cong (V_1 \otimes V_2 \otimes ...\otimes V_n)^{**} \cong (L(V_1,V_2,...,V_n;\bR))^*.$$  If all the vector spaces in the family are given non-degenerate inner products, then each vector space has a preferred isomorphism with its dual, so that the inner product determined on the vector space of real valued multilinear maps then allows one to identify tensors with multilinear maps, as is usually done in physics.  Thus, if $V$ is a given vector space with non-degenerate inner product, and if each $V_k$ in the family is either $V$ or $V^*,$ then the resulting tensor product is naturally isomorphic to the vector space of multilinear maps which are classically called tensors in physics.  In general, if  $V$ is a vector space and $V_k=V$ for each $k\leq m,$ then we define $L^m(V;W)=L(V_1,V_2,...,V_m;W).$  We denote by $L^m_{alt}(V;W)$ the vector subspace of alternating or anti-symmetric members of $L^m(V;W)$ and denote by $L^m_{sym}(V;W)$ the vector subspace of symmetric members.

For simplicity, let us call $X$ a {\it space} if $X \subset N,$ for some manifold $N.$  Allowing manifolds modeled on Banach spaces will mean that our simple definition includes all metrizable spaces.  We will generally deal only with finite dimensional manifolds, but much here is true for infinite dimensional Banach manifolds as well.    If $N_1$ and $N_2$ are manifolds, $X_k \subset N_k,~~k=1,2,$ and if $f:X_1 \lra X_2,$ is a map, we say $f$ is {\it smooth} provided that for each $x \in X_1$ there are open subsets $W_k \subset N_k,~~k=1,2$ and a genuine smooth map $g:\W_1 \lra W_2,$ so that $f$ and $g$ agree on the overlap $X_1 \cap W_1.$  Such a $g$ is called a local smooth extension of $f$ at $x \in X_1.$  If $X$ is a submanifold of $N,$ then a smooth map on $X$ as a space is easily seen to agree with the definition of smooth as an ordinary map defined on a manifold.  Moreover, by Lemma 2.4 of \cite{DUPREGLAZE1}, together with the tubular neighborhood theorem, it follows that $X \subset N$ is a smooth submanifold of the finite dimensional smooth manifold $N$ if and only if there is an open subset $U \subset N$ with $X \subset U$ and a smooth retraction $r: U \lra U$ such that $r(U)=X$ and here we remind the reader that to say $r$ is a retraction is to say $r \circ r=r.$  It then follows immediately, that $[Tr]$ is a retraction of $TU \subset TN$ onto $TX \subset TU,$ and so on, for all the higher tangent bundles. 

A vector space is a {\it topological vector space} provided that it is a space and a vector space and the operations of addition and scalar multiplication are continuous maps.  Any finite dimensional vector space is a topological vector space in a unique way, since it has a unique manifold structure for which the vector operations are continuous.  In particular, this is true of finite dimensional vector spaces which have semi-Riemannian metrics.

We say that $B$ is a {\it bundle} over $X$ provided that $B$ is also a space and there is a given smooth map $p=p_B:B \lra X,$ called the { \it projection} of the bundle.  Thus technically, the bundle consists of the triple $(p,B,X).$  We call $B$ the {\it total space} of the bundle and call $X$ the {\it base space} of the bundle.  If $B$ and $C$ are bundles over $X$ and $Y$ respectively, then by a {\it bundle map} $h:B \lra C$ {\it covering} $f:X \lra Y$ we mean a pair of maps $(h;f)$ so that $p_C \circ h=f \circ p_B.$  In case $X=Y,$ we say that $h$ is over $X$ to mean it covers the identity map, $id_X,$ of $X.$    If $B$ is a bundle over $X$ and $A \subset S,$ then we denote by $B|A$ the restriction of $B$ to $A$ whose total space is simply $p_B^{-1}(A),$ so $B|A$ is a bundle over $A.$ By a section $s$ of $B$ we mean a smooth map $s:X \lra B$ such that $p \circ s=id_X.$  We denote by $\Gamma(B)$ the set of all sections of $B.$  If $x \in X,$ then $B_x=p^{-1}(x)$ is the fiber over $x \in X.$  For any space $F,$ we denote by $\epsilon(X;F)$ the product bundle over $X$ with fiber $F$ whose total space is $X \times F$ and whose projection is simply the natural projection of the product onto its first factor.  We say that $B$ is trivial over $A \subset X$ provided that $B|A$ is isomorphic over $A$ to a product bundle.  We say that $B$ is locally trivial provided that its base can be covered by open subsets over which $B$ is trivial.  If all these trivializations have the same fiber $F,$ when we say that $B$ is locally trivial with fiber $F.$  The bundle isomorphisms of the trivial bundle $\epsilon(X;F)$ are easily seen to be determined by maps of $X$ into the diffeomorphism group of $F.$  When $G$ is a Lie group and all the overlap trivializations of $B$ are given by smooth maps into the Lie group $G,$ then we say the bundle has group $G.$  Thus, a locally trivial bundle with fiber $F$ and group $G$ is what is usually called a {\it fiber bundle} as defined by Steenrod \cite{STEENROD}.  A section $s$ of the trivial bundle with fiber $F$ over $X$ is simply a pair $(id_X,f)$ where $f:X \lra F$ is an ordinary smooth map, which is called the {\it principal part} of $s.$

If each fiber is a vector space and the vector operations in the fibers all together give smooth bundle maps over the base, then the bundle is called a vector bundle.  If all the fibers of the vector bundle $E$ are of the same finite dimension, then the bundle is locally trivial and is in fact a fiber bundle with group a general linear group. Putting a metric on such a vector bundle then amounts to reducing the group of the bundle to a Lie subgroup of the general linear group. All the above constructions for vector spaces and their natural maps and isomorphisms then extend to constructions on vector bundles over the space $X,$ with the obvious notation.  For instance, if $E_1,E_2,...,E_m$ are all vector bundles over $X,$ then $E_1\otimes E_2 \otimes ... \otimes E_m$ is the tensor product of these vector bundles over $X,$
$$(E_1\otimes E_2 \otimes ... \otimes E_m)_x=(E_1)_x\otimes (E_2)_x \otimes ...\otimes (E_m)_x,$$ and likewise
$$[L(E_1,E_2,...,E_m;E)]_x=L([E_1]_x,[E_2]_x,...,[E_m]_x;E_x),~~x \in X.$$  For $V$ a  topological vector space, we note that $\epsilon(N,V),$ the product bundle over $N$ with fiber $V,$ is a vector bundle. In case the base space is a manifold and each $E_k$ is either the tangent bundle of the base or its dual, sections of $L(E_1,E_2,...,E_m;E)$ are called $E-$valued tensor fields.  Thus the ordinary tensor fields of classical physics are simply $\epsilon(N,\bR)-$valued tensor fields.  Sections of $L^p_{alt}(TN;E)$ are called $E-$valued $p-$forms.  Of course the main example of a vector bundle is the tangent bundle of a manifold. It is useful to know that for any vector bundle with all finite dimensional fibers, say $E_1$ over $X,$ there is always another such vector bundle $E_2$ over $X$ such that $E_1 \oplus E_2$ is trivial, and thus all these vector bundles can be very profitably thought of as simply vector subbundles of trivial vector bundles.  In fact, the same is true with infinite dimensional fibers under suitable restrictions.  In the case of the tangent bundle to a manifold $N,$ when we embed it in a high dimensional Euclidean space, the normal bundle does the trick.  In particular, if $E \subset \epsilon(X,V),$ then we have for each $x \in X$ a vector subspace $E(x)$ of $V$ such that $E_x=\{x\}\times E(x).$ Putting a Euclidean metric on $V,$ we can form a bundle map $P$ of the trivial bundle to itself whose image is $E,$ and which is simply orthogonal projection of $V$ onto $E(x)$ for each $x \in X.$  Thus $P$ is determined by its principal part $P_{pr}:X \lra L(V;V)$ via $P(x,w)=(x,P_{pr}(x)v).$  We have then $P^2=P=P^*,$ since $P_{pr}(x)^2=P_{pr}(x)=P_{pr}(x)^*,$ for each $x \in X.$  We call $P$ so constructed the {\it orthogonal bundle projection} onto the subbundle $E$ of the trivial bundle.  The principal part of $P$ is simply a smooth map into a Grassmann manifold of self adjoint idempotent operators on the high dimensional Euclidean space or semi-Riemanniann vector space as the case may be, and iss called the {\bf classifying map} of the bundle.  It is easy to see that applying the various functors to the classifying maps gives the classifying map of the bundle resulting from applying a functor to construct new bundles, such as tensor product  bundles or dual bundles or bundles of multilinear maps and so forth.

For general references on differential geometry, semi-Riemannian
and Lorentz geometry, we refer to \cite{LANG}, \cite{GREUB}, \cite{KRIELE}, \cite{HAWKINGELLIS},
\cite{MTW}, \cite{WALD}, and \cite{ONEIL}. To begin, we say that $E$ is a connected vector bundle or vector bundle with connection provided that the base of $E$ is a smooth manifold $N$ and $E$ has a specified Kozul connection, $\nabla.$ This means that for each tangent vector field $v$ on the base and for each section $s \in \Gamma(E),$ we have the covariant derivative of $s$ along $v$ given by $\nabla_v s.$  By definition, this covariant differentiation operator for fixed $v$ is linear on sections and satisfies the Leibniz Rule:
$$\nabla_v (fs)=(D_vf)s+f\nabla_v s$$ for any scalar function $f$ defined on the base manifold $N.$  Enforcing the Leibniz rule over all evaluations and tensor products allows all the constructions for vector bundles to apply to connected vector bundles.  The trivial bundle with fiber $V$ over $N$ will always be assumed to have the obvious trivial connection where covariant differentiation simply amounts to ordinary differentiation of the principal part of a section.  To do this, restrict $V$ to be a Banach space, even though the theory of differential calculus can be applied to more general topological vector spaces, the theory becomes technical.  As any convex combination of Kozul connections is again a Kozul connection, it is easy to see using a smooth partition of unity and local trivailizations, that any vector bundle can be given some connection.  Alternately, we can simply embed $E$ in a trivial bundle $\epsilon$ with fiber $V$ a high dimensional Euclidean space. With the trivial connection on this trivial bundle denoted $D$ and with the vector bundle map $P:\epsilon \lra \epsilon$ taken to be orthogonal projection onto $E,$ in effect, a classifying map for the bundle, we simply define

$$\nabla_v s=PD_v s.$$  
We can note here, that for a Riemannian submanifold of a Euclidean space of any dimension, as the tangent bundle is a subbundle of the tangent bundle of the Euclidean space which is trivial, the previous projection procedure always gives the Levi-Civita connection on the manifold.  By the Nash Embedding Theorem, all Riemannian manifolds can be isometrically embedded in a Euclidean space, so Riemannian geometry can be thought of as the geometry of submanifolds of Euclidean space.  The same is also true for semi-Riemannian manifolds in general \cite{GREENE}.  This means that we can usefully think of any semi-Riemannian manifold as a submanifold of a finite dimensional semi-Riemannian vector space.  As far as producing Kozul connections on various functors of vector bundles is concerned, we simply use the Leibniz rule in all cases, which is easily seen to be consistent with the construction using classifying maps as noted above as to their use in constructing the various functors of vector bundles.  In particular, if $E_1,...E_n,F$ are all smooth vector bundles over $M,$ if $t$ is a section of $L(E_1,...E_n;F)$ all bundles having given connections all denoted 
$\nabla,$ then for any sections $s_1 \in \Gamma(E_1),...,s_n,$ and any smooth vector field $v$ in $\Gamma(TM),$ we have

$$\nabla_v [t(s_1,...,s_n)]=[\nabla_v t](s_1,...,s_n)+t(\nabla_v s_1, s_2,...,s_n)+...+t(s_1,...,s_{n-1},\nabla_v s_n),$$
by the Leibniz rule, which serves to determine $\nabla_v t.$  We then define the covariant derivative of $t$ denoted $\nabla t$ to e the section of $L(TM,E_1,...E_n;F)$ via

$$[\nabla t](v,s_1,s_2,...s_n)=[\nabla_v t](s_1,...s_n).$$
In particular, this shows that if $E$ and $F$ both have connections, then the connection on $E \otimes F$ is simply determined by the requirement that for any smooth section $s$ of $E$ and any smooth section $t$ of $F,$ and any smooth vector field $v$ of $M,$ we have

$$\nabla_v (s \otimes t)=(\nabla_v s) \otimes t + s \otimes (\nabla_v t).$$
We can note here that in general, $\nabla(s \otimes t)$ is not simply $[\nabla s] \otimes t + s \otimes [\nabla t],$ as in general, this last summation is undefined, the first term being a section of $L(TM;E) \otimes F$ and the second term being a section of $E \otimes L(TM;F).$  In order to fix this in cases of interest, one must introduce permutations acting on tensors.

We denote by $P(m)=Sym(m)$ the symmetric group of all permutations of $\{1,2,...,m\}$ and by $Cycl(m)$ the subgroup of order $m$ consisting of cycles.   For $t$ any tensor or tensor field in $L^m(E;F)$ and $\sigma \in Sym(m),$ we define
$$\sigma(t)(v_1,v_2,...,v_m)=t(v_{\sigma(1)},v_{\sigma(2)},...,v_{\sigma(m)}).$$  For $\tau \in Sym(m),$ we denote by $sgn(\tau)$ its parity as plus one for even and negative one for odd permutations.  Thus $sgn$ is a group homomorphism of $Sym(m)$ into the group $\{-1,1\}.$  If $m$ is odd and $\tau \in Cycl(m)$ then $sgn(\tau)=1.$  Define the symmetrization operator $Sym$ and $Alt,$ the alternation operator, by
$$Sym(t)=\frac{1}{m!}\sum_{\tau \in Sym(m)} \tau(t) \mbox{ and } Alt(t)=\frac{1}{m!}\sum_{\tau \in Sym(m)} sgn(\tau)\tau(t)$$ and if $m$ is odd, define $Cycl$ the cyclic permutation operator by
$$Cycle(t)=\sum_{\tau \in Cycl(m)} \tau(t).$$

If $A$ is a smooth section of $L^r(E;\bR) $ and $B$ is a smooth section of $L^s(E;\bR)$ and $\sigma_r=[1,2,3,....,r+1]$ is the cycle which increase each index except the last which is sent to 1, then we can now say that $sgn(\sigma_r)=r,$ and

$$\nabla[A \otimes B]=[\nabla A] \otimes B + \sigma( A \otimes [\nabla B]).$$
Consequently,

$$Alt(\nabla[A \otimes B])=Alt([\nabla A] \otimes B) +(-1)^r Alt(A \otimes [\nabla B]).$$

 Thus, for the ordinary tensor fields of classical physics, the covariant differentiation of tensor fields is that due to the Levi-Civita connection due to the metric tensor on tangent vector fields as it naturally extends to all tensor fields in a unique way preserving the Leibniz rule, and more generally, for $E-$valued tensor fields, the connection on the $E-$valued tensor bundle is defined using the Levi-Civita connection on the tangent bundle and its dual and the specified Kozul connection on $E.$  For the general vector bundle $E$ over $N$ with Kozul connection $\nabla,$ we define the curvature operator $\R_E$ as an $L(E;E)-$valued second rank tensor field given by
$$\R_E(u,v)=[\nabla_u,\nabla_v]-\nabla_{[u,v]}.$$  In general for $E-$valued alternating tensor fields on $M$ we can define the exterior covariant derivative $d_E$ as acting on $E-$valued $p-$forms on $M$ by the formula
$$d_E\omega=(p+1)(-1)^pAlt( \nabla \omega).$$  We then have the general Bianchi identity
$$d_{L(E;E)}\R_E=0.$$  We define the third rank $E-$valued tensor $R_E$ by
$$R_E(u,v,w)=\R(u,v)w$$ when $u$ and $v$ are tangent vector fields and $w$ is a differentiable section of $E.$  In case that $E=TN$ with the Levi-Civita connection from a metric tensor, then the Bianchi identity is called the second Bianchi identity as there is the first Bianchi identity which says that in fact
$$Cycle(R_{TN})=0.$$  Obviously, the first Bianchi identity cannot be formulated for Kozul connections on general vector bundles, as $Cycl(R_E)$ is nonsense unless $E=TN.$

For any third rank $E-$valued tensor field, $\omega,$ letting $\sigma$ denote the permutation which simply interchanges one and two and fixes three, we have
$$(3!)Alt(\omega)=Cycle(\omega)-Cycle(\sigma(\omega))$$ and therefore if $\omega$ is an alternating $E-$valued 2-form, then as $$\sigma(\nabla \omega)=-\nabla \omega,$$
$$(3!)Alt(\nabla \omega)=2Cycl(\nabla \omega)$$ and therefore
$$d_E\omega=Cycl(\nabla \omega).$$  Thus the Bianchi identity is equivalent to $$Cycl(\nabla \R_E)=0.$$  This is the form in which the second Bianchi identity is usually stated when $E=TM,$ which somewhat obscures the fact that it is really saying that the exterior covariant derivative of the curvature operator vanishes.  The vanishing of an exterior covariant derivative is a kind of generalized conservation law since in the case of ordinary forms, we can apply Stokes Theorem for their integrals.  The proof of the general Bianchi identity can be reduced to an application of the Jacobi identity of Lie algebra by working at a specific point of $M$ and choosing sections and vector fields which have vanishing covariant derivative at the specific point.  Notice that if $w$ and $z$ are sections of $E$ and $E^*,$ respectively, then for any $L(E;E)-$valued $p-$form  $\lambda$ we can form the ordinary $p-$form $z(\lambda(w)).$

Of course, the sections of $E$ are considered to be $E-$valued 0-forms on $N,$ and obviously if $s$ is a section of $E,$ then
$$d_Es=\nabla s.$$  For $\omega$ an $E-$valued 1-form, and $v_1,v_2$ any tangent vector fields on $N,$
$$[d_E\omega](v_1,v_2)=\nabla \omega (v_2,v_1)-\nabla \omega (v_1,v_2)=[\nabla_{v_1}\omega](v_2)-[\nabla_{v_2}\omega](v_1)$$
$$=\nabla_{v_1}[\omega (v_2)]-\nabla_{v_2}[\omega(v_1)]-\omega(\nabla_{v_1}v_2-\nabla_{v_2}v_1)$$
$$=\nabla_{v_1}[\omega (v_2)]-\nabla_{v_2}[\omega(v_1)]-\omega([v_1,v_2]),$$  where we used the fact that the Levi-Civita connection on $TN$ is torsion free in the last step.  Thus,
$$[d_E\omega](v_1,v_2)=\nabla_{v_1}[\omega (v_2)]-\nabla_{v_2}[\omega(v_1)]-\omega([v_1,v_2]),$$ for any $E-$valued 1-form $\omega$ on $N.$  More generally, we find that due to the antisymmetry of the exterior covariant derivative, when expanded in terms of covariant derivatives of tangent vector fields using the Levi-Civita connection on the tangent bundle, all the terms with Levi-Civita covariant derivatives combine in pairs as above so as to be expressible as Lie brackets.  This means the exterior covariant derivative is actually independent of the choice of semi-Riemannian metric on $N$ as it can be expressed purely in terms of covariant derivatives of sections of $E$ with the Kozul connection on $E$ and Lie brackets of tangent vector fields.  Moreover, if $s$ is any section of $E,$ then combining the two previous formulas, we have
$$d^2_E s=\R_E s.$$

Suppose $E$ and $F$ are both vector bundles with Kozul connection over $N$ and $H=E \otimes F$ has the unique Kozul connection described above.   If $\omega$ is an $E-$valued $p-$form and $\lambda$ is an $F-$valued $q-$form, then we can form the $H-$valued wedge product, denoted $\omega \wedge \lambda$ with the obvious generalization of the ordinary wedge product of ordinary forms, by taking $\omega \otimes \lambda$ and applying the alternation operator and using the same normalization factors as with the ordinary wedge product.  We have by the Leibniz rule
$$\nabla_v(\omega \otimes \lambda)=([\nabla_v \omega] \otimes \lambda)+(\omega \otimes [\nabla_v \lambda]),$$ and therefore
$$d_H(\omega \wedge \lambda)=d_E\omega \wedge \lambda +(-1)^p \omega \wedge d_F \lambda.$$  On the other hand, if $B$ is a section of $L(E;F),$ and instead we take $H=L(E;F)$ with the unique Kozul connection described above, then $B\omega$ is a section of $F$ and by the Leibniz rule
$$\nabla_v[B\omega]=[\nabla_vB]\omega+B[\nabla_v \omega].$$  Of course, $B$ is itself an $H-$valued 0-form, and we have here
$$d_F[B\omega]=[d_H B]\wedge \omega +B[d_E \omega].$$  If $E$ itself is a tensor product of connected vector bundles on $N,$ say $E=E_1 \otimes E_2,$ then $B$ amounts to a bilinear map, so applied to a wedge product gives a $F-$ valued wedge product depending on the choice of $B$ which we denote by $\wedge_B.$  For instance, if $\omega$ is an $L(E;F)-$valued $p-$form and $s$ is a section of $E,$ then evaluation of bundle maps on sections give a section $ev$ of $H=L(L(E;F)\otimes E;F)$ so that
$$ev(\omega \otimes s)=\omega(s).$$  By the Leibniz rule we find $d_H ev=\nabla ev=0,$ so
$$d_F[\omega s]=ev(d_T [\omega \otimes s])=ev([d_L \omega] \wedge s + \omega \wedge d_E s)=[d_L \omega]s+\omega \wedge_{ev} d_Es.$$  Thus, if $\lambda$ is a section of $F^*,$ then as $\lambda(\omega s)$ is an ordinary $p-$form we have
$$d[\lambda(\omega s)]=[d_{F^*}\lambda]\wedge [\omega s]+\lambda(d_F[\omega s])=[d_{F^*}\lambda] \wedge_{ev} [\omega s]
+\lambda([d_H \omega]s+(-1)^p \omega \wedge_{ev} d_Es)$$
$$=[d_{F^*}\lambda]\wedge_{ev} [\omega s] +\lambda([d_H \omega]s)+(-1)^p \lambda (\omega \wedge_{ev}[d_Es]).$$  In case $E=F$ and $\omega=\R_E,$ by the Bianchi identity, we have
$$d[\lambda(\R_E s)]=[d_{E^*}\lambda] \wedge_{ev} [\R_E s]
+(-1)^p \R_E \wedge_{ev} d_Es).$$  If $\omega$ is any $E-$valued $p-$form, then our previous formula $d^2_Es=\R_Es$ generalizes to give
$$d^2_E \omega=\R_E \wedge_{ev} \omega.$$ So actually, in the obvious sense, $$d_E^2=\R_E \wedge_{ev}.$$

\med

%%%%%%%%%%%%%%%%%%%%%%%%%%%%%%%%%%%%%%%%%%%%%%%%%%%%%%%%%%%%%%%%%%%%%%%%%%%%%%%%%%%%%%%%%%%%%%%%%%%%%%%%%%%%%
%%%%%%%%%%%%%%%%%%%%%%%%%%%%%%%%%%%%%%%%%%%%%%%%%%%%%%%%%%%%%%%%%%%%%%%%%%%%%%%%%%%%%%%%%%%%%%%%%%%%%%%%%%%%%%%
%%%%%%%%%%%%%%%%%%%%%%%%%%%%%%%%%%%%%%%%%%%%%%%%%%%%%%%%%%%%%%%%%%%%%%%%%%%%%%%%%%%%%%%%%%%%%%%%%%%%%%%%%%
%%%%%%%%%%%%%%%%%%%%%%%%%%%%%%%%%%%%%%%%%%%%%%%%%%%%%%%%%%%%%%%%%%%%%%%%%%%%%%%%%%%%%%%%%%%%%%%%%%%%

\end{document}